\documentclass[a4paper,11pt,fleqn,titlepage,booktabs]{report}

\usepackage{graphicx} 
\usepackage{color}

\usepackage{fancyhdr}

\usepackage[ansinew]{inputenc}
\usepackage{amsthm}
\usepackage{nicefrac}
\usepackage{amsmath}
\usepackage{amssymb}
\usepackage{amsthm}
\usepackage{amsfonts}

\usepackage{graphicx} 
\usepackage{tikz}
\usepackage{tikz-3dplot}
\usetikzlibrary{matrix}
\usepackage{color}
\usepackage{multicol} 
\usepackage{float} 
\usepackage{longtable}

\usepackage[english]{babel}

\usepackage{makeidx} 

\usepackage{hyperref} 

\newtheorem{Theorem}{Theorem}[section]
\newtheorem{Proposition}[Theorem]{Proposition}
\newtheorem{Lemma}[Theorem]{Lemma}
\newtheorem{Corollary}[Theorem]{Corollary}
\newtheorem{Conjecture}[Theorem]{Conjecture}
\newtheorem{Observation}[Theorem]{Observation}
\newtheorem{Definition}[Theorem]{Definition}
\newtheorem{Example}[Theorem]{Example}

\newtheorem{Remark}[Theorem]{Remark}
\newtheorem{Question}[Theorem]{Question}
\newtheorem{Concept}[Theorem]{Concept}

\newcommand{\bigslant}[2]{{\raisebox{.2em}{$#1$}\left/\raisebox{-.2em}{$#2$}\right.}}

\newcommand{\mf}[1]{\mathfrak{#1}}
\newcommand{\mb}[1]{\mathbb{#1}}

\newcommand{\al}{\mf{a}}
\newcommand{\be}{\mf{b}}
\newcommand{\ga}{\mf{c}}
\newcommand{\fal}{F_\al}
\newcommand{\fbe}{F_\be}
\newcommand{\fga}{F_\ga}
\newcommand{\gal}{{g_\al}}
\newcommand{\gbe}{{g_\be}}
\newcommand{\gga}{{g_\ga}}
\newcommand{\nal}{{\nu_\al}}
\newcommand{\nbe}{{\nu_\be}}
\newcommand{\nga}{{\nu_\ga}}
\newcommand{\dal}{d_\al}
\newcommand{\dbe}{d_\be}
\newcommand{\dga}{d_\ga}

\newcommand{\Gal}{\Gamma_\al}
\newcommand{\Gbe}{\Gamma_\be}
\newcommand{\Gga}{\Gamma_\ga}
\newcommand{\ii}{\II_i}
\newcommand{\ial}{\II_\al}
\newcommand{\ibe}{\II_\be}
\newcommand{\iga}{\II_\ga}

\newcommand{\dt}{|\theta_1|}
\newcommand{\dtt}{|\theta_2|}
\newcommand{\dti}{|\theta_i|}
\newcommand{\dri}{|\rho_i|}

\newcommand{\bd}[1]{\mathrm{Dic}_{#1}}
\newcommand{\bt}{\mathbb{T}^{\flat}}
\newcommand{\bo}{\mathbb{O}^{\flat}}
\newcommand{\by}{\mathbb{Y}^{\flat}}
\newcommand{\bti}{\mathbb{T}_{1}}
\newcommand{\btii}{\mathbb{T}_{2}}
\newcommand{\btiii}{\mathbb{T}_{3}}
\newcommand{\btiiii}{\mathbb{T}^{\flat}_{4}}
\newcommand{\btiiiii}{\mathbb{T}^{\flat}_{5}}
\newcommand{\btiiiiii}{\mathbb{T}^{\flat}_{6}}
\newcommand{\btiiiiiii}{\mathbb{T}_{7}}
\newcommand{\tnat}{\btiiii}
\newcommand{\boi}{\mathbb{O}_{1}}
\newcommand{\boii}{\mathbb{O}_{2}}
\newcommand{\boiii}{\mathbb{O}_{3}}
\newcommand{\boiiii}{\mathbb{O}^{\flat}_{4}}
\newcommand{\boiiiii}{\mathbb{O}^{\flat}_{5}}
\newcommand{\boiiiiii}{\mathbb{O}_{6}}
\newcommand{\boiiiiiii}{\mathbb{O}_{7}}
\newcommand{\boiiiiiiii}{\mathbb{O}^{\flat}_{8}}
\newcommand{\onat}{\boiiii}
\newcommand{\byi}{\mathbb{Y}_{1}}
\newcommand{\byii}{\mathbb{Y}^{\flat}_{2}}
\newcommand{\byiii}{\mathbb{Y}^{\flat}_{3}}
\newcommand{\byiiii}{\mathbb{Y}_{4}}
\newcommand{\byiiiii}{\mathbb{Y}_{5}}
\newcommand{\byiiiiii}{\mathbb{Y}_{6}}
\newcommand{\byiiiiiii}{\mathbb{Y}^{\flat}_{7}}
\newcommand{\byiiiiiiii}{\mathbb{Y}_{8}}
\newcommand{\byiiiiiiiii}{\mathbb{Y}^{\flat}_{9}}
\newcommand{\ynat}{\byii}

\newcommand{\mero}{\mathcal{M}(\overline{\CC})}

\newcommand{\splitk}{\CC}
\newcommand{\triv}{\epsilon}

\newcommand{\natrep}{U}

\newcommand\regrep[1]{\splitk {#1}}
\newcommand\regchar[1]{\chi_{\mathrm{reg}{#1}}}

\newcommand\palia[3][G]{\left(\mf{#2}({#3})\otimes R\right)^{#1}}
\newcommand\alia[4][G]{\left(\mf{#2}({#3})\otimes \mero_{#4}\right)^{#1}}
\newcommand\salia[4][G]{\overline{\mf{#2}({#3})}_{#4}^{#1}}
\newcommand\p[1]{\mathcal{P}_{#1}}
\newcommand\h[1]{\mathcal{H}_{#1}}

\newcommand{\AL}[2][\omega^2]{\mf{P}_{#1}(#2)}
\newcommand{\roots}{\Phi}
\newcommand{\sroots}{\Delta}

\newcommand{\Id}{\mathrm{Id}}
\newcommand{\id}{\mathrm{id}}
\newcommand{\tr}{\mathrm{tr}\,}
\newcommand{\im}{\mathrm{Im }} 
\newcommand{\GL}{\mathrm{GL}}
\newcommand{\SL}{\mathrm{SL}}
\newcommand{\SO}{\mathrm{SO}}
\newcommand{\NSO}{\mathrm{O}}

\newcommand{\PSL}{\mathrm{PSL}}
\newcommand{\PGL}{\mathrm{PGL}}
\newcommand{\Hom}{\mathrm{Hom}}
\newcommand{\End}{\mathrm{End}}
\newcommand{\Aut}{\mathrm{Aut}}
\newcommand{\Inn}{\mathrm{Inn}}
\newcommand{\Dyn}{\mathrm{Dyn}}
\newcommand{\Ad}{\mathrm{Ad}}
\newcommand{\ad}{\mathrm{ad}}
\newcommand{\Irr}{\mathrm{Irr}}

\newcommand{\lcm}{\mathrm{\,lcm}}
\newcommand{\codim}{\mathrm{codim}\,}
\newcommand{\diag}{\mathrm{diag}}
\newcommand{\rank}{\mathrm{rank}\,}

\newcommand{\cA}{\mathcal A}

\newcommand{\cG}{\mathcal G}

\newcommand{\cM}{\mathcal M}
\newcommand{\cN}{\mathcal N}
\newcommand{\cQ}{\mathcal Q}

\newcommand{\cV}{\mathcal V}

\newcommand\zn[1]{\bigslant{\ZZ}{#1}}

\newcommand{\NN}{\mathbb N}       
\newcommand{\ZZ}{\mathbb Z}       
\newcommand{\RR}{\mathbb R}
\newcommand{\QQ}{\mathbb Q}        
\newcommand{\CC}{\mathbb C}

\newcommand{\II}{\mathbb I}

\newcommand{\DD}{\mathbb D}        
\newcommand{\TT}{\mathbb T}  
\newcommand{\OO}{\mathbb O}  
\newcommand{\YY}{\mathbb Y}
\newcommand{\BB}{\mathbb B}  

\newcommand{\nodescale}{0.6}
\newcommand{\coloral}{red}
\newcommand{\colorbe}{green}
\newcommand{\colorga}{blue}

\makeindex
\makeglossary
\begin{document}
%
%
%

\pagestyle{empty}
\begin{center}

{\scshape\huge\mbox{Invariants of Automorphic Lie Algebras}} 
\vskip 6 cm
{\scshape\huge VJA Knibbeler}
\vskip 7.3 cm
{\scshape\huge  PhD} 
\vskip 2 cm
{\scshape\huge 2014}

\end{center}

%
%
%

\pagestyle{empty}
\begin{center}

\vskip 3 cm
{\huge\bf \mbox{Invariants of Automorphic Lie Algebras}} 
\vskip 3 cm
{\huge  Vincent Jan Ambrosius Knibbeler}
\vskip 3 cm
{\huge  A thesis submitted in partial fulfilment of the requirements of the University of Northumbria at Newcastle for the degree of Doctor of Philosophy}
\vskip 3 cm
{\huge  Research undertaken in the Department of Mathematics and Information Sciences}
\vskip 3 cm
{\huge September 2014}

\end{center}

\pagestyle{fancy}
\pagenumbering{roman}
\cleardoublepage\phantomsection
\addcontentsline{toc}{chapter}{Abstract}
\centerline{\textbf{\large{Abstract}}}

Automorphic Lie Algebras arise in the context of reduction groups introduced in the late 1970s \cite{Mikhailov81} in the field of integrable systems. 
They are subalgebras of Lie algebras over a ring of rational functions, defined by invariance under the action of a finite group, the reduction group.
Since their introduction in 2005 \cite{LM05comm, Lombardo}, mathematicians aimed to classify Automorphic Lie Algebras. 
Past work shows remarkable uniformity between the Lie algebras associated to different reduction groups.
That is, many Automorphic Lie Algebras with nonisomorphic reduction groups are isomorphic \cite{bury2010automorphic,LS10}.  
In this thesis we set out to find the origin of these observations by searching for properties that are independent of the reduction group, called invariants of Automorphic Lie Algebras.

The uniformity of Automorphic Lie Algebras with nonisomorphic reduction groups starts at the Riemann sphere containing the spectral parameter, restricting the finite groups to the polyhedral groups. 
Through the use of classical invariant theory and the properties of this class of groups it is shown that Automorphic Lie Algebras are freely generated modules over the polynomial ring in one variable. Moreover, the number of generators equals the dimension of the base Lie algebra, yielding an invariant.
This allows the definition of the 
determinant of invariant vectors which will turn out to be another invariant. A surprisingly simple formula is given expressing this determinant as a monomial in 
ground forms.


All invariants are used to set up a structure theory for Automorphic Lie Algebras. This naturally leads to a cohomology theory for root systems. A first exploration of this structure theory narrows down the search for Automorphic Lie Algebras significantly. Various particular cases are fully determined by their invariants, including most of the previously studied Automorphic Lie Algebras, thereby providing an explanation for their uniformity. 
In addition, the structure theory advances the classification project. For example, it clarifies the effect of a change in pole orbit resulting in various new 
Cartan-Weyl normal form generators for Automorphic Lie Algebras. 
From a more general perspective, the success of the structure theory and root system cohomology in absence of a field promises interesting theoretical developments for Lie algebras over a graded ring.


\cleardoublepage\phantomsection
\addcontentsline{toc}{chapter}{Contents}
\tableofcontents
\cleardoublepage\phantomsection
\addcontentsline{toc}{chapter}{List of Tables}
\listoftables
\cleardoublepage\phantomsection
\addcontentsline{toc}{chapter}{List of Figures}
\listoffigures
\cleardoublepage\phantomsection
\addcontentsline{toc}{chapter}{Acknowledgements}
\chapter*{Acknowledgements\markboth{Acknowledgements}{}}
\label{ch:Acknowledgements}

This research is the result of a three year PhD project at Northumbria University in Newcastle upon Tyne, United Kingdom. It has been a wonderful experience and I am indebted to many people for giving me this opportunity.

First of all I would like to thank Sara Lombardo, for inviting me to Newcastle and for introducing me to the academic world in general and to the fascinating subject of Automorphic Lie Algebras in particular. 
Each of the many hours that Sara and I worked together has been a pleasure. She organised several conferences during my research period at Northumbria, giving me the opportunity to participate and meet a wider scientific community, e.g.~NEEDS (Nonlinear Evolution Equations and Dynamical Systems) 2012, Integrable Systems In Newcastle 2013 and 2014, and the workshop ``An algebraic view of dynamics'' (to celebrate the 65th birthday of Jan Sanders).
The professional advice and guidance that Sara gave me will never lose significance and her character stays a source of inspiration. 

Jan Sanders' creativity seems to be limitless. I feel extremely privileged to have been a part of the research group on Automorphic Lie Algebras together with such an amazing scientist and I am very grateful for his generosity in spending time with me. Both Jan and Sara have made me feel welcome from the start, created a good research environment and played vital parts in the establishment of the results in this thesis.

Many thanks are due to all the members of the reading group on integrable systems at Northumbria University: Sara Lombardo, Matteo Sommacal, Benoit Huard, Richard Brown, Antonio Moro and James Atkinson. I am grateful for the effort everybody invested in order to make this reading group successful. 

My friends and fellow PhD students in lab F7, many thanks to them for the great times, the good conversations, the football matches, the badminton matches and the parties. They make F7 work.

Many thanks to my wonderful parents, family and friends in the Netherlands. It was always a treat to come back because of them. In particular thanks to Alwin ter Huurne for welcoming me into his home, giving me a place to stay and work, and of course for the many cappuccini.

To everyone involved in the MAGIC group, thank you for giving so many students across the UK the opportunity to see excellent lectures on a wide variety of mathematical subjects. I owe thanks in particular to the staff of Northumbria University responsible for enrolling the mathematics department into the MAGIC group.

I am grateful to the London Mathematical Society for organising the LMS-EPSRC short course series and waiving the fees for students.

This thesis would not have been written without the financial support of Northumbria University. 
I am also grateful to the NWO (The Netherlands Organisation for Scientific Research) for additinal support to visit the Vrije Universiteit Amsterdam (grant 639.031.622, PI Sara Lombardo), to the NEEDS network for a grant to participate in NEEDS 2012 in Crete, to the Graduate School for a travel grant to visit the University of Cagliari and to our Department for travel support to attend the LMS-EPSRC course on computational group theory in St Andrews. 

Last but not least I want to thank Sandra ter Huurne, for her companionship during this adventure, her continuous encouragement and support. She enriches the lives of everyone around her more than she could believe.
\cleardoublepage\phantomsection
\addcontentsline{toc}{chapter}{Declaration}
\centerline{\textbf{\large{Declaration}}}
\thispagestyle{plain}

I declare that the work contained in this thesis has not been submitted for any other award and that it is all my own work. I also confirm that this work fully acknowledges opinions, ideas and contributions from the work of others.
Any ethical clearance for the research presented in this thesis has been approved. Approval has been sought and granted by the Faculty Ethics Committee on $16$ December $2011$.


\vskip 1 cm
Name: Vincent Knibbeler
\vskip 2 cm
Signature:
\vskip 2 cm
Date:

\pagenumbering{arabic}


\chapter[Symmetric Symmetries]{Symmetric Symmetries}
\label{ch:intro}

Lie algebras are ubiquitous in physics and mathematics, often describing symmetries of a system of equations. They can be regarded as linear approximations of the more intricate \emph{Lie groups}: groups with a smooth manifold structure. The simpler Lie algebra, conventionally denoted by a lower-case letter in \emph{fraktur} such as $\mf{g}$, turns out to carry most of the information of the Lie group \cite{fulton1991representation}, which is one reason for its popularity.
We define Lie algebras over a ring conform Bourbaki \cite{bourbaki1998lie}.
\begin{Definition}[$R$-algebra]
If $R$ is a commutative ring then we call $\cA$ a $R$-algebra if it is a $R$-module equipped with a product $[\cdot,\cdot]:\cA\times\cA\rightarrow\cA$ which is $R$-bilinear, i.e.
\begin{align*}
&[ra+sb,c]=r[a,c]+s[b,c],\\
&[a,rb+sc]=r[a,b]+s[a,c],
\end{align*}
for all $a,b,c\in\cA$ and $r,s\in R$.
\end{Definition}
Notice that the product is not required to be associative.
\begin{Definition}[Lie algebra]
A Lie algebra $\mathfrak{g}$ is a $R$-algebra with antisymmetric product that satisfies Jacobi's identity, i.e.
\begin{align*}
&[a,a]=0,\\
&[a,[b,c]]+[b,[c,a]]+[c,[a,b]]=0,
\end{align*}
for all $a,b,c\in\mathfrak{g}$. The product is called the \emph{Lie bracket}.
\end{Definition}
More commonly a Lie algebra is defined to be a vector space (i.e.~$R$ is assumed to be a field).
As an example consider the space $\End(V)$ of linear endomorphisms on a vector space $V$. Together with the commutator bracket \[[A,B]=AB-BA\] this is a Lie algebra called the \emph{general linear} Lie algebra $\mf{gl}(V)$. A famous result known as \emph{Ado's theorem} \cite{fulton1991representation} shows that any complex Lie algebra is (isomorphic to) a Lie subalgebra of $\mf{gl}(V)$.

The \emph{special linear} Lie algebra is the subspace of traceless endomorphisms \[\mf{sl}(V)=\{A\in \mf{gl}(V)\;|\;\tr A=0\}\] with Lie bracket inherited from $\mf{gl}(V)$. Other examples are defined by a nondegenerate bilinear form $B$: \[\mf{g}_B(V)=\{A\in \mf{gl}(V)\;|\;A^TB+BA=0\}.\] If $B^T=B$ we speak of the \emph{orthogonal} Lie algebra $\mf{so}(V)$ and if $B^T=-B$ it is a \emph{symplectic} Lie algebra $\mf{sp}(V)$. These examples are the \emph{classical Lie algebras}.

Besides describing symmetries of a system of equations, a Lie algebra can be defined by a symmetry of its own. If for instance we have a homomorphism of groups $\rho:G\rightarrow \Aut(\mf{g})$, that is, $\rho(g)[a,b]=[\rho(g)a,\rho(g)b]$ for all $g\in G$ and all $a,b\in\mf{g}$, then one can define the Lie algebra \[\mf{g}^G=\{a\in \mf{g}\;|\;\rho(g)a=a,\,\forall g\in G\}.\] 
This space is closed under the Lie bracket due to the assumption that the group $G$ acts by Lie algebra morphisms. 
Automorphic Lie Algebras, to be defined in the next section, are examples of such Lie algebras.

\section{Motivation for Symmetry}

In various branches of mathematical physics, such as conformal field theories, integrable systems, and the areas founded upon these disciplines, it has been proven fruitful to include a complex parameter $\lambda$ in the Lie algebras. For instance, one can take a classical Lie algebra (as described above), call it $\mf{g}(V)$, and take the tensor product with the space of Laurent polynomials \[\mf{g}(V)\otimes \CC[\lambda,\lambda^{-1}]\] together with the Lie bracket defined by extending the bracket of $\mf{g}(V)$ linearly over $\CC[\lambda,\lambda^{-1}]$ \footnote{A convenient and commonly used notation for this Lie algebra is $\mf{g}[\lambda,\lambda^{-1}]$. However, in Automorphic Lie Algebra theory, the vector space $V$ is a group-module and plays a crucial role. Therefore we do not want to suppress it in the notation at this stage.}. As illustration, a typical element of $\mf{sl}(\CC^2)\otimes \CC[\lambda,\lambda^{-1}]$ can be represented by \[\begin{pmatrix}f&g\\h&-f\end{pmatrix}, \qquad f,g,h\in \CC[\lambda,\lambda^{-1}]\] and the bracket is simply the commutator. Equivalently, elements of this space can be regarded as Laurent polynomials with coefficients in the Lie algebra $\mf{g}(V)$ \[\sum_{i\in\ZZ} A_i \lambda^i,\qquad A_i\in\mf{g}(V).\] Such Lie algebras are called \emph{current algebras of Krichever-Novikov type} \cite{schlichenmaier2003higher,vasil2014harmonic} (or sometimes \emph{loop algebras} \cite{babelon2003introduction}). In fact, one can replace the Laurent polynomials by any ring of functions. In this thesis we consider the ring of rational functions with restricted poles
\[\mero_\Gamma=\{f:\overline{\CC}\rightarrow\overline{\CC}\;|\;f\text{ analytic outside }\Gamma\subset\overline{\CC}\}.\]
Notice that $\mero_{\{0,\infty\}}=\CC[\lambda,\lambda^{-1}]$.

Automorphic Lie Algebras were introduced in the context of the classification of integrable partial differential equations by Lombardo \cite{Lombardo} and Lombardo and Mikhailov \cite{LM04iop,LM05comm}. The Zakharov-Shabat / Ablowitz-Kaup-Newell-Segur scheme, 
used to integrate these equations, requires a pair of elements $X,T\in\mf{g}(V)\otimes\mero_\Gamma$ and the equations take the form of a zero curvature condition \begin{equation}
\label{eq:zero curvature}
[\partial_x-X,\partial_t-T]=0.
\end{equation}
Since a general pair of such $\lambda$-dependent matrices gives rise to an under determined system of differential equations (in the entries of the matrices), one requires additional constraints. By a well established scheme introduced by Mikhailov \cite{Mikhailov81}, and further developed in collaboration with Lombardo \cite{LM05comm}, this can be achieved by imposing a group symmetry on the matrices. The Lie subalgebras of $\mf{g}(V)\otimes\mero_\Gamma$ consisting of all such symmetric matrices are called \emph{Automorphic Lie Algebras}, in analogy with automorphic functions $\mero_\Gamma^G$. Since their introduction they have been extensively studied (see the work by Lombardo and Sanders \cite{LS10} and references therein, but also the thesis of Bury \cite{bury2010automorphic} and Chopp \cite{chopp2011lie}). 

\begin{Definition}[Automorphic Lie Algebra \cite{LM05comm,Lombardo}]
\label{def:alias1} Let $G$ be a finite subgroup of $\Aut(\overline{\CC})$ and $\Gamma\subset\overline{\CC}$ a $G$-orbit.
Consider a homomorphism $\psi:G\rightarrow \Aut(\mf{g}(V))$ where $\mf{g}(V)$ is a Lie algebra and define $\rho(g)=\psi(g)\otimes g^\ast\in\Aut\left(\mf{g}(V)\otimes\mero_\Gamma\right)$. 
The \emph{Automorphic Lie Algebra}\footnote{Later we will call this space of invariant matrices the \emph{natural representation of the Automorphic Lie Algebra}, and define the actual Automorphic Lie Algebra as the Lie algebra defined by this faithful representation, analogous to the classical case. A somewhat subtle distinction that will not be relevant before Chapter \ref{ch:B}.} is the space of $G$-invariant elements 
\[\alia{g}{V}{\Gamma}=
\left\{a\in\mf{g}(V)\otimes\mero_\Gamma\;|\;\rho(g)a=a, \;\forall g\in G\right\}\] of the current algebra. The \emph{base Lie algebra} of this Automorphic Lie Algebra is $\mf{g}(V)$.
\end{Definition}

Automorphic Lie Algebras are related to the well known Kac-Moody algebras \cite{kac1994infinite}: Lie algebras associated to a \emph{generalised Cartan matrix}. These matrices are classified and divided into three types: \emph{finite}, \emph{affine} and \emph{indefinite}. Generalised Cartan matrices of finite type correspond to the classical Cartan matrices \cite{humphreys1972introduction} and the associated Lie algebras are the simple complex Lie algebras (all nonexceptional cases were described above as the classical Lie algebras).
All Kac-Moody algebras associated to a generalised Cartan matrix of affine type (affine Kac-Moody algebras) can be realised using a Lie algebra of the form $\left(\mf{g}\otimes\CC[\lambda,\lambda^{-1}]\right)^{\zn{N}}$, where $\zn{N}$ acts (faithfully) on a simple complex Lie algebra $\mf{g}$ by a Dynkin diagram automorphism ($N\in\{1,2,3\}$) and on the Laurent polynomials by $\lambda\mapsto \omega\lambda$, $\omega^N=1$. An affine Kac-moody algebra is obtained when this Lie algebra is nontrivially extended by a one-dimensional centre, and adjoined by a derivation which kills the centre and acts as $\lambda\frac{d}{d\lambda}$ on the Laurent series \cite{kac1994infinite}. Notice that $\left(\mf{g}\otimes\CC[\lambda,\lambda^{-1}]\right)^{\zn{N}}$ is not an Automorphic Lie Algebra by Definition \ref{def:alias1} as it has poles on two orbits of $\zn{N}$: $\{0\}$ and $\{\infty\}$, yet it is closely related. In fact, in \cite{LM05comm} Automorphic Lie Algebras are alowed to have multiple pole orbits and an example with two pole orbits is described.

From a Lie algebraic perspective the current algebra $\mf{g}\otimes \mero_\Gamma$ is as interesting as its base Lie algebra $\mf{g}$ because the Lie brackets are the same. However, an extension of the current algebra, such as an affine Kac-Moody algebra, can have a richer structure and representation theory. Similarly, a Lie subalgebra of the infinite dimensional current algebra can have a very intricate Lie structure, and Automorphic Lie Algebras in particular promise to do so. This thesis is a first exploration of the algebraic structures of Automorphic Lie Algebras.

\section{The History of Automorphic Lie Algebras}

In the famous 1974-paper \cite{zakharov1974scheme} Zakharov and Shabat already notice that a general Lax pair does not always result in a (physically) meaningful integrable equation (here the existence of a Lax pair is taken as a definition of integrability) and some form of reduction is advised.
Mikhailov notices in 1979 \cite{mikhailov1979integrability} that a Lax pair of the Toda chain has a discrete group symmetry which is preserved by the flow. In the 1980-Letter \cite{mikhailov1980reduction} he defines the \emph{reduction group} (called $G$ in Definition \ref{def:alias1}) and explains how it can help in the classification of integrable systems through the restriction of Lax pairs. An example of a Lax pair with tetrahedral symmetry is computed.

About 25 years later Lombardo \cite{Lombardo} and Lombardo and Mikhailov \cite{LM04iop,LM05comm} introduce the Lie algebra of all matrices with rational dependence on the spectral parameter $\lambda$ that are invariant under the action of a finite reduction group $G$, and they name it the Automorphic Lie Algebra. The reduction group is discussed in more detail than it has been before. Extra attention is given to simplifications that can be achieved if there is a normal subgroup of $G$ that acts trivially on either the spectral parameter or on the matrices.

Using the Levi-decomposition for Lie algebras \cite{fulton1991representation} Lombardo and Mikhailov explain that it is the semisimple summand of the base Lie algebra that accounts for the nonlinear equations in the zero curvature condition (\ref{eq:zero curvature}). Moreover, these equations can be studied separately from the equations related to the radical of the base Lie algebra. If the first equations are solved, the latter follow easily. Therefore only semisimple base Lie algebras are considered.

Various Automorphic Lie Algebras are computed, e.g.~$\alia[\DD_N]{sl}{V}{\Gamma}$ where $\dim V$ is $2$ or $3$. This is done by averaging a selection of elements in the current algebra over the finite reduction group. Laurent expansions together with an induction argument on the order of the poles are used to show that the complete Lie algebra is accounted for by these averaged elements together with \emph{simple automorphic functions}. Both inner and outer automorphisms (cf.~Section \ref{sec:inner automorphisms}) are considered, as well as $\lambda$-dependent bases for the representations of $\DD_N$ (cf.~Section \ref{sec:representation theory}) for which the name \emph{twisted reduction group} is introduced.

In 2010 two important contributions to the subject appeared\footnote{Chronologically \cite{LS10} came before \cite{bury2010automorphic} but to discuss the development of the theory of Automorphic Lie Algebras it is more fitting to start with \cite{bury2010automorphic}.}. The second PhD-thesis \cite{bury2010automorphic} on the subject was finished, where Bury resumes the study of Automorphic Lie Algebras at the point where Lombardo \cite{Lombardo} left it. The methods used are similar but many more examples are calculated and explicit isomorphisms between different Automorphic Lie Algebras are found, significantly advancing the classification project. In particular, Bury finds that all $\mf{sl}_2(\CC)$-based Automorphic Lie Algebras whose poles lie in an exceptional orbit are isomorphic. Also preliminary results on the classification of $\mf{sl}_n(\CC)$-based Automorphic Lie Algebras with $n\ge 3$ are presented.
Bury puts the emphasis in \cite{bury2010automorphic} on integrable equations. Many of the computed Automorphic Lie Algebras are used to obtain systems of integrable partial differential equations which are subsequently studied.

The second milestone of 2010 is the paper \cite{LS10}. Lombardo and Sanders give the subject a new face by moving the workspace from the Riemann sphere to $\CC^2$. This puts the problem in the setting of \emph{classical invariant theory} (Section \ref{sec:classical invariant theory}). Tools from this field allow for more efficient ways to calculate invariants compared to averaging, e.g.~using \emph{transvectants}. Moreover, generating functions and \emph{Molien's theorem} provided a rigorous way to prove that the full Automorphic Lie Algebras are found. 

The paper constitutes the first occasion where Automorphic Lie Algebras are determined for a general reduction group. The authors show that all $\mf{sl}_2(\CC)$-based Automorphic Lie Algebras with poles restricted to the smallest orbit (in this thesis denoted by $\alia{sl}{V}{\Gal}$ with $\dim V=2$) are isomorphic. The proof contains the assumption that $V$ is the same representation as the one used to act on the spectral parameter, which enables one to find a single invariant for all reduction groups. This can only be done if $\dim V=2$, making this method for treating a general reduction group unsuitable for Automorphic Lie Algebras based on larger matrices.

A final important contribution of \cite{LS10} is the introduction of a \emph{Cartan-Weyl normal form} (Section \ref{sec:normal form}). Lombardo and Sanders find a set of generators for $\mf{sl}_2(\CC)$-based Automorphic Lie Algebras 
with similar properties as the Cartan-Weyl basis for simple complex Lie algebras.

Chopp provided a third PhD-thesis \cite{chopp2011lie} in 2011. Half of the thesis covers the subject of \emph{Witt-type algebras} and the other half concerns Automorphic Lie Algebras. The author considers arbitrary compact Riemann surfaces, with emphasis on the Riemann sphere and the torus. An implicit expression of a basis in terms of the averaging map for a set of generators for the Lie algebra is established. The proof relies on Laurent expansions similar to the work of Lombardo and Mikhailov. However, it includes a proof establishing independence of the proposed basis elements.
Apart from \cite{LS10}, the independence was not explicitly discussed before.

The most recent contribution, due to Lombardo, Sanders and the author \cite{knibbeler2014automorphic} continues in the style of \cite{LS10}, using classical invariant theory to compute Automorphic Lie Algebras with dihedral symmetry. Moreover, it contains a method to treat all possible pole orbits in one computation. The normal form is constructed for all pole orbits, including generic orbits. The contents of \cite{knibbeler2014automorphic} are included as a backbone of examples in this thesis, around which we build a theory.

We finish this brief historic account with a glimpse of the future: the paper \cite{knibbeler2014higher} that is in preparation during the time of writing.
Lombardo, Sanders and the author describe the Automorphic Lie Algebras with a reduction group isomorphic to the tetrahedral, octahedral or icosahedral group and base Lie algebra $\mf{sl}(V)$.
The explicit computations, given an irreducible group representation $V$,
are too complicated to do by hand. To overcome this obstacle, Sanders has written a FORM \cite{Form00} program. Calling on GAP \cite{GAP} and Singular \cite{MR2363237}, this program computes a generating set of invariant 
traceless matrices of degree $|G|$, where $G$ is the reduction group.  It then computes the corresponding \emph{matrices of invariants}, and such matrices find their first application in this paper. They are used to construct a Cartan subalgebra and compute the root vectors,
resulting in a \emph{Cartan-Weyl normal form} of the Automorphic Lie algebras. These computational results allowed for educated guesses that led to many of the general results of this thesis.
Together with the paper \cite{knibbeler2014automorphic} on dihedral symmetry, this completes the classification of the
$\mf{sl}_n(\CC)$-based Automorphic Lie Algebras with exceptional pole orbits.


\section{The Perspective of This Thesis}

The spectral parameter is a resident of the Riemann sphere. An assumption that the theory of Automorphic Lie Algebras has inherited from the theory of Lax connections in integrable systems. The finite group that defines the symmetry of the Automorphic Lie Algebra acts on the Riemann sphere. If this action is faithful then the finite group is 
a \emph{polyhedral group}. Their classification is a classical result \cite{dolgachev2009mckay,Klein56, Klein93,toth2002finite,toth2002glimpses}.

In this thesis we discover how the properties of this class of groups go all the way to the Lie algebra structure, severely limiting the possible structure constants and in some cases completely defining the Automorphic Lie Algebras. We aim to explain the uniformity in the Lie algebras for different reduction groups. This leads us to the following concept.
\begin{Concept}[Invariant of Automorphic Lie Algebras]
\label{conc:invariant of alias}
A property of an Automorphic Lie Algebra $\alia{g}{V}{\Gamma}$ is called an \emph{invariant of Automorphic Lie Algebras} if it depends solely on the type of orbit $\Gamma$ and the base Lie algebra $\mf{g}(V)$ up to Lie algebra isomorphism.
\end{Concept}
In this concept the \emph{type of orbit} is either \emph{generic}, $\al$, $\be$ or $\ga$. An orbit of size $|\Gamma|=G$ is generic and the remaining orbits, of which there are two if $G$ is cyclic and three otherwise, are of type $\al$, $\be$ or $\ga$, ordered by nondecreasing size. Nongeneric orbits are called exceptional.


The question whether the observed uniformity of Automorphic Lie Algebras over different reduction groups and their representations 
reaches as far as it possibly could,
can now be formulated as follows.

\begin{Question}[Isomorphism question]
\label{q:isomorphism question}
Is the Lie algebra structure an invariant of Automorphic Lie Algebras?
\end{Question}
By the Lie algebra structure we mean all the information that is encoded in the Lie bracket or equivalently the \emph{structure constants} \cite{fulton1991representation, humphreys1972introduction}.

The isomorphism question is an important motivation for the classification project and a distant goal of the research in this thesis. Unofficially there have been various conjectures similar to the affirmation of the isomorphism question circulating among the experts.

Our approach is to pass from the Riemann sphere to $\CC^2$, taking polyhedral groups to binary polyhedral groups, and obtain access to classical invariant theory, in line with the work by Lombardo and Sanders \cite{LS10}. Among the many powerful results is Hochster and Aegon's theorem (cf.~Section \ref{sec:classical invariant theory}) which shows that the spaces of invariants are finitely generated free modules over a polynomial ring (Cohen-Macaulay). This means for instance that Automorphic Lie Algebras are defined by finitely many structure constants (and in particular quasigraded), making it more natural to study these infinite dimensional Lie algebras in the framework developed for finite dimensional semisimple Lie algebras.

The specific class of finite groups allows us to do a lot better than this.
Using Clebsch-Gordan decompositions for $\SL_2(\CC)$-modules we show that the number of generators 
of the $\mero_\Gamma^G$-module $(V\otimes\mero_{\Gamma})^G$ equals the dimension of the base vector space $V$.
In particular this holds for Automorphic Lie Algebras, where $V$ is replaced by a Lie algebra $\mf{g}(V)$, and we arrive at a first invariant.

The particular number of generating invariant vectors permits the definition of the \emph{determinant of invariant vectors}. It takes considerable effort to determine this determinant by a direct computation, which adds to the value of the remarkably simple formula that is obtained, expressing this determinant as a monomial in \emph{ground forms}. By taking a simple Lie algebra as base vector space, this formula allows us to see that the determinant of invariant vectors is an invariant of Automorphic Lie Algebras as well.

An Automorphic Lie Algebra $\alia{g}{V}{\Gamma}$ is also a family of Lie subalgebras of $\mf{g}(V)$ parametrised by the Riemann sphere. Indeed, for all $\mu\in\CC$ one can evaluate the space of all invariant matrices to obtain a finite dimensional complex Lie algebra $\alia{g}{V}{\Gamma}(\mu)$. We provide a full classification of these families and find that it is another invariant. The Lie algebra structure of $\alia{g}{V}{\Gamma}(\mu)$ depends in fact only on the orbit type of $G\mu\subset\overline{\CC}$ and the base Lie algebra $\mf{g}(V)$ up to isomorphism.

The convenient number of generators of an Automorphic Lie Algebra begs the question whether it shares more properties with its base Lie algebra. To pursue this idea we define a \emph{Cartan-Weyl normal form} for Automorphic Lie Algebras, as similar to the classical Cartan-Weyl basis as possible, following \cite{knibbeler2014higher, knibbeler2014automorphic, LS10}. Particularly in \cite{knibbeler2014higher}, at the frontier of the Automorphic Lie Algebra classification project, many Cartan-Weyl normal forms are explicitly computed and represented in terms of \emph{matrices of invariants}. 
The aforementioned determinant of invariant vectors yields the precise number and type of automorphic functions appearing in this representation. The formula for this determinant thus provides a way to predict crucial information about the matrices of invariants, circumventing a tremendous amount of computations.

All the invariants of Automorphic Lie Algebras are combined to exploit their predictive power, which can be done conveniently using the framework of root systems. 
However, root systems alone do not carry enough information to reconstruct Lie algebras over rings, due to the absence of multiplicative inverses of the structure constants.
Adding the missing information naturally leads to a cohomology theory for root systems. The invariants can then be expressed in cohomological terms, and this way we get a handle on the Lie algebra structure. 

The continuation of the computational classification project (cf.~\cite{knibbeler2014higher}) can be checked against these invariants, just as the invariants were originally checked against the available computational results during their development. 
Moreover, the invariants are so restrictive that for most base Lie algebras and pole orbit types there are just a hand full of candidates for the Automorphic Lie Algebra, and in some cases there is indeed just a single candidate, thereby reproducing the results of various papers, e.g.~\cite{knibbeler2014automorphic,LS10}.




This work has been organised as follows:

The current introductory chapter is followed by Chapter \ref{ch:preliminaries}, Preliminaries, were a summary is given of selected background material. Chapter \ref{ch:A} discusses general features of representations of binary polyhedral groups (that is, finite subgroups of $\SL_2(\CC)$). Polynomial invariants are discussed in Section \ref{sec:invariant vectors} and these are used to devise a method to study invariants over meromorphic functions with all type of pole restrictions by computing one particular space of invariants, in Section \ref{sec:homogenisation}. This method is described in terms of two operators: \emph{prehomogenisation} and \emph{homogenisation}.
The last two sections lay the foundation for the main results of this thesis; the invariants of Automorphic Lie Algebras. In Section \ref{sec:squaring the ring} we show that Automorphic Lie Algebras are free modules over a polynomial ring and in Section \ref{sec:determinant of invariant vectors} the \emph{determinant of invariant vectors} is described.

In Chapter \ref{ch:B} we explore the outline of the classification project, discussing various possibilities and impossibilities to generalise the definition of an Automorphic Lie Algebra.
We consider classical Lie algebras as modules of the polyhedral groups and investigate how the Lie structure and module structure interact. 
Section \ref{sec:group decomposition of simple lie algebras} shows how to decompose complex Lie algebras into irreducible representations of the reduction group. In Section \ref{sec:inner automorphisms} we discuss inner and outer automorphisms of the base Lie algebras and explain how the reduction group can be represented in this context. Finally, in Section \ref{sec:evaluating alias} we classify the evaluations of the \emph{natural representation of Automorphic Lie Algebras}.


Chapter \ref{ch:C} combines the polynomial invariants from Chapter \ref{ch:A} with the Lie algebras from Chapter \ref{ch:B} and the first Automorphic Lie Algebras are discussed. 
From \emph{Polynomial Automorphic Lie Algebras} in Section \ref{sec:palias} we move to Automorphic Lie Algebras in Section \ref{sec:alias}. There we define the generalised Cartan-Weyl normal from and give the explicit computation and results for all dihedral Automorphic Lie Algebras with a general orbit of poles. Furthermore we conjecture the existence of the Cartan-Weyl normal form for all Automorphic Lie Algebras. 
Finally we develop a structure theory for Automorphic Lie Algebras in Section \ref{sec:structure theory for alias} using a new cohomology theory on root systems. The implications are studied for all root system that are involved in the isomorphism question. 


\chapter[Preliminaries]{Preliminaries}\label{ch:preliminaries}

In this chapter we discuss some notions of well established fields such as representation theory for finite groups, classical invariant theory and a touch of cohomology theory for groups. We expect the reader to be comfortable with linear algebra, including direct sums and tensor products.
Only basic knowledge of group theory including representation theory of finite groups over the complex field is needed.

The goal of this chapter is to remind the reader of results that will be of importance for the development of a theory for Automorphic Lie Algebras. Each section contains references that provide more detailed expositions and proofs.
Apart from Proposition \ref{prop:stanley} it does not contain original results and the reader can skip any section that they are familiar with, or just skim through in order to see the chosen notations.
In the main chapters, Chapter \ref{ch:A}, \ref{ch:B} and \ref{ch:C}, we refer back to the relevant preliminary section when appropriate.


\section{Polyhedral Groups}
\label{sec:polyhedral groups}
In the context of Automorphic Lie Algebras we are interested in groups of automorphisms of a Riemann surface $\overline{\CC}$, which historically was the domain of the spectral parameter of inverse scattering theory. For now we will consider finite groups of holomorphic bijections $\overline{\CC}\rightarrow\overline{\CC}$, i.e.~automorphisms of the Riemann sphere. These maps are also known as \emph{fractional linear transformations} or \emph{M\"obius transformations}. 

We first realise that $\Aut(\overline{\CC})\cong\PSL_2(\CC)\cong \SO(3)$, and use the latter perspective, rotations of a sphere in $\RR^3$, which has the advantage that we can have a mental picture of the group action.

Following e.g.~\cite{dolgachev2009mckay,etinghof,Klein56,Klein93,toth2002finite,toth2002glimpses}, we sketch the classification of finite subgroups $G$ of rotations of the 2-sphere.
Each nontrivial group element fixes a pair of antipodal points on the sphere. 
Therefore one can express the number of nontrivial group elements $|G|-1$ (where $|G|$ is the order of the group) in terms of stabiliser subgroups \[G_\lambda=\{g\in G\;|\;g\lambda=\lambda\},\qquad \lambda\in\overline{\CC},\] as follows. 
Because the group is finite, there can only be a finite number of points on the sphere whose stabiliser subgroup is nontrivial. Therefore the sum $\sum_{\lambda\in\overline{\CC}}\left(|G_\lambda|-1\right)$ has finitely many nonzero terms and is well defined. Moreover, since each nontrivial group element fixes two points on the sphere it is counted twice in this sum, thus we have a formula
\begin{equation}\label{eq:nontrivial group elements}2\left(|G|-1\right)=\sum_{\lambda\in\overline{\CC}}\left(|G_\lambda|-1\right).\end{equation}

Let $\Omega$ be an index set for the orbits $\Gamma_i$, $i\in \Omega$ on the sphere of elements with nontrivial stabiliser groups, \emph{exceptional orbits} hereafter. Moreover, let $d_i=|\Gamma_i|$ be the size of such an orbit and $\nu_i$ the order of the stabiliser subgroups at points in $\Gamma_i$. In particular
\[d_i\nu_i=|G|,\qquad i\in\Omega.\]
The sum in formula (\ref{eq:nontrivial group elements}) can be restricted to all $\lambda$ with nontrivial stabiliser group $\bigcup_{i\in\Omega}\Gamma_i$.
The formula becomes $2(|G|-1)=\sum_{i\in\Omega}\sum_{\lambda\in\Gamma_i}(\nu_i-1)=\sum_{i\in\Omega}d_i(\nu_i-1)$ or
\begin{equation}
\label{eq:finite subgroups of SO(3)}
2\left(1-\frac{1}{|G|}\right)=\sum_{i\in\Omega}\left(1-\frac{1}{\nu_i}\right).
\end{equation}
This equation can also be deduced from the Riemann-Hurwitz formula \cite{toth2002glimpses}. In terms of $d_i$ it reads
\[\sum_{i\in\Omega}d_i=(|\Omega|-2)|G|+2.\]
This equation is very restrictive, since we require all variables to be natural numbers and $\nu_i$ to divide $|G|$. It follows for instance that there are either $2$ or $3$ exceptional orbits. If $|\Omega|=2$ one finds the cyclic groups and if $|\Omega|=3$ one obtains the symmetry groups of the Platonic solids and regular polygons embedded in $\RR^3$.
\begin{center}
\begin{table}[h!] 
\caption{The orders of the polyhedral groups.}
\label{tab:orders of polyhedral groups}
\begin{center}
\begin{tabular}{ccc} \hline
$G$&$(\nu_i\;|\;i\in\Omega)$&$|G|$\\
\hline
$\zn{N}$&$(N, N)$&$N$\\
$\DD_{N}$&$(N,2,2)$&$2N$\\
$\TT$&$(3,3,2)$&$12$\\
$\OO$&$(4,3,2)$&$24$\\
$\YY$&$(5,3,2)$&$60$\\
\hline 
\end{tabular}
\end{center}
\end{table}
\end{center}
In this thesis, most attention will go to the non-cyclic groups, the groups with three exceptional orbits. For convenience we put in this situation 
\[\Omega=\{\al,\be,\ga\}\]
and $\nal\ge\nbe\ge\nga$.
Notice the Euler characteristic for the sphere
\begin{equation}
\label{eq:Euler characteristic}
\dal+\dbe-\dga=2.
\end{equation}

We will adopt the term \emph{polyhedral groups}, but these groups are also known as \emph{spherical von Dyck groups} $D(\nal, \nbe, \nga)$, which can be defined as the subgroups of words of even length (orientation preserving as isometries of the sphere) in the generators of the \emph{spherical triangle groups} $\Delta(\nal, \nbe, \nga)$.
The list contains alternating and symmetric groups: $\TT=A_4$, $\OO=S_4$ and $\YY=A_5$.

Polyhedral groups allow a presentation of the form 
\[G=\langle r,s\;|\;r^{\nal}=(rs)^{\nbe}=s^{\nga}=1\rangle.\]
The cyclic groups included, with $(\nal, \nbe, \nga)=(N,N,1)$, even though $s$ is redundant in that case. This presentation was favoured in previous works on Automorphic Lie Algebras, \cite{bury2010automorphic,chopp2011lie,LM04iop,LM05comm,LS10,Lombardo}, although the reader must be wary of different conventions: the role of $r$ and $s$ might be the other way around.

An alternative presentation, more in the style of triangle groups, is given by
\[G=\langle \gal, \gbe, \gga\;|\; \gal^{\nal}=\gbe^{\nbe}=\gga^{\nga}=\gal \gbe \gga=1\rangle.\]
It is related to the previous presentation through the isomorphism 
\begin{align*}
\gal\leftrightarrow r, \quad
\gbe\leftrightarrow (sr)^{-1},\quad
\gga\leftrightarrow s.
\end{align*}
Notice that any two of the three generators $\{ \gal, \gbe, \gga\}$ will generate the whole group, since the third can be constructed thanks to the relation $\gal \gbe \gga=1$.
This latter presentation turns out to be very convenient in the development of a theory for Automorphic Lie Algebras, and we will stick to it.

For future reference we give the \emph{abelianisation} and the \emph{exponent} of the polyhedral groups and summarise various properties of the groups in Table \ref{tab:various properties of polyhedral groups}. The definition of the \emph{Schur multiplier} $M(G)$ shown in Table \ref{tab:various properties of polyhedral groups} will be postponed to Section \ref{sec:schur covers}.

\begin{Definition}[abelianisation]
\label{def:abelianistaion}
The \emph{abelianisation} $\cA G$ of a group $G$ is the quotient group 
\[\cA G = \bigslant{G}{[G,G]}\]
where $[G,G]=\langle g^{-1}h^{-1}gh\;|\;g,h\in G\rangle$ denotes the commutator subgroup.
It is the largest abelian quotient group of $G$.
\end{Definition}
One can find the abelianisations of the polyhedral groups for instance by considering homomorhpisms into an abelian group, e.g.~$G\rightarrow \splitk^\ast$.

\begin{Definition}[Exponent of a group]
\label{def:exponent of a group}
The least common multiple of the orders of elements of a group $G$ is called the \emph{exponent} of $G$ and denoted \(\|G\|\), 
\[\|G\|=\min\{n\in\NN\;|\;g^n=1,\,\forall g\in G\}.\]
The exponent divides the order of a finite group.
\end{Definition}


As an example, the exponent of a cyclic group equals the group order, since there is a group element of that order. For general polyhedral groups we notice that each group element is contained in a subgroup $\zn{\nu_i}$ for some $i\in\Omega$. Hence the exponent is the least common multiple of $\{\nu_i\;|\;i\in\Omega\}$.

\begin{center}
\begin{table}[h!] 
\caption{The polyhedral groups, orders, exponent, Schur multiplier and abelianisation.}
\label{tab:various properties of polyhedral groups}
\begin{center}
\begin{tabular}{lllllllllllllllllllllllll} \hline
$G$&$|\Omega|$&$(\nu_i\;|\;i\in\Omega)$&$(d_i\;|\;i\in\Omega)$&$|G|$&$\|G\|$&$M(G)$&$\cA G$\\
\hline
$\zn{N}$&$2$&$(N,N)$&$(1,1)$&$N$&$N$&$1$&$\zn{N}$\\
$\DD_{N=2M-1}$&$3$&$(N,2,2)$&$(2,N,N)$&$2N$&$2N$&$1$&$\zn{2}$\\
$\DD_{N=2M}$&$3$&$(N,2,2)$&$(2,N,N)$&$2N$&$N$&$\zn{2}$&$\zn{2}\times \zn{2}$\\
$\TT$&$3$&$(3,3,2)$&$(4,4,6)$&$12$&$6$&$\zn{2}$&$\zn{3}$\\
$\OO$&$3$&$(4,3,2)$&$(6,8,12)$&$24$&$12$&$\zn{2}$&$\zn{2}$\\
$\YY$&$3$&$(5,3,2)$&$(12,20,30)$&$60$&$30$&$\zn{2}$&$1$\\
\hline 
\end{tabular}
\end{center}
\end{table}
\end{center}


\section{Representations of Finite Groups}
\label{sec:representation theory}

There are many excellent texts on representation theory for finite groups, e.g.~\cite{etinghof,fulton1991representation,serre1977linear}. The treatment of this subject depends rather strongly on the underlying field of the vector spaces. For our purposes it is sufficient to consider $\CC$ and $\RR$. We present a brief recap of some important results that will be used in the sequel.


\subsection{Group-Modules}

A \emph{representation} of a group $G$ is a homomorphism \[\rho:G\rightarrow\GL(V),\] where $\GL(V)$ is the \emph{general linear} group, that is, the group of invertible linear maps $V\rightarrow V$.
The group $G$ is represented by such linear transformations. It is common practice to call the vector space $V$ a representation as well, even though it is perhaps better to say that $V$ is a $G$-module. In that case we often omit the map $\rho$ in the notation for the action on $V$. Some common notations for this action are $v\mapsto\rho(g)v=\rho_g v=g\cdot v=gv$, where  $g\in G$ and $v\in V$.

A subrepresentation of $\rho:G\rightarrow\GL(V)$ is a subspace $U<V$ preserved by $G$, meaning $\rho(g)U\subset U$ for all $g$ in $G$.
A representation is called \emph{irreducible} (or simple) if it has no proper subrepresentation. Otherwise the representation is called \emph{reducible}. A representation is called \emph{completely reducible} (or semisimple) if it is a direct sum of irreducible representations.
\begin{Theorem}[Maschke's Theorem]
A representation $\rho:G\rightarrow\GL(V)$ of a finite group $G$ is completely reducible if and only if the characteristic of the field of $V$ does not divide the order of the group.
\end{Theorem}
In this thesis we only deal with fields of characteristic zero, hence all representations of finite groups are semisimple.

A representation $\rho:G\rightarrow\GL(V)$ is \emph{faithful} if it is a monomorphism, i.e.~if $\rho(G)\cong G$. Faithfulness and irreducibility are independent properties.


Given two $G$-modules $U$ and $V$, we define the set of \emph{$G$-linear} maps $U\rightarrow V$ by
\[\Hom_{G}(U,V)=\{f\in\Hom(U,V)\;|\;fg=gf,\;\forall g\in G\}\]
and $\End_{G}(V)=\Hom_{G}(V,V)$. Two $G$-modules $U$ and $V$ are isomorphic if and only if $\Hom_{G}(U,V)$ contains an invertible element. This is an isomorphism.

\begin{Lemma}[Schur's Lemma]
\label{lem:schur}
If $U$ and $V$ are simple $G$-modules, then $\Hom_{G}(U,V)$ is a division ring, i.e.~then any nonzero $G$-linear map $U\rightarrow V$ is invertible.
\end{Lemma}
A finite dimensional division ring $D$ over an algebraically closed field $k$ is isomorphic to the field, $D=k 1_D$. In particular, by Schur's Lemma, all $G$-linear maps in $\End(V)$ over the complex numbers are scalars \[\End_{G}(V)=\CC\Id\] whenever $V$ is an irreducible representation.
This fact can be used to show that any complex representation $V$ has a \emph{unique} decomposition
\[V=U_1\oplus\cdots\oplus U_n,\] where $U_i$ is in turn a direct sum of irreducible representations which are pairwise isomorphic, but irreducible components from different summands, $U_i$ and $U_j$, are not isomorphic \cite{fulton1991representation,serre1977linear}. Such a summand is called an \emph{isotypical component} of $V$. 

Representation theory over the real numbers leads to division rings over the real numbers through Schur's Lemma. A real division ring is isomorphic to either the real numbers, the complex numbers or the quaternions. At this stage already the real theory is more complicated than the complex theory. The study of invariant bilinear forms is one way to get a handle on real representations. Before turning to this subject we introduce the \emph{character} of a representation.


\subsection{Character Theory}
\label{sec:character theory}
Representations are characterised by their trace. Fittingly the trace of a representation $\rho:G\rightarrow\GL(V)$ is called the \emph{character}. It is a map $\chi:G\rightarrow k$ defined by
\[\chi(g)=\tr\rho(g),\]
where $k$ is the field of the vector space $V$. 
Two representations are isomorphic if and only if they have the same character \cite{fulton1991representation, serre1977linear}.
We say that $\rho$ is a \emph{representation affording $\chi$} or $V$ is a \emph{module affording $\chi$}. It is often convenient to write $\chi_V$ for the character of $V$, or to write $V_\chi$ or $\rho_\chi$ for the representation affording $\chi$.  A character is called irreducible if the related representation is irreducible.
We notice that
\begin{align*}
\chi_V(1)&=\dim V,\\
\chi(g^{-1})&=\overline{\chi(g)},\\
\chi(hgh^{-1})&=\chi(g),
\end{align*}
where the bar denotes complex conjugation. 
For the middle equality we assume that the field of $V$ is a subfield of $\CC$. The equality then follows from the fact that $g$ has finite order.
The last equality shows that a character is a \emph{class function}, that is, it is constant on the conjugacy classes of the group.

A group action on two vector spaces $U$ and $V$ induces an action on their direct sum and tensor product by $g(u\oplus v)=gu\oplus gv$ and $g(u\otimes v)=gu\otimes gv$. The symmetric and alternating square of a tensor product are subrepresentations: $V\otimes V=S^2V\oplus \wedge^2 V$. Their characters are related by
\begin{align*}
\chi_{U\oplus V}&=\chi_U+\chi_V,\\
\chi_{U\otimes V}&=\chi_U\chi_V,\\
\chi_{S^2V}(g)&=\nicefrac{1}{2}\left(\chi_V(g)^2+\chi_V(g^2)\right),\\
\chi_{\wedge^2 V}(g)&=\nicefrac{1}{2}\left(\chi_V(g)^2-\chi_V(g^2)\right).
\end{align*}

On the space of functions $G\rightarrow \CC$ we define a Hermitian inner product
\begin{equation}
\label{eq:inner product}
(\phi,\psi)=\frac{1}{|G|}\sum_{g\in G}\phi(g)\overline{\psi(g)}.
\end{equation}
The power of character theory is largely due to the orthogonality relations.  Let \[\Irr(G)=\{\text{irreducible characters of }G\}.\]
\begin{Theorem}[First orthogonality relation]
\label{thm:characters orthonormal}
The irreducible characters of a finite group form an orthonormal basis with respect to (\ref{eq:inner product}) for the space of class functions of the group,
\[(\chi,\psi)=\delta_{\chi\psi},\qquad \chi,\psi\in\Irr(G).\]
In particular, the number of irreducible representations $|\Irr(G)|$ equals the dimension of the space of class functions, that is, the number of conjugacy classes of $G$.
\end{Theorem}

There is a particular representation that can be used to find more interesting properties of characters, the \emph{regular representation}. Consider the vector space with basis $\{e_g\;|\;g\in G\}$ and turn it into a $G$-module by defining \[ge_h=e_{gh}.\] In this basis each group element is represented by a permutation matrix and only the identity element fixes any basis vectors. In particular, we see that the character is
\[
\regchar{}(g)=\left\{\begin{array}{ll}
|G|&g=1,\\
0& g\ne1.\end{array}\right.
\]
Using the first orthogonality relation, Theorem \ref{thm:characters orthonormal}, one can find the decomposition of the regular representation into irreducible representations. Indeed, if $\regchar{}=\sum_{\chi\in\Irr(G)}n_\chi\chi$ then
\[n_\chi=(\chi,\regchar{})=\frac{1}{|G|}\sum_{g\in G}\chi(g)\overline{\regchar{}(g)}=\chi(1)=\dim V_\chi.\]
When we evaluate the regular character we find the identity
\[
\sum_{\chi\in\Irr(G)}\chi(1)\chi(g)=\left\{\begin{array}{ll}
|G|&g=1,\\
0& g\ne1.\end{array}\right.
\]
In fact, a more general orthogonality statement holds.
\begin{Proposition}[Second orthogonality relation]
\label{prop:second orthogonality relation}
\[\sum_{\chi\in\Irr(G)}\chi(g)\overline{\chi(h)}=\left\{\begin{array}{ll}
|C_G(g)|&[g]=[h],\\
0& [g]\ne[h].\end{array}\right.\]
\end{Proposition}
Here $C_G(g)=\{h\in G\;|\;hg=gh\}$ is the centraliser of $g$ and $[g]=\{g'\in G\;|\;\exists h\,:\;hg'=gh\}$ is the conjugacy class of $g$.

\subsection{Invariant Bilinear Forms}
\label{sec:invariant bilinear forms}

An action of a group $G$ on a $k$-space $V$ induces an action of the group on its dual $V^\ast=\Hom (V,k)$ by the requirement that the natural pairing of $V$ and $V^\ast$ is respected. If $v\in V$, $v^\ast\in V^\ast$ and $g\in G$, we require $v^\ast(v)=(gv^\ast)(gv)$ so that\[gv^\ast=v^\ast\circ g^{-1}.\]
This implies that \[\chi_{V^\ast}=\overline{\chi_V}\]
and we see that $V$ and $V^\ast$ are isomorphic as $G$-modules if and only if the character $\chi_V$ is real valued.

This idea can be generalised to $\Hom(U,V)$. As linear spaces there is an isomorphism \[\Hom(U,V)\cong V\otimes U^\ast.\] If $U$ and $V$ are $G$-modules then so is $V\otimes U^\ast$ by the above constructions. By requiring the linear isomorphism $\Hom(U,V)\cong V\otimes U^\ast$ to be a $G$-module isomorphism, we obtain an action of $g\in G$ on $f\in\Hom(U,V)$. Indeed, if $f$ corresponds to a pure tensor $v\otimes u^\ast \in V\otimes U^\ast$ then \[gf=g(v\otimes u^\ast)=gv\otimes gu^\ast=gv\otimes u^\ast\circ g^{-1}=g\circ f\circ g^{-1},\]
and we extend linearly to the full space.
In particular, this gives another description of $G$-linear maps:
\[\Hom_G(U,V)\cong (V\otimes U^\ast)^G.\]

We are interested in bilinear forms on a complex $G$-module $V$, that is, elements of $V^\ast\otimes V^\ast$. In particular, we are interested in \emph{invariant} bilinear forms. If we assume that $V$ is irreducible then Schur's Lemma gives
\[(V^\ast\otimes V^\ast)^G\cong\Hom_G(V,V^\ast)=\delta_{\chi_V,\overline{\chi_V}}\CC\Id.\]
In other words, if $\chi_V$ is real valued then there is a unique invariant bilinear form and it is nondegenerate. If $\chi_V$ is not real valued then there is no invariant bilinear form.
Let $\chi_V$ be real valued. By the decomposition of $G$-modules $V^\ast\otimes V^\ast=S^2V^\ast\oplus \wedge^2V^\ast$ we know that the invariant bilinear form is either symmetric or antisymmetric. To distinguish between these cases, one can use the so called \emph{Frobenius-Schur indicator} $\iota:\Irr(G)\rightarrow\{-1,0,1\}$ defined by
\begin{equation}
\label{eq:fs indicator}
\iota_\chi=\dim (S^2 V_\chi)^G-\dim(\wedge^2 V_\chi)^G=\frac{1}{|G|}\sum_{g\in G}\chi(g^2).
\end{equation}
The second equality can be derived from the character formulas of the previous section as follows. Let $\triv:G\rightarrow \{1\}$ be the trivial character. Then
$\dim (S^2 V)^G-\dim(\wedge^2 V)^G=(\chi_{S^2 V}-\chi_{\wedge^2 V},\triv)=\frac{1}{|G|}\sum_{g\in G}\chi_{S^2 V}(g)-\chi_{\wedge^2 V}(g)=\frac{1}{|G|}\sum_{g\in G}\chi(g^2)$. 

Complex irreducible representations with Frobenius-Schur indicator $1$, $0$ or $-1$ are respectively known as representations of \emph{real type}, \emph{complex type} or \emph{quaternionic type}. This is due to the classification of real division rings, such as the $G$-linear endomorphisms $\End_G(V)$ on a real simple $G$-module $V$ \cite{fulton1991representation}.

\section{Projective Representations}
\label{sec:schur covers}


In this section we present a classical and elementary exposition of the theory of projective representation, as in the original work by Schur \cite{schur1904darstellung,schur1911darstellung},  following the clear account of Curtis \cite{curtis1999pioneers}.

\subsection{The Schur Multiplier}
Projective representations of a group $G$ are homomorphisms
\[\rho:G\rightarrow\PGL(V),\] where $\PGL(V)$ is the \emph{projective general linear group}, the quotient of $\GL(V)$ by the scalar maps $k\Id$. If we take a nonzero representative of $\rho(g)$ in $\GL(V)$ for each $g\in G$ one obtains a map
\[\tilde{\rho}:G\rightarrow\GL(V).\]
Because $\rho$ is a morphism we have \[\tilde{\rho}(g)\tilde{\rho}(h)=c(g,h)\tilde{\rho}(gh)\] for a map $c:G\times G\rightarrow \CC^\ast$ (where $\CC^\ast$ is the multiplicative group $\CC\setminus\{0\}$) and by associativity one finds that
\[c(g_1,g_2)c(g_1g_2,g_3)=c(g_1,g_2g_3)c(g_2,g_3).\]
This is the defining property of a $2$-\emph{cocycle} $c$ of $G$. A map $\tilde{\rho}:G\rightarrow\GL(V)$ with the property that $\tilde{\rho}(g)\tilde{\rho}(h)=c(g,h)\tilde{\rho}(gh)$ for some cocycle $c$ is also called a projective representation. One obtains a homomorphism $G\rightarrow\PGL(V)$ by composing $\tilde{\rho}$ with the quotient map $\GL(V)\rightarrow\PGL(V)$.

Of course, we had a choice in constants to define $\tilde{\rho}$ from $\rho$. Any other choice is given by $\tilde{\rho}'(g)=b(g)\tilde{\rho}(g)$ for some map $b:G\rightarrow \CC^\ast$. The cocycle related to $\tilde{\rho}'$ is seen to be \[c'(g,h)=\frac{b(g)b(h)}{b(gh)}c(g,h).\]
That is, $c$ and $c'$ differ by a \emph{coboundary} $db(g,h)=\frac{b(g)b(h)}{b(gh)}.$ Therefore, the projective representation $\rho$ determines the cocycle $c$ uniquely up to coboundaries, i.e.~$\rho$ defines an element in the second cohomology group \[H^2(G,\CC^\ast).\]
Multiplication gives $H^2(G,\CC^\ast)$ the structure of an abelian group. This group is also known as the \emph{Schur multiplier} of $G$ and denoted $M(G)$.

Now that we have established equivalence, we will make no distinction between the homomorphism $\rho$ and the pair $(\tilde{\rho},[c])$, where $[c]$ is the related element in the Schur multiplier. 

\subsection{Central Extensions}
The group $G^\flat$ is called an \emph{extension} of $G$ by $Z$ if there exists a short exact sequence
\[1\rightarrow Z\rightarrow G^\flat \xrightarrow{\pi} G\rightarrow 1,\]
i.e., if $Z\triangleleft G^\flat$ and $G\cong\bigslant{G^\flat}{Z}$.
When (the image of) $Z$ is in the centre of $G^\flat$, the extension is called \emph{central}.

The isomorphism $\phi:G\rightarrow\bigslant{G^\flat}{Z}$ implies the existence of a section 
\[{s}:G\rightarrow G^\flat,\qquad \pi\circ s=\text{id},\qquad s(1)=1,\qquad \phi(g)=s(g)Z\] sending $g\in G$ to a representative of the coset $\phi(g)=s(g)Z$.
The fact that $\phi$ is a morphism says $s(g)s(h)Z=s(g)Zs(h)Z=\phi(g)\phi(h)=\phi(gh)=s(gh)Z$ and thus ensures the existence of a map \cite{humphreys1996course} $z:G\times G\rightarrow Z$, defined by \[s(g)s(h)=z(g,h)s(gh)\]
with the properties
\begin{align*}
&z(g,1)=1=z(1,g)\\
&z(g_1,g_2)z(g_1g_2,g_3)=z(g_1,g_2g_3)z(g_2,g_3).
\end{align*}

Now consider an irreducible linear representation \[\tau:G^\flat\rightarrow \GL(V).\] If $z\in Z<Z(G^\flat)$ then Schur's Lemma ensures that there is a scalar $c\in\CC^\ast$ such that $\tau(z)=c\Id$. In particular, we can define a $2$-cocycle $c:G\times G \rightarrow \CC^\ast$ by \[\tau(z(g,h))=c(g,h)\Id.\] 
Then the map \[\tilde{\rho}:G\rightarrow \GL(V),\qquad \tilde{\rho}(g)=\tau(s(g)),\]
has the property that 
\begin{align*}
\tilde{\rho}(g)\tilde{\rho}(h)&=\tau(s(g))\tau(s(h))=\tau(s(g)s(h))\\
&=\tau(z(g,h)s(gh))=\tau(z(g,h))\tau(s(gh))=c(g,h)\tilde{\rho}(gh).
\end{align*}
In other words, $\tilde{\rho}$ is a projective representation of $G$. We see that \emph{each irreducible representation of a central extension $G^\flat$ of $G$ induces a projective representation of $G$.}

A century ago, Schur studied projective representations \cite{curtis1999pioneers, schur1904darstellung} and wondered whether an extension $G^\flat$ of $G$ exists such that all projective representations of $G$ are induced by linear representations of $G^\flat$, in the way described in this section. One would say $G$ is \emph{sufficiently extended}. Particullarly he was interested in sufficient extensions of minimal order. 
He found that such groups exist, and nowadays the sufficient extensions $G^\flat$ of minimal order are called \emph{Schur covers} of $G$ (or \emph{Schur extensions}, formerly known as \emph{Darstellunggruppe} or \emph{representation groups}). This group is not unique in general, contrary to the \emph{Schur multiplier} $M(G)$, which takes the place of $Z$ in the exact sequence of such a minimal extension. 

\subsection{Schur Covers and Other Sufficient Extensions}
We are now familiar with projective representations, central extensions, and how the latter can be used to find the former. However, this does not help us to find an extension which provides all projective representations, let alone a Schur cover. The missing information is given by the next theorem, which is due to Schur \cite{schur1904darstellung}. We present a reformulation by Curtis \cite{curtis1999pioneers}.
\begin{Theorem}
\label{thm:sufficiently extended groups}
Let $G^\flat$ be a central extension of $G$ with kernel $Z$. Then the intersection of the derived subgroup $[G^\flat,G^\flat]$ and the kernel $Z$ is isomorphic to a subgroup of the Schur multiplier $M(G)$. There is an isomorphism $[G^\flat,G^\flat]\cap Z\cong M(G)$ if and only if all projective representations of $G$ are induced by linear representations of $G^\flat$.
\end{Theorem}

The finite groups of our interest are the polyhedral groups. We state their Schur multipliers (or second cohomology groups) \cite{read1976schur, schur1911darstellung}.
\begin{Theorem}
\label{thm:schur multipliers of polyhedral groups}
The Schur multiplier of $\zn{N}$ and $\DD_{2M-1}$ is trivial and the Schur multiplier of $\DD_{2M}$, $\TT$, $\OO$ and $\YY$ is somorphic to $\zn{2}$.
\end{Theorem}
Notice in particular that the order of polyhedral groups equal their exponent times the order of the Schur multiplier.
\begin{equation}
\label{eq:order is exponent times schur multiplier}
|G|=\|G\||M(G)|
\end{equation}
cf.~Table \ref{tab:various properties of polyhedral groups}

With this information one can find sufficiently extended groups for the polyhedral groups. But first a lemma for future reference.

\begin{Lemma}
\label{lem:identical abelianisation}
If $G^\flat$ is a sufficient extension of $G$ then their abelianisations are isomorphic.
\end{Lemma}
\begin{proof}
A one-dimensional projective representation $\tilde{\rho}:G\rightarrow \CC^\ast$ is equivalent to a linear representation. Indeed, the multiplication $\tilde{\rho}(g)\tilde{\rho}(h)=c(g,h)\tilde{\rho}(gh)$ immediately shows that the cocycle $c$ is a coboundary: $c(g,h)=\frac{\tilde{\rho}(g)\tilde{\rho}(h)}{\tilde{\rho}(gh)}$. In other words, the set of one-dimensional representations of $G$ and $G^\flat$ coincide. Therefore $\mathcal{A}G^\flat\cong\mathcal{A}G$.
\end{proof}

\section{Binary Polyhedral Groups}
\label{sec:binary polyhedral groups}

In the light of our classification ambitions regarding Automorphic Lie Algebras it is crucial that we find all projective representation of polyhedral groups $G$. 
For various reasons, both computational and theoretical, it is desirable to work with \emph{linear} representations rather than \emph{projective} representations.
Theorem \ref{thm:sufficiently extended groups} guarantees that we can do this provided that we find a central extension $G^\flat$ of $G$ with kernel $Z$, such that $[G^\flat,G^\flat]\cap Z\cong M(G)$. It is not important to us to have the smallest group $G^\flat$ satisfying these conditions, that is, we do not necessarily need a Schur cover of $G$.




If a group has a trivial Schur multiplier, then any projective representation is equivalent to a linear representation. Indeed, if any cocycle is a coboundary, then the cocycle that occurs in the multiplication of the projective representation can be transformed to be identically $1$. In other words, such a group equals its own unique Schur cover.
Moreover, every central extension of a group with trivial Schur multiplier is sufficiently extended. This also follows immediately from Theorem \ref{thm:sufficiently extended groups}.

The cyclic groups and dihedral groups with odd parameter $N=2M-1$ have trivial Schur multiplier. All other polyhedral groups have a Schur multiplier isomorphic to $\zn{2}$ (Theorem \ref{thm:schur multipliers of polyhedral groups}).
For these groups, we need to find a central extension $G^\flat$ of $G$ with kernel $Z$ such that $[G^\flat,G^\flat]\cap Z$ has two elements, according to Theorem \ref{thm:sufficiently extended groups}. It turns out that the \emph{binary polyhedral groups}, defined as follows, will suffice.

A polyhedral group $G<\PSL_2(\CC)$ uniquely defines a \emph{binary polyhedral group} $\BB G<\SL_2(\CC)$ through the exact sequence
\[1\rightarrow \{\pm\Id\}\rightarrow \SL_2(\CC) \xrightarrow{q} \PSL_2(\CC)\rightarrow 1,\]
i.e.~$\BB G=q^{-1}(G)$. The embedding of $\BB G$ in $\SL_2(\CC)$ is called the \emph{natural representation} of $\BB G$ \cite{dolgachev2009mckay,MR2500567}.

Let us first investigate the centre $Z(\BB G)$. The natural representation is reducible if and only if the group is abelian, i.e.~$Z(\BB G)=\BB G$. Otherwise, if the natural representation is irreducible, then Schur's Lemma ensures that central elements are scalars,
\[Z(\BB G)\subset \CC\Id\cap \SL_2(\CC)=\{\pm \Id\}=Z(\SL_2(\CC)).\]
On the other hand, by definition,
\[Z(\SL_2(\CC))\cap \BB G\subset Z(\BB G),\]
and because $Z(\SL_2(\CC))=\{\pm \Id\}=q^{-1}(1)\subset \BB G$ we obtain
\[Z(\BB G)=\{\pm \Id\}=Z(\SL_2(\CC)).\]


The group $[\BB G,\BB G]\cap Z$ is now determined by whether $-\Id$ is a commutator of $\BB G$ or not. By Theorem \ref{thm:sufficiently extended groups} this cannot be the case for polyhedral groups with trivial Schur multiplier.

Generally, a finite subgroup of $\SL_2(\CC)$ contains $-\Id$ if (and only if) it has even order, since the Sylow theorems (cf.~\cite{humphreys1996course}) ensure that such a group has an element of order two, and $-\Id$ is the only involution in $\SL_2(\CC)$.

One can use the computer package GAP (Groups, Algorithms, Programing, cf.~\cite{GAP}) to determine the derived subgroups of the binary polyhedral groups $\BB \TT$, $\BB \OO$ and $\BB \YY$. To this end it is useful to know the code of these groups in the \emph{Small group library}, \cite{beschesmallgroups}:
\begin{align*}
&\BB \TT=\texttt{SmallGroup}(24,3),\\
&\BB \OO=\texttt{SmallGroup}(48,28),\\
&\BB \YY=\texttt{SmallGroup}(120,5).
\end{align*}
One finds
\begin{align*}
&{[}\BB \TT,\BB \TT]=\BB \DD_2=Q_8,\\
&{[}\BB \OO,\BB \OO]=\BB\TT,\\ 
&{[}\BB \YY,\BB \YY]=\BB\YY.
\end{align*}
In particular, these groups have even order, hence also contain $-\Id$ in $\SL_2(\CC)$.

The binary dihedral group, also known as the \emph{dicyclic} group $\bd{N}$ of order $4N$, has a presentation
\[\bd{N}=\langle r,s\;|\;r^{2N}=1,\, s^2=r^N,\, rs=sr^{-1}\rangle.\]
Similar to the dihedral groups, all elements are of the form $r^j$ or $r^js$ with $0\le j<2N$. With this concise description we can quickly find the derived subgroup. Indeed, 
\begin{align*}
&[r^is,r^j]=(r^is)^{-1}(r^{j})^{-1}r^isr^j=s^{-1}r^{-i-j+i}sr^{j}=s^{-1}sr^{j}r^{j}=r^{2j},\\ 
&[r^is,r^js]=s^{-1}r^{-i}s^{-1}r^{-j}r^isr^js=s^{-1}s^{-1}r^{i}r^{-j}r^ir^{-j}ss=r^{-N+i-j+i-j+N}=r^{2(i-j)},
\end{align*}
so the commutators in $\bd{N}$ are the even powers of $r$:
\[[\bd{N},\bd{N}]=\zn{N}.\]
In particular, if $N$ is even then the embedding of this group in $\SL_2(\CC)$ contains $-\Id$.
We can conclude as follows.
\begin{Observation}
A binary polyhedral group $\BB G$ is a sufficient extension of the related polyhedral group $G$. Moreover, it is a Schur cover of $G$ if and only if $M(G)\cong\zn{2}$.
\end{Observation}

This fact allows us to use results of $\SL_2(\CC)$ theory. Specifically the \emph{Clebsch-Gordan decomposition} (cf.~\cite{fossum, suter2007quantum} and Section \ref{sec:squaring the ring}) will play an important part in the study of Automorphic Lie Algebras.
However, sometimes a different choice of central extension is more convenient. For instance, the polyhedral groups with trivial Schur multiplier do not need an extension at all like the dihedral groups with odd parameter. But the other dihedral groups have a Schur cover that is a bit easier to work with as well. Take the presentation $\DD_{2N}=\langle r,s\;|\;r^{2N}=1,\, s^2=1,\, rs=sr^{-1}\rangle$ and the homomorphism $\pi:\DD_{2N}\rightarrow\DD_{N}<\DD_{2N}$ defined by $\pi(r)=r^2$, $\pi(s)=s$. The image of $\pi$ is $\DD_{N}$ and the kernel $Z=\{1,r^N\}\cong\zn{2}$ is central. Moreover, just as with the dicyclic group we find the commutator subgroup $[\DD_{2N},\DD_{2N}]=\langle r^2\rangle\cong\zn{N}$ which contains $Z$ if and only if $N$ is even, in which case $\DD_{2N}$ is another Schur cover of $\DD_N$, an alternative to $\bd{N}$. The dihedral group is easier to handle in explicit calculations, compared to the dicyclic group, and this will be our preferred choice in the many examples throughout this thesis.

For compactness we will hereafter denote the binary polyhedral group related to $G$ by $G^\flat$ instead of $\BB G$. If $G=\langle \gal, \gbe, \gga\;|\; \gal^{\nal}=\gbe^{\nbe}=\gga^{\nga}=\gal \gbe \gga=1\rangle$
then $G^\flat$ allows a presentation
\[G^\flat=\langle \gal, \gbe, \gga\;|\; \gal^{\nal}=\gbe^{\nbe}=\gga^{\nga}=\gal \gbe \gga\rangle\]
(see for instance \cite{suter2007quantum}). Here we have made the dangerous decision to use the same symbols for the generators of two different groups. We trust however that this will not create confusion as the context will show whether the subject is $G$ or $G^\flat$.

The unique nontrivial central element in this presentation is given by $z=\gal\gbe\gga$ since $g_i\gal\gbe\gga=g_i^{\nu_i+1}=\gal\gbe\gga g_i$. Notice that a homomorphism $\pi:G^\flat\rightarrow G$ with $\pi(z)=1$ maps the relations to $(\pi \gal)^{\nal}=(\pi \gbe)^{\nbe}=(\pi \gga)^{\nga}=(\pi \gal)(\pi \gbe)(\pi \gga)=1$.



\section{Selected Character Tables}
\label{sec:characters}
In this section we list the characters that will be used in the sequel.
The characters of $\zn{N}=\langle r\;|\;r^N=1\rangle$ are the homomorphisms $\chi:\zn{N}\rightarrow \splitk^\ast$ which are given by $\{\chi:r\mapsto \omega_N^j\;|\;0\le j< N\}$ where $\omega_N=e^{\frac{2\pi i}{N}}$ or any other primitive $N$-th root of unity.

Next we consider the dihedral group $\DD_N=\langle r, s\;|\;r^N=s^2=(rs)^2=1 \rangle$.
There are significant differences between the cases where $N$ is odd or even, as we have already seen from their projective representations.
If $N$ is odd then $\DD_N$ has two one-dimensional characters, $\chi_1$ and $\chi_2$. If $N$ is even there are two additional one-dimensional characters, $\chi_3$ and $\chi_4$. 
\begin{center}
\begin{table}[h!]
\caption{One-dimensional characters of $\DD_N$.}
\label{eq:D_N one-dimensional characters}
\begin{center}
\begin{tabular}{ccccc}
\hline
 $g$ & $\chi_1$ & $\chi_2$ & $\chi_3$ &$\chi_4$\\
\hline
$r$ & $1$ & $1$ & $-1$ & $-1$\\
$s$ & $1$ & $-1$ & $1$ & $-1$\\
\hline
\end{tabular}
\end{center}
\end{table}
\end{center}

The remaining irreducible characters are two-dimensional (indeed, there is a normal subgroup of index $2$: $\zn{N}$, \cite{serre1977linear}). We denote these characters by $\psi_j$, for $1\le j <\nicefrac{N}{2}$. They take the values
\begin{equation}
\label{eq:D_N two-dimensional characters}
\psi_j(r^i)=\omega_N^{ji}+\omega_N^{-ji},\qquad \psi_j(sr^i)=0\,,
\end{equation}
where $\omega_N= e^{\frac{2\pi i}{N}}.$
The characters of the faithful (i.e.~injective) representations are precisely those $\psi_j$ for which $\gcd(j,N)=1$, where $\gcd$ stands for greatest common divisor.

Finally, we move to the Schur covering groups of $\TT$, $\OO$ and $\YY$ and we choose the binary polyhedral groups for this. Their character tables can be produced by the computer package GAP \cite{GAP}.
In the character table for $\bt$ we use the definition $\omega_3= e^{\frac{2\pi i}{3}}$ 
and in anticipation of the character table of $\by$ we define the golden section $\phi^+$ and its conjugate $\phi^-$ in $\QQ(\sqrt{5})$
\[\phi^\pm= \frac{1\pm\sqrt{5}}{2},\] which are the roots of $p(x)=x^2-x-1$. 
If $\omega_5=e^{\frac{2\pi i}{5}}$ then $\omega_5^2+\omega_5^3 \text{ and } \omega_5+\omega_5^4$ both satisfy the equation $p(-x)=x^2+x-1=0$. Since the first is negative and the second positive we conclude that $\omega_5^2+\omega_5^3=-\phi^+$ and $\omega_5+\omega_5^4=-\phi^-$.
\begin{center}
\begin{table}[h!]
\caption{Irreducible characters of the binary tetrahedral group $\bt$.}
\label{tab:ctbt}
\begin{center}
\begin{tabular}{c|ccccccc|cccc} \hline
$g$& $1$ &$\gal^2$ &$\gga$&$z$&$\gbe^2$&$\gbe$&$\gal$\\
$|C_G(g)|$&$24$&$6$&$4$&$24$&$6$&$6$&$6$ &$\iota$&$\det$&$\im$\\
\hline
$\bti$ & $1$ &$1$&$1$&$1$&$1$&$1$&$1$&$1$&$\bti$&$1$\\
$\btii$ & $1$ &$\omega_3$&$1$&$1$&$\omega_3^2$&$\omega_3$&$\omega_3^2$&$0$&$\btii$&$\zn{3}$\\
$\btiii$  & $1$ &$\omega_3^2$&$1$&$1$&$\omega_3$&$\omega_3^2$&$\omega_3$&$0$&$\btiii$&$\zn{3}$\\
$\underline{\btiiii}$ & $2$ & $-1$ & $0$ & $-2$ & $-1$ & $1$ & $1$ & $-1$&$\bti$&$\bt$\\
$\btiiiii$ & $2$ & $-\omega_3^2$ & $0$ & $-2$ & $-\omega_3$ & $\omega_3^2$ & $\omega_3$&$0$&$\btii$&$\bt$\\
$\btiiiiii$ & $2$ & $-\omega_3$ & $0$ & $-2$ & $-\omega_3^2$ & $\omega_3$ & $\omega_3^2$&$0$&$\btiii$&$\bt$\\
$\btiiiiiii$ & $3$ & $0$ & $-1$ & $3$ & $0$ & $0$ & $0$ &$1$&$\bti$&$\TT$\\
\hline
\end{tabular}
\end{center}
\end{table}
\end{center}

The characters will be denoted by the symbol used for the group and a numbering in the subscript. For example, the characters of the binary tetrahedral group are denoted $\Irr(\bt)=\{\bti,\ldots,\btiiiiiii\}$. This way it is easier to discuss representations of different groups. We add a superscript ``$\flat$'' if and only if the character is \emph{spinorial}, cf.~Definition \ref{def:spinorial}, e.g.~$\btiiii$. 
The order of the centraliser in the second row is found by the second orthogonality relation, Proposition \ref{prop:second orthogonality relation}.
The column headed ``$\iota$'' contains the Frobenius-Schur indicator (\ref{eq:fs indicator}). In the next column we have the homomorphism ``$\det$'' defined by composing the representation with $\det:\GL(V)\rightarrow\CC^\ast$. The right most column describes the group structure of the image of the representation. In particular, there we see which representations are faithful. Finally, we have underlined the character of the natural representation $G^\flat\hookrightarrow \SL_2(\CC)$.


\begin{center}
\begin{table}[h!]
\caption{Irreducible characters of the binary octahedral group $\bo$.}
\label{tab:ctbo}
\begin{center}
\begin{tabular}{c|cccccccc|cccc}\hline
$g$& $1$ & $\gga$ & $\gbe^2$ & $\gal^2$ & $z$ & $\gal^3$ & $\gbe$ & $\gal$ \\
$|C_G(g)|$&$48$&$4$&$6$&$8$&$48$&$8$&$6$&$8$&$\iota$&$\det$&$\im$\\
\hline
$\boi$ & $1$ &$1$&$1$&$1$&$1$&$1$&$1$&$1$&$1$&$\boi$&$1$\\
$\boii$ & $1$ & $-1$ & $1$ & $1$ & $1$ &$-1$&$1$&$-1$ &$1$&$\boii$&$\zn{2}$\\
$\boiii$ & $2$ &$0$&$-1$&$2$&$2$&$0$&$-1$&$0$&$1$&$\boii$&$\DD_3$\\
$\underline{\boiiii}$ & $2$ &$0$&$-1$&$0$&$-2$&$-\sqrt{2}$&$1$&$\sqrt{2}$&$-1$&$\boi$&$\bo$\\
$\boiiiii$ & $2$ &$0$&$-1$&$0$&$-2$&$\sqrt{2}$&$1$&$-\sqrt{2}$&$-1$&$\boi$&$\bo$\\
$\boiiiiii$ & $3$ &$1$&$0$&$-1$&$3$&$-1$&$0$&$-1$&$1$&$\boii$&$\OO$\\
$\boiiiiiii$ & $3$ &$-1$&$0$&$-1$&$3$&$1$&$0$&$1$&$1$&$\boi$&$\OO$\\
$\boiiiiiiii$ & $4$ &$0$ & $1$ & $0$ & $-4$ & $0$ & $-1$ & $0$ & $-1$&$\boi$&$\bo$\\\hline
\end{tabular}
\end{center}
\end{table}
\end{center}

\begin{center}
\begin{table}[h!]
\caption{Irreducible characters of the binary icosahedral group $\by$.}
\label{tab:ctby}
\begin{center}
\begin{tabular}{c|ccccccccc|cccc} \hline
$g$& $1$ & $\gal^2$ &$\gal^4$&$\gbe$&$\gga$&$\gbe^2$&$\gal^3$&$z$&$\gal$\\
$|C_G(g)|$&$120$&$10$&$10$&$6$&$4$&$6$&$10$&$120$&$10$&$\iota$&$\det$&$\im$\\
\hline
$\byi$ & $1$ &$1$&$1$&$1$&$1$&$1$&$1$&$1$&$1$&$1$&$\byi$&$1$\\
$\underline{\byii}$ & $2$ &$-\phi^-$&$-\phi^+$&$1$&$0$&$-1$&$\phi^-$&$-2$&$\phi^+$&$-1$&$\byi$&$\by$\\
$\byiii$ & $2$ &$-\phi^+$&$-\phi^-$&$1$&$0$&$-1$&$\phi^+$&$-2$&$\phi^-$&$-1$&$\byi$&$\by$\\
$\byiiii$ & $3$ &$\phi^+$&$\phi^-$&$0$&$-1$&$0$&$\phi^+$&$3$&$\phi^-$&$1$&$\byi$&$\YY$\\
$\byiiiii$ & $3$ &$\phi^-$&$\phi^+$&$0$&$-1$&$0$&$\phi^-$&$3$&$\phi^+$&$1$&$\byi$&$\YY$\\
$\byiiiiii$ & $4$&$-1$&$-1$&$1$&$0$&$1$&$-1$&$4$&$-1$&$1$&$\byi$&$\YY$\\
$\byiiiiiii$ & $4$&$-1$&$-1$&$-1$&$0$&$1$&$1$&$-4$&$1$&$-1$&$\byi$&$\by$\\
$\byiiiiiiii$ & $5$&$0$&$0$&$-1$&$1$&$-1$&$0$&$5$&$0$&$1$&$\byi$&$\YY$\\
$\byiiiiiiiii$ & $6$&$1$&$1$&$0$&$0$&$0$&$-1$&$-6$&$-1$&$-1$&$\byi$&$\by$\\\hline
\end{tabular}
\end{center}
\end{table}
\end{center}

\newpage

\section{Classical Invariant Theory}
\label{sec:classical invariant theory}

The theory of polynomial invariants, developed in the second half of the $19$th century, 
is a powerful aid in the study of Automorphic Lie Algebras. We list some important notions and results.
The main references for this section are \cite{neusel2007invariant,MR1328644,stanley1979invariants}.

Let $\natrep$ be a finite dimensional vector space and consider the direct sum of vector spaces
\[\CC[\natrep]=\CC\oplus \natrep^\ast\oplus S^2 \natrep^\ast\oplus S^3 \natrep^\ast\oplus\ldots,\]
where $S^d\natrep$ denotes the $d$-th symmetric tensor of $\natrep$. Given a basis $\{X_1,\ldots, X_n\}$ for $\natrep^\ast$, the summand $S^d \natrep^\ast$ corresponds to homogeneous polynomials in the variables $\{X_1,\ldots, X_n\}$ of degree $d$, also known as \emph{forms}. 
The vector space $\CC[\natrep]$ is a graded ring under the usual multiplication. We introduce the notation $R_d=S^d \natrep^\ast$ and
\[R=\bigoplus_{d\ge 0}R_d,\qquad R_dR_e\subset R_{d+e}.\]

If $\natrep$ is a $G$-module by the representation \[\sigma:G\rightarrow \GL(\natrep)\]
then there is an induced action on the polynomial ring $R=\CC[\natrep]$ given by
\[g\cdot p=p\circ \sigma(g)^{-1},\qquad p\in R,\,g\in G.\]
We are interested in the decomposition of $R$ into the isotypical components $R^\chi$ of this group action,
\[R=\bigoplus_{\chi \in \Irr(G)} R^\chi.\]

The action of $G$ on $R$ respects the ring structure. Moreover, it preserves the grading. That is, each homogeneous component $R_d=S^d\natrep^\ast$ is a submodule. For this last reason we may confine our investigations to forms.


\begin{Example}[$\zn{2}$]
The group of two elements has two conjugacy classes, hence two irreducible characters. Say
\[G=\zn{2}=\left \langle\begin{pmatrix}0&1\\1&0\end{pmatrix}\right\rangle, \qquad \Irr(G)=\{\triv, \chi\},\] where $\triv$ denotes the trivial character and $\chi$ the nontrivial character, sending the nontrivial group element to $-1\in\CC^\ast$.
The isotypical components of $R$ are
\[R^\triv=\CC[X+Y,XY],\qquad R^\chi=\CC[X+Y,XY](X-Y).\]
Indeed, one checks that $X+Y$ and $XY$ are invariant and $X-Y$ has character $\chi$. Hence $R^\triv\supset\CC[X+Y,XY]$ and $R^\chi\supset\CC[X+Y,XY](X-Y)$. In particular $\CC[X+Y,XY]\cap\CC[X+Y,XY](X-Y)=0$. Now notice that $X+Y$ and $XY$ are algebraically independent and compute the sum of the Poincar\'e series accordingly
\begin{align*}
&P(\CC[X+Y,XY],t)+P(\CC[X+Y,XY](X-Y),t)\\
&=\frac{1}{(1-t)(1-t^2)}+\frac{t}{(1-t)(1-t^2)}=\frac{1}{(1-t)^2}=P(R,t).
\end{align*}
Thus we have indeed found the decomposition.
\end{Example}

\begin{Example}[$\DD_3$]
\label{ex:D3 decomposition}
Section \ref{sec:characters} describes the three irreducible characters of
\[G=\DD_3=\left \langle\begin{pmatrix}\omega_3&0\\0&\omega_3^{-1}\end{pmatrix},\begin{pmatrix}0&1\\1&0\end{pmatrix}\right\rangle,\;\omega_3=e^{\frac{2\pi i}{3}}, \qquad \Irr(G)=\{\epsilon, \chi, \psi\}.\]
We obtain the components
\begin{align*}
&R^\triv=\CC[XY,X^3+Y^3],\\
&R^\chi=\CC[XY,X^3+Y^3]\left(X^3-Y^3\right),\\
&R^\psi=\CC[XY,X^3+Y^3]\left(Y\oplus X\oplus X^2\oplus Y^2\right).
\end{align*}
Again we check the sum of the Poincar\'e series
\begin{align*}
&P(R^\triv,t)+P(R^\chi,t)+P(R^\psi,t)=
\frac{1+t^3+2t+2t^2}{(1-t^2)(1-t^3)}=
\frac{(1+t)(1+t+t^2)}{(1-t^2)(1-t^3)}\\&=
\frac{1}{(1-t)^2}=P(R,t),
\end{align*}
establishing the decomposition.
\end{Example}



Molien's Theorem is a central result in the theory of invariants. It provides a way to determine the Poincar\'e series of an isotypical component of the polynomial ring without having to find the component first.
\begin{Theorem}[Molien's Theorem, \cite{neusel2007invariant,MR1328644,stanley1979invariants}]
\label{thm:molien} Let the group $G$ act on $R=\CC[\natrep]$ as induced by the representation $\sigma:G\rightarrow\GL(\natrep)$. Then the Poincar\'e series of an isotypical component $R^\chi$ is given by 
\[P(R^\chi,t)=\frac{\chi(1)}{|G|}\sum_{g\in G}\frac{\chi(g)}{\det (1-\sigma(g)t)}\]
for any character $\chi$ of $G$.
\end{Theorem}

\begin{Example}
[Ring of invariants of the dihedral group, cf.~\cite{MR1328644}, p.93]
\label{ex:D_N invariant forms}
Consider the cyclic group $\zn{N}$ generated by the matrix $\diag(\omega_N,\omega_N^{N-1})$, where $\omega_N$ is a primitive $N$-th root of unity. By direct investigation one finds that $R^\zn{N}=\CC[X,Y]^{\zn{N}}=\bigoplus_{i=0}^{N-1}\CC[X^N,Y^N](XY)^i$ and therefore \[P\left(R^\zn{N},t\right)=\frac{\sum_{i=0}^{N-1}t^{2i}}{(1-t^N)^2}=\frac{1-t^{2N}}{(1-t^2)(1-t^N)^2}=\frac{1+t^{N}}{(1-t^2)(1-t^N)}.\]
With this in mind one can determine the invariants for the group 
\[\DD_N=\left\langle r=\begin{pmatrix}\omega_N&0\\0&\omega_N^{N-1}\end{pmatrix},\;s=\begin{pmatrix}0&1\\1&0\end{pmatrix}\right\rangle\]
using Molien's Theorem. Indeed, it follows that
\begin{align*}
P\left(R^{\DD_N},t\right)&=\frac{1}{2N}\sum_{g\in\DD_N}\frac{1}{\det(1-\sigma_gt)}\\
&=\frac{1}{2N}\sum_{i=0}^{N-1}\frac{1}{\det(1-r^i t)}+\frac{1}{2N}\sum_{i=0}^{N-1}\frac{1}{\det(1-r^i s t)}\\
&=\frac{1}{2}P\left(R^\zn{N},t\right)+\frac{1}{2N}\sum_{i=0}^{N-1}\begin{vmatrix}1&-\omega_N^i t\\-\omega_N^{N-i}t&1\end{vmatrix}^{-1}\\
&=\frac{1}{2}\left(\frac{1+t^{N}}{(1-t^2)(1-t^N)}+\frac{1}{(1-t^2)}\right)\\
&=\frac{1}{(1-t^2)(1-t^N)}\,.
\end{align*}
If one can find two algebraically independent $\DD_N$-invariant forms, one of degree $2$ and one of degree $N$, e.g.~$XY$ and $X^N+Y^N$, then the above calculation proves that any $\DD_N$-invariant polynomial is a polynomial in these two forms, i.e.~$R^{\DD_N}=\CC[XY, X^N+Y^N]$.
\end{Example}




The isotypical summands of $R$ in the first examples have a very simple structure. There are classical theorems that guarantee this is always the case.

An ideal $I$ in a ring $R$ is called \emph{prime} if $a,b\in R$ and $ab\in I$ imply that $a\in I$ or $b\in I$. 
\begin{Definition}[Krull dimension]
The \emph{Krull dimension} of a commutative algebra is the supremum of the length of all chains of proper inclusions of prime ideals.
\end{Definition}
The standard (but not obvious) example is that the Krull dimension of $\CC[\natrep]$ is $\dim \natrep$. Here is another example:
\begin{Theorem}[\cite{stanley1979invariants}, Theorem 1.1]
The Krull dimension of $R^G$ is $\dim \natrep$, i.e.~there exists exactly $\dim \natrep$ algebraically independent invariants. 
\end{Theorem}

We say that an $R$-module $M$ is generated by $m_1,\ldots,m_k$ if every element in $M$ can be written as a sum $r_1m_1+\ldots+r_km_k$ where $r_i\in R$, and we write \[M=\sum_{i=1}^kR m_i.\] If in addition the elements $r_1,\ldots,r_k$ are uniquely determined we write \[M=\bigoplus_{i=1}^kR m_i\] and say that $M$ is a \emph{free} $R$-module generated by $m_1,\ldots,m_k$.

\begin{Definition}[Homogeneous system of parameters, \cite{paule2008algorithms}]
Let $n$ be the Krull dimension of a graded algebra $A$. A set of homogeneous elements $\theta_1,\ldots,\theta_n$ of positive degree is a \emph{homogeneous systems of parameters} for the algebra if $A$ is finitely generated as a $\CC[\theta_1,\ldots\theta_n]$-module.
\end{Definition}
\begin{Proposition}[\cite{neusel2007invariant} Proposition 12.15, \cite{paule2008algorithms} Theorem 2.3.1, \cite{stanley1979invariants} Proposition 3.1]
\label{prop:hsop}
If an algebra is a finitely generated \emph{free} module over \emph{one} homogeneous system of parameters, then it is a finitely generated \emph{free} module over \emph{any} of its homogeneous systems of parameters.
\end{Proposition}
An algebra that satisfies the property of Proposition \ref{prop:hsop} is said to be \emph{Cohen-Macaulay}.
\begin{Theorem}[\cite{stanley1979invariants}, Theorem 3.2]
The ring of invariants $R^G$ is Cohen-Macaulay. That is, there are independent invariants $\theta_1,\ldots,\theta_m, \eta_1,\ldots,\eta_k$ such that \[R^G=\bigoplus_{i=1}^k\CC[\theta_1,\ldots,\theta_m]\eta_i\]
where $m$ is the Krull dimension of $R^G$.
\end{Theorem}
Stanley \cite{stanley1979invariants} attributes this theorem to Hochster and Aegon, cf.~\cite{hochster1971cohen}. We call $\theta_1,\ldots,\theta_m$ \emph{primary invariants} and $\eta_1,\ldots,\eta_k$ \emph{secondary invariants}. Notice that primary and secondary invariants are chosen rather than fixed. Consider for instance two descriptions of one ring $\CC[X]=\CC[X^2](1\oplus X)$. In the setting of the left hand side $X$ is primary and on the right hand side $X$ is secondary.

See (\ref{eq:invariant forms}) for more examples of rings of invariants $R^G$.
There we see that these objects are not necessarily polynomial ($k$ might be greater than $1$). The famous Chevalley-Shephard-Todd theorem states that $R^G$ is polynomial if and only if $G$ is generated by \emph{pseudoreflections}, matrices with at most one eigenvalue unequal to one \cite{MR1328644}.


Notice that $R^G\cdot R^\chi\subset R^\chi$, i.e.~$R^\chi$ is a $R^G$-module. In fact, $R^\chi$ is finitely generated as such: $R^\chi=R^G\rho_1+\ldots+ R^G\rho_{k'}$ (cf.~\cite{stanley1979invariants}).
Unfortunately $R^\chi$ need not be a \emph{free} $R^G$-module, in the sense that the sum need not be direct over the ring. There exists however a choice of primary invariants over whose polynomial ring this isotypical component \emph{is} free.
\begin{Theorem}[\cite{stanley1979invariants}, Theorem 3.10]
\label{thm:cohen-macaulay}
An isotypical component $R^\chi$ is a Cohen-Macaulay module, i.e.~there are independent forms $\theta_1,\ldots,\theta_m, \rho_1,\ldots,\rho_{k_\chi}$ such that
\[R^\chi=\bigoplus_{i=1}^{k_\chi}\CC[\theta_1,\ldots,\theta_m]\rho_i\]
where $m$ is the Krull dimension of $R$.
\end{Theorem}

The following proposition is a modification of a result by Stanley \cite{stanley1979invariants} (Lemma 4.2 and Corollary 4.3) where arbitrary finite groups $G$ in $\GL(\natrep)$ and their rings of invariants $R^G$ are considered. We require information on all the $R^G$-modules $R^\chi$, rather than just the invariants. This generalisation ruins the simple form of the second formula in Corollary 4.3 of \cite{stanley1979invariants} (the analogue of (\ref{eq:stanley2}) below). However, restricting to subgroups $G^\flat$ of  $\SL_2(\CC)$ simplifies the situation again, due to the lack of \emph{pseudoreflections}. We give a full proof.

\begin{Proposition}
\label{prop:stanley}
Let $G^\flat$ be a finite subgroup of $\SL_2(\CC)$ and let $\chi$ be one of its irreducible characters. $R^\chi$ is a Cohen-Macaulay module of Krull dimension 2, cf.~Theorem \ref{thm:cohen-macaulay}. Say  
\[R^\chi=\bigoplus_{i=1}^{k_\chi}\CC[\theta_1,\theta_2]\rho_i.\]
Then 
\begin{align}
\label{eq:stanley1}
\frac{k_\chi}{\chi(1)^2}&=\frac{\dt \dtt}{|G^\flat|},\\
\label{eq:stanley2}
\frac{2}{k_\chi}\sum_{i=1}^{k_\chi}\dri&=\dt+\dtt-2,
\end{align}
where $\dti=\deg \theta_i$ and $\dri=\deg \rho_i$.
\end{Proposition}
\begin{proof}
The two equations follow from the first two coefficients, $A$ and $B$, of the Laurent expansion around $t=1$ of the Poincar\'e series
\[P(R^\chi,t)=\frac{A}{(1-t)^2}+\frac{B}{1-t}+\mathcal{O}(1).\]
We have two ways to express $P(R^\chi,t)$. Namely by Molien's Theorem (Theorem \ref{thm:molien}) and by the expression of $R^\chi$ as a Cohen-Macaulay module.
Molien's theorem states
\[P(R^\chi,t)=\frac{\chi(1)}{|G^\flat|}\sum_{g\in G^\flat}\frac{\chi(g)}{\det (1-tg)}.\] Considering $g$ to be diagonal we see that the only contribution to the term of order $(1-t)^{-2}$ in the Laurent expansion comes from the identity element $g=1$, so $A=\frac{\chi(1)^2}{|G^\flat|}.$ The terms $\frac{\chi(g)}{\det (1-tg)}$ that contribute to the coefficient of $(1-t)^{-1}$ in the Laurent expansion come from elements $g$ that have precisely one eigenvalue equal to 1; pseudoreflections. Here we use the assumption that $G^\flat$ is a subgroup of $\SL_2(\CC)$, which has no pseudoreflections. Thus $B=0$.

On the other hand we calculate 
\begin{align*}
&P\left(\bigoplus_{i=1}^{k_\chi}\CC[\theta_1,\theta_2]\rho_i,t\right)=\frac{\sum_{i=1}^{k_\chi}t^{\dri}}{(1-t^{\dt})(1-t^{\dtt})}\\
&=\frac{\sum_{i=1}^{k_\chi}t^{\dri}}{(1-t)^2\sum_{j=0}^{\dt-1}t^j\sum_{j=0}^{\dtt-1}t^j}\\
&=\frac{\sum_{i=1}^{k_\chi}\left[1-\dri (1-t)+\mathcal{O}((1-t)^2)\right]}
{(1-t)^2\sum_{j=0}^{\dt-1}\left[1-j (1-t)+\mathcal{O}((1-t)^2)\right]
\sum_{j=0}^{\dtt-1}\left[1-j (1-t)+\mathcal{O}((1-t)^2)\right]}\\
&=\frac{\sum_{i=1}^{k_\chi}\left[1-\dri (1-t)\right]}
{(1-t)^2\left[\dt-\frac{1}{2}\dt(\dt-1) (1-t)\right]
\left[\dtt-\frac{1}{2}\dtt(\dtt-1)(1-t)\right]}+\mathcal{O}(1)\\
&=\frac{\sum_{i=1}^{k_\chi}\left[1-\dri (1-t)\right]
\left[\frac{1}{\dt}+\frac{\dt-1}{2\dt} (1-t)\right]
\left[\frac{1}{\dtt}+\frac{\dtt-1}{2\dtt} (1-t)\right]}{(1-t)^2}+\mathcal{O}(1)\\
&=\frac{\frac{k_\chi}{\dt\dtt}}{(1-t)^2}+\frac{\frac{k_\chi}{2 \dt\dtt}(\dt-1)+\frac{k_\chi}{2\dt\dtt}(\dtt-1)-\frac{1}{\dt\dtt}\sum_{i=1}^{k_\chi}\dri}{(1-t)}+\mathcal{O}(1)
\end{align*}
and therefore $A=\frac{k_\chi}{\dt\dtt}$ and $B=\frac{k_\chi}{2 \dt\dtt}(\dt-1)+\frac{k_\chi}{2\dt\dtt}(\dtt-1)-\frac{1}{\dt\dtt}\sum_{i=1}^{k_\chi}\dri$. Equating the two expressions for $A$ respectively $B$ results in (\ref{eq:stanley1}) respectively (\ref{eq:stanley2}).
\end{proof}
Even though primary and secondary invariants, and therefore their degrees also
, can be chosen to a certain extent, this proposition gives a relation between the two that is fixed.
For example, the quotient (\ref{eq:stanley1}) predicts the number of generators of a particular space of invariants as a module over the chosen ring of primary invariants. It is interesting that the right hand sides of (\ref{eq:stanley1}, \ref{eq:stanley2}) are independent of the character.
The formulas can be compared to the invariant forms (\ref{eq:invariant forms}).


\section{Ground Forms}
\label{sec:ground forms}

In this section we consider the polynomial ring \[R=\CC[U]\] as in Section \ref{sec:classical invariant theory}. This time we fix the module $U$ to the natural one of the binary polyhedral groups, i.e.~we have a monomorphism \[\sigma:G^\flat\rightarrow \SL(U)=\SL_2(\CC),\]
e.g.~$U=\tnat$, $\onat$ or $\ynat$.
We will relate orbits of the polyhedral group in the Riemann sphere to forms in $R$, called \emph{ground forms}. 
\begin{Definition}[Algebraic sets]
\label{def:algebraic set}
Let $S\subset \CC[\natrep]$ and $A\subset \natrep$.
We denote the set of common zeros of polynomials in $S$ by
\[\cV(S)=\{u\in \natrep\;|\;f(u)=0,\,\forall f\in S\},\]
which is known as an \emph{algebraic set}.
Conversely, the set of polynomials vanishing on $A$ is denoted
\[\cN(A)=\{f\in \CC[\natrep]\;|\;f(u)=0,\,\forall u\in A\}\]
and constitutes an ideal in $\CC[\natrep]$.
\end{Definition}

The zeros of a form $p\in R$ is a collection of lines in $\natrep\cong \CC^2$, since, by homogeneity, $p(tu)=t^{\deg p}p(u)$, $t\in\CC$, $u\in\natrep$. These zeros are thus well defined on the Riemann sphere $\overline{\CC}\cong \CC P^1$. In other words, if $q:\CC^2 \rightarrow \CC P^1$ is the quotient map, a form $p$ has the property $q^{-1}(q(\cV(p)\setminus\{0\}))=\cV(p)\setminus\{0\}$, whereas in general one can only say $q^{-1}(q(\cV(p)\setminus\{0\}))\supset\cV(p)\setminus\{0\}$. In what follows we will suppress $q$ in the notation and switch between points on the Riemann sphere and lines in $\natrep$ without mentioning it.

\begin{Definition}[The form vanishing on an orbit]
\label{def:form of orbit}
For any $G$-orbit $\Gamma$ on the Riemann sphere, we define 
$F_\Gamma\in R_{|\Gamma|}$
as the form vanishing on $\Gamma$, i.e.~the generator of the ideal  
\[\cN(\Gamma)=RF_\Gamma.\]
This form $F_\Gamma$ is unique up to multiplicative constant. 
Moreover, there is a number $\nu_\Gamma\in\NN$ such that 
\begin{equation*}
\nu_\Gamma |\Gamma|=|G|\,.
\end{equation*}
For exceptional orbits we use a shorthand notation $F_i=F_{\Gamma_i}$ and $\nu_i=\nu_{\Gamma_i}$, $i\in\Omega$, in agreement with Section \ref{sec:polyhedral groups}. 
\end{Definition}

The forms $F_i$ are relative invariants of $G^\flat$. That is, $g\,F_i=\chi_i(g) F_i$, $\forall g\in G^\flat$, where $\chi_i$ is a one-dimensional character of $G^\flat$. Indeed, since the action of $G$ permutes the zeros of $F_i$, $g\,F_i$ must be a scalar multiple of $F_i$. In fact, these forms generate the same ring as the relative invariants of any Schur cover $G^\flat$ do, as follows from an induction argument on the zeros \cite{dolgachev2009mckay, toth2002glimpses}
\[R^{[G^\flat,G^\flat]}=\CC[F_i\;|\; i\in\Omega].\]
This is trivially true for the cyclic case, where $[G^\flat,G^\flat]=1$, $|\Omega|=2$ and $F_i$ is linear. 

The forms $F_i^{\nu_i}$ and $F_j^{\nu_j}$ of degree $|G|$ vanish on different orbits if $i\ne j$. Therefore, any orbit of zeros can be achieved by taking a linear combination of these two. If the orbit has size equal to the order $|G|$, the zeros are simple. This can be done with any pair, $F_j^{\nu_j}$ and $F_k^{\nu_k}$, obtaining the same form. Hence, if $|\Omega|=3$, there is a linear relation between $\fal^\nal$, $\fbe^\nbe$ and $\fga^\nga$, as is well known \cite{toth2002glimpses}, and we may choose our constants such that
\begin{equation}
\label{eq:relation ground forms}
\sum_{i\in \Omega} F_i^{\nu_i}=0,\qquad |\Omega|=3.
\end{equation}
There is no other algebraic relation between the forms $F_i^{\nu_i}$.

In the non-cyclic case, i.e.~$|\Omega|=3$, we have elements of order $\nga=2$, which gives us the following well known structure.
\begin{Proposition}
\label{prop:ring of relative invariant forms}
For all non-cyclic binary polyhedral groups $G^\flat$,
\[
R^{[G^\flat,G^\flat]}=\CC[\fal,\fbe]\oplus \CC[\fal,\fbe]\fga,
\]
where $\fal$, $\fbe$ and $\fga$ are the ground forms of $G^\flat$ and there is no exceptional orbit larger than $\Gga$.
\end{Proposition}
The structure of a particularly interesting subring, the ring of invariants $R^{G^\flat}$, is more complicated. We first consider another consequence of (\ref{eq:relation ground forms}), namely, if $g F_i=\chi_i(g)F_i$, then $\chi_i^{\nu_i}=\chi_j^{\nu_j}$, for all $i,j\in\Omega$.
This leads us to the next lemma.

\begin{Lemma}
\label{lem:finuiinvariant}
If $\cN(\Gamma_i)=RF_i$ for $i\in\Omega$, then 
\[F_i^{\nu_i}\in R^{G^\flat}\]
where $\nu_i=\frac{|G|}{|\Gamma_i|}$.
\end{Lemma}
\begin{proof}
For the cyclic group the statement follows immediately since then $\nu_i=N=\left|\zn{N}\right|$ and therefore $\chi^{\nu_i}=\triv$ for any homomorphism $\chi:\zn{N}\rightarrow \CC^\ast$, where $\triv:G\rightarrow\{1\}\subset\CC^\ast$ denotes the trivial character.

By definition,  $\chi_j^{\|\cA G^\flat\|}=\triv$, where $\cA G$ is the abelianisation of $G$ and $\|H\|$ denotes the exponent of $H$. Indeed, the image of $\chi_j$ is an abelian quotient of $G^\flat$ and therefore a subgroup of $\cA G^\flat$.
To show that $F_i^{\nu_i}$ is invariant for the non-cyclic groups one only needs to find one $j\in\Omega$ such that $\chi_j^{\nu_j}=\triv$, thanks to the above observation that $\chi_i^{\nu_i}=\chi_j^{\nu_j}$.
Thus it will suffice to find one particular order $\nu_i$ that is a multiple of $\| \cA G^\flat\|$.
The abelianisations $\cA G^\flat$ (isomorphic to $\cA G$ by Lemma \ref{lem:identical abelianisation}) are shown in Table \ref{tab:various properties of polyhedral groups}. One finds $\| \cA \DD_N\|=\| \cA \OO\|=2=\nga$, $\| \cA \TT\|=3=\nbe$ and $\| \cA \YY\|=1$, proving the Lemma. \end{proof}



The rings of invariants are
\begin{equation}
\label{eq:invariant forms}
\begin{array}{llll}
R^{\zn{N}}&=&\CC[\fal\fbe, \fal^N+\fbe^N](1\oplus \fal^N-\fbe^N),\\
R^{\DD_N}&=&\CC[\fal,\fbe]\quad\,\text{if $N$ is odd},\\
R^{\DD_{2N}}&=&\CC[\fal,\fbe^2]\quad\text{if $N$ is even},\\
R^{\bt}&=&\CC[\fga,\fal\fbe](1\oplus \fal^3),\\
R^{\bo}&=&\CC[\fbe,\fal^2](1\oplus \fal\fga),\\
R^{\by}&=&\CC[\fal,\fbe](1\oplus \fga).
\end{array}
\end{equation}
The degrees $d_i=\deg F_i$ can be found in Table \ref{tab:various properties of polyhedral groups} and give the Poincar\'e series
\begin{equation}
\label{eq:gf of invariant forms}
\begin{array}{lll}
P\left(R^{\zn{N}},t\right)&=&\frac{1+t^N}{(1-t^2)(1-t^N)},\\
P\left(R^{\DD_N},t\right)&=&\frac{1}{(1-t^2)(1-t^N)},\\
P\left(R^{\bt},t\right)&=&\frac{1+t^{12}}{(1-t^6)(1-t^8)},\\
P\left(R^{\bo},t\right)&=&\frac{1+t^{18}}{(1-t^8)(1-t^{12})},\\
P\left(R^{\by},t\right)&=&\frac{1+t^{30}}{(1-t^{12})(1-t^{20})}.
\end{array}
\end{equation}
One can confirm (\ref{eq:stanley1}) for all of these rings, and (\ref{eq:stanley2}) for all but the dihedral invariants. Indeed, $\DD_N$ is not a subgroup of $\SL_2(\CC)$.



\chapter[Invariants of Polyhedral Groups]{Invariants of Polyhedral Groups}
\label{ch:A}

This chapter lays the foundation of the invariants of Automorphic Lie Algebras (see Concept \ref{conc:invariant of alias}) that we will find in this thesis, and is in this sense the most important chapter. It only concerns the theory of finite groups and their invariants; Lie algebras will not enter the story until the next chapter. 
It is this group theory that explains the group-independence of Automorphic Lie Algebras found to date.

Henceforth $G$ will be a polyhedral group (cf.~Section \ref{sec:polyhedral groups}) and $G^\flat$ the corresponding binary group (cf.~Section \ref{sec:binary polyhedral groups}), unless otherwise stated. We recall the presentations
\begin{align*}
&G=\langle \gal, \gbe, \gga\;|\;\gal^\nal=\gbe^\nbe=\gga^\nga=\gal\gbe\gga=1\rangle\\
&G^\flat=\langle \gal, \gbe, \gga\;|\;\gal^\nal=\gbe^\nbe=\gga^\nga=\gal\gbe\gga\rangle
\end{align*}
where in the noncyclic case the numbers $\nu_i$, $i\in\Omega=\{\al, \be, \ga\}$, are given by Table \ref{tab:orders of non-cyclic polyhedral groups}.
\begin{center}
\begin{table}[h!] 
\caption{Orders of non-cyclic polyhedral groups.}
\label{tab:orders of non-cyclic polyhedral groups}
\begin{center}
\begin{tabular}{ccc} \hline
$G$&$(\nal,\nbe,\nga)$&$|G|$\\
\hline
$\DD_{N}$&$(N,2,2)$&$2N$\\
$\TT$&$(3,3,2)$&$12$\\
$\OO$&$(4,3,2)$&$24$\\
$\YY$&$(5,3,2)$&$60$\\
\hline 
\end{tabular}
\end{center}
\end{table}
\end{center}


\section{Representations of Binary Polyhedral Groups}

In this section we discuss both linear and projective representations of (binary) polyhedral groups. We find some interesting general properties of these representations, most notably Lemma $\ref{lem:dimV^g}$ and Theorem $\ref{thm:multiplicities of eigenvalues}$, specific to this class of groups, suggesting these results are related to the group of automorphisms of the Riemann sphere.

Before we get there, we want to distinguish two types of linear representations of binary polyhedral groups $G^\flat$. Since we have a homomorphism $\pi:G^\flat\rightarrow G$, any representation of $G$ can be extended to a representation of $G^\flat$ by composition with $\pi$. The other representations of $G^\flat$ will be called \emph{spinorial}, following \cite{suter2007quantum}. It is useful to elaborate on this point.

We recall from Section \ref{sec:schur covers} that an irreducible representation $\tau:G^\flat\rightarrow\GL(V)$ of a binary polyhedral group induces a projective representation $\tilde{\rho}:G\rightarrow \GL(V)$ of a polyhedral group through a section $s:G\rightarrow G^\flat$ by defining $\tilde{\rho}=\tau\circ s$. The section defines the cocycle of $\tilde{\rho}$. The map $\tilde{\rho}$ is related to a homomorphism $\rho:G\rightarrow \PGL(V)$ by $\rho=q\circ \tilde{\rho}$ where $q:\GL(V)\rightarrow\PGL(V)$ is the quotient map. 
One can go from the homomorphism $\rho$ to the map $\tilde{\rho}$ using a section $b:\PGL(V)\rightarrow\GL(V)$ by defining $\tilde{\rho}=b\circ\rho$. 
To a projective representation $\rho$ or $\tilde{\rho}$ is related a linear representation $\tau:G^\flat\rightarrow\GL(V)$ by Theorem \ref{thm:sufficiently extended groups}.
\begin{center}
\begin{tikzpicture}
  \matrix (m) [matrix of math nodes,row sep=3em,column sep=4em,minimum width=2em]
  {
    Z& G^\flat & G \\
    \CC\Id &\GL(V) &  \PGL(V) \\};
  \path[-stealth]
    (m-1-1.east|-m-1-2) edge node []{} (m-1-2)
    (m-1-2) edge node [left] {$\tau$} (m-2-2)
    (m-1-2.east|-m-1-3) edge node [above] {$\pi$} (m-1-3)
    (m-2-1.east|-m-2-2) edge node [] {} (m-2-2)
     (m-2-2.east|-m-2-3) edge node [below] {$q$} (m-2-3)
     (m-1-3) edge node [above] {$\tilde{\tau}$}node [below] {$\tilde{\rho}$} (m-2-2)
             edge node [right] {$\rho$} (m-2-3)
	     edge [bend right=40] node [above] {$s$} (m-1-2)
	(m-2-3)edge [bend left=40] node [below] {$b$} (m-2-2);
\end{tikzpicture}
\end{center}
\begin{Definition}[Spinorial representations]
\label{def:spinorial}
The linear representation $\tau$ or the projective representation $\tilde{\rho}$ or $\rho$ is called \emph{nonspinorial} if one of the following equivalent statements hold. Otherwise it is called \emph{spinorial}.


\begin{enumerate}
\item\label{def:spinorial1}
The representation $\tau:G^\flat\rightarrow \GL(V)$ factors through $G$ via the epimorphism $\pi:G^\flat\rightarrow G$. That is, there exists a homomorphism $\tilde{\tau}:G\rightarrow \GL(V)$ such that $\tau=\tilde{\tau}\circ\pi$, which is true if and only if $Z<\ker \tau$.
\item\label{def:spinorial2}
The cocycle $c$ in $\tilde{\rho}(g)\tilde{\rho}(h)=c(g,h)\tilde{\rho}(gh)$, $\forall\,g, h\in G$, is a coboundary, i.e.~there exists an equivalent projective representation $\tilde{\tau}:G\rightarrow\GL(V)$ which is also an ordinary representation.
\item\label{def:spinorial3}
The homomorphism $\rho:G\rightarrow \PGL(V)$ factors through $\GL(V)$ via the quotient map $q:\GL(V)\rightarrow \PGL(V)$, that is, there exists a homomorphism $\tilde{\tau}:G\rightarrow \GL(V)$ such that $\rho=q\circ\tilde{\tau}$.
\end{enumerate}
\end{Definition}
Let us show the equivalence.
Suppose we have a nonspinorial representation $\tau=\tilde{\tau}\circ\pi$ as in (\ref{def:spinorial1}). Then $\tilde{\rho}=\tau\circ s=\tilde{\tau}\circ\pi\circ s=\tilde{\tau}$. This is a homomorphism, hence we have (\ref{def:spinorial2}). 

For the other direction, let $\tilde{\rho}=\tau\circ s$ be a homomorphism, i.e.~satisfy (\ref{def:spinorial2}). We know $s$ is not a homomorphism, since otherwise $G$ is a finite subgroup of $\SL_2(\CC)$, contradicting the classification results. Therefore there exists $g,h\in G$ such that $s(g)s(h)=z(g,h)s(gh)$ where $z(g,h)$ is the nontrivial element in $Z$. Then 
\begin{align*}
\tilde{\rho}(gh)&=\tilde{\rho}(g)\tilde{\rho}(h)=\tau(s(g))\tau(s(h))\\
&=\tau(s(g)s(h))=\tau(z(g,h)s(gh))\\
&=\tau(z(g,h))\tau(s(gh))=\tau(z(g,h))\tilde{\rho}(gh)
\end{align*}
and therefore $\tau(z(g,h))=\Id$ and $Z<\ker \tau$. Hence (\ref{def:spinorial1}).
The equivalence of (\ref{def:spinorial2}) and (\ref{def:spinorial3}) follows from the relation $q\circ\tilde{\rho}=\rho$.

One can spot the spinorial and nonspinorial representations in the character table of $G^\flat$ (cf.~sections \ref{sec:character theory} and \ref{sec:characters}) in the columns where all values have maximal norm $|\chi(z)|=\chi(1)$, $\forall \chi\in\Irr G^\flat$. These are the columns of the central elements.
Indeed, $\chi(z)$ is a sum of $\chi(1)$ roots of unity, which implies $|\chi(z)|\le\chi(1)$ and if $\chi$ is faithful then \[|\chi(z)|=\chi(1)\Leftrightarrow \rho_\chi(z)=c\Id\Leftrightarrow z\in Z(G^\flat).\]
The last implication from left to the right uses faithfulness of $\chi$ and the implication from right to left uses Schur's Lemma and irreducibility.

In Section \ref{sec:binary polyhedral groups} we found that $Z=Z(G^\flat)=\zn{2}$ if $Z(G^\flat)\ne G^\flat$ in which case $\chi(z)=\pm\chi(1)$ if $z\in Z$. 
This leads us to the following.
\begin{Observation}
An irreducible character $\chi$ of a non-abelian binary polyhedral group is spinorial if and only if 
$\chi(z)\ne \chi(1)$, where $z=\gal \gbe \gga$ is the unique nontrivial central group element.
\end{Observation}

Some interesting properties of representations of binary polyhedral groups can be deduced in general. 
\begin{Lemma}
\label{lem:odd dimensional nonspinorial}
All odd-dimensional irreducible projective representations of polyhedral groups are nonspinorial.
\end{Lemma}
\begin{proof}
In the setting of this chapter, the kernel of the extension is $Z\cong\zn{2}$.
Consider an odd-dimensional irreducible projective representation of $G$. Its cocycle equals the cocycle of the one-dimensional projective representation defined by taking the determinant of the former representation. But cocycles of one-dimensional projective representations are always trivial, by Lemma \ref{lem:identical abelianisation}, as desired.
\end{proof}

The following lemma and its consequences about spinoral and nonspinorial representations will be used throughout the rest of the thesis. The special case of the icosahedral group is discussed by Lusztig \cite{lusztig2003homomorphisms} where it is attributed to Serre.
\begin{Lemma}
\label{lem:dimV^g}
If $V$ is a spinorial $G^\flat$-module then $V^{\langle g_i\rangle}=0$ for all $i\in\Omega$. If $V$ is a nonspinorial $G^\flat$-module then
\[(|\Omega|-2)\dim V+2\dim V^{G^\flat}=\sum_{i\in\Omega}\dim V^{\langle g_i\rangle}.\]
\end{Lemma}
This result will mostly be used for nontrivial irreducible representations $V$ of non-cyclic polyhedral groups. In that case the formula simplifies to \[\dim V=\sum_{i\in\Omega}\dim V^{\langle g_i\rangle}.\]
\begin{proof}
Averaging over the group gives a projection $\frac{1}{2\nu_i}\sum_{j=0}^{2\nu_i-1}g_i^j:V\rightarrow V^{\langle g_i\rangle}$. If $V$ is irreducible then the central element $z\in G^\flat$ is represented by a scalar (Schur's Lemma \ref{lem:schur}) and since $z^2=1$ this scalar is $-\Id$ if $V$ is spinorial and $\Id$ when $V$ is nonspinorial. In the former case $V^{\langle g_i\rangle}=\im \frac{1}{2\nu_i}\sum_{j=0}^{2\nu_i-1}g_i^j=\im \frac{1}{2\nu_i}\left(\sum_{j=0}^{\nu_i-1}g_i^j-g_i^j\right)=0$.

Now suppose that $V$ is nonspinorial, so that the action is effectively from $G$. We will use the description of $G$ as being covered by stabiliser subgroups $G_\lambda=\{g\in G\;|\;g\lambda=\lambda\}$ of the action on the Riemann sphere $\overline{\CC}\ni\lambda$ (cf.~Section \ref{sec:polyhedral groups}). That is, \[G\setminus\{1\}=\bigcup_{\lambda \in \overline{\CC}}G_\lambda\setminus\{1\}=\bigcup_{i\in\Omega,\;\lambda\in\Gamma_i}G_\lambda\setminus\{1\}.
\]
Notice also that each nontrivial group element appears exactly twice on the right hand side, since each nontrivial rotation of the sphere fixes two points.

Stabiliser subgroups $G_{\lambda}$ and $G_\mu$ are conjugate if $\lambda$ and $\mu$ share an orbit ($G\lambda=G\mu\Leftrightarrow\exists g\in G : g\lambda=\mu\Rightarrow G_\mu=g G_\lambda g^{-1}$). In particular $gV^{G_\lambda}=V^{G_\mu}$ so that the dimension $\dim V^{G_\lambda}$ is constant on orbits. Therefore it is for the purpose of this proof sufficient to take one representative $\langle g_i\rangle$ of the collection of stabiliser subgroups belonging to each exceptional orbit $\Gamma_i$.

Again we use the projection operator $\frac{1}{|G|}\sum_{g\in G}g:V\rightarrow V^{G}$, this time with the general fact that a projection operator is diagonalisable and each eigenvalue is either $0$ or $1$. Thus one can see that its trace equals the dimension of its image: $\tr \frac{1}{|G|}\sum_{g\in G}g=\dim V^{G}$. 
Let $\chi$ be the character of $V$ (cf.~Section \ref{sec:character theory}). Then
\begin{align*}
2\dim V^G
&=2\;\tr\frac{1}{|G|}\sum_{G}g=\frac{2}{|G|}\sum_{G}\chi(g)\\
&=\frac{2}{|G|}\left( \chi(1)+\sum_{G\setminus\{1\}}\chi(g)\right)\\
&=\frac{2\chi(1)}{|G|}+\frac{1}{|G|}\sum_{i\in\Omega}\sum_{\lambda\in \Gamma_i}\sum_{g\in G_\lambda\setminus\{1\}}\chi(g)\\
&=\frac{2\chi(1)}{|G|}+\frac{1}{|G|}\sum_{i\in\Omega}\sum_{\lambda\in \Gamma_i}\left(\nu_i\dim V^{\langle g_i\rangle}-\chi(1)\right)\\
&=\frac{2\chi(1)}{|G|}+\frac{1}{|G|}\sum_{i\in\Omega}d_i\left(\nu_i\dim V^{\langle g_i\rangle}-\chi(1)\right)\\
&=\frac{2\chi(1)}{|G|}+\sum_{i\in\Omega}\dim V^{\langle g_i\rangle}-\frac{\chi(1)}{|G|}\sum_{i\in\Omega} d_i.
\end{align*}
Now we use (\ref{eq:finite subgroups of SO(3)}) to see that $\sum_{i\in\Omega}d_i=(|\Omega|-2)|G|+2$ and obtain the result.
\end{proof}

One consequence of this result is that any vector in a representation of a relevant group is the sum of a $\gal$-, a $\gbe$- and a $\gga$-invariant.
\begin{Corollary} 
\label{cor:sumV^g}
If V is a representation of a non-cyclic polyhedral group then 
\[V=V^{\langle \gal \rangle}+V^{\langle \gbe \rangle}+V^{\langle \gga \rangle}\]
and the sum is direct if and only if $V^G=0$.
\end{Corollary}
\begin{proof} Since any two of the three generators of $G$ generate the whole group, that is, $\langle g_i, g_j\rangle=G$ if $i,j\in\Omega$, $i\ne j$, we have $V^{\langle g_i \rangle}\cap V^{\langle g_j \rangle}=V^G$. Now we count the dimension
\begin{align*}
&\dim \left(V^{\langle \gal \rangle}+V^{\langle \gbe \rangle}+V^{\langle \gga \rangle}\right)\\
&=\sum_{i\in\Omega}\dim V^{\langle g_i \rangle}
-
\sum_{\{i,j\}\subset\Omega,\,i\ne j}\dim \left(V^{\langle g_i \rangle}\cap V^{\langle g_j \rangle}\right)
+
\dim \left(V^{\langle \gal \rangle}\cap V^{\langle \gbe \rangle}\cap V^{\langle \gga \rangle}\right)\\
&=\sum_{i\in\Omega}\dim V^{\langle g_i \rangle}-2\dim V^G=\dim V
\end{align*}
where the last equality is given by Lemma \ref{lem:dimV^g}. 
\end{proof}
Since we have all the character tables available (cf.~Section \ref{sec:characters}), it is not hard to find the dimension of the space $V^{\langle g_i\rangle}$ of $g_i$-invariants in a $G$-module $V$. Indeed, these are simply the traces of the projection operators: $\dim V^{\langle g\rangle}=\frac{1}{\nu_i}\sum_{j=0}^{\nu_i-1}\chi(g^j_i)$
and we list them in Table \ref{tab:dimV^g}. Notice that the second row in this table equals the sum of the last three, in agreement with Lemma \ref{lem:dimV^g}.
\begin{center}
\begin{table}[h!] 
\caption{Dimensions $\dim V^{\langle g_i\rangle},\; i\in\Omega$ for all nontrivial irreducible representations $V$ of non-cyclic polyhedral groups.}
\label{tab:dimV^g}
\begin{center}
\begin{tabular}{cccccccccccccccccccc} \hline
$ $&$\chi_2$&$\chi_3$&$\chi_4$&$\psi_j$&$\btii$&$\btiii$&$\btiiiiiii$&$
\boii$&$\boiii$&$\boiiiiii$&$\boiiiiiii$&$\byiiii$&$\byiiiii$&$\byiiiiii$&$\byiiiiiiii$\\
\hline
$1$&$1$&$1$&$1$&$2$&$1$&$1$&$3$&$1$&$2$&$3$&$3$&$3$&$3$&$4$&$5$\\
\hline
$\gal$&$1$&$0$&$0$&$0$&$0$&$0$&$1$&$0$&$1$&$0$&$1$&$1$&$1$&$0$&$1$\\
$\gbe$&$0$&$0$&$1$&$1$&$0$&$0$&$1$&$1$&$0$&$1$&$1$&$1$&$1$&$2$&$1$\\
$\gga$&$0$&$1$&$0$&$1$&$1$&$1$&$1$&$0$&$1$&$2$&$1$&$1$&$1$&$2$&$3$\\
\hline 
\end{tabular}
\end{center}
\end{table}
\end{center}

The next lemma is not restricted to polyhedral groups.
\begin{Lemma}
\label{lem:dimkG^H}
Let $\regrep{G}$ denote the regular representation of a finite group $G$. If $H<G$ is a subgroup then
\[\dim \regrep{G}^{H}=[G:H]\]
where $[G:H]=\frac{|G|}{|H|}$ is the index of $H$ in $G$.
\end{Lemma}
\begin{proof}
Denote an arbitrary vector $v\in\regrep{G}$ by $\sum_{g\in G}c_g g$. Invariance of $v$ under $h\in H$ means 
\[\sum_{g\in G}c_g g=v=h^{-1} v =h^{-1} \sum_{g\in G}c_g g=\sum_{g\in G}c_g h^{-1}g=\sum_{g\in G}c_{h g} g.\]
In other words, $v\in\regrep{G}^H$ if and only if $c_{h g}=c_g$ for all $g\in G$ and $h\in H$, i.e.~$c:G\rightarrow \splitk$ is constant on cosets of $H$, leaving a vector space of dimension $\frac{|G|}{|H|}$.
\end{proof}
Applied to a polyhedral group and a stabiliser subgroup of its action on the sphere, Lemma \ref{lem:dimkG^H} reads
\begin{equation}
\label{eq:dimkG^g}
\dim \regrep{G}^{\langle g_i \rangle}=\dim \regrep{G^\flat}^{\langle g_i \rangle}=d_i
\end{equation}
were we consider $g_i$ in $G$ or $G^\flat$ accordingly. The dimension is the same because the order of the group and the subgroup are both doubled when going to the binary case.
Notice how this result satisfies Lemma \ref{lem:dimV^g}, considering Equation (\ref{eq:finite subgroups of SO(3)}).

This section will be concluded with the first instance of a recurring theme in this text. A seemingly group-dependent object turns out to be group-independent. The value of the theorem will become apparent in Chapter \ref{ch:B}.

\begin{Theorem}
\label{thm:multiplicities of eigenvalues}
Let $\tau:G^\flat\rightarrow\GL(V)$ be an irreducible representation of a binary polyhedral group. Given a section $s:G\rightarrow G^\flat$, the multiplicities of eigenvalues of $\tau(s(g_i))$ depend only on $i\in\Omega$ and $\dim V$, not on the choice of binary polyhedral group $G^\flat$ nor the choice of the section. The partitions are listed in Table \ref{tab:multiplicities of eigenvalues}.
\begin{center}
\begin{table}[h!] 
\caption{Eigenvalue multiplicities of $\tau(g_i)$ for irreducible representations  $\tau:G^\flat\rightarrow\GL(V)$.}
\label{tab:multiplicities of eigenvalues}
\begin{center}
\begin{tabular}{cccc} \hline
$\dim V$ & $\al$ & $\be$ & $\ga$\\
\hline
$1$ & $(1)$ & $(1)$ & $(1)$ \\
$2$ & $(1,1)$ & $(1,1)$ & $(1,1)$ \\
$3$ & $(1,1,1)$ & $(1,1,1)$ & $(2,1)$ \\
$4$ & $(1,1,1,1)$ & $(2,1,1)$ & $(2,2)$\\
$5$ & $(1,1,1,1,1)$ & $(2,2,1)$ & $(3,2)$\\
$6$ & $(2,1,1,1,1)$ & $(2,2,2)$ & $(3,3)$\\
\hline 
\end{tabular}
\end{center}
\end{table}
\end{center}
\end{Theorem} 
\begin{proof}
Cyclic groups only have one-dimensional irreducible representations, hence these only play a part in the trivial first row of Table \ref{tab:multiplicities of eigenvalues}. In the proof below we assume that $G^\flat$ is non-cyclic.

We denote the multiplicities of eigenvalues of $\tau(s(g_i))$ by $(m_{i,1},\ldots,m_{i,k_i})$. Consider the action of $s(g_i)$ on $V\otimes V^\ast=\End (V)$. Different choices of coset representative of $s(g_i)Z$ only differ by a central element, which acts trivially on $\End (V)$ by Schur's Lemma. In other words, the action of $G^\flat$ on $\End (V)$ is nonspinorial; it is a $G$-module. 

In a basis where $\tau(s(g_i))$ is diagonal, one can quickly check that 
\[\dim \End (V)^{\langle g_i\rangle}=\sum_{r=1}^{k_i} m_{i,r}^2.\] 
By Schur's Lemma, $\dim \End (V)^{G^\flat}=1$ so that Lemma \ref{lem:dimV^g} for the non-cyclic groups reads
\begin{align*}
\sum_{i\in \Omega}\sum_{r=1}^{k_i} m_{i,r}^2
&=\sum_{i\in \Omega}\dim \End (V)^{\langle g_i\rangle}\\
&=\dim\End (V)+2\dim \End (V)^{G^\flat}=n^2+2,
\end{align*}
where $n=\dim V$.

Now, using the order $\nu_i$ of $g_i$, one can find a lower bound for each dimension $\sum_{r=1}^{k_i} m_{i,r}^2$. This can be done by arguing that $\tau(s(g_i))$ can have at most $\min\{\nu_i,n\}$ distinct eigenvalues, i.e.\[k_i\le \min\{\nu_i,n\}.\] Indeed $s(g_i)^{\nu_i}\in\ker\pi< Z(G^\flat)$ and $\tau$ is irreducible so Schur's Lemma ensures $\tau(s(g_i))^{\nu_i}$ is a scalar. Each eigenvalue of $\tau(g_i)$ is a $\nu_i$-th root of this scalar, which leaves $\nu_i$ options.

Observe that $\min\{\nu_i,n\}$ is independent of the choice of binary polyhedral group. 
This follows because when the dimension $n$ of the irreducible representation increases, certain groups drop out of the equation, see Section \ref{sec:characters}. This happens precisely at the number $n$ where one would otherwise see variation in $\min\{\nu_i,n\}$, cf.~Table \ref{tab:number of distinct eigenvalues}.
\begin{center}
\begin{table}[h!] 
\caption{The minima $\min\{\nu_i,n\}$ of the orders $\nu_i$ and the dimension of a simple module of a binary polyhedral group.}
\label{tab:number of distinct eigenvalues}
\begin{center}
\begin{tabular}{ccccc} \hline
$n$ & involved groups & $\al$ & $\be$ & $\ga$\\
\hline
$1$ & $\zn{N},\;\bd{N},\;\bt,\;\bo,\;\by$ & $1$ & $1$ & $1$ \\
$2$ & $\bd{N},\;\bt,\;\bo,\;\by$ & $2$ & $2$ & $2$ \\
$3$ & $\bt,\;\bo,\;\by$ & $3$ & $3$ & $2$ \\
$4$ & $\bo,\;\by$ & $4$ & $3$ & $2$\\
$5$ & $\by$ & $5$ & $3$ & $2$\\
$6$ & $\by$ & $5$ & $3$ & $2$\\
\hline 
\end{tabular}
\end{center}
\end{table}
\end{center}

Consider the map $D:(m_1,\ldots,m_k)\mapsto \sum_{r=1}^k m_r^2$. We will look for a partition $(m_1,\ldots,m_k)$ of $n$ that minimises $D$. If we move a unit in the partition then the value of $D$ changes by
\[D\left((m_1,\ldots,m_k)\right)-D\left((m_1,\ldots,m_r+1,\ldots,m_{r'}-1,\ldots,m_k)\right)=2(m_r-m_{r'}+1).\] 
We see that $D$ can be decreased by this move if and only if there are two parts $m_r$ and $m_{r'}$ such that $|m_r-m_{r'}|>1$. Such a partition is therefore not minimizing. 

But for each $K \le n$ there is only one partition $(m_1,\ldots,m_K)$ of $n$ where any two parts differ by at most one. Hence this is a minimiser and the global minimiser over all partitions with at most $K$ parts. These minimisers, with $K=K_i=\min\{\nu_i, n\}$, are listed in Table \ref{tab:multiplicities of eigenvalues}.
Finally, we check that all the minimisers together satisfy the equation $\sum_{i\in \Omega}\sum_{r=1}^{k_i} m_{i,r}^2=n^2+2$ that we established earlier, therefore no other partitions are possible.
\end{proof}


\section{Invariant Vectors and Fourier Transforms}
\label{sec:invariant vectors}

In order to study Automorphic Lie Algebras through classical invariant theory, one needs to get a handle on invariant vectors, also known as equivariant vectors,
\[(V_\chi\otimes R)^{G^\flat},\]
where $V_\chi$ is a finite dimensional $G^\flat$-module and $R=\CC[U]$ the polynomial ring on the natural representation $U$ of $G^\flat$. 
We start with the example of the dihedral group, where all these spaces of invariant vectors can be found by hand. They are apparently well known, but a good reference is hard to find.

\subsection{Vectors with Dihedral Symmetry}
\label{sec:D_N-invariant vectors}


In Section \ref{sec:characters} we described the characters of the dihedral group. Now we want to consider explicit matrices for the representation $\rho:\DD_N\rightarrow \GL(V_{\psi_j})$. The matrices are only determined up to conjugacy, i.e.~choice of basis. We will consider here the choice
\begin{equation}
\label{eq:standard matrices}
 \rho_r=\begin{pmatrix}\omega_N^{{j}}&0\\0&\omega_N^{N-{j}}\end{pmatrix} ,\qquad \rho_s=\begin{pmatrix}0&1\\1&0\end{pmatrix},
\end{equation}
where $\omega_N=e^{\frac{2\pi i}{N}}$.
By confirming that the group relations $\rho_r^N=\rho_s^2=(\rho_r\rho_s)^2=\Id$ hold and that the trace $\psi_j=\tr \circ\rho$ is given by (\ref{eq:D_N two-dimensional characters}) one can check that this is indeed the correct representation. 

If one knows the space of invariants $V^G$, when $G$ is represented by $\rho$, then one can readily find the space of invariants belonging to an equivalent representation $\rho'$, because the invertible transformation $T$ that relates the equivalent representations, $T\rho_g=\rho'_g T$, also relates the spaces of invariants: $V^{\rho'(G)}=TV^{\rho(G)}$.

Suppose $\{X,Y\}$ is a basis for $V_{\psi_1}^\ast$, corresponding to (\ref{eq:standard matrices}).
One then finds the relative invariant forms (\ref{eq:Fs}) as in Example \ref{ex:D_N invariant forms}:
\begin{equation}
\label{eq:Fs}
\fal=XY\,,\qquad\fbe=\frac{X^N+Y^N}{2}\,,\qquad \fga=\frac{X^N-Y^N}{2}.
\end{equation} 
These forms satisfy one algebraic relation:
\begin{equation}
\label{eq:abg relation}
\fal^N-\fbe^2+\fga^2=0.
\end{equation}
We will show that the invariant vectors over $R=\CC[X,Y]$ are given by
\begin{equation}
\label{eq:D_N-invariant vectors}
\big(V_{\psi_j}\otimes  R\big)^{\DD_N}=\left( \begin{pmatrix} X^{j}\\Y^{j} \end{pmatrix}\oplus \begin{pmatrix}Y^{N-j}\\X^{N-j}\end{pmatrix}\right) \otimes\CC[\fal,\fbe]
\end{equation}
where the sum is direct over the ring $\CC[\fal,\fbe]$.

An object is invariant under a group action if it is invariant under the action of all generators of a group. To find all $\DD_N=\langle r, s \rangle$-invariant vectors we first look for the $\langle r \rangle=\zn{N}$-invariant vectors and then average over the action of $s$ to obtain all dihedral invariant vectors.

The space of invariant vectors is a module over the ring of invariant forms. When searching for $\zn{N}$-invariant vectors, one can therefore look for invariants modulo powers of the $\zn{N}$-invariant forms $XY$, $X^N$ and $Y^N$. 
Moreover, we use a basis $\{e_1, e_2\}$ that diagonalises the cyclic action, cf.~(\ref{eq:standard matrices}).
Therefore, one only needs to investigate the vectors $X^d e_i$ and $Y^{d} e_i$ for $d \in\{0,\ldots,N-1\}$ and $i\in\{1,2\}$.

We have $r X^d=\omega_N^{-d} X^d$ and $r e_1=\omega_N^{{j}}e_1$ hence
\[r X^d e_1 =  \omega_N^{j-d} X^d e_1.\] We want to solve $\omega_N^{j-d}=1$, i.e.~$d-j\in \ZZ N$, and find $d\in(j+\ZZ N)\cap\{0,\ldots,N-1\}=j$. That is, $X^{j} e_1$ is invariant under the action of $\langle r \rangle\cong \zn{N}$.

Now consider the next one, $rY^{d} e_1=\omega_N^{d+j}Y^{d} e_1$. We solve $\omega_N^{{d+j}} =1$ i.e.~$d+j \in \ZZ N$. This implies $d\in (N-j+\ZZ N)\cap \{0, \ldots, N-1\}=N-j$, therefore $rY^{N-j} e_1=Y^{N-j} e_1$ is invariant.

Similarly one finds the invariant vectors $Y^{j} e_2$ and $X^{N-j} e_2$, and thus we show that the invariant vectors $\left(V_{\psi_j}\otimes\CC[X,Y]\right)^{\zn{N}}$ are given by 
\begin{equation*} 
\left( \begin{pmatrix} X^{j}\\0 \end{pmatrix} + \begin{pmatrix} 0\\Y^{j} \end{pmatrix} + \begin{pmatrix} Y^{N-j}\\0 \end{pmatrix} + \begin{pmatrix}0\\X^{N-j}\end{pmatrix}\right) \otimes\CC[X,Y]^{\zn{N}}\,.
\end{equation*}
If we use the fact $\CC[X,Y]^{\zn{N}}=(1\oplus \fga)\otimes\CC[\fal,\fbe]$ this space is generated as a $\CC[\fal,\fbe]$-module by the vectors
\begin{align*}
&\begin{pmatrix} X^{j}\\0 \end{pmatrix}, 
\begin{pmatrix} 0\\Y^{j} \end{pmatrix}, 
\begin{pmatrix} Y^{N-j}\\0 \end{pmatrix}, 
\begin{pmatrix}0\\X^{N-j}\end{pmatrix},\\
&\fga\begin{pmatrix} X^{j}\\0 \end{pmatrix}, 
\fga\begin{pmatrix} 0\\Y^{j} \end{pmatrix}, 
\fga\begin{pmatrix} Y^{N-j}\\0 \end{pmatrix}, 
\fga\begin{pmatrix}0\\X^{N-j}\end{pmatrix}.
\end{align*}
It turns out the above vectors are dependent over the ring $\CC[\fal,\fbe]$. One finds for instance that 
\[\fga\begin{pmatrix} X^{j}\\0 \end{pmatrix}=P(\fal,\fbe)\begin{pmatrix} X^{j}\\0 \end{pmatrix}+Q(\fal,\fbe)\begin{pmatrix} Y^{N-j}\\0 \end{pmatrix}\] 
if $P(\fal,\fbe)=\fbe$ and $Q(\fal,\fbe)=-\fal^{j}$, hence this vector is redundant. Similarly, the other vectors with a factor $\fga$ can be expressed in terms of the vectors without this factor.

The remaining vectors between the brackets are independent over the ring $\CC[\fal,\fbe]$. Indeed, let $P,Q\in\CC[\fal,\fbe]$ and consider the equation 
\[PX^{j}+QY^{N-j}=0\,.\] 
If the equation is multiplied by $Y^{j}$ we find 
\[P\fal^{j} +QY^{N}=P\fal^{j} +Q (\fbe-\fga)=0\,.\] Now one can use the fact that all terms are invariant under the action of $s$ except for $Q\fga$ to see that $Q=0$, and hence $P=0$.  Similarly $PY^{j}+QX^{N-j}=0$ implies $P=Q=0$. Therefore we have a direct sum
\[
\left(V_{\psi_j}\otimes R\right)^{\zn{N}}
=\left( \begin{pmatrix} X^{j}\\0 \end{pmatrix}\oplus \begin{pmatrix} 0\\Y^{j} \end{pmatrix}\oplus \begin{pmatrix} Y^{N-j}\\0 \end{pmatrix}\oplus \begin{pmatrix}0\\X^{N-j}\end{pmatrix}\right) \otimes\CC[\fal,\fbe]\,.
 \] 

To obtain $\DD_N$-invariants we apply the projection $\frac{1}{2}(1+s)$.
Observe that the $\DD_N$-invariant polynomials move through this operator so that one only needs to compute 
\[\frac{1}{2}(1+s)\begin{pmatrix} X^{j}\\0 \end{pmatrix}
=\frac{1}{2}\begin{pmatrix} X^{j}\\Y^{j} \end{pmatrix}
=\frac{1}{2}(1+s)\begin{pmatrix} 0\\Y^{j} \end{pmatrix}\]
and 
\[\frac{1}{2}(1+s)\begin{pmatrix} Y^{N-j}\\0 \end{pmatrix}
=\frac{1}{2}\begin{pmatrix} Y^{N-j}\\X^{N-j} \end{pmatrix}
=\frac{1}{2}(1+s)\begin{pmatrix} 0\\X^{N-j} \end{pmatrix}.
\]
These two vectors are independent over the ring by the previous reasoning, and we have obtained (\ref{eq:D_N-invariant vectors}). 

\begin{Remark}
If one allows the representation with basis $\{X,Y\}$ to be non-faithful, several more cases appear. However, they are not more interesting than what we have seen so far, which is why we decided not to include them in the discussion. Summarised, the more general situation \[\left(V_{\psi_j}\otimes \CC[V_{\psi_{j'}}]\right)^{\DD_N}\]  is as follows.
Let $\rho_j$ be the representation with character $\psi_j$. If $\bigslant{\DD_N}{\ker \rho_j}$ is a quotient group of $\bigslant{\DD_N}{\ker \rho_{j'}}$ then everything is the same as above except that $N$ will be replaced by $N'=\frac{N}{\gcd(N,j')}$ since $\rho_{j'}(\DD_N)\cong \DD_{N'}$. If on the other hand  $\bigslant{\DD_N}{\ker \rho_j}$ is not a quotient group of $\bigslant{\DD_N}{\ker \rho_{j'}}$, then there are no nonzero invariants. Indeed, then the character $\overline{\psi}_j=\psi_j$ cannot appear in the $\bigslant{\DD_N}{\ker \rho_{j'}}$-module $\CC[V_{\psi_{j'}}]$, cf.~Section \ref{sec:fourier transform}.
\end{Remark}

The results of this subsection are summarised in Table \ref{tab:invariants, N odd} and Table \ref{tab:invariants, N even}, where the ground forms are given by (\ref{eq:Fs}).

\begin{table}[h!]
\caption{Module generators $\eta_i$ in $(V_\chi\otimes R)^{\DD_N}=\bigoplus_i\CC[\fal,\fbe]\eta_i$, $N$ odd.}
\label{tab:invariants, N odd}
\begin{center}
\begin{tabular}{cccc}\hline
$\chi$ & $\chi_1$ & $\chi_2$ & $\psi_j$\\
\hline
$\eta_i$ & $1$ & ${\fga}$ & $\begin{pmatrix}X^j\\Y^j\end{pmatrix},\;\begin{pmatrix}Y^{N-j}\\X^{N-j}\end{pmatrix}$\\\hline
\end{tabular}
\end{center}
\end{table}

\begin{table}[h!]
\caption{Module generators $\eta_i$ in $(V_\chi\otimes R)^{\DD_{2N}}=\bigoplus_i\CC[\fal,\fbe^2]\eta_i$.}
\label{tab:invariants, N even}
\begin{center}
\begin{tabular}{cccccc}\hline
$\chi$ & $\chi_1$ & $\chi_2$ & $\chi_3$ & $\chi_4$ & $\psi_j$\\
\hline
$\eta_i$ & $1$ & $\fbe\fga$ & $\fbe$ & $\fga$ & $\begin{pmatrix}X^j\\Y^j\end{pmatrix},\;\begin{pmatrix}Y^{2N-j}\\X^{2N-j}\end{pmatrix}$\\\hline
\end{tabular}
\end{center}
\end{table}

With regard to Table \ref{tab:invariants, N even} we recall that the ground forms $F_i$ are defined by the action of the group on the Riemann sphere, and is therefore related to the polyhedral group $G=\DD_N$ rather than the Schur cover of choice $G^\flat=\DD_{2N}$.

\subsection{Fourier Transform}
\label{sec:fourier transform}

All the information of invariant vectors is contained in the polynomial ring $R$. Indeed, the module of invariant vectors $(V_\chi\otimes R)^G$ and the isotypical component $R^{\overline{\chi}}$ are equivalent in the following sense. 
There is precisely one invariant in $V_\psi\otimes R$ for each copy of $V_{\overline\psi}$ in $R$ and nothing more, because $(V_\psi\otimes V_\chi)^G\cong\Hom_G(V_{{\psi}},V_{\overline{\chi}})$ and Schur's Lemma (Section \ref{sec:representation theory}) states the latter space is one-dimensional if $\chi= \psi$ and zero otherwise. In particular, we have the relation of Poincar\'e series
\begin{equation}
\label{eq:generating function of forms and vectors}
P(R^{\overline{\chi}},t)=\chi(1)P((V_{\chi}\otimes R)^G,t).
\end{equation}
But we have in fact more than that. Fourier decomposition allows us to construct invariant vectors from the forms in $R^{\overline{\chi}}$.

\begin{Example}[]
If $e_1=\begin{pmatrix}1&0\end{pmatrix}^T$ and $e_2=\begin{pmatrix}0&1\end{pmatrix}^T$ are basis vectors for an irreducible $G$-module and $\{X, Y\}$ is its dual basis, then \[X e_1+Y e_2=\begin{pmatrix} X\\ Y\end{pmatrix}\] is the unique invariant vector in the tensor product of the two representations, the trace of the bases.

There is of course more in this four-dimensional tensor product. For example, if the original representation is the two-dimensional irreducible representation $\psi$ of $\DD_3$ then the tensor product has character $\overline{\psi}\psi=\epsilon+\chi+\psi$. The $\chi$-component is $\CC(Xe_1-Ye_2)$ and the $\psi$-component has basis $\{ Xe_2,Ye_1\}$.

Another (dual) basis for $\psi$ in $R$ is $\{Y^2,X^2\}$, cf.~Example \ref{ex:D3 decomposition}, so a second invariant vector is \[Y^2 e_1+X^2 e_2=\begin{pmatrix}Y^2\\X^2\end{pmatrix}.\] Compare with (\ref{eq:D_N-invariant vectors}).
\end{Example}

In general the Fourier transform can be described as follows.
Let \(W\) be a finite dimensional module of a finite group \(G^\flat\) and let  $\{w_i\;|\; i=1,\ldots,\dim W\}$ be a basis of $W$. Then \( W \) can be decomposed as a direct sum of irreducible representations of \(G^\flat\)
as follows. 

Let \(V\) be such an irreducible  \(G^\flat\)--representation and let $\{v^j\;|\;j=1,\ldots,\dim V^\ast\}$ be a basis of \(V^\ast\). Let  \((\chi_W,\chi_V)\) be the multiplicity of $V$ in $W$ (that is, \(V\)  occurs as a direct summand in \(W\) \((\chi_W,\chi_V)\) times) and consider the space of invariants $(W\otimes V^{\ast})^{G^\flat} $ with basis
\[
\Big\{u^{k}=\sum_{i,j}u^{k}_{i,j} \,w_i \otimes v^j\;|\; k=1,\ldots,(\chi_W,\chi_V)\Big\}\,.
\]
The $u^k$ are traces of the basis of $V^\ast$ and its canonical dual basis, a basis for $V$. By the expression for $u^k$ we find $(\chi_W,\chi_V)$ $V$-bases $\{v_j^k=\sum_{i}\eta_{i,j}^k w_i\;|\;j=1,\ldots,\dim V\}$ in $W$.
In practice we take a general element $\sum_{i,j} C_{i,j} \,w_i \otimes v^j$ in $W\otimes V^\ast$ and require this element to be invariant under the action of the generators of $G^\flat$ to obtain elements $U^k=\sum_{i,j}\eta^{k}_{i,j} \,w_i \otimes v^j$.

If we now do the same construction for $U\otimes V$ we find $V^\ast$-bases in $U$. Taking the trace with each $V$-basis in $W$ results in  $(\chi_W,\chi_V)(\chi_U,\overline{\chi_V})$ linearly independent elements of $(W\otimes U)^{G^\flat}$. The space spanned by these elements will be denoted by  $(W\otimes U)_{\chi_V}^{G^\flat}$. We have
\[(W\otimes U)^{G^\flat}=\bigoplus_{\chi\in\Irr G^\flat}(W\otimes U)_\chi^{G^\flat}.\]
This method can be applied to $V\otimes R$ to find invariant vectors.

\subsection{Evaluations}

Obtaining explicit descriptions of invariant vectors as a set of generators over a ring of invariant forms, e.g.~Section \ref{sec:D_N-invariant vectors}, is often a substantial computational problem. On the other hand, describing the finite dimensional vector space that one would get after evaluating the space of invariant vectors in a single point turns out to be much simpler. One can circumvent the problem of determining all invariant vectors by using the regular representation of the group and an interpolating function.

\begin{Proposition}
\label{prop:evaluating invariant vectors}
Suppose a finite group $G$ acts on a vector space $V$ and the Riemann sphere $\overline{\CC}$. The space of invariant $V$-valued rational maps $\left(V\otimes\mero_\Gamma\right)^{G}$ can be evaluated at a point $\mu$ in its holomorphic domain $\overline{\CC}\setminus \Gamma$ to obtain a vector subspace of $V$. This results in
\[\left(V\otimes\mero_\Gamma\right)^{G}(\mu)=V^{G_\mu}\]
where $G_\mu=\{g\in G\;|\;g\mu=\mu\}$ is the stabiliser subgroup.
\end{Proposition}
\begin{proof}
By definition \begin{equation}\label{eq:evaluation trivial inclusion}\left(V\otimes\mero_\Gamma\right)^{G}(\mu)\subset V^{G_\mu}.\end{equation} Lemma \ref{lem:dimkG^H} gives information on the right hand side of this inclusion: 
\[\sum_{\chi\in \Irr(G)}\chi(1)\dim V_\chi^{G_\mu}=\dim\left(\bigoplus_{\chi\in \Irr(G)} \chi(1)V_\chi\right)^{G_\mu}=\dim(\regrep{G})^{G_\mu}=[G:G_\mu].\]
For brevity we define
\[d_\mu=[G:G_\mu]=|G\mu|.\] 
Now it is sufficient to show that 
\[\dim\left(\regrep{G} \otimes\mero_\Gamma\right)^{G}(\mu)=\sum_{\chi\in\Irr(G)}\chi(1)\dim\left(V_\chi\otimes\mero_\Gamma\right)^{G}(\mu)\ge d_\mu\]
so that the left hand side of (\ref{eq:evaluation trivial inclusion}) is at least as big as the right hand side.

A $|G|$-tuple of functions $f_g\in\mero_\Gamma$ defines a vector $\sum_{g\in G} f_g g\in \regrep{G}\otimes \mero_\Gamma$, which is invariant if and only if $\sum_{g\in G} f_g(\lambda) g=h\sum_{g\in G} f_g(\lambda) g=\sum_{g\in G} f_g(h^{-1}\lambda) hg=\sum_{g\in G} f_{h^{-1}g}(h^{-1}\lambda) g$, i.e.
\[f_g(\lambda)=f_{hg}(h\lambda),\qquad \forall g,h\in G,\; \forall \lambda\in \overline{\CC}.\]
In particular, a tuple of rational functions $(f_g\,|\,g\in G)$ related to an invariant is defined by one function, e.g.~$f_1$, since $f_g(\lambda)=f_{g^{-1}g}(g^{-1}\lambda)=f_1(g^{-1}\lambda)$. Conversely, any one function $f_1\in\mero_\Gamma$ gives rise to an invariant in $\regrep{G}\otimes \mero_\Gamma$.

Consider a left transversal $\{h_1,\ldots,h_{d_\mu}\}\subset G$ of $G_\mu$. That is, a set of representatives of left $G_\mu$-cosets. We have a disjoint union
\[G=\bigsqcup_{i=1}^{d_\mu} h_i G_\mu.\]
Define ${d_\mu}$ vectors $v_i\in\CC^{|G|}$ by $v_i=(f(h_i^{-1}g^{-1}\mu)\,|\,g\in G)$ where $f\in\mero_\Gamma$. Then $v_i\in\left(\regrep{G}\otimes\mero_\Gamma\right)^{G}(\mu)$. We are done if we can show that these vectors are linearly independent for at least one choice of $f$. 

If $g^{-1}$ and $g'^{-1}$ are in the same coset $h_j G_\mu$ then $f(h_i^{-1}g^{-1}\mu)=f(h_i^{-1}g'^{-1}\mu)=f(h_i^{-1}h_j\mu)$. Therefore we might as well restrict our attention to the square matrix $f(h_i^{-1}h_j\mu)$.

The point $\mu$ appears precisely once in every row and every column of $(h_i^{-1}h_j \mu)$. Indeed, if $h_i^{-1}h_j \mu=h_i^{-1}h_{j'} \mu$ then $h_j^{-1}h_{j'}\in G_\mu$, i.e.~$h_jG_\mu\ni h_{j'}$ and therefore $h_j=h_{j'}$. On the other hand, if $h_i^{-1}h_j \mu=h_{i'}^{-1}h_{j} \mu=\mu$ then $h_{i}^{-1},h_{i'}^{-1}\in G_\mu h_j^{-1}$ i.e.~$h_{i},h_{i'}\in h_j G_\mu$ which implies $h_{i}=h_{i'}=h_j$.

The proof follows by showing existence of a function $f\in\mero_\Gamma$ such that the determinant $\det(f(h_i^{-1}h_j \mu))\ne 0$.
But for any numbers $(c_1, \ldots, c_{d_\mu})$ there exists an interpolating polynomial $p$ of degree $\le d_\mu\le |G|$ such that $p(h_i \mu)=c_i$ if the points $\{h_i \mu\;|\;i=1,\ldots,d_\mu\}$ are all distinct. Indeed, this follows from the well known \emph{Vandermonde determinant}. This polynomial $p$, which has a pole at infinity, defines an invariant vector with poles at the orbit $G\infty$. To move the pole of $p$ to a point $\bar{\mu}\in\Gamma$ we simply multiply by a power of $\frac{1}{\lambda-\bar{\mu}}$:
\[f=\frac{p}{(\lambda-\bar{\mu})^{d_\mu}}.\]

If, in this construction, we choose $p$ to take the values $c_i=\delta_{1i}$, Kronecker delta, then $f(h_i^{-1}h_j \mu)$ is a matrix with precisely one nonzero entry at each row and each column, and we see that $\det(f(h_i^{-1}h_j \mu))\ne 0$, as desired.
\end{proof}

\begin{Remark}
Proposition \ref{prop:evaluating invariant vectors} can be generalised to other than just rational function spaces. For example, multivariate polynomials $\CC[U]$. The proof will only differ in the existence of an interpolating function in the chosen function space. This can get rather complicated and may introduce more constraints on the orbit of the transversal $\{h_i \mu\;|\;i=1,\ldots,d_\mu\}$, cf.~\cite{olver2006multivariate}.
\end{Remark}

\begin{Example}
To illustrate the proof of Proposition \ref{prop:evaluating invariant vectors}, we look at the dihedral group $\DD_3=\langle r,s\;|\;r^3=s^2=(rs)^2=1\rangle$ and consider the representation $V=V_{\psi_1}$ with basis corresponding to (\ref{eq:standard matrices}) and we use the action on the Riemann sphere \[r\lambda=\omega_3 \lambda,\qquad s\lambda=\frac{1}{\lambda},\]
where $\omega_3$ is a primitive cube root of unity.

If $\mu\in \overline{\CC}$ is an arbitrary point, we have orbit
\[
\DD_3 \mu=\left\{\mu, r\mu, r^2\mu, s\mu, rs\mu, r^2s\mu\right\}
=\left\{\mu, \omega_3\mu, \omega_3^2\mu, \frac{1}{\mu}, \frac{\omega_3}{\mu}, \frac{\omega_3^2}{\mu}\right\}.
\]
The stabiliser subgroups are
\[
(\DD_3)_\mu=\left\{
\begin{array}{ll}
\langle r \rangle & \text{if }\mu\in\{0,\infty\},\\
\langle s \rangle & \text{if }\mu\in\{1,-1\},\\
\langle rs \rangle & \text{if }\mu\in\{\omega_3^2,-\omega_3^2\},\\
\langle r^2s \rangle & \text{if }\mu\in\{\omega_3,-\omega_3\},\\
1 & \text{otherwise.}\\
\end{array}\right.
\]
In the proof of Proposition \ref{prop:evaluating invariant vectors} we consider invariant vectors in the regular representation $\left(\CC\DD_3\otimes \mero_\Gamma\right)^{\DD_3}$ and evaluate them at $\mu$. The claim is that the dimension of the space one ends up with equals the size of the orbit, $|\DD_3 \mu|$. 

We will work out the details of two cases. First we take \[\mu=1\] so that $(\DD_3)_\mu=\langle s \rangle$ and $|\DD_3 \mu|=3$. A transversal for the stabiliser subgroup is \[\{1, r, r^2\}\] which happens to be another subgroup, making it easier to see that the matrix $(r^{-i}r^j \mu)$ has precisely one $\mu$ in each row and in each column.

For any polynomial $p$, the vector $(p, p\circ r^{-1}, p\circ r^{-2} , p\circ s^{-1} , p\circ (rs)^{-1} , p\circ (r^2s)^{-1})$ is an invariant in $\CC\DD_3\otimes \mero$.
In the proof we claim that there is an interpolating polynomial $p$ of degree at most $2$ such that $p(r^i \mu)=\delta_{0i}$. This works with
\[p=\nicefrac{1}{3}(1+\lambda+\lambda^2)
.\]
If we allow poles at $\Gamma=\{0,\infty\}$ then $p=f$ is already in the correct function space $\mero_\Gamma$.
The three polynomials $p,\; p\circ r^{-1}$ and $p\circ r^{-2}$ yield invariant vectors which, evaluated at $\mu=1$, are 
\begin{align*}
&(1,0,0,1,0,0)\\
&(0,0,1,0,0,1)\\
&(0,1,0,0,1,0)
\end{align*}
and therefore $\dim(\CC\DD_3\otimes \mero_\Gamma)^{\DD_3} (\mu)\ge3=|\DD_3 \mu|$.

For the second example we consider a full orbit, \[\mu=i,\] so that $(\DD_3)_\mu=1$ and $|\DD_3 \mu|=6$. The transversal of the stabiliser group is the full group $\DD_3$.
The polynomial 
\[p=\nicefrac{1}{6}(1-i\lambda-\lambda^2+i\lambda^3+\lambda^4-i\lambda^5)\]
satisfies $p(g^{-1}\mu)=\delta_{1g}$ and the invariant vectors generated by $p\circ g^{-1}$ evaluate at $\mu=i$ to the six canonical basis elements of $\CC\DD_3$.
\end{Example}

\section{Homogenisation}
\label{sec:homogenisation}

Automorphic Lie Algebras are rational objects on a Riemann surface. In this section we take some preparatory steps to aid the transition from polynomial to rational structures.

It is well known that the field of meromorphic, or rational functions on the Riemann sphere $\mero$ can be identified with the field of quotients of forms in two variables of the same degree, as follows. We start by defining
\[\lambda=\frac{X}{Y}\] to identify the Riemann sphere $\overline{\CC}\ni\lambda$ with the projective space $\CC P^1\ni(X,Y)$. 
To a polynomial $p(\lambda)$ of degree $d$ one can relate a form $P(X,Y)$ of the same degree by \begin{equation}
\label{eq:polynomial to form}
p\left(\frac{X}{Y}\right)=Y^{-d}P(X,Y).
\end{equation}
In particular $p(\lambda)=P(\lambda,1)$.
Then, any rational function in $\lambda$ becomes \[\frac{p(\lambda)}{q(\lambda)}=\frac{Y^{-\deg p}P(X,Y)}{Y^{-\deg q}Q(X,Y)}=\frac{Y^{\deg q}P(X,Y)}{Y^{\deg p}Q(X,Y)}\] a quotient of two forms (automatically of the same degree). There are however different forms resulting in the same quotient. The other way around, a quotient of two forms of identical degree $d$ is a rational function of $\lambda$, \[\frac{P(X,Y)}{Q(X,Y)}=\frac{Y^{-d}P(X,Y)}{Y^{-d}Q(X,Y)}=\frac{p(\lambda)}{q(\lambda)}.\]  

M\"obius transformations on $\lambda$ and linear transformation on $(X,Y)$ by the same matrix $g=\begin{pmatrix}a&b\\c&d\end{pmatrix}$ commute with the map  $\phi:(X,Y)\mapsto \lambda$:
\[g\phi(X,Y)=g \lambda=\frac{a\lambda+b}{c\lambda+d}=\frac{aX+bY}{cX+dY}=\phi(aX+bY,cX+dY)=\phi(g(X,Y)).\]
However, two matrices yield the same M\"obius transformation if and only if they are scalar multiples of one another. Therefore we allow the action on $(X,Y)\in \CC^2$ to be projective in order to cover all possible actions on $\overline{\CC}$; we require homomorphisms $\rho:G\rightarrow \PGL(\CC^2)$.

The polynomial framework has some advantages over the rational framework. Most notably, the results of classical invariant theory (cf.~Section \ref{sec:classical invariant theory}) are at our disposal. Trying to exploit this, we will work over the polynomial ring for as long as it is fruitful, before going to the rational functions. 

The transition process from forms to rational functions is called \emph{homogenisation} (even though this term is overused and there is some ambiguity). A form in two variables, which is a polynomial of homogeneous degree in $X$ and $Y$, is \emph{homogenised} if it is divided by a form of the same degree, to obtain an object of degree $0$, or a rational function of $\frac{X}{Y}=\lambda$. Equation (\ref{eq:polynomial to form}) is an example of this process. The denominator determines the location of the poles on the Riemann sphere, e.g.~$\infty$ in the case of (\ref{eq:polynomial to form}). 

By defining two operators, we formalise this process while at the same time creating an intermediate position, which will be our preferred place to work.
Concretely, we define \emph{prehomogenisation} $\p{}$ and \emph{homogenisation} $\h{}$, which will take us from $(V\otimes R)^{G^\flat}$ to $(V~\!\otimes~\!\mero_{\Gamma})^G$ in two steps. By taking just the first step, one can study $(V\otimes \mero_{\Gamma})^G$ while holding on to the degree information of the homogeneous polynomials in $R$. 

\begin{Definition}[Prehomogenisation] 
\label{def:prehomogenisation}
Let $U$ be a vector spaces and $R=\CC[U]$ a polynomial ring.
Define $\p{d}:R\rightarrow R$ to be the linear projection operator   
\[\p{d} P=\left\{
\begin{tabular}{ll}
$P$&$\text{if }d\,|\deg P$,\\
$0$&$\text{otherwise,}$
\end{tabular}
\right.\]
killing all forms of degree not divisible by $d$.
\end{Definition}
Elements of $V$ are considered forms of degree zero in the tensor product $V\otimes R$ with a polynomial ring. The prehomogenisation operator is thus extended $\p{d}:V\otimes R\rightarrow V\otimes R$ sending a basis element $v\otimes P$ to \[\p{d}(v\otimes P)=\p{d}v\otimes \p{d}P=v\otimes \p{d}P.\]
Notice that the prehomogenisation operator is not a morphism of (ring) modules, since it does not respect products. Indeed, one can take two forms such that their degrees are no multiples of $d$ but the sum of their degrees is, so that both forms are annihilated by $\p{d}$, but their product is not. This is the problematic part of the transition between forms and rational functions. But now that this is captured in the prehomogenisation procedure, we can define a second map that \emph{does} behave well with respect to products.

\begin{Definition}[Homogenisation]
\label{def:homogenisation}
Let $\Gamma\subset\overline{\CC}$ be a finite subset and $F_\Gamma$ the form of degree $|\Gamma|$ related to $f(\lambda)=\prod_{\mu\in\Gamma}(\lambda-\mu)$ through (\ref{eq:polynomial to form}), cf.~Definition \ref{def:form of orbit}. If $|\Gamma| \;|\; d$ we define 
\[\h{\Gamma}:\p{d}R\rightarrow \mero_\Gamma\]
to be the linear map sending a form $P$ to 
\[\h{\Gamma} P=\frac{P}{F_\Gamma^r}\]
where $r|\Gamma|=\deg P$.
\end{Definition}
The homogenisation map generalises to $\h{\Gamma}:V\otimes \p{d}R\rightarrow V\otimes\mero_\Gamma$ in the same manner as the prehomogenisation map by acting trivially on $V$.

The next lemma shows one can find all invariant vectors $(V\otimes\mero_{\Gamma})^G$ by considering quotients of invariant vectors and invariant forms whose degrees are multiples of $|G|$.
\begin{Lemma}
\label{lem:only quotients of invariants}
Let $V$ and $\overline{\CC}$ be $G$-modules, $\Gamma\in\bigslant{\overline{\CC}}{G}$ and $\bar{v}\in (V\otimes\mero_{\Gamma})^G$. Then there exists a number $e\in \NN\cup\{0\}$ and an invariant vector $v\in(V\otimes R)^{G^\flat}$ such that $\bar{v}=v F_\Gamma^{-e\nu_\Gamma}$, with $F_\Gamma$ and $\nu_\Gamma$ as in Definition \ref{def:form of orbit}. In particular, the map 
\[\h{\Gamma}:\left(V\otimes \p{d}R\right)^{G^\flat}\rightarrow \left(V\otimes\mero_\Gamma\right)^{G}\] is surjective if $|G|$ divides $d$.
\end{Lemma}
\begin{proof}
If $\bar{v}$ is constant then the statement follows by taking $e=0$ and $v=\bar{v}$. Suppose $\bar{v}$ is not constant. By the identification $\overline{\CC}\cong\CC P^1$, there is a vector $v'\in V\otimes R$ and a polynomial $F'\in R$ of the same homogeneous degree, such that $\bar{v}=\frac{v'}{F'}$, as described in the introduction of this section. The assumption on the poles of $\bar{v}$ gives $\cV(F')\subset\Gamma$ and since $\bar{v}$ is not constant $\cV(F')\ne \emptyset$. The invariance of $\bar{v}$ implies that the order of the poles is constant on an orbit. Therefore $F'=F_\Gamma^p$ for some $p\in\NN$.
Choose $e\in\NN$ such $q=e\nu_\Gamma-p\ge 0$. Put $v=v'F_\Gamma^q$. Then $\bar{v}=\frac{v'}{F'}=\frac{v'}{F_\Gamma^p}=\frac{v'F_\Gamma^q}{F_\Gamma^{p+q}}=\frac{v}{F_\Gamma^{e\nu_\Gamma}}$. By Lemma \ref{lem:finuiinvariant} $F_\Gamma^{\nu_\Gamma}$ is invariant, hence invariance of $\bar{v}$ implies invariance of $v$.

To see that $\h{\Gamma}:\left(V\otimes \p{d}R\right)^{G^\flat}\rightarrow \left(V\otimes\mero_\Gamma\right)^{G}$ is surjective if $d=d'|G|$, notice that one can choose $e$ such that $e=e'd'$, as there is only a lower bound for $e$. Then $\deg v=e|G|=e'd'|G|=e'd$, hence $v\in\left(V\otimes \p{d}R\right)^{G^\flat}$.
\end{proof}

\begin{Example}
\label{ex:automorphic functions}
If we apply Lemma \ref{lem:only quotients of invariants} to the trivial $G$-module $V=\CC$, one obtains the automorphic functions of $G$ with poles on a specified orbit. If $f\in\mero^G_{\Gamma}$ the invariance and pole restriction imply that the denominator of $f$ is a power of $F_{\Gamma}$. By Lemma \ref{lem:only quotients of invariants} we can assume that it is in fact a power of $F_\Gamma^{\nu_\Gamma}$
and that the numerator is a form in $\CC[\fal^{\nal},\fbe^{\nbe},\fga^{\nga}]$. The numerator factors into elements of the two-dimensional vector space $\bigslant{\CC\fal^{\nal}+\CC\fbe^{\nbe}+\CC\fga^{\nga}}{\fal^{\nal}+\fbe^{\nbe}+\fga^{\nga}}$. One line in this space is $\CC F_\Gamma^{\nu_\Gamma}$. Any vector outside this line will generate the ring of automorphic functions, e.g.
\[
\mero^G_{\Gamma}=\CC\left[\h{\Gamma}F_{i}^{\nu_{i}}\right],\qquad \Gamma_i\ne\Gamma\,,
\] 
where we recall that $\h{\Gamma}F_{i}^{\nu_{i}}=\frac{F_{i}^{\nu_{i}}}{F_\Gamma^{\nu_\Gamma}}$.

\end{Example}

We define equivalence classes \emph{$\textrm{mod }F$} on $R$ by 
\[P\textrm{ mod }F = Q\textrm{ mod }F\Leftrightarrow\exists r\in\ZZ : P=F^r Q.\]
That is, $P$ and $Q$ are equivalent $\textrm{mod }F$ if their quotient is a unit in the \emph{localisation} \cite{eisenbud1995commutative} $R_F=\left\{\frac{P}{F^r}\;|\;P\in R, r\in\NN_0\right\}$, where $\NN_0=\NN\cup\{0\}$.
\begin{Lemma}
\label{lem:prehomogenisation mod f is homogenisation}
Let $V$ and $\overline{\CC}$ be $G$-modules and $\Gamma\subset\overline{\CC}$ a $G$-orbit. If $|G|$ divides $d$ then there is an isomorphism of modules
\[\p{d}(V\otimes R)^{G^\flat}\textrm{mod }F_\Gamma \cong \left(V\otimes \mero_\Gamma\right)^G,\]
where $F_\Gamma$ is given in Definition \ref{def:form of orbit}.
\end{Lemma}
\begin{proof}
The linear map $\h{\Gamma}:\p{d}(V\otimes R)^{G^\flat}\rightarrow(V\otimes \mero_\Gamma)^{G}$ respects products and is therefore a homomorphism of modules. Restricted to the ring $\p{d} R$, the kernel $\h{\Gamma}^{-1}(1)$ is the equivalence class of the identity, $1\textrm{ mod }F_\Gamma$, and by Lemma \ref{lem:only quotients of invariants} the map is surjective.
\end{proof}

The prehomogenisation projection $\p{d}$ only concerns degrees. Therefore it makes sense to define it on Poincar\'e series. Since there is little opportunity for confusion, we use the same name,
\[\p{d}:\NN_0[t]\rightarrow\NN_0[t]\,,\quad \sum_{r\ge0}k_r t^r\mapsto\sum_{d\,|\,r}k_r t^r\,.\]
Equivalently, one can define the prehomogenisation of Poincar\'e series by \[\p{d}P(R,t)=P(\p{d}R,t)\] for any graded ring $R$.

\begin{Example}[$\mero^{\DD_N}_{\Gamma}$]
\label{ex:automorphic forms} The case of the trivial representation $V=V_{\triv}$ corresponds to automorphic functions, which were already found in Example \ref{ex:automorphic functions}. This time we apply Lemma \ref{lem:prehomogenisation mod f is homogenisation}.
First we assume that $N$ is odd so that we may assume that $\DD_N^\flat=\DD_N$ and $ R^{\DD_N^\flat}= R^{\DD_N}=\CC[\fal,\fbe]$. Recall that $\dal=\deg \fal=2$ and $\dbe=\deg \fbe=N$ and the two forms are algebraically independent. 
The prehomogenisation projection $\p{2N}$ maps
\[
P\left( R^{\DD_N^\flat},t\right)=\frac{1}{(1-t^2)(1-t^N)}=\frac{(1+t^2+t^4+\ldots+t^{2N-2})(1+t^N)}{(1-t^{2N})^2}\mapsto\frac{1}{(1-t^{2N})^2}
\]
and one finds \[\p{2N}\CC[\fal,\fbe]=\CC[\fal^{\nal},\fbe^{\nbe}]\,.\]
If $N$ is even we may use the Schur cover $\DD_{2N}$ (cf.~Section \ref{sec:binary polyhedral groups}) with invariant forms $ R^{\DD_{2N}}=\CC[\fal,\fbe^{2}]$, which are mapped onto the same ring under $\p{2N}$.

The quotient $\CC[\fal^{\nal},\fbe^{\nbe}]\textrm{ mod }F_\Gamma$ is equivalent to $\CC[F_i^{\nu_i}]$ for $i\in\Omega$ such that $F_i\ne F_\Gamma$. Notice that this ring is isomorphic to $\mero^G_{\Gamma}=\CC\left[\h{\Gamma}{F_{i}^{\nu_{i}}}\right]$, $\Gamma_i\ne\Gamma$, from Example \ref{ex:automorphic functions}.

\end{Example}

\section{Squaring the Ring}
\label{sec:squaring the ring}

This section contains the core of various invariants of Automorphic Lie Algebras.
We start by studying the exponent $\|G^\flat\|$ of the binary polyhedral groups (cf.~Definition \ref{def:exponent of a group}).
\begin{Lemma}
\label{lem:exponent}
Let $G^\flat$ be a binary polyhedral group corresponding to the polyhedral group $G$. Then
\[\|G^\flat\|=2\lcm(\nal, \nbe, \nga)=2\|G\|=\frac{2|G|}{|M(G)|}=\left\{\begin{array}{ll}|G|&\text{if }G\in\left\{\DD_{2M},\TT,\OO,\YY\right\}\\2|G|&\text{if }G\in\left\{\zn{N},\DD_{2M+1}\right\}\end{array}\right.\]
where $M(G)$ is the Schur multiplier given in Theorem \ref{thm:schur multipliers of polyhedral groups} and $\lcm$ stands for \emph{least common multiple}.
\end{Lemma}
\begin{proof}
We will prove that $\|G^\flat\|=2\lcm(\nal, \nbe, \nga)$. The other equalities can be found in Table $\ref{tab:various properties of polyhedral groups}$. The polyhedral group
$G$ is covered by stabiliser subgroups $G_\lambda$, $\lambda\in\overline{\CC}$. Thus, if $\pi:G^\flat\rightarrow G$ is any extension with kernel $Z$,  then $G^\flat$ is covered by the preimages $\pi^{-1}G_\lambda$, $\lambda\in\overline{\CC}$. The order of an element $h\in G^\flat$ thus divides the order of the subgroup $\pi^{-1}G_\lambda$ containing it, which is $|Z|\nu_i$. Hence the exponent divides the least common multiple of these;
\[\|G^\flat\|\;|\;\lcm(|Z|\nal, |Z|\nbe, |Z|\nga)=|Z|\lcm(\nal, \nbe, \nga).\]
If we assume that $G^\flat$ is the binary polyhedral group then it is clear from the presentation $G^\flat=\langle \gal, \gbe, \gga\;|\; \gal^{\nal}=\gbe^{\nbe}=\gga^{\nga}=\gal \gbe \gga\rangle$ and the fact that $\gal\gbe\gga\in\zn{2}$ that the order of $g_i$ is $2\nu_i$, so that
\[2\lcm(\nal, \nbe, \nga)=\lcm(2\nal,2 \nbe, 2\nga)\;|\;\|G^\flat\|\] 
and $\|G^\flat\|=2\lcm(\nal, \nbe, \nga)$ as desired.
\end{proof}
Notice that if we would define $G^\flat$ to be a Schur cover, equal to the binary polyhedral group when possible, then $\|G^\flat\|=|G|$. For now we wish to use some results that are specific to $\SL_2(\CC)$.
Let 
$\chi$ be the character of the natural representation of $G^\flat$, i.e.~the monomorphism $\sigma:G^\flat\rightarrow \SL(\natrep)=\SL_2(\CC)$, and denote its symmetric tensor of degree $h$ by \[\chi_h=\chi_{S^hU}.\]
The \emph{Clebsch-Gordan decomposition} \cite{fossum} for $\SL_2(\CC)$-modules
\begin{equation}
\label{eq:CG}
\natrep\otimes S^h\natrep=S^{h+1}\natrep\oplus S^{h-1}\natrep, \quad h\ge 2,
\end{equation} 
can be conveniently used to find the decomposition of $\chi_h$ when $h$ is a multiple of $\|G^\flat\|$.
To this end, we write the Clebsch-Gordan decomposition in terms of the characters 
\begin{equation}
\label{eq:CGmatrix}
\begin{pmatrix} \chi_h\\\chi_{h-1}\end{pmatrix}=
\begin{pmatrix}\chi&-1\\1&0\end{pmatrix}
\begin{pmatrix} \chi_{h-1}\\\chi_{h-2}\end{pmatrix}
\end{equation}
with boundary conditions
\[\chi_{-1}=0,\qquad \chi_0=\triv.\]
The function $\chi_{-1}$ might not have meaning as a character but it gives a convenient boundary condition resulting in the correct solution, with $\chi_1=\chi$. 


\begin{Lemma}\label{lem:character degree reduction} 
If $g\in G^\flat$ has order $\nu>2$ and $\chi_h$ is the character of the $h$-th symmetric power of the natural representation of $G^\flat$, then \[\chi_{h+\nu}(g)=\chi_h(g)\] for all $h\in\ZZ$.
\end{Lemma}
\begin{proof}
Let $\omega$ and $\omega^{-1}$ be the eigenvalues of $g$'s representative in $\SL_2(\CC)$. In particular $\omega^\nu=1$. The matrix which defines the linear recurrence relation (\ref{eq:CGmatrix}) becomes \[M(g)=\begin{pmatrix}\omega+\omega^{-1}&-1\\1&0\end{pmatrix}.\] We check that this matrix has eigenvalues $\omega^{\pm 1}$ as well.
If the eigenvalues are distinct, i.e.~$\omega^2  \ne 1$, i.e.~$\nu>2$, then $M(g)$ is similar to $\diag(\omega,\omega^{-1})$ and has order $\nu$, proving the lemma. Notice that $M(g)$ is not diagonalisable if $\nu=1$ or $\nu=2$. 
\end{proof}

Recall from Section \ref{sec:binary polyhedral groups} the short exact sequence
$1\rightarrow \zn{2}\rightarrow G^\flat \xrightarrow{\pi} G\rightarrow 1$.
The only maps in  $\SL_2(\CC)$ whose square is the identity are $\pm\Id$. Therefore
\begin{equation} 
\label{lem:order 2 in ker pi}
g\in \ker\pi \Leftrightarrow g^2=1.
\end{equation}
Moreover, if $1\ne z\in\ker\pi$ then its $\SL_2(\CC)$-representative is $-\Id$ and its action on a form $F\in S^hU\cong S^hU^\ast$ is given by \[zF=(-1)^h F.\]
Let $\regchar{G}$ and $\regchar{G^\flat}$ be the characters of the regular representations (cf.~Section \ref{sec:character theory}) of $G$ and $G^\flat$ respectively, i.e.
\[
\pi^\ast \regchar{G}(g)=\left\{\begin{array}{ll} |G|&\pi(g)=1,\\ 0&\pi(g)\ne 1,\end{array}\right.\qquad
\regchar{G^\flat}(g)=\left\{\begin{array}{ll} 2|G|&g=1,\\ 0&g\ne 1,\end{array}\right.\]
where $\pi^\ast$ denotes the pullback of $\pi$, pulling back a map $f$ on $G$ to a map $\pi^\ast f$ on $G^\flat$ given by $\pi^\ast f(g)=f(\pi(g))$.

\begin{Theorem}
\label{thm:dimV}
Let $G$ be a polyhedral group and $G^\flat$ its corresponding binary group, with epimorphism $\pi:G^\flat\rightarrow G$.
Let $\chi_h$ be the character of the $h$-th symmetric power of the natural representation of $G^\flat$, and $\regchar{G}$ and $\triv$ the character of the regular and trivial representation respectively. Then
\begin{align*}
\chi_{m|G|}&=m \pi^\ast \regchar{G}+\triv,\\
\chi_{m|G|-1}&=m (\regchar{G^\flat}-\pi^\ast \regchar{G}),
\end{align*}
for all $m\in\frac{2}{|M(G)|}\NN_0$.
\end{Theorem}
\begin{proof} The statement is trivial for $m=0$. We prove the first equality for $m\in\frac{2}{|M(G)|}\NN$ by evaluating both sides of the equation at each element of $G^\flat$. 

First of all, if $\pi(g)=1$ then $(m \pi^\ast \regchar{G}+\triv)(g)=m|G|+1$. On the other hand, since $m|G|$ is an even number, the action of $g\in\ker\pi$ on $\chi_{m|G|}$ is trivial and $\chi_{m|G|}(g)=\dim \chi_{m|G|}=m|G|+1$.

If $\pi(g)\ne 1$ then $(m \pi^\ast \regchar{G}+\triv)(g)=1$. 
Also, by (\ref{lem:order 2 in ker pi}), $g$ has order $\nu>2$. Since $\nu\,|\,\|G^\flat\|=\frac{2|G|}{|M(G)|}\;|\;m|G|$ by Lemma \ref{lem:exponent}, one can apply Lemma \ref{lem:character degree reduction} to find $\chi_{m|G|}(g)=\chi_0(g)=1$.

The second equality follows in the same manner. Left and right hand sides clearly agree on the trivial group element. Inserting the nontrivial central element $1\ne z\in\ker\pi$ in the right hand side gives $m (\regchar{G^\flat}(z)-\pi^\ast \regchar{G}(z))=m(0-|G|)=-m|G|$. On the other hand, this element $z$ acts as multiplication by $-1$ on forms of odd degree $m|G|-1$, hence $\chi_{m|G|-1}(z)=-\dim \chi_{m|G|-1}=-m|G|$.

Now let $\pi(g)\ne 1$ so that $m (\regchar{G^\flat}(g)-\pi^\ast \regchar{G}(g))=0$.
Again, by (\ref{lem:order 2 in ker pi}), $g$ has order $\nu>2$, and since $\nu\,|\,\|G^\flat\|=\frac{2|G|}{|M(G)|}\;|\;m|G|$ by Lemma \ref{lem:exponent}, one can apply Lemma \ref{lem:character degree reduction} to find $\chi_{m|G|-1}(g)=\chi_{-1}(g)=0$ by the boundary conditions of the recurrence relation (\ref{eq:CGmatrix}).
\end{proof}

\begin{Example}
As an illustration we compute the first twelve symmetric powers of the natural representation of the tetrahedral group, using the Clebsch-Gordan decomposition (\ref{eq:CG}), and list them in Table \ref{tab:decomposition of symmetric powers} (possibly the easiest way to do this is using the \emph{McKay correspondence}, cf.~\cite{dolgachev2009mckay,springer1987poincare}). We check Theorem \ref{thm:dimV} for $G=\TT$ and $m=1$. 
Notice also that all even powers are nonspinorial and all odd powers are spinorial. 

For specific irreducible characters $\psi$, one can find $\psi(1)$ copies at lower degrees. For instance, $\TT_4$ appears twice at degree $6$, $\TT_5$ and $\TT_6$ appear twice at degree $9$, and $\TT_7$ can be found thrice at degree $10$.
\begin{center}
\begin{table}[h!]
\caption{Decomposition of symmetric powers $S^h\btiiii$, $h\le \|\TT\|$.}
\label{tab:decomposition of symmetric powers}
\begin{center}
\begin{tabular}{ccc}\hline
$h$ & $S^h\btiiii$ & $\btiiii\otimes S^{h}\btiiii$  \\
\hline
$-1$ & $0$ & $0$\\
$0$ & $\bti$ & $\btiiii$ \\
$1$ & $\btiiii$ & $\btiiii\btiiii=\bti + \btiiiiiii$\\
$2$ & $\btiiiiiii$ & $\btiiii\btiiiiiii=\btiiii + \btiiiii+\btiiiiii$ \\
$3$ & $\btiiiii+\btiiiiii$ & $\btii+\btiii+2\btiiiiiii$ \\
$4$ & $\btii+\btiii+\btiiiiiii$ & $\btiiii+2\btiiiii+2\btiiiiii$\\
$5$ & $\btiiii+\btiiiii+\btiiiiii$ & $\bti+\btii+\btiii+3\btiiiiiii$\\
$6$ & $\bti+2\btiiiiiii$ & $\btiiii+2(\btiiii+\btiiiii+\btiiiiii)$\\
$7$ & $2\btiiii+\btiiiii+\btiiiiii$ & $2\bti+4\btiiiiiii+\btii+\btiii$\\
$8$ & $\bti+\btii+\btiii+2\btiiiiiii$ & $3(\btiiii+\btiiiii+\btiiiiii)$ \\
$9$ & $\btiiii+2(\btiiiii+\btiiiiii)$ & $\bti+2(\btii+\btiii)+5\btiiiiiii$ \\
$10$& $\btii+\btiii+3\btiiiiiii$ & $4(\btiiiii+\btiiiiii)+3\btiiii$ \\
$11$& $2(\btiiii+\btiiiii+\btiiiiii)$ &  $2(\bti+\btii+\btiii+3\btiiiiiii)$\\
$12$& $2\bti+\btii+\btiii+3\btiiiiiii$ &\\\hline
\end{tabular}
\end{center}
\end{table}
\end{center}
\end{Example}
Theorem \ref{thm:dimV} has many nice consequences. For instance, the first equation in the theorem is equivalent to the following Poincar\'e series of the image of the prehomogenisation operator $\p{m|G|}$ of a space of invariant vectors,
\begin{equation}
\label{eq:prehomogenised generating function}
\p{m|G|}P\left(R^\chi,t\right)=
\left\{
\begin{array}{ll}
0& \text{if $\chi$ is spinorial,}\\
\frac{1+(m-1)t^{m|G|}}{(1-t^{m|G|})^2} & \text{if $\chi=\triv$,}\\
\frac{m\chi(1)^2 t^{m|G|}}{(1-t^{m|G|})^2}  & \text{otherwise,}
\end{array}\right.
\end{equation}
where $m$ is a multiple of $\frac{2}{|M(G)|}$. Notice that$\frac{2}{|M(G)|}=1$ for most of the interesting groups, namely for $\TT,\;\OO,\;\YY$ and $\DD_{2M}$, in which case one can obtain the simplest formula by taking $m=1$.

The Poincar\'e series (\ref{eq:prehomogenised generating function}) suggest, but do not prove, the existence of a set of primary and secondary invariants such that the modules of invariant vectors have the form given in Proposition \ref{prop:dimV generators} below.
The fact that such invariants exist nonetheless can be established using the same technical result that underlies the fact that isotypical components of $R$ are Cohen-Macaulay, namely Proposition \ref{prop:hsop}.
\begin{Proposition}
\label{prop:dimV generators}
Let $G^\flat$ be a binary polyhedral group, $V$ a $G^\flat$-module and $U$ the natural $G^\flat$-module. Let $R=\CC[U]$ be the polynomial ring containing the ground forms $F_i$, $i\in\Omega$ of degree $\frac{|G|}{\nu_i}$. Then
\[\p{m|G|}(V\otimes R)^{G^\flat}=
\left\{
\begin{array}{ll}
0& \text{if $V$ is spinorial,}\\
\CC[F_i^{m\nu_i},F_j^{m\nu_j}](1\bigoplus_{r=1}^{m-1} F_i^{r\nu_i}F_j^{(m-r)\nu_j})& \text{if $V$ is trivial,}\\
\bigoplus_{r=1}^{m\dim V} \CC[F_i^{m\nu_i},F_j^{m\nu_j}]\zeta_r & \text{otherwise,}
\end{array}
\right.
\]
for some invariant vectors $\zeta_1,\ldots,\zeta_{m\dim V}$ of degree $m|G|$.
\end{Proposition}
\begin{proof}
For the case of the trivial representation, recall that $F_i^{\nu_i}$ is an invariant form,  by Lemma \ref{lem:finuiinvariant}, and has degree $|G|$. Therefore \[\CC[F_i^{m\nu_i},F_j^{m\nu_j}](1\bigoplus_{r=1}^{m-1} F_i^{r\nu_i}F_j^{(m-r)\nu_j})\subset\p{m|G|}(V_\triv\otimes R)^{G^\flat}.\] Since $F_i$ and $F_j$ are algebraically independent, the Poincar\'e series (\ref{eq:prehomogenised generating function}) gives equality.

In the last case, recall first that each copy of $V_{\overline{\chi}}$ in $R$ contributes one invariant vector to $(V_\chi \otimes R)^{G^\flat}$, cf.~(\ref{eq:generating function of forms and vectors}).


The original module, before prehomogenisation, is Cohen-Macaulay (Theorem \ref{thm:cohen-macaulay}) \[(V\otimes R)^{G^\flat}=\bigoplus_{r=1}^{k}\CC[\theta_1, \theta_2]\eta_r.\] 
We will show that the prehomogenised module is Cohen-Macaulay as well, allowing us to apply Proposition \ref{prop:hsop}.

The prehomogenisation operator is linear so that it moves through the sum. Moreover, since it merely eliminates certain elements, no relations can be introduced, therefore the sum stays direct.
\[\p{m|G|}(V\otimes R)^{G^\flat}=
\p{m|G|}\bigoplus_{r=1}^{k}\CC[\theta_1, \theta_2]\eta_r=
\bigoplus_{r=1}^{k} \p{m|G|}\CC[\theta_1, \theta_2]\eta_r=
\bigoplus_{r=1}^{k'} \CC[\theta_1^{\nu_1}, \theta_2^{\nu_2}]\tilde{\eta}_r\]
where $m|G|$ divides $\nu_i\deg\theta_i$ (this always holds for some integer $\nu_i$ because $\theta_i$ is a polynomial in ground forms, all whose degrees divide $|G|$) and 
\[\{\tilde{\eta}_1,\ldots\tilde{\eta}_{k'}\}=\{\theta_1^{a_1} \theta_2^{a_2}\eta_j\;|\; m|G|\text{ divides } \deg\theta_1^{a_1} \theta_2^{a_2}\tilde{\eta}_j,\;0\le a_i<\nu_i,\;0\le j\le k\}.\]
This shows that $\p{m|G|}(V\otimes R)^{G^\flat}$ is a free $\CC[\theta_1^{\nu_1}, \theta_2^{\nu_2}]$-module.
Therefore, by Proposition \ref{prop:hsop}, $\p{m|G|}(V\otimes R)^{G^\flat}$ is free over any homogeneous system of parameters. The proof is done if we show that $\{F_i^{m\nu_i},F_j^{m\nu_j}\}$ is a homogeneous system of parameters, i.e.~we need to show that $\p{m|G|}(V\otimes R)^{G^\flat}$ is finitely generated over $\CC[F_i^{m\nu_i},F_j^{m\nu_j}]$. But the parameters $\theta_i^{\nu_i}$ are invariant forms whose degrees divide $m|G|$. Therefore $\CC[\theta_1^{\nu_1}, \theta_2^{\nu_2}]\subset \p{m|G|}R^{G^\flat}=\CC[F_i^{m\nu_i},F_j^{m\nu_j}](1\bigoplus_{r=1}^{m-1} F_i^{r\nu_i}F_j^{(m-r)\nu_j})$ and
\begin{align*}
\p{m|G|}(V\otimes R)^{G^\flat}&=\bigoplus_{r=1}^{k'} \CC[\theta_1^{\nu_1}, \theta_2^{\nu_2}]\tilde{\eta}_r\\
&=\sum_{r=1}^{k'} \CC[F_i^{m\nu_i},F_j^{m\nu_j}](1\bigoplus_{s=1}^{m-1} F_i^{s\nu_i}F_j^{(m-s)\nu_j})\tilde{\eta}_r\\
&=\sum_{r=1}^{k''} \CC[F_i^{m\nu_i},F_j^{m\nu_j}]\tilde{\tilde{\eta}}_r
\end{align*}
where $\{\tilde{\tilde{\eta}}_r\;|\;r=1\ldots k''\}=\{\tilde{\eta}_r,\; F_i^{s\nu_i}F_j^{(m-s)\nu_j}\tilde{\eta}_r\;|\;r=1\ldots k',\; s=1,\ldots, m-1\}$, i.e.~$\{F_i^{\nu_i},F_j^{\nu_j}\}$ is indeed a homogeneous system of parameters. Now, by Proposition \ref{prop:hsop} there exist elements $\zeta_r$ that freely generate $\p{m|G|}(V\otimes R)^{G^\flat}$ over $\CC[F_i^{m\nu_i},F_j^{m\nu_j}]$ and the Poincar\'e series (\ref{eq:prehomogenised generating function}) tells us the number and degrees of these generators.
\end{proof}

\begin{Example}[$R^{\btii}$]
The $\btii$-component of $R=\CC[\tnat]$ can be expressed in terms of ground forms by
\[R^{\btii}=\CC[\fga,\fal\fbe](\fal\oplus \fbe^2)\]
and we have the Poincar\'e series
\begin{align*}
P\left(R^{\btii},t\right)&=\frac{\al t^4+\be^2t^{8}}{(1-\ga t^6)(1-\al\be t^8)}\\
&=\frac{(\al t^4+\be^2t^{8})(1+\ga t^6)(1+\al\be t^8+(\al\be)^2 t^{16})}{(1-\ga^2 t^{12})(1-(\al\be)^3 t^{24})},
\end{align*}
where $\al,\be,\ga$ are dummy symbols added for convenience.
Eliminating all but multiples of $12$ one finds
\[\p{12}P(R^{\btii},t)=\frac{(\al (\al\be) t^{12}+\be^2(\al\be)^2t^{24})}{(1-\ga^2 t^{12})(1-(\al\be)^3 t^{24})}\]
and the corresponding module
\[\p{12}R^{\btii}=\CC[\fga^2,(\fal\fbe)^3](\fal^2\fbe \oplus \fal^2\fbe^4).\]
By Proposition \ref{prop:dimV generators} we can also express this as
\[\p{12}R^{\btii}=\CC[\fal^3,\fbe^3]\fal^2\fbe.\]
\end{Example}
In order to appreciate the result on independence of Proposition \ref{prop:dimV generators} an example where $\chi(1)>1$ would be better suited. Unfortunately, these examples quickly grow out of control and do not make nice reading. A recurring practical problem for this subject.

In the next theorem we combine the results so far to show that, as a module over the automorphic functions, the Automorphic Lie Algebras have a very simple structure.
But first we strip down the notation,
\begin{equation}
\label{eq:automorphic function}
\ii=\h{\Gamma}F_i^{\nu_i}=\frac{F_i^{\nu_i}}{F_\Gamma^{\nu_\Gamma}},\qquad i\in\Omega,
\end{equation}
where we use $\h{\Gamma}$ from Definition \ref{def:homogenisation}. In other words, $\ii$ is a meromorphic function on $\overline{\CC}$ with divisor $\nu_i\Gamma_i-\nu_\Gamma \Gamma$.
Thus we suppress the orbit of poles in the notation but keep the orbit of zeros. Recall that the dependence on the group $G$ was already suppressed in the notation for the ground forms $F_i$.
It is clear that \[\ii=1\Leftrightarrow \Gamma=\Gamma_i.\] If $\ii\ne1$ then it is an example of a \emph{simple automorphic function} \cite{ford1951automorphic} or a \emph{primitive automorphic function} \cite{LM05comm,Lombardo}.

In Section \ref{sec:ground forms}, in particular (\ref{eq:relation ground forms}),  we found that there exist complex numbers $c_i$ and $z_i$ such that  $F_\Gamma^{\nu_\Gamma}=c_{\al}\fal^\nal+c_\be\fbe^\nbe$ and
\[c_\al\ial+c_\be\ibe=1,\qquad z_\al \ial+z_\be\ibe+z_\ga\iga=0. \] Using these relations any of these automorphic functions can be linearly expressed in any other non-constant one. In particular, if $\Gamma_i\ne\Gamma\ne\Gamma_j$, then \[\CC[\II_i]=\CC[\II_j]\] and from now on we will write $\CC[\II]$ for this polynomial ring.

\begin{Theorem}
\label{thm:dimV generators}
If $G<\Aut(\overline{\CC})$ is a finite group acting on a vector space $V$, then the space of invariant vectors \[\left(V\otimes\mero_\Gamma\right)^G\] is a free $\CC[\II]$-module generated by $\dim V$ vectors. In particular this holds for Automorphic Lie Algebras.
\end{Theorem}
\begin{proof}
Most of the work that goes into this proof has been done above. If $G$ is one of the groups $\DD_{2M}$, $\TT$, $\OO$ or $\YY$ we can take the result of Proposition \ref{prop:dimV generators}, with $m=1$, as starting point. In Example \ref{ex:automorphic functions} we showed that $\h{\Gamma}\CC[F_i^{\nu_i},F_j^{\nu_j}]=\CC[\II]$. For nontrivial irreducible $G$-modules $V$ we first apply Lemma \ref{lem:only quotients of invariants} and then Proposition \ref{prop:dimV generators} to find \begin{align*}
(V\otimes\mero_\Gamma)^G&=\h{\Gamma}\p{|G|}(V\otimes R)^{G^\flat}
=\h{\Gamma}\bigoplus_{r=1}^{\dim V}\CC[F_i^{\nu_i},F_j^{\nu_j}]\zeta_r\\
&=\sum_{r=1}^{\dim V}\CC[\h{\Gamma}F_i^{\nu_i},\h{\Gamma}F_j^{\nu_j}]\h{\Gamma}\zeta_r
=\sum_{r=1}^{\dim V}\CC[\II]\bar{\zeta}_r
\end{align*}
where $\bar{\zeta}_r=\h{\Gamma}\zeta_r=\frac{\zeta_r}{F_\Gamma}$. 
The remaining groups, $\zn{N}$ and $\DD_{2M-1}$, equal their own Schur cover by Theorem \ref{thm:schur multipliers of polyhedral groups}, so we may replace $G^\flat$ by $G$ in the above computation and considering the dihedral invariants found in Section \ref{sec:D_N-invariant vectors} one finds the same result, namely $\dim V$ generators $\bar{\zeta}_r$. This prehomogenisation will be carried out explicitly in Example \ref{ex:dihedral alia}.

To prove the claim we need to establish independence of the invariant vectors $\bar{\zeta}_r$ over $\CC[\II]$. Suppose \[p_1(\II)\bar{\zeta}_1+\ldots+p_{\dim V}(\II)\bar{\zeta}_{\dim V}=0, \qquad p_r\in\CC[\II].\]
Evaluated anywhere in the domain $D=\overline{\CC}\setminus(\Gamma\cup\Gal\cup\Gbe\cup\Gga)$, the vectors $\bar{\zeta}_r$ are independent over $\CC$, by Proposition \ref{prop:evaluating invariant vectors}. Hence the functions $p_r$ all vanish on $D$. Because $D$ is not a discrete set in $\overline{\CC}$ and the functions $p_r$ are rational, they must be identically zero.
\end{proof}

\section{Determinant of Invariant Vectors}
\label{sec:determinant of invariant vectors}

In this section use the notion of a \emph{divisor} on a Riemann surface $\Sigma$. A divisor is a formal $\ZZ$-linear combination of points of $\Sigma$. If a meromorphic function $f\in\cM(\Sigma)$ has zeros $\lambda_1,\ldots,\lambda_l\in\Sigma$ of order $d_1,\ldots,d_l\in\NN$ respectively and poles at $\mu_1,\ldots,\mu_m\in\Sigma$ of order $e_1,\ldots,e_m\in\NN$ respectively than the divisor of $f$ is given by 
\[(f)=d_1\lambda_1+\ldots+d_m\lambda_m-e_1\mu_1-\ldots-e_m\mu_m.\]
The results of the previous section allow us to assign a divisor on the Riemann sphere to each character of a polyhedral group. This will play an important role in the development of the theory of Automorphic Lie Algebras in this thesis.
\begin{Definition}[Determinant of invariant vectors]
\label{def:determinant of invariant vectors}
Let $\chi$ be a nontrivial irreducible character of a polyhedral group $G$ and let $\{P_1^i,\ldots,P_{\chi(1)}^i\}$, $i=1,\ldots,\chi(1)$ be bases 
of $G$-modules affording $\chi$, such that $\{P^i_j\;|\;i,j=1,\ldots,\chi(1)\}$ spans $R^\chi_{|G|}$, where $R$ is the polynomial ring on the natural representation of the binary polyhedral group $G^\flat$. Moreover, suppose all of the bases $\{P_1^i,\ldots,P_{\chi(1)}^i\}$ are the same in a $G$-module sense: the matrices representing $G$ with respect to $\{P_1^i,\ldots,P_{\chi(1)}^i\}$ are identical for all $i$. Define
\[\det\left(R^\chi_{|G|}\right)=\CC\det(P^i_j)\]
and extend the definition to reducible characters of $G$ by the rules
\begin{align*}
&\det\left(R^{\triv}\right)=\CC,\\
&\det\left(R^{\chi+\psi}_{|G|}\right)=\det\left(R^{\chi}_{|G|}\right)\det\left(R^{\psi}_{|G|}\right).
\end{align*}
We will call $\det\left(R^\chi_{|G|}\right)$ the \emph{determinant of invariant $\chi$-vectors}.
\end{Definition}
The determinant $\det\left(R^\chi_{|G|}\right)$ of Definition \ref{def:determinant of invariant vectors} is well defined and
\begin{equation}
\label{eq:determinant of invariant vectors is relative invariant}
\det\left(R^\chi_{|G|}\right)\subset R^{\det\chi}_{(\chi(1)-(\chi,\triv))|G|}
\end{equation}
for all characters $\chi$ of $G$, where $\det\chi(g)=\det(\rho_\chi(g))$ if $\rho_\chi:G^\flat\rightarrow \GL(V)$ is the representation affording the character $\chi$.
The degree of $\det\left(R^\chi_{|G|}\right)$ follows immediately and the way the group acts on it follows readily as well. To be explicit, 
let $\{P_1^i,\ldots,P_{\chi(1)}^i\}$, $i=1,\ldots,\chi(1)$, be bases for $G^\flat$-modules satisfying the conditions of Definition \ref{def:determinant of invariant vectors}. The action of $g\in G^\flat$ is
\begin{align*}
g\det(P^i_j)&=\det(gP^i_j)\\
&=\det\left(\rho_{\chi}(g)(P^1_1,\ldots,P^1_{\chi(1)}),\ldots,\rho_{\chi}(g)(P^{\chi(1)}_1,\ldots,P^{\chi(1)}_{\chi(1)})\right)\\
&=\det(\rho_{\chi}(g))\det(P^i_j)=\det\chi(g)\det(P^i_j),
\end{align*}
hence (\ref{eq:determinant of invariant vectors is relative invariant}).

Because we forget about scalar factors in this section it is natural to use the language of divisors.
\begin{Definition}[Divisor of invariant vectors]
If $\chi$ is a character of a finite group $G<\Aut(\overline{\CC})$ and $\Gamma\in\bigslant{\overline{\CC}}{G}$ then
\[(\chi)_\Gamma=\left(\h{\Gamma}\det\left(R^\chi_{|G|}\right)\right)\]
is called the \emph{divisor of invariant vectors.}
\end{Definition}
Notice that $(\triv)_\Gamma=0$ and $(\chi+\psi)_\Gamma=(\chi)_\Gamma+(\psi)_\Gamma$. A general formula for the divisor of invariant vectors can be expressed using the following half integers.
\begin{Definition}[$\kappa(\chi)$]
\label{def:kappa(chi)}
Let $V$ be a module of a polyhedral group $G$ affording the character $\chi$. We define the half integer $\kappa(\chi)_i$ by
\[\kappa(\chi)_i=\nicefrac{1}{2}\;\codim V^{\langle g_i\rangle}=\frac{\chi(1)}{2}-\frac{1}{2\nu_i}\sum_{j=0}^{\nu_i-1}\chi(g_i^j),\qquad i\in\Omega,\]
where $\codim V^{\langle g_i\rangle}=\dim V-\dim V^{\langle g_i\rangle}$.
\end{Definition}
Lemma \ref{lem:dimV^g} translates to 
\begin{equation}\label{eq:summ_i}\sum_{i\in\Omega}\kappa(\chi)_i=\dim V_\chi-\dim V^G_\chi=\chi(1)-(\chi,\epsilon).\end{equation}

\begin{Theorem}
\label{thm:determinant of invariant vectors}
Let $\chi$ be a character of a polyhedral group $G<\Aut(\overline{\CC})$ and let $\Gamma\in\bigslant{\overline{\CC}}{G}$. Then $\chi$ is real valued if and only if
\[(\chi)_\Gamma=\sum_{i\in\Omega}\kappa(\chi)_i(\nu_i\Gamma_i-\nu_\Gamma \Gamma),\]
i.e.~$\det\left(R^\chi_{|G|}\right)=\CC\prod_{i\in\Omega}F_i^{\nu_i \kappa(\chi)_i}$ and $\h{\Gamma}\det\left(R^\chi_{|G|}\right)=\CC\prod_{i\in\Omega}\ii^{\kappa(\chi)_i}$. Here we identify a subset $\Gamma\subset\overline{\CC}$ with the divisor $\sum_{\gamma\in\Gamma}\gamma$.
\end{Theorem}
\begin{proof}
Let $g\in G$ and $\omega\in\CC^\ast$ have order $\nu$. If $\tau:G\rightarrow\GL(V)$ is the representation affording a real valued character $\chi$ then the eigenvalues of $\tau(g)$ are powers of $\omega$ and the nonreal eigenvalues come in conjugate pairs. We denote the multiplicity of the eigenvalue $1$ and $-1$ by $\kappa_+$ and $\kappa_{-}$ respectively and the number of conjugate pairs by $\kappa_c$. That is, $\kappa_+=\dim V^{\langle g\rangle}$, $\kappa_++\kappa_-=\dim V^{\langle g^2\rangle}$ and $\kappa_++\kappa_-+2\kappa_c=\chi(1)$. If $g=g_i$ then $\nicefrac{1}{2}\,\kappa_-+\kappa_c=\kappa(\chi)_i$. 

There is a basis for $V$ such that \[\tau(g)=\diag(1,\ldots,1,-1,\ldots,-1,\omega^{j_1},\omega^{-j_1},\ldots,\omega^{j_{\kappa_c}},\omega^{-j_{\kappa_c}}).\] 
Denote this basis by $\{u_1,\ldots, u_{\kappa_+},v_1,\ldots, v_{\kappa_-},w_1,x_1,\ldots,w_{\kappa_c},x_{\kappa_c}\}$. Let the meromorphic invariant vectors be $\left(V_\chi\otimes\mero_\Gamma\right)^G=\CC[\II]\langle \bar{v}_1,\ldots, \bar{v}_{\chi(1)}\rangle$ and express each generator $\bar{v}_r$ in the aforementioned basis
\begin{multline}
\bar{v}_r=f^r_1 u_1+\ldots+ f^r_{\kappa_+}u_{\kappa_+}+h^r_1v_1+\ldots+h^r_{\kappa_-}v_{\kappa_-}\\+p^r_1 w_1+q^r_1 x_1+\ldots+p^r_{\kappa_c} w_{\kappa_c}+q^r_{\kappa_c} x_{\kappa_c}
\end{multline}
where $f^r_s,\:h^r_s,\;p^r_s,\;q^r_s\in\mero_\Gamma$. Notice that $\h{\Gamma}\det\left(R^\chi_{|G|}\right)$ is given by
\begin{equation}
\label{eq:determinant of invariant vectors}
\begin{vmatrix}
f^1_1&\cdots& f^1_{\kappa_+}&h^1_1&\cdots& h^1_{\kappa_-}&p^1_1 &q^1_1&\cdots&p^1_{\kappa_c} & q^1_{\kappa_c} \\
\vdots&&\vdots&\vdots&&\vdots&\vdots&\vdots&&\vdots&\vdots\\
f^{\chi(1)}_1&\cdots& f^{\chi(1)}_{\kappa_+}&h^{\chi(1)}_1&\cdots& h^{\chi(1)}_{\kappa_-}&p^{\chi(1)}_1 &q^{\chi(1)}_1&\cdots&p^{\chi(1)}_{\kappa_c} & q^{\chi(1)}_{\kappa_c} 
\end{vmatrix}\end{equation} up to scalar.

Let $\mu\in\overline{\CC}$ and $\langle g \rangle=G_\mu$. By Proposition \ref{prop:evaluating invariant vectors} $\left(V_\chi\otimes\mero_\Gamma\right)^G(\mu)=V^{\langle g \rangle}$ hence $h^r_s(\mu)=p^r_s(\mu)=q^r_s(\mu)=0$. We fix one meromorphic function $h^r_s$. There exists a local coordinate transformation $\phi:U\rightarrow V\subset\CC$ from a neighbourhood $U$ of $\mu$ on the Riemann sphere to a neighbourhood $V$ of $\phi(\mu)=0$ that transforms the function $h^r_s$ to $\tilde{h}^r_s(t)=h^r_s\circ \phi^{-1}(t)=t^d$. This follows from the well known characterisation of the local behaviour of holomorphic mappings. We are interested in the order $d$ of the zero. Since $g(\mu)=\mu$ we can define the map $\tilde{g}=\phi\circ g\circ\phi^{-1}$ between two neighbourhoods of $0$ in $V$, and on a small enough neighbourhood we have $\tilde{g}^\nu=\id$.

If we plug in the invariance assumption $h^r_s\circ g^{-1}(\lambda)=-h^r_s(\lambda)$ we find, locally, $\tilde{h}^r_s\circ \tilde{g}^{-1}(t)=-\tilde{h}^r_s(t)$, that is $(\tilde{g}^{-1}(t))^d=-t^d$. Therefore $\tilde{g}$ is linear, and given its order $\nu$ we conclude $\tilde{g}^{-1}(t)=\omega^{-k}t$ where $\gcd(k,\nu)=1$. Again using $(\tilde{g}^{-1}(t))^d=-t^d$ shows $dk\in \nicefrac{\nu}{2}+\ZZ\nu\subset \ZZ\nicefrac{\nu}{2}$. In particular, $d\ge \nicefrac{\nu}{2}$ (since $\gcd(k,\nicefrac{\nu}{2})=1$) and each column $h^\cdot_s$ in (\ref{eq:determinant of invariant vectors}) adds at least $\frac{\nu}{2}\mu$ to the divisor $(\chi)_\Gamma$.

A similar local description $\tilde{p}^r_s(t)=t^d$ of $p^r_s$ near $\mu$ leads to
\[(\omega^{-k} t)^d=\tilde{p}^r_s(\tilde{g}^{-1}(t))=\omega^{-j_s} \tilde{p}^r_s=\omega^{-j_s}t^d,\] which implies $kd-j_s\in\ZZ\nu$. Repeating this trick for a function $q^{r'}_s$ shows it has a zero at $\lambda=\mu$ of order $d'$ such that  $kd'+j_s\in\ZZ\nu$. 
Therefore $k(d+d')\in\ZZ\nu$ and since $\gcd(k,\nu)=1$ the product $p^r_s q^{r'}_s$ has a zero of order $d+d'\ge\nu$.
In particular, the pair of columns $p^\cdot_s$ and $q^{\cdot}_s$ in equation (\ref{eq:determinant of invariant vectors}) adds at least $\nu\mu$ to the divisor $(\chi)_\Gamma$. Combined with the functions $h^r_s$ the coefficient of $\mu$ in $(\chi)_\Gamma$ is at least $\left(\nicefrac{1}{2}\,\kappa_-+\kappa_c\right)\nu$.


The determinant of invariant vectors has the form $\det\left(R^\chi_{|G|}\right)=\prod_{i\in\Omega}F_i^{\delta_i}$, for some orders $\delta_i\in\NN_0$, since it is a relative invariant, cf.~(\ref{eq:determinant of invariant vectors is relative invariant}), and by Proposition \ref{prop:evaluating invariant vectors} can only vanish on an exceptional orbit. By the above we have $\delta_i\ge \nu_i \kappa(\chi)_i$. Now we show that equality must occur because $\sum_{i\in\Omega}\frac{\delta_i}{\nu_i}=\sum_{i\in\Omega}\kappa(\chi)_i$. Indeed  $\sum_{i\in\Omega}\frac{\delta_i}{\nu_i}=\sum_{i\in\Omega}\frac{\delta_i\deg{F_i}}{|G|}=|G|^{-1}\deg \det\left(R^\chi_{|G|}\right)=\chi(1)-(\chi,\triv)$ and this equals $\sum_{i\in\Omega}\kappa(\chi)_i$ by equation (\ref{eq:summ_i}).

There are precisely two nonreal valued irreducible characters of polyhedral group: $\btii$ and $\btiii$. 
Example \ref{ex:determinant of invariant forms} shows that here the formula of the theorem does not apply.
\end{proof}

The numbers $\kappa(\chi)_i$ are important for Automorphic Lie Algebras due to Theorem \ref{thm:determinant of invariant vectors}. We list them in Table \ref{tab:1/2codimV^g}. Notice that  $\nu_i \kappa(\chi)_i$, $i\in\Omega$, are integers if and only if $\chi$ is real valued (cf.~Section \ref{sec:invariant bilinear forms} and \ref{sec:characters}).
\begin{center}
\begin{table}[h!] 
\caption{Half integers $\kappa(\chi)_i$, in rows $3$, $4$ and $5$, and $\nu_i \kappa(\chi)_i$ in rows $6$, $7$ and $8$.}
\label{tab:1/2codimV^g}
\begin{center}
\begin{tabular}{cccccccccccccccccccc} \hline
$ $&$\chi_2$&$\chi_3$&$\chi_4$&$\psi_j$&$\btii$&$\btiii$&$\btiiiiiii$&$
\boii$&$\boiii$&$\boiiiiii$&$\boiiiiiii$&$\byiiii$&$\byiiiii$&$\byiiiiii$&$\byiiiiiiii$\\
\hline
 &$1$&$1$&$1$&$2$&$1$&$1$&$3$&$1$&$2$&$3$&$3$&$3$&$3$&$4$&$5$\\
\hline
$\al$&$0$&$\nicefrac{1}{2}$&$\nicefrac{1}{2}$&$1$&$\nicefrac{1}{2}$&$\nicefrac{1}{2}$&$1$&

$\nicefrac{1}{2}$&$\nicefrac{1}{2}$&$\nicefrac{3}{2}$&$1$&$1$&$1$&$2$&$2$\\

$\be$&$\nicefrac{1}{2}$&$\nicefrac{1}{2}$&$0$&$\nicefrac{1}{2}$&$\nicefrac{1}{2}$&$\nicefrac{1}{2}$&$1$

&$0$&$1$&$1$&$1$&$1$&$1$&$1$&$2$\\

$\ga$&$\nicefrac{1}{2}$&$0$&$\nicefrac{1}{2}$&$\nicefrac{1}{2}$&$0$&$0$&$1$&

$\nicefrac{1}{2}$&$\nicefrac{1}{2}$&$\nicefrac{1}{2}$&$1$&$1$&$1$&$1$&$1$\\

\hline 

$\al$&$0$&$\nicefrac{N}{2}$&$\nicefrac{N}{2}$&$N$&$\nicefrac{3}{2}$&$\nicefrac{3}{2}$&$3$&

$2$&$2$&$6$&$4$&$5$&$5$&$10$&$10$\\

$\be$&$1$&$1$&$0$&$1$&$\nicefrac{3}{2}$&$\nicefrac{3}{2}$&$3$

&$0$&$3$&$3$&$3$&$3$&$3$&$3$&$6$\\

$\ga$&$1$&$0$&$1$&$1$&$0$&$0$&$2$&

$1$&$1$&$1$&$2$&$2$&$2$&$2$&$2$\\

\hline 
\end{tabular}
\end{center}
\end{table}
\end{center}
We end this chapter calculating various determinants of invariant vectors.
\begin{Example}[Determinants of $\DD_N$-invariant vectors]\label{ex:determinant of dihedral invariant vectors}
For the dihedral group we have explicit descriptions for the isotypical components $R^\chi$ at hand, see (\ref{eq:Fs}) and Table \ref{tab:invariants, N odd} and Table \ref{tab:invariants, N even}. Therefore we can readily check Theorem \ref{thm:determinant of invariant vectors} for this group by computing all the determinants and compare in each case the exponents of $\fal$, $\fbe$ and $\fga$ to the relevant column in Table \ref{tab:1/2codimV^g}.  First notice that all representations of $\DD_N$ are of real type.


We start with $\chi_2$. If $N$ is odd then $R^{\chi_2}_{2N}=\left(\CC[\fal,\fbe]\fga\right)_{2N}=\CC\fbe\fga$. If $N$ is even, we use the extension $G^\flat=\DD_{2N}$ and also find $R^{\chi_2}_{2N}=\left(\CC[\fal,\fbe^2]\fbe\fga\right)_{2N}=\CC\fbe\fga$. The exponents $(0,1,1)$ are indeed identical to $(\nal \kappa(\chi_2)_\al,\, \nbe \kappa(\chi_2)_\be,\, \nga \kappa(\chi_2)_\ga)$ as found in the $\chi_2$-column of Table \ref{tab:1/2codimV^g}.


For $\chi_3$ and $N$ even, $G^\flat=\DD_{2N}$, one finds $R^{\chi_3}_{2N}=\left(\CC[\fal,\fbe^2]\fbe\right)_{2N}=\CC\fal^{\frac{N}{2}}\fbe$, and the last linear character gives $R^{\chi_4}_{2N}=\left(\CC[\fal,\fbe^2]\fga\right)_{2N}=\CC\fal^{\frac{N}{2}}\fga$.

Now we consider the two-dimensional representations, when $N$ is odd;
\begin{align*}
\det R^{\psi_j}_{2N}&=\det\left(\CC[\fal,\fbe]\begin{pmatrix}X^j&Y^j\\Y^{N-j}&X^{N-j}\end{pmatrix}\right)_{2N}\\
&=\left\{
\begin{array}{ll}
\CC\det\begin{pmatrix}\fal^{\frac{N-j}{2}}\fbe X^j&\fal^{\frac{N-j}{2}}\fbe Y^j\\\fal^{\frac{N+j}{2}} Y^{N-j}&\fal^{\frac{N+j}{2}}X^{N-j}\end{pmatrix}= \CC\fal^N\fbe\fga& j\text{ odd,} \\
\CC\det\begin{pmatrix}\fal^{\frac{2N-j}{2}} X^j&\fal^{\frac{2N-j}{2}}Y^j\\\fal^{\frac{j}{2}}\fbe Y^{N-j}&\fal^{\frac{j}{2}}\fbe X^{N-j}\end{pmatrix}=\CC\fal^N\fbe\fga& j\text{ even.}
\end{array}\right.
\end{align*}

If $N$ is even and $G^\flat=\DD_{2N}$, we only consider $2j$ because the statement of Theorem \ref{thm:determinant of invariant vectors} only concerns representations of $G$, not for instance spinorial representation of $G^\flat$;
\begin{align*}
\det R^{\psi_{2j}}_{2N}&=\det\left(\CC[\fal,\fbe^2]\begin{pmatrix}X^{2j}&Y^{2j}\\Y^{2N-{2j}}&X^{2N-{2j}}\end{pmatrix}\right)_{2N}\\
&=
\CC\det\begin{pmatrix}\fal^{N-j} X^{2j}&\fal^{N-j} Y^j\\\fal^{j} Y^{2N-2j}&\fal^{j}X^{2N-2j}\end{pmatrix}=\CC\fal^N (X^{2N}-Y^{2N})=\CC\fal^N\fbe\fal.
\end{align*}
This confirms Theorem \ref{thm:determinant of invariant vectors} for all $\DD_N$-representations.
\end{Example}

The other examples that are feasible by hand are the remaining one-dimensional characters. These include the only representations of the polyhedral groups that are not real valued, namely $\btii$ and $\btiii$, cf.~Section \ref{sec:characters}. 
\begin{Example}[One-dimensional characters]\label{ex:determinant of invariant forms}
Let the characters of the ground forms be indexed by $\Omega$, \[gF_i=\chi_i(g)F_i,\qquad i\in\Omega.\] 
For the tetrahedral group, the degrees of the ground forms are $\dal=\dbe=4$ and $\dga=6$, cf.~Table \ref{tab:various properties of polyhedral groups}. Necessarily $\chi_\al=\btii$ and $\chi_\be=\btiii$, or the other way around. There is no way to distinguish (but they cannot be equal because there is a degree 8 invariant, cf.~(\ref{eq:gf of invariant forms})) so we take the first choice. The last ground form is invariant, $\chi_\ga=\bti$, because $\chi_\ga^\nga=\bti$,  $\nga=2$, and the only element in $\cA\TT=\zn{3}$ that squares to the identity is the identity. Thus we get
\begin{align*}
&R^{\btii}_{12}=\left(\CC[\fal\fbe,\fga](\fal\oplus \fbe^2)\right)_{12}=\CC\fal^2\fbe,\\
&R^{\btiii}_{12}=\left(\CC[\fal\fbe,\fga](\fal^2\oplus\fbe)\right)_{12}=\CC\fal\fbe^2,\\
&\det R^{\btii+\btiii}_{12}=R^{\btii}_{12}R^{\btiii}_{12}=\CC\fal^3\fbe^3.
\end{align*}
We see that the formula of Theorem \ref{thm:determinant of invariant vectors} indeed does not hold for the characters $\btii$ and $\btiii$ of complex type, but it does for the real valued character $\btii+\btiii$.

The remaining one-dimensional character to check is $\boii$. We have $\chi_\al=\boii$, $\chi_\be=\boi$ and $\chi_\ga=\boii$ and
$R^{\boii}_{24}=\left(\CC[\fal^2,\fbe](\fal\oplus\fga)\right)_{24}=\CC\fal^2\fga$,
as promised.
\end{Example}


\chapter[Group Actions and Lie Brackets]{Group Actions and Lie Brackets}\label{ch:B}

In this chapter we consider Lie algebras, represented by endomorphisms on a finite dimensional vector space. Taking this space $V$ to be a module of a group $G$, one has an induced action of $G$ on the Lie algebra of all endomorphisms $\End(V)\cong V\otimes V^\ast$. Now it is important to know which Lie subalgebras $\mf{g}(V)<\End(V)$ are also submodules of $G$. Or the other way around, which $G$-submodules of $V\otimes V^\ast$ are Lie algebras? Let us start with an example. 
\begin{Example}
Let $V$ be an irreducible representation of the dihedral group $\DD_N$. We are interested in subspaces $\mf{g}(V)<\mf{gl}(V)$ which are both a Lie subalgebra and a $\DD_N$-submodule. 
The dimension of $V$ is $1$ or $2$ (cf.~Section \ref{sec:characters}). In the first case the action on $\mf{gl}(V)$ is trivial and the Automorphic Lie Algebras $\alia{g}{V}{\Gamma}=\mero^G_\Gamma$ are the rings of automorphic functions which are analytic outside $\Gamma$. 

Let $V$ now be a two-dimensional $\DD_N$-module affording the character $\psi_j$. Then the $\DD_N$-module $\mf{gl}(V)$ has character $\psi_j^2=\chi_1+\chi_2+\psi_{2j}$ and all $\DD_N$-submodules $\mf{g}(V)<\mf{gl}(V)$ are given by a subset of these characters. 
In the basis corresponding to (\ref{eq:standard matrices}), the decomposition is given by
\[
\mf{gl}(V)^{\chi_1}=\splitk \begin{pmatrix}1&0\\0&1\end{pmatrix},\quad
\mf{gl}(V)^{\chi_2}=\splitk \begin{pmatrix}1&0\\0&-1\end{pmatrix},\quad
\mf{gl}(V)^{\psi_{2j}}=\splitk \begin{pmatrix}0&1\\0&0\end{pmatrix}
\oplus \splitk \begin{pmatrix}0&0\\1&0\end{pmatrix}.
\]

Now we add the condition that $\mf{g}(V)$ be a Lie algebra. To this end we compute some commutator brackets. Clearly $\mf{gl}(V)^{\chi_1}\subset Z(\mf{gl}(V))$ and for the other components we find
\begin{align*}
&[\mf{gl}(V)^{\chi_2},\mf{gl}(V)^{\chi_2}]=0,\\ 
&[\mf{gl}(V)^{\chi_2},\mf{gl}(V)^{\psi_{2j}}]=\mf{gl}(V)^{\psi_{2j}},\\
&[\mf{gl}(V)^{\psi_{2j}},\mf{gl}(V)^{\psi_{2j}}]=\mf{gl}(V)^{\chi_{2}}.
\end{align*}
The only restriction on the submodule $\mf{g}(V)$ coming from the requirement that it be a Lie subalgebra, is that if $\mf{g}(V)$ contains the $\psi_{2j}$-summand, then it also contains the $\chi_2$-summand.

Lie algebras of all dimensions $\le 4$ are available. The only noncommutative cases are $\mf{sl}(V)$, affording $\chi_2+\psi_{2j}$, and $\mf{gl}(V)$. 
Since we have a Lie algebra direct sum \[\alia{gl}{V}{\Gamma}=\mero^G_\Gamma\Id\oplus\alia{sl}{V}{\Gamma}\] considering $\mf{sl}(V)$ is sufficient for Automorphic Lie Algebras with dihedral symmetry.
\end{Example}

The two conditions for the linear subspace $\mf{g}(V)<\End(V)$ are independent. In the example we have seen $G$-submodules which are not Lie algebras, e.g.~$\mf{gl}(V)^{\psi_{2j}}$. As an example of a Lie subalgebra of $\End(V)$ which is not a submodule, one can consider for instance a one-dimensional subspace of $\mf{gl}(V)^{\psi_{2j}}$, or something more interesting such as $\mf{so}(V)<\End(V)$ when $V$ is not a representation of real type (cf.~Section \ref{sec:invariant bilinear forms}).

In Section \ref{sec:group decomposition of simple lie algebras} we will study the spaces $\mf{g}(V)$ assuming they have the two structures we need. The results discussed in this section are especially useful for the purpose of computing Automorphic Lie Algebras explicitly. Section \ref{sec:inner automorphisms} explains why the restriction to automorphisms on $\mf{g}(V)$ induced by $V$ is not so severe as it might seem, as polyhedral groups can only act by \emph{inner automorphisms} on many of the classical simple Lie algebras. This chapter will be concluded with Section \ref{sec:evaluating alias} where we describe the Lie algebras of $G$-invariant matrices over $\mero$, evaluated in a point of the Riemann sphere. In order to understand the final chapter and the main results of this thesis one only needs Lemma \ref{lem:determinant of automorphism} and Theorem \ref{thm:evaluated alias} of this chapter.

\section{Group Decomposition of Simple Lie Algebras}
\label{sec:group decomposition of simple lie algebras}

If a space $\mf{g}(V)$ is assumed to be a Lie algebra and $G$-module, one can deduce a few facts about the relation between these structures, which is the content of this section.
We start with a lemma on how to extract the group-module structure from the Lie algebra structure.
\begin{Lemma}
\label{lem:characters of Lie algebras} Let $G$ be a finite group, $V$ a $G$-module and suppose a Lie algebra $\mf{g}<V\otimes V^\ast$ is also preserved by $G$. If $\mf{g}$ is $\mf{sl}(V)$, $\mf{so}(V)$ or $\mf{sp}(V)$ then the character of this representation is given by 
\begin{equation}
\label{eq:characters of Lie algebras}
\begin{array}{l}
\chi_{\mf{sl}(V)}(g)=\chi_V(g)\overline{\chi_V}(g)-1,\\
\chi_{\mf{so}(V)}(g)=\nicefrac{1}{2}\left(\chi_V(g)^2-\chi_V(g^2)\right),\\
\chi_{\mf{sp}(V)}(g)=\nicefrac{1}{2}\left(\chi_V(g)^2+\chi_V(g^2)\right),
\end{array}
\end{equation}
for all $g\in G$, respectively. 
\end{Lemma}
\begin{proof}
Finding the first character takes little effort. Indeed, $\mf{gl}(V)=\End(V)\cong V\otimes V^*$ has character $\chi_V\overline{\chi_V}$, and the identity matrix is invariant.

The other characters are harder to find, but we can do both of them at once.
Consider the Lie algebra defined by a nondegenerate bilinear form $B$, which we represent by an element in $\GL(V)$; \[\mf{g}(V)=\{A\in \End(V)\;|\; A^TB+BA=0\}.\] If $B$ is symmetric ($B^T=B$) we say $\mf{g}(V)=\mf{so}(V)$ and if $B$ is antisymmetric ($B^T=-B$) we say $\mf{g}(V)=\mf{sp}(V)$. The defining condition for the Lie algebra can be rephrased as a symmetry condition on $BA$ for $A\in \mf{g}(V)$. Indeed, $(BA)^T=A^TB^T=\pm A^TB=\mp BA$, where the last equality defines the Lie algebra, and the choice of sign determines the choice between orthogonal and symplectic Lie algebras.
Now we have \[B\mf{g}(V)=\{BA\;|\;A\in\mf{g}(V)\}=\{M\in \End(V)\;|\;M^T=\pm M\}\]
i.e.~$B\mf{so}(V)$ is the space of antisymmetric matrices and $B\mf{sp}(V)$ is the space of symmetric matrices. As a vector space, $B\mf{g}(V)$ is isomorphic to $\mf{g}(V)$ because $B$ is nondegenerate. We can define an action of $G$ on $B\mf{g}(V)$ by requiring $B\mf{g}(V)$ to be isomorphic to $\mf{g}(V)$ as a representation. This gives the following result.

The representation $\tau:G\rightarrow \GL(V)$ is such that $\mf{g}(V)$ is a submodule of $V\otimes V^\ast$ if and only if $\tau_g^T B\tau_g=B$ for all elements $g\in G$. Therefore, the induced action on $M=BA\in B\mf{g}(V)$ reads
\[g\cdot M=g\cdot BA=Bg\cdot A=B\tau_g A \tau_g^{-1}=\tau_g^{-T} BA \tau_g^{-1}=\tau_g^{-T} M \tau_g^{-1},\]
where we use a shorthand notation $\tau_g^{-T}=(\tau_g^{T})^{-1}=(\tau_g^{-1})^T$.

Consider a basis $\{e_i\}$ for $V$ diagonalising $\tau_g=\diag(\mu_i)$ for a fixed group element $g\in G$. We find a diagonal action of $g$ on the basis $\{E_{i,j}-E_{j,i}\;|\; i<j\}$ for $B\mf{so}(V)$ or on the basis $\{E_{i,j}+E_{j,i}\;|\; i\le j\}$ for $B\mf{sp}(V)$ given by multiplication by $\overline{\mu_i\mu_j}$. Hence the trace is respectively given by \[\sum_{i<j}\overline{\mu_i\mu_j}=\frac{1}{2}\left(\left(\sum_{i}\overline{\mu_i}\right)^2-\sum_{i}\overline{\mu_i}^2\right)=\frac{1}{2}\left(\overline{\chi_V(g)}^2-\overline{\chi_V(g^2)}\right)\]
and
\[\sum_{i\le j}\overline{\mu_i\mu_j}=\frac{1}{2}\left(\left(\sum_{i}\overline{\mu_i}\right)^2+\sum_{i}\overline{\mu_i}^2\right)=\frac{1}{2}\left(\overline{\chi_V(g)}^2+\overline{\chi_V(g^2)}\right).\]
Finally, notice that $\chi_V(g)$ is real. Indeed, this is a basic fact about orthogonal and symplectic matrices, $\chi_V(g)=\tr \tau_g=\tr(B^{-1}\tau_g^{-T}B)=\tr \tau_g^{-1}=\overline{\chi_V(g)}$, thus the proof is complete.
\end{proof}
With the information of Section \ref{sec:characters} on the characters of the binary polyhedral groups and Equation (\ref{eq:characters of Lie algebras}) we can calculate the character decompositions of the Lie algebras of our interest. After an example we present all decompositions in Table \ref{tab:group decompositions of simple lie algebras}. Notice that all irreducible summands occurring in the Lie algebras are nonspinorial, as expected.
\begin{Example}[The character of $\mf{sp}(\boiiiiiiii)$]
By Schur's Lemma there is no invariant in $\mf{sp}(\boiiiiiiii)$, that is, $(\mf{sp}(\boiiiiiiii),\boi)=0$. Moreover, because the actions on the Lie algebras are nonspinorial, by design, we also know that $(\mf{sp}(\boiiiiiiii),\bo_j)=0$. Thus, we can say beforehand that 
\begin{equation}
\label{eq:decomposition of sp(O8)}
\mf{sp}(\boiiiiiiii)=n_2\boii+n_3\boiii+n_6\boiiiiii+n_7\boiiiiiii
\end{equation}
where the numbers $n_j$ are nonnegative integers.

Since we only need to find $4$ integers it will be sufficient to find as many equations (and likely one can do with less). We proceed using Formula (\ref{eq:characters of Lie algebras}) and the character table of the binary octahedral group, Table \ref{tab:ctbo}. At the trivial group element, (\ref{eq:characters of Lie algebras}) reads  $\chi_{\mf{sp}(\boiiiiiiii)}(1)=\nicefrac{1}{2}(4^2+4)=10$ and at $g=\gga$ it says $\chi_{\mf{sp}(\boiiiiiiii)}(\gga)=\nicefrac{1}{2}(0-4)=-2$. The next column of the character table belongs to the conjugacy class $[\gbe^2]$. In order to use (\ref{eq:characters of Lie algebras}) one must first figure out which class $\gbe^4=z\gbe$ belongs to. Alternatively, one can skip this column: $4$ equations are sufficient anyway. For the group elements $\gal^2$ and $\gal$ one can immediately compute $\chi_{\mf{sp}(\boiiiiiiii)}(\gal^2)=\nicefrac{1}{2}(0-4)=-2$ and $\chi_{\mf{sp}(\boiiiiiiii)}(\gal)=\nicefrac{1}{2}(0-0)=0$. By evaluating (\ref{eq:decomposition of sp(O8)}) in the group elements $1$, $\gga$, $\gal^2$ and $\gal$ respectively, one finds the system of equations
\[\begin{pmatrix}1&2&3&3\\-1&0&1&-1\\1&2&-1&-1\\-1&0&-1&1\end{pmatrix}\begin{pmatrix}n_2\\n_3\\n_6\\n_7\end{pmatrix}=\begin{pmatrix}10\\-2\\-2\\0\end{pmatrix}\]
which has the unique solution $(n_2, n_3, n_6, n_7)=(1,0,1,2)$, cf.~Table \ref{tab:group decompositions of simple lie algebras}.
\end{Example}

\begin{center}
\begin{table}[h!] 
\caption{Character decompositions of simple Lie algebras $\mf{g}(V)$, ordered by $\dim V$.}
\label{tab:group decompositions of simple lie algebras}
\begin{center}
\begin{tabular}{ll} \hline
$2$&$
\begin{array}{lccccccccc}
V&\psi_j&\btiiii&\btiiiii&\btiiiiii&\boiii&\boiiii&\boiiiii&\byii&\byiii\\
\mf{sl}(V)&\chi_2+\psi_{2j}&\btiiiiiii&\btiiiiiii&\btiiiiiii&\boii+\boiii&\boiiiiiii&\boiiiiiii&\byiiiii&\byiiii\\
\end{array}$\\
\hline 
$3$&$
\begin{array}{lccccccc}
V&\btiiiiiii&\boiiiiii&\boiiiiiii&\byiiii&\byiiiii\\
\mf{sl}(V)&\btii+\btiii+2\btiiiiiii&\boiii+\boiiiiii+\boiiiiiii&\boiii+\boiiiiii+\boiiiiiii&\byiiii+\byiiiiiiii&\byiiiii+\byiiiiiiii\\
\mf{so}(V)&\btiiiiiii&\boiiiiiii&\boiiiiiii&\byiiii&\byiiiii\\
\end{array}$\\
\hline
$4$&$
\begin{array}{lccccccccc}
V&\boiiiiiiii&\byiiiiii&\byiiiiiii\\
\mf{sl}(V)&\boii+\boiii+2\boiiiiii+2\boiiiiiii&\byiiii+\byiiiii+\byiiiiii+\byiiiiiiii&\byiiii+\byiiiii+\byiiiiii+\byiiiiiiii\\
\mf{so}(V)&&\byiiii+\byiiiii&\\
\mf{sp}(V)&\boii+\boiiiiii+2\boiiiiiii&&\byiiii+\byiiiii+\byiiiiii\\
\end{array}$\\
\hline
$5$&$
\begin{array}{lccccccccccc}
V&\byiiiiiiii\\
\mf{sl}(V)&\byiiii+\byiiiii+2\byiiiiii+2\byiiiiiiii\\
\mf{so}(V)&\byiiii+\byiiiii+\byiiiiii\\
\end{array}$\\
\hline
$6$&$
\begin{array}{lccccccccccc}
V&\byiiiiiiiii\\
\mf{sl}(V)&2\byiiii+2\byiiiii+2\byiiiiii+3\byiiiiiiii\\
\mf{sp}(V)&2\byiiii+2\byiiiii+\byiiiiii+\byiiiiiiii\\
\end{array}$\\
\hline
\end{tabular}
\end{center}
\end{table}
\end{center}


The following result can be used to obtain information on the Lie algebra structure from information on the group module structure.
\begin{Lemma}
\label{lem:bracket of isotypical components}
Let $\mf{g}$ be a Lie algebra and suppose $G<\Aut(\mf{g})$.
The bracket of two $G$-submodules of $\mf{g}$ is a submodule of $\mf{g}$.
If $\mf{g}$ is semisimple and $\mf{g}^\chi\cong nV_\chi$ and $\mf{g}^\psi\cong mV_\psi$ are isotypical components of $\mf{g}$ (which are uniquely determined by the characters, cf.~Section \ref{sec:representation theory}), then $[\mf{g}^\chi,\mf{g}^\psi]$ is a representation, isomorphic to a subrepresentation of $nV_\chi\wedge mV_\psi$.
\end{Lemma}
\begin{proof}
The first claim follows from the assumption that $G$ acts by Lie algebra automorphisms, $g[V,W]=[gV,gW]=[V,W]$.

For the second part, consider bases $\{e^\chi_i\}$ and $\{e^\psi_j\}$ for the isotypical components. Then $\{[e^\chi_i,e^\psi_j]\}$ is a set spanning $[\mf{g}^\chi,\mf{g}^\psi]$. Thus a subset thereof is a basis for this space.

For any particular group element $g$ we can assume the bases $\{e^\chi_i\}$ and $\{e^\psi_j\}$ diagonalise $g$. But then the action on $[e_i,e_j]$ equals the action on $\{e_i\otimes e_j\}$ which is a basis for  $nV_\chi\otimes mV_\psi$. If $\chi=\psi$ we can restrict to $\{[e^\chi_i,e^\chi_j]\;|\;i<j\}$ inside $nV_\chi\wedge mV_\psi$, by antisymmetry of the Lie bracket.
\end{proof}

\begin{Example}[$\mf{so}(\byiiiiii)$]
The orthogonal Lie algebra based on $\byiiiiii$ has group decomposition
\[\mf{so}(\byiiiiii)=\byiiii\oplus\byiiiii\]
according to Table \ref{tab:group decompositions of simple lie algebras}.
This is a rather special case since it is the only irreducible $4$-dimensional representation of a binary polyhedral group that preserves a symmetric bilinear form, i.e.~the only $\mf{so}_4(\CC)$-case. Moreover, this Lie algebra is not simple; \[\mf{so}_4(\CC)=\mf{sl}_2(\CC)\oplus\mf{sl}_2(\CC).\] One could ask whether this decomposition as a Lie algebra coincides with the decomposition as a group module.

Using Lemma \ref{lem:bracket of isotypical components} we check that $[\byiiii,\byiiii]<\wedge^2 \byiiii=\byiiii$, and $[\byiiiii,\byiiiii]<\wedge^2 \byiiiii=\byiiiii$, so the two summands of the first decomposition are in fact three-dimensional (and perfect) Lie algebras. To show that the direct sum of $G$-modules is a Lie algebra direct sum as well, we must show that elements from different components commute. Indeed they do. By Lemma \ref{lem:bracket of isotypical components} $[\byiiii,\byiiiii]<\byiiii\byiiiii=\byiiiiii\oplus\byiiiiiiii$ and this $\YY$-module has zero intersection with $\mf{so}(\byiiiiii)=\byiiii\oplus\byiiiii$.
\end{Example}

\section{Inner Automorphisms and the Reduction Group}
\label{sec:inner automorphisms}

In this section we use a few standard notions from the theory of Lie groups $\cG$ and their Lie algebras $\mf{g}$, such as the \emph{adjoint representation} of the Lie group $\Ad:\cG\rightarrow\Aut (\mf{g})$ and of the Lie algebra $\ad:\mf{g}\rightarrow\Aut (\mf{g})$, the \emph{exponent map} $\exp:\mf{g}\rightarrow \cG$ (also written as $\exp(a)=e^a$), the \emph{Killing form} given by $K(a,b)=\tr\ad(a)\ad(b)$, $a,b\in\mf{g}$ (a $\cG$-invariant bilinear form on $\mf{g}$), \emph{Cartan subalgebras} $\mf{h}<\mf{g}$ and the relationship between semisimple Lie algebras and \emph{Dynkin diagrams}. Some good references are \cite{fulton1991representation,humphreys1972introduction,jacobson1979lie,knapp2002lie}. We recall the following definition.

\begin{Definition}[Inner automorphisms of Lie algebras]
Automorphisms of a semisimple Lie algebra $\mf{g}$ of the form $\Ad(e^a)$, where $a\in\mf{g}$, are called \emph{inner}. The set of all inner automorphisms is denoted by $\Inn(\mf{g})$. Elements of the complement $\Aut(\mf{g})\setminus \Inn(\mf{g})$ are sometimes called \emph{outer} automorphisms.
\end{Definition}
It is well known that $\Inn(\mf{g})$ is a normal subgroup of $\Aut(\mf{g})$ and the quotient is the automorphism group of the Dynkin diagram \cite{fulton1991representation,humphreys1972introduction,jacobson1979lie}. 
\[\bigslant{\Aut(\mf{g})}{\Inn(\mf{g})}\cong \Aut(\Dyn (\mf{g}))\]
Dynkin diagrams are arranged in families of type $A$, $B$, $C$, $D$, and $E$ (as are many other objects, in particular the closely related root systems). These families relate to classical Lie algebras by 
\begin{equation*}
\begin{array}{ll}
A_\ell=\Dyn (\mf{sl}_{\ell+1}(\CC)),\\
B_\ell=\Dyn (\mf{so}_{2\ell+1}(\CC)),\\
C_\ell=\Dyn (\mf{sp}_{2\ell}(\CC)),\\
D_\ell=\Dyn (\mf{so}_{2\ell}(\CC)).\\
\end{array}
\end{equation*}
The number $\ell$ is called the \emph{rank} of the Dynkin diagram or of the associated Lie algebra, and equals the dimension of its Cartan subalgebra.
The automorphism groups of Dynkin diagrams are also well known \cite{fulton1991representation,humphreys1972introduction}. 
They are all trivial with the following exceptions.
\begin{equation}
\label{eq:AutDynkin}
\begin{array}{ll}
\Aut(A_\ell)=\zn{2},& \ell\ge 2,\\
\Aut(D_4)=S_3,&\\
\Aut(D_\ell)=\zn{2}&\ell\ge 5,\\
\Aut(E_6)=\zn{2}.
\end{array}
\end{equation}

One can find explicit descriptions of inner and outer automorphisms in \cite{jacobson1979lie}.
Inner automorphisms of the special linear algebra $\mf{sl}_n(\CC)$, naturally represented in $\End(\CC^n)$, are precisely the conjugations $A\mapsto \Ad(B)A=BAB^{-1}$ where $B\in\SL_n(\CC)$. If $n\ge3$ there exists outer automorphisms, which are conjugations composed with the map $A\mapsto -A^T$.
Automorphisms of $\mf{so}_{2n}(\CC)$, $n\ne 4$, are conjugations by orthogonal matrices. These are inner if and only if the orthogonal matrix has determinant $1$. In other words $\Aut(\mf{so}_{2n}(\CC))\cong\Ad(\NSO_{2n}(\CC))$ and $\Inn(\mf{so}_{2n}(\CC))\cong\Ad(\SO_{2n}(\CC))$. The exception $\mf{so}_{8}(\CC)$ allows more outer automorphisms, cf.~\cite{fulton1991representation}.

\begin{Lemma}
\label{lem:determinant of automorphism}
An automorphism $\phi\in\Aut(\mf{g})$ of a semisimple Lie algebra has determinant $\pm1$. If the automorphism is inner, it has determinant $1$. There exists outer automorphisms with determinant of both signs.
\end{Lemma}
\begin{proof}
Any Lie algebra automorphism respects the Killing form $K$ of the Lie algebra, which is nondegenerate if and only if the Lie algebra is semisimple (a fact known as \emph{Cartan's criterion}). In terms of matrices, these statements read $\phi^T K \phi=K$ and $\det K \ne 0$ if $\mf{g}$ is semisimple. Taking the determinant of the first equation gives $(\det \phi)^2=1$.

Now suppose the automorphism is inner, i.e.~$\phi=\Ad(e^a)$ where $a\in\mf{g}$. By semisimplicity, $\mf{g}=[\mf{g},\mf{g}]$, so the adjoint action of $\mf{g}$ on itself is traceless: $\tr\ad(a)=\tr\ad([b,c])=\tr[\ad(b),\ad(c)]=0$. Therefore 
$\det \Ad(e^a)=\det e^{\ad(a)}=e^{\tr\ad(a)}=1$.

Finally, as an example of an automorphism with determinant $-1$, consider $\phi\in\Aut(\mf{sl}_n(\CC))$ defined by
$\phi:A\mapsto-A^T$. Then $\det \phi=(-1)^{\nicefrac{1}{2}n(n+1)-1}$. 
Indeed, in the usual \emph{Chevalley} \cite{humphreys1972introduction} basis for the Lie algebra and the usual basis of its representation as $n\times n$ matrices, a basis element of the Cartan subalgebra is a diagonal matrix, hence mapped to minus itself by $\phi$, contributing a factor $-1$ to $\det \phi$. Each nondiagonal basis vector $E_{i,j}$ is mapped to $-E_{j,i}$ and such a pair also contributes a factor $\det\begin{pmatrix}0&-1\\-1&0\end{pmatrix}=-1$ to $\det \phi$, counting a total of $\dim\mf{h}+\frac{\dim\mf{sl}_n-\dim\mf{h}}{2}=n-1+\frac{n^2-1-(n-1)}{2}=\nicefrac{1}{2}n(n+1)-1$.
\end{proof}

If $G < \Aut(\mf{g})$ is a group of automorphisms, then the structure of the group $G$ often limits the options of how its representation is divided into inner and outer automorphisms. Indeed, the inner part is a normal subgroup: $G\cap \Inn(\mf{g})\triangleleft G$. A well known identity in group theory, often called the \emph{second isomorphism theorem}, states that $\bigslant{G}{G\cap \Inn(\mf{g})}\cong \bigslant{G\,\Inn(\mf{g})}{\Inn(\mf{g})}$. The right hand side is a subgroup of automorphisms of the Dynkin diagram $\bigslant{\Aut(\mf{g})}{\Inn(\mf{g})}$, so we have
\begin{equation}
\label{eq:subgroup of Aut(Dyn)}
\bigslant{G}{G\cap \Inn(\mf{g})}<\Aut(\Dyn(\mf{g})).
\end{equation} 
The structure of $\Aut(\Dyn(\mf{g}))$ is known (\ref{eq:AutDynkin}). Comparing with the quotients groups of $G$ can inform us about the possibility of outer automorphisms.

For example, let $\cQ G=\left\{\bigslant{G}{N}\;|\;N\triangleleft G\right\}$ denote the set of all quotient groups of $G$.
For the polyhedral groups we find
\begin{equation}
\begin{array}{lll}
\cQ\zn{N}&=&\left\{\zn{M}\;|\;M\text{ divides }N\right\},\\
\cQ\DD_N&=&\left\{1,\DD_{M}\;|\;M\text{ divides }N\right\},\\
\cQ\TT&=&\left\{1,\zn{3},\TT\right\},\\
\cQ\OO&=&\left\{1,\zn{2},S_3,\OO\right\},\\
\cQ\YY&=&\left\{1,\YY\right\}.
\end{array}
\end{equation}
If for instance the tetrahedral group $\TT$ acts faithfully on a simple Lie algebra other than $\mf{so}_8(\CC)$, then it acts solely by inner automorphisms, since $\zn{2}$ is not in $\cQ\TT$. By the same argument the icosahedral group can only act by inner automorphism on a classical simple Lie algebra.

On the other hand, this information can help to find actions involving outer automorphisms, for instance a $\zn{2N}$ or $\DD_N$ action on $\mf{sl}_n(\CC)$ where the normal subgroup $\zn{N}$ of index $2$ acts by inner automorphism, and the other half of the group elements are represented by outer automorphisms, i.e.~are of the form $A\mapsto -\Ad(B)A^T$. This dihedral case is studied in \cite{MPW2014}. 
We also notice that the largest group of Dynkin diagram automorphisms is a polyhedral group: $\Aut(\Dyn(\mf{so}_8(\CC)))=\Aut(D_4)=S_3=\DD_3$ (cf.~\emph{triality}, \cite{fulton1991representation}).

If we stick to the format where $G$ acts on $\mf{g}(V)$ as induced by a $G$-module $V$ (which is a restriction of Definition \ref{def:alias1}) then there is only conjugation and many outer automorphisms are excluded. In fact, for classical simple Lie algebras only $\mf{so}(V)$ with $\dim V \in 2\NN$ might still be acted upon by an outer automorphism. But we found in Section \ref{sec:characters} that the only even dimensional irreducible representation of a polyhedral group that preserves a bilinear form is $V=\byiiiiii$, and because $\det{\byiiiiii}=\byi$, we can conclude that if we restrict further by requiring $V$ to be irreducible then this action on the base Lie algebra is completely inner.
\begin{Observation}
\label{obs:only inner}
Let $V$ be an irreducible representation of a binary polyhedral group $G^\flat$ 
and let $\mf{g}(V)<\mf{gl}(V)$ be a $G^\flat$-submodule and classical simple 
Lie algebra. Then $G$ acts by inner automorphisms on this Lie algebra.
In particular the determinant of the action on $\mf{g}(V)$ is trivial; $\det\chi_{\mf{g}(V)}=\triv.$
\end{Observation} 

\begin{Remark} In the above we implicitly make the assumption that the action of $G$ on $\mf{g}(V)$ is faithful. There is a single case in the current setup where this is not so, as can be seen in Table \ref{tab:group decompositions of simple lie algebras} together with the $\im$-column of the character tables in Section \ref{sec:characters}. This is $\mf{sl}(\boiii)=\boii+\boiii$ where we act by the quotient group $\DD_3=S_3$ rather than $\OO$. Given that all automorphisms of $\mf{sl}_2(\CC)$ are inner, this has no consequences for the preceding discussion.
There is however a convenient implication of this fact when one is interested in the invariant matrices $\palia{sl}{\boiii}$. In Table \ref{tab:ctbo} we see that the kernel of the action on $\mf{sl}(\boiii)$, a normal subgroup in $\OO$ of order $24/6=4$,  is given by $\{1,\gal^2,z,\gal^2z\}=\langle \gal^2\rangle \cong \zn{4}$. This subgroup acts solely on $R$. In particular, the entries of the invariant matrices $\palia{sl}{\boiii}$ are ${\zn{4}}$-invariant forms. A more thorough discussion of these and related phenomena can be found in \cite{LM05comm,Lombardo}.
\end{Remark}

If we drop the condition that $V$ is irreducible, the embedding $G<\Aut(\mf{g}(V))$ can contain outer automorphisms. The smallest simple case is $\mf{so}_6(\CC)$.
\begin{Example}[Outer automorphisms on a simple Lie algebra]
Consider the octahedral group. The representations $\boiiiiii$ and $\boiiiiiii$ both preserve a symmetric bilinear form. Therefore, so does their sum $\boiiiiii\oplus\boiiiiiii$. 
The determinant of this representation is $\det{\boiiiiii\oplus\boiiiiiii}=\det{\boiiiiii}\det{\boiiiiiii}=\boii\boi=\boii$. In particular, there are endomorphisms with determinant $-1$ which induce outer automorphisms of $\mf{so}(\boiiiiii\oplus\boiiiiiii)$.
\end{Example}



All in all this section shows that in the present setup all Automorphic Lie Algebras are invariant solely under inner automorphisms. If the setup is generalised to the case where a polyhedral group is embedded $G\hookrightarrow\Aut(\mf{g})$ in an arbitrary way, the classification problem becomes infinite, and a few cases with outer automorphisms appear, namely the embeddings of $\zn{2M}$, $\DD_{2M}$ and $\OO$ in the automorphism groups of $\mf{sl}_n(\CC)$ with $n\ge 3$, $\mf{so}_{2n}(\CC)$ with $n\ge 2$ and $\mf{g}_2(\CC)$ may contain outer automorphisms,  as do embeddings of $\zn{2M}$, $\zn{3M}$, $\DD_{2M}$, $\DD_{3M}$ and $\TT$ in $\Aut(\mf{so}_{8}(\CC))$.

\section{A Family of Reductive Lie Algebras}
\label{sec:evaluating alias}

In this section we determine the map
\[\mf{f}:\mu\mapsto \alia{g}{V}{\Gamma}(\mu)\]
up to automorphisms of Lie algebras. That is, we evaluate the space of invariant matrices in a point of the Riemann sphere and determine the Lie algebra structure of the resulting finite dimensional vector space. This is related to what was done in Proposition \ref{prop:evaluating invariant vectors}, where we evaluated the space of invariant vectors. In fact, Proposition \ref{prop:evaluating invariant vectors} will be the starting point. The main difference is that we now want to determine the Lie algebra structure, rather than just the vector space.

On first glance the map $\mf{f}$ seems to be terribly complicated because of the many dependencies. Besides the point $\mu$ of the Riemann sphere the Lie algebra $\mf{f}(\mu)$ may depend on the choice of polyhedral group $G$, one of its representations $V$ and one of the orbits $\Gamma\subset\overline{\CC}$ under $G$ and finally a complex Lie algebra $\mf{g}$.

It turns out that the situation is drastically easier than that. First of all, from the definition we have $\mf{f}(g\mu)=g\mf{f}(\mu)\cong \mf{f}(\mu)$, i.e.~this map is constant on orbits and can be defined on the orbifold $\bigslant{\overline{\CC}}{G}$. We will find in this section that, after the identification $\bigslant{\overline{\CC}}{G}\cong\overline{\CC}$, the group $G$ and its representation $V$ do not play a role anymore. In fact, the value of $\mf{f}$ depends only on the orbit type of $G\mu$ and a choice of simple Lie algebra $\mf{g}$. In particular, $\mf{f}$ defines an invariant of Automorphic Lie Algebras (see Concept \ref{conc:invariant of alias}).

To asses the value of this result we need to revisit our definition of Automorphic Lie Algebras. In the introduction of this thesis we defined $\mf{sl}_n(\CC)$ as $n\times n$ traceless matrices. However, it would be better to say that  $\mf{sl}_n(\CC)$ is an $(n^2-1)$-dimensional complex vector space together with a set of structure constants defined by commutators of $n\times n$ traceless matrices. These matrices are merely a representation of the Lie algebra. Since the Lie algebra is defined by this representation one could call it the \emph{natural representation}. This distinction between the Lie algebra and a natural representation has not been relevant in this thesis until now. Therefore it was ignored.

Analogous to the classical Lie algebras, we will slightly modify Definition \ref{def:alias1} and call the space of invariant matrices $\alia{g}{V}{\Gamma}$ the \emph{natural representation of the Automorphic Lie Algebra}, whereas the actual Automorphic Lie Algebra is an infinite dimensional complex vector space together with a set of structure constants in $\mero_{\Gamma}$ as defined by this natural representation.
Furthermore, now that the subject of the thesis comes into focus, we introduce the acronym \emph{ALiAs} for Automorphic Lie Algebras.

The results on the map $\mf{f}$ of this section tell us a lot about the natural representation of 
ALiAs, but less about these Lie algebras themselves. Nevertheless, these concepts are intimately related and the map $\mf{f}$ will be of great help in the next chapter.

\begin{Lemma}
\label{lem:g(V)^g}
Let $V$ be an $n$-dimensional vector space and let $g\in\GL(V)$ have finite order and eigenvalues
\[g\cong \diag(1,\ldots,1,-1,\ldots,-1,\mu_3,\ldots,\mu_3,\ldots,\mu_k,\ldots,\mu_k)\]
with respective multiplicities $(m_1, m_2,m_3,\ldots m_k)$.
If $g$ preserves $\mf{so}(V)$ or $\mf{sp}(V)$ we may assume that $m_r=m_{r+1}$ and $\overline{\mu_r}=\mu_{r+1}$ if $r\ge 3$ is odd. If $g$ preserves $\mf{sp}(V)$ we can moreover assume that $m_1$ and $m_2$ are even. Up to isomorphism, the Lie algebras of $g$-invariants are
\begin{align*}
&\mf{sl}(V)^{\langle g\rangle}\cong \bigslant{\left(\bigoplus_{r=1}^{k}\mf{gl}_{m_{r}}\right)}{\CC \Id_n}\cong \bigoplus_{r=1}^{k}\mf{sl}_{m_{r}}\oplus \bigoplus_{r=1}^{k-1}\CC,\\
&\mf{so}(V)^{\langle g\rangle}\cong\mf{so}_{m_1}\oplus  \mf{so}_{m_2}\oplus\bigoplus_{r=3,\,r\text{ odd}}^{k-1} \mf{gl}_{m_{r}},\\
&\mf{sp}(V)^{\langle g\rangle}\cong\mf{sp}_{m_1}\oplus  \mf{sp}_{m_2}\oplus\bigoplus_{r=3,\,r\text{ odd}}^{k-1} \mf{gl}_{m_{r}},
\end{align*}
where we use Lie algebra direct sums, i.e.~elements from distinct summands commute.
In particular one finds the dimensions 
\begin{align*}
&\dim \mf{sl}(V)^{\langle g\rangle}=-1+\sum_{r=1}^k m_r^2,\\
&\dim \mf{so}(V)^{\langle g\rangle}=\frac{1}{2}\left(m_1(m_1-1)+m_2(m_2-1)+\sum_{r=3}^k m_r^2\right),\\
&\dim \mf{sp}(V)^{\langle g\rangle}=\frac{1}{2}\left(m_1(m_1+1)+m_2(m_2+1)+\sum_{r=3}^k m_r^2\right).
\end{align*}
satisfying $\dim \mf{gl}(V)^{\langle g\rangle}=\dim \mf{sl}(V)^{\langle g\rangle}+1=\dim \mf{so}(V)^{\langle g\rangle} +\dim\mf{sp}(V)^{\langle g\rangle}$.
\end{Lemma}
\begin{proof}
The description of $\mf{sl}(V)^{\langle g\rangle}$ follows immediately from the observation that \[\End(V)^{\langle g\rangle}=\mf{gl}(V)^{\langle g\rangle}\cong \bigoplus_{r=1}^{k}\mf{gl}_{m_{r}}(\CC).\]
If $g$ preserves a bilinear form $B$ then its nonreal eigenvalues come in conjugate pairs. 
Now we choose a basis such that $g=\diag(g_1,\ldots,g_n)$ is the given diagonal matrix. The condition $g^T B g=B$, reads $g_ig_jB_{ij}=B_{ij}$, which implies $B_{ij}=0$ if $g_ig_j\ne 1$. This gives $B$ a block structure. A change of basis $g\mapsto P^{-1}g P$ transforms the bilinear form as $B\mapsto P^TBP$. Restricting $P$ to the group $C_{\GL(V)}(g)=\GL(V)^{\langle g\rangle}$ means that the matrix for $g$ does not change, yet there is plenty of freedom to find a convenient matrix representing $B$.

From here we distinguish between the orthogonal case ($B^T=B$) and the symplectic case ($B^T=-B$).
In the orthogonal case one can transform $B$ into 
\[B=\diag\left(\Id_{m_1},\Id_{m_2},\begin{pmatrix}0&\Id_{m_3}\\\Id_{m_3}&0\end{pmatrix},\begin{pmatrix}0&\Id_{m_5}\\\Id_{m_5}&0\end{pmatrix},\ldots,\begin{pmatrix}0&\Id_{m_k}\\\Id_{m_k}&0\end{pmatrix}\right).\]
Indeed, the block structure we start with, given by $B_{ij}=0$ if $g_ig_j\ne 1$, is precisely this block structure where each $\Id_{m}$ is replaced by a square matrix, and since $B$ is nondegenerate and symmetric, so are all these blocks. The transformation group $C_{\GL(V)}(g)$ allows us free reign over all these blocks.

Now we can describe $\mf{so}(V)^{\langle g\rangle}= \mf{so}(V)\cap\mf{gl}(V)^{\langle g\rangle}$
. Recall that $\mf{gl}(V)^g=\mf{gl}_{m_1}\oplus\ldots\oplus\mf{gl}_{m_k}$.
We describe $\mf{so}(V)$ at the blocks. First, at the $m_1$ and $m_2$ block we have antisymmetric matrices, contributing the summand $\mf{so}_{m_1}\oplus\mf{so}_{m_2}$. For $m_r$, $r\ge 3$, $r$ odd, we consider the blocks $\begin{pmatrix}a_{11}&0\\0&a_{22}\end{pmatrix}$ in $ \mf{so}(V)\cap\mf{gl}(V)^{\langle g\rangle}$ of size $2m_r$. The condition that $BA$ is antisymmetric for $A\in\mf{so}(V)$ implies here that $a_{11}^T=-a_{22}$. 
Thus we obtain the description of $\mf{so} (V)^g$.

In the symplectic case,
$B$ is an antisymmetric nondegenerate form. The same block structure,  $B_{ij}=0$ if $g_ig_j\ne 1$, holds, and each block must be nondegenerate and antisymmetric as well. This implies that $m_1$ and $m_2$ are even numbers.  The antisymmetric bilinear form $B$ can be transformed to
\[B=\diag\left(J_{\frac{m_1}{2}},J_{\frac{m_2}{2}},J_{m_3},J_{m_5}\ldots, J_{m_k}\right)\]
while keeping $g$ diagonal. Here $J_m=\begin{pmatrix}0&\Id_{m}\\-\Id_{m}&0\end{pmatrix}$. This time we seek $A=\begin{pmatrix}a_{11}&a_{12}\\a_{21}&a_{22}\end{pmatrix}$ such that $JA$ is symmetric. We obtain the same condition for the diagonal blocks: $a_{11}^T=-a_{22}$, and the other blocks are required to be symmetric: $a_{12}^T=a_{12}$, $a_{21}^T=a_{21}$. When intersecting with $\mf{gl} (V)^g$ we find the Lie algebra described in the Lemma.
\end{proof}

Now we have gathered enough information to classify evaluations of the natural representation of ALiAs. 

\begin{Theorem}
\label{thm:evaluated alias}
Let $V$ be an irreducible representation of a binary polyhedral group $G^\flat$ and $\mf{g}(V)<\mf{gl}(V)$ a simple Lie algebra and $G^\flat$-submodule. 
Denote the natural representation of the corresponding Automorphic Lie Algebra, holomorphic outside $\Gamma\subset\overline{\CC}$, by $\alia{g}{V}{\Gamma}$.
Define the map $\mf{f}_{\mf{g}}:\bigslant{\overline{\CC}}{G}\rightarrow\{\text{Lie subalgebras of }\mf{g}(V) \}$ by
\[\mf{f}_{\mf{g}}(\Gamma)=\alia{g}{V}{\Gamma'}(\mu),\qquad \mu\in\Gamma\ne\Gamma'.\]
This map is well defined up to Lie algebra isomorphism. That is, $\mf{f}_{\mf{g}}(\Gamma)$ is independent of the element $\mu\in\Gamma$, of the chosen orbit of poles $\Gamma'$, of the representation $V$ and in particular independent of the polyhedral group $G$.

If $\Gamma$ is a generic orbit then $\mf{f}_{\mf{g}}(\Gamma)=\mf{g}$.
There is a \emph{linear direct sum} of the values at the exceptional orbits
\[\mf{g}=\bigoplus_{i\in\Omega} \mf{f}_{\mf{g}}(\Gamma_i).\] 
Moreover, the Lie algebra $\mf{f}_{\mf{g}}(\Gamma_i)$ is as described in Table \ref{tab:evaluated alias}, where $(m_{i,1},\ldots,m_{i,k_i})$ are the multiplicities of eigenvalues of $g_i\in G^\flat$ given in Table \ref{tab:multiplicities of eigenvalues}.
\begin{center}
\begin{table}[h!] 
\caption{The natural representation evaluated at exceptional orbits: $\mf{f}_{\mf{g}}(\Gamma_i)$.}
\label{tab:evaluated alias}
\begin{center}
\begin{tabular}{cccccc} \hline
 &$ \al $&$ \be $&$ \ga$\\
\hline
$\mf{sl}_n 
$&$\bigoplus_{r=1}^{k_\al}\mf{sl}_{m_{\al,r}}\oplus\bigoplus_{r=1}^{k_\al-1}\CC
$&$\bigoplus_{r=1}^{k_\be}\mf{sl}_{m_{\be,r}}\oplus\bigoplus_{r=1}^{k_\be-1}\CC$&$
\bigoplus_{r=1}^{k_\ga}\mf{sl}_{m_{\ga,r}}\oplus\bigoplus_{r=1}^{k_\ga-1}\CC$\\
\hline
$\mf{so}_3 $&$ \CC $&$ \CC $&$ \CC$\\
$\mf{so}_4 $&$ \CC\oplus\CC $&$ \CC\oplus\CC $&$ \CC\oplus\CC $\\
$\mf{so}_5 $&$ \CC\oplus\CC $&$ \mf{sl}_2\oplus\CC $&$  \mf{sl}_2\oplus\CC$\\
\hline
$\mf{sp}_2 $&$ \CC $&$ \CC $&$ \CC$\\
$\mf{sp}_4 $&$ \CC\oplus\CC $&$ \mf{sl}_2\oplus\CC $&$ \mf{sl}_2\oplus\CC$\\
$\mf{sp}_6 $&$ \mf{sl}_2\oplus\CC\oplus\CC $&$ \mf{sl}_2\oplus \mf{sl}_2\oplus\CC $&$ \mf{sl}_3\oplus\CC$\\
\hline 
\end{tabular}
\end{center}
\end{table}
\end{center}

In particular, we have the dimensions of $\dim \mf{f}_{\mf{g}}(\Gamma_i)$ in Table \ref{tab:dimg(V)^g}. 
\begin{center}
\begin{table}[h!] 
\caption{Dimensions of single element invariants: $\dim \mf{f}_{\mf{g}}(\Gamma_i)$.}
\label{tab:dimg(V)^g}
\begin{center}
\begin{tabular}{cccccccccccc} \hline
$ $&$\mf{sl}_2$&$\mf{sl}_3$&$\mf{sl}_4$&$\mf{sl}_5$&$\mf{sl}_6
$&$\mf{so}_3$&$\mf{so}_4$&$\mf{so}_5
$&$\mf{sp}_2$&$\mf{sp}_4$&$\mf{sp}_6$\\
\hline
$\al$&$1$&$2$&$3$&$4$&$7$&$
{1}$&${2}$&${2}
$&$1$&${2}$&${5}$\\
$\be$&$1$&$2$&$5$&$8$&$11$&$
{1}$&${2}$&${4}
$&$1$&${4}$&${7}$\\
$\ga$&$1$&$4$&$7$&$12$&$17$&$
{1}$&${2}$&${4}
$&$1$&${4}$&${9}$\\
\hline
$\Sigma$&$3$&$8$&$15$&$24$&$35$&$3$&$6$&$10$&$3$&$10$&$21$\\
\hline 
\end{tabular}
\end{center}
\end{table}
\end{center}
\end{Theorem}
\begin{proof}
This result is a combination of previous results. Firstly Proposition \ref{prop:evaluating invariant vectors}, which shows that $\mf{f}_{\mf{g}}(\Gamma)=\mf{g}$ if $\Gamma$ is a generic orbit, and $\mf{f}_{\mf{g}}(\Gamma_i)=\mf{g}(V)^{\langle g_i\rangle}$ otherwise. Secondly, the description of $\mf{g}(V)^{\langle g_i\rangle}$ given in Lemma \ref{lem:g(V)^g}.

The linear direct sum $\mf{g}(V)=\bigoplus_{i\in\Omega} \mf{g}(V)^{\langle g_i\rangle}$ follows immediately from Schur's Lemma and Corollary \ref{cor:sumV^g}. Notice that this direct sum is not respected by the Lie bracket if $\mf{g}(V)$ is a simple Lie algebra.

In Theorem \ref{thm:multiplicities of eigenvalues} we found the multiplicities of eigenvalues of $g_i\in G^\flat$, cf.~Table \ref{tab:multiplicities of eigenvalues}, which are independent of the group.
This gives us the Lie algebra structure of $\mf{sl}(V)^{\langle g_i\rangle}$ as shown in Table \ref{tab:evaluated alias}. However, this is not enough information to determine the Lie algebra structure of the other cases: $\mf{so}(V)^{\langle g_i\rangle}$ and $\mf{sp}(V)^{\langle g_i\rangle}$, thus Lemma \ref{lem:g(V)^g} comes into play. By this lemma it is sufficient to determine the multiplicities of \emph{real} eigenvalues of irreducible representations. This can be done in an ad hoc manner, going through all cases and occasionally using the fact that  $\dim\mf{g}(V)=\sum_{i\in\Omega}\dim \mf{g}(V)^{\langle g_i\rangle}$. The findings are sketched in the table at the end of this proof.

For the orthogonal cases we may assume that nonreal eigenvalues come in conjugate pairs, but real eigenvalues can be single.
If $n=\dim V=3$ and $i\in\{\al,\be\}$ then the multiplicities of eigenvalues of $g_i$ are $(1,1,1)$, cf.~Table \ref{tab:multiplicities of eigenvalues}. Given that there are only two real roots of unity, $\pm1$, there must be precisely one conjugate pair of nonreal eigenvalues and one real eigenvalue. For $\gga$ we have multiplicities $(2,1)$, which can only correspond to real eigenvalues. Using Lemma \ref{lem:g(V)^g} we have now established the $\mf{so}_3$-row of Table \ref{tab:evaluated alias}.

Now consider $n=4$. The group element $\gal$ has eigenvalue multiplicities $(1,1,1,1)$. The eigenvalues can either be two conjugate pairs, which would imply $\dim \mf{so}_4^{\langle \gal\rangle}=2$, or two reals and one conjugate pair, yielding $\dim \mf{so}_4^{\langle \gal\rangle}=1$. 
At $\be$ we have multiplicities $(2,1,1)$, implying two identical reals and one conjugate pair and $\dim \mf{so}_4^{\langle \gbe\rangle}=2$. 
Finally, the multiplicities $(2,2)$ for $\ga$ yield all nonreal or all real eigenvalues, implying respectively $\dim \mf{so}_4^{\langle \gga\rangle}=4$ and $\dim \mf{so}_4^{\langle \gga\rangle}=2$. The fact that the dimensions add up to $6$ fixes everything.

If $n=5$ we can see that $\gbe$, with multiplicities $(2,2,1)$, has precisely one real eigenvalue and $\gga$, with multiplicities $(3,2)$, has all real eigenvalues. The $\al$ situation is then fixed by dimension.
The case $n=6$ is impossible for the orthogonal algebra. In Section \ref{sec:characters} we confirm that there is no orthogonal irreducible representation of this dimension.

Now we go through all the symplectic cases. This time we have the additional condition that real eigenvalues occur in even numbers.
If $n=2$ we have multiplicities $(1,1)$. The eigenvalues are not real. Therefore $\mf{sp}(V)^{g_i}\cong \mf{gl}_1$ for all $i\in\Omega$. 
If $n=4$ then $\gal\sim\diag(z_1, \bar{z}_1, z_2, \bar{z}_2)$ and $\gbe\sim\diag(\pm1, \pm1, z_3, \bar{z}_3)$.
We know that \[\dim\mf{sp}(V)^{\gga}=\dim\mf{sp}(V)-\dim\mf{sp}(V)^{\gal}-\dim\mf{sp}(V)^{\gbe}=10-2-4=4.\] This means we can decide whether the multiplicities $(2,2)$ come from all real eigenvalues $(1,1,-1,-1)$ or a double nonreal eigenvalue with conjugates, since the first contributes a dimension of $\frac{1}{2}(2\cdot3+2\cdot3)=6$ and the latter $\frac{1}{2}(2^2+2^2)=4$.
The case $n=6$ is easier. At $\al$, $(2,1,1,1,1)$, we have one pair of real eigenvalues and otherwise distinct nonreal eigenvalues. At $\be$, $(2,2,2)$, there is also exactly one real pair and at $\ga$, $(3,3)$ there cannot be any real eigenvalues.

The eigenvalue structure can be schematically summarised as follows,
\begin{center}
\begin{tabular}{cccc}
 & $\al$ & $\be$ & $\ga$\\
\hline
$\mf{so}_3$ & $(\pm1, z, \bar{z})$ & $(\pm1, z, \bar{z})$ & $(\pm1,\pm1, \mp1)$ \\
$\mf{so}_4$ & $(z_1, \bar{z}_1,z_2, \bar{z}_2)$ & $(\pm1, \pm1, z, \bar{z})$ & $(1, 1, -1, -1)$ \\
$\mf{so}_5$ & $(\pm1, z_1, \bar{z}_1,z_2, \bar{z}_2)$ & $(\pm1, z,z, \bar{z},\bar{z})$ & $(\pm1,\pm1,\pm1,\mp1, \mp1)$ \\
$\mf{sp}_2$ & $(z, \bar{z})$ & $(z, \bar{z})$ & $(z, \bar{z})$ \\
$\mf{sp}_4$ & $(z_1, \bar{z}_1,z_2, \bar{z}_2)$ & $(\pm1, \pm1, z, \bar{z})$ & $(z,z, \bar{z},\bar{z})$ \\
$\mf{sp}_6$ & $(\pm1,\pm1, z_1, \bar{z}_1,z_2, \bar{z}_2)$ & 
$(\pm1, \pm1, z,z, \bar{z},\bar{z})$ & 
$(z,z, z, \bar{z},\bar{z},\bar{z})$ \\
\end{tabular}
\end{center} 
where $z,z_i\in\CC\setminus\RR$ and $z_1\ne z_2$. At each position in the table the numbers $z$, $z_1$ and $z_2$ are redefined.
\end{proof}

We end this chapter by attaching an $|\Omega|$-tuple of integers to the relevant selection of simple Lie algebras, abusing the notation of Definition \ref{def:kappa(chi)}.
\begin{Definition}[$\kappa(\roots)$]
\label{def:kappa(roots)}
Let $V$ be an irreducible representation of a binary polyhedral group $G^\flat$ and $\mf{g}(V)$ a  $G^\flat$-submodule of $\mf{gl}(V)$ and a simple Lie subalgebra in the isomorphism class $\roots$. We define \[\kappa(\roots)_i=\nicefrac{1}{2}\,\codim\mf{g}(V)^{\langle g_i\rangle},\qquad i\in\Omega.\]
Equivalently, using Definition \ref{def:kappa(chi)} we have $\kappa(\roots)=\kappa(\chi_{\mf{g}(V)})$.
\end{Definition}
Theorem \ref{thm:evaluated alias} shows that $\kappa(\roots)$ is well defined, that is, the dimension of $\mf{g}(V)^{\langle g_i\rangle}$ only depends on $\roots$ and $i\in\Omega$. This can also be deduced from the decompositions of $\mf{g}(V)$ into irreducible representations, cf.~Table \ref{tab:group decompositions of simple lie algebras}, for which all these dimensions have been calculated and summarised in Table \ref{tab:dimV^g}. 
However, this does not provide the full Lie algebra structure.
The analysis of this section enabled us to see that 
the Lie algebra structure of $\mf{f}_{\mf{g}}(\Gamma)$ is an invariant of ALiAs. Moreover, we see that all values of $\mf{f}_{\mf{g}}$ contain a Cartan subalgebra for $\mf{g}$, and that they are reductive \cite{fulton1991representation} (this holds true in bigger generality, as shown by Kac in 1969 \cite{kac1969automorphisms, kac1994infinite}). In particular, their codimensions are even, since both semisimple and reductive Lie algebras are linear direct sums of a Cartan subalgebra and an even number of one-dimensional weight spaces. Thus $\kappa(\roots)_i$ is an integer.





\chapter[Structure Theory for Automorphic Lie Algebras]{Structure Theory for Automorphic Lie Algebras}
\label{ch:C}
In this chapter we discuss several classes of Lie algebras. First we define Polynomial Automorphic Lie Algebras in Section \ref{sec:palias}, which is a natural continuation from the previous chapters. Secondly, we move on to ALiAs in Section \ref{sec:alias}, the primary subject of study. We define a normal form for these Lie algebras and find examples of explicit representations, namely the natural representation, i.e.~the invariant matrices, and a much simpler representation, which we will call the \emph{matrices of invariants}.

In Section \ref{sec:structure theory for alias} all invariants of ALiAs (Concept \ref{conc:invariant of alias}) are combined in Theorem \ref{thm:alias satisfy game of roots} to obtain a list of constraints for the Lie algebra structure. A natural language in which to express these constraints is that of cochains and their boundaries on root systems, leading to root system cohomology and a class of Lie algebras $\AL{\roots}$ associated to $2$-cocycles $\omega^2$ on a root system $\roots$. The isomorphism question (Question \ref{q:isomorphism question}) is then revisited in light of the new information.
At the end of this chapter we will be able to summarise the current state of the art and specify the open problems.

\section{Polynomial Automorphic Lie Algebras}
\label{sec:palias}

In Chapter \ref{ch:A} we discussed the polynomial ring $R=\CC[U]$ were $U$ is a representation of a finite group. More specifically, in Section \ref{sec:invariant vectors} we discussed polynomial invariant vectors $(V_\chi\otimes R)^G\cong R^{\overline{\chi}}$. If we now replace the $G$-module $V_\chi$ by one that is also a Lie algebra, like the spaces $\mf{g}(V)$, which formed the topic of Chapter $\ref{ch:B}$, then one obtains a Lie algebra. Indeed, the Lie bracket of $\mf{g}(V)$ induces a Lie bracket on $\mf{g}(V)\otimes R$ by linear extension over $R$. The fact that $G$ acts by Lie algebra isomorphisms on $\mf{g}(V)$ implies that  $\palia{g}{V}$ is closed under this bracket, that is, if $ga=a$ and $gb=b$ then $g[a,b]=[ga,gb]=[a,b]$.
\begin{Definition}[Polynomial Automorphic Lie Algebra]
Let $V$ and $U$ be finite dimensional representations of $G$. Suppose $\mf{g}(V)$ is a Lie subalgebra and $G$-submodule of $\mf{gl}(V)$ and $R=\CC[U]$ the polynomial ring. The Lie algebra of invariants 
\[\palia{g}{V}\]
is called a \emph{Polynomial Automorphic Lie Algebra} based on $\mf{g}(V)$.
\end{Definition}
Classical results from invariant theory (cf.~Section \ref{sec:classical invariant theory}) demonstrate that the space of invariants $\palia{g}{V}$ is a finitely generated free module over the polynomial ring in a choice of primary invariants of $G$ in $R$. Given such a choice, one can predict the number of generators, which is a fixed multiple of $\dim\mf{g}(V)$ (cf.~Proposition \ref{prop:stanley}). The Lie algebra induced by this space, which we also denote by $\palia{g}{V}$, is defined by a related number of structure constants, all being polynomials in primary invariants.

We continue our string of dihedral examples with the corresponding Polynomial ALiA.
\begin{Example}[${\palia[\DD_N]{sl}{V_{\psi_j}}}$]
\label{ex:dihedral palia}
Let $\psi_j$ be an arbitrary two-dimensional irreducible character of the dihedral group $\DD_N$. In the introduction of Chapter \ref{ch:B} we found explicit bases for the decomposition  $\mf{sl}(V_{\psi_j})=V_{\chi_2}\oplus V_{\psi_{2j}}$. Just to recall, $V_{\chi_2}$ is in this context the space of diagonal traceless $2\times 2$ matrices and  $V_{\psi_{2j}}$ is the space of $2\times 2$ matrices with zeros on the diagonal. The basis 
$\left\{\begin{pmatrix}0&1\\0&0\end{pmatrix},\begin{pmatrix}0&0\\1&0\end{pmatrix}\right\}$ for the latter space 
corresponds to our preferred basis (\ref{eq:standard matrices}) for $\psi_{2j}$.

In order to find the polynomial invariant matrices,
one can at this stage simply plug in the information of dihedral invariant vectors summarised in Table \ref{tab:invariants, N odd} and Table \ref{tab:invariants, N even}. 
We find generators
\begin{center}
\begin{tabular}{lcr}
$\eta_1=  \begin{pmatrix} 0 & X^{2j} \\ Y^{2j} & 0 \end{pmatrix},$ &
$\eta_2 =\begin{pmatrix} 0 & Y^{N-2j} \\ X^{N-2j} & 0 \end{pmatrix},$ & 
$\eta_3 =\nicefrac{1}{2}(X^N-Y^N) \begin{pmatrix}1 & 0 \\ 0 &-1 \end{pmatrix},$
\end{tabular}
\end{center}
and the space of invariant matrices
\[\palia{sl}{V_{\psi_j}}=\CC[\fal,\fbe]\eta_1\oplus\CC[\fal,\fbe]\eta_2\oplus\CC[\fal,\fbe]\eta_3,\] 
where the primary invariants $\fal$ and $\fbe$ are given by (\ref{eq:Fs}).

The Lie algebra structure is found by computing the commutator brackets 
\begin{align*}
&[\eta_2,\eta_3]=2(-{F_\al}^{N-2j}\eta_1+{F_\be } \eta_2),\\
&[\eta_3,\eta_1]=2({F_\be } \eta_1-{F_\al}^{2j} \eta_2),\\
&[\eta_1,\eta_2]=2\eta_3,
\end{align*}
and we notice that the structure constants are indeed in $\CC[\fal,\fbe]$.


\end{Example}

If we pick a different extension of the dihedral group we can find a Lie algebra with more generators. For instance, in the next example we consider the binary dihedral group, which we refer to as the \emph{dicyclic group} $\bd{N}$. The lowest possible degree of the invariants in a homogeneous system of parameters $\{\theta_1,\theta_2\}$ is $\deg\theta_1=4$ and $\deg{\theta_2}=2N$. Then Proposition \ref{prop:stanley} provides the number of secondary invariants: $\frac{\deg \theta_1\deg \theta_2}{|\bd{N}|}=\frac{4\cdot 2N}{4N}=2$ times the dimension of the base vector space. For example, if this base vector space is $\mf{sl}_2(\CC)$ then there are $2\cdot 3=6$ generators.
\begin{Example}[${\palia[\bd{N}]{sl}{V}}$]
\label{ex:dicyclic palia}
Just like the dihedral group, the dicyclic group \[\bd{N}=\langle r,s\;|\;r^{2N}=1,\,s^2=r^N,\,sr=r^{-1}s\rangle\] has an abelian subgroup of index 2: $\langle r\rangle\cong \zn{2N}$. Therefore, the dimension of an irreducible representation of $\bd{N}$ is at most $2$. The two-dimensional cases are conjugate to \[r=\begin{pmatrix}\omega_{2N}^j&0\\0&\omega_{2N}^{-j}\end{pmatrix},\qquad s=\begin{pmatrix}0&i^j\\i^j&0\end{pmatrix}\] where $\omega_{2N}$ is a primitive $2N$-th root of unity. We take $j=1$ for the natural representation $U$ with dual basis $\{X,Y\}$. Let $V$ be the above representation, with parameter $j$ arbitrary. Then the space of polynomial invariant traceless matrices takes the form
\[\palia[\bd{N}]{sl}{V}=\CC[\theta_1,\theta_2]\left(\zeta_1\oplus\zeta_2\oplus\zeta_3\oplus\zeta_4\oplus\zeta_5\oplus\zeta_6\right)\]
where
\begin{align*}
&\theta_1= (XY)^2,\\
&\theta_2= X^{2N}+(-1)^NY^{2N},\\
&\zeta_1=\begin{pmatrix}XY&0\\0&-XY\end{pmatrix},\\
&\zeta_2=\begin{pmatrix}X^{2N}-(-1)^NY^{2N}&0\\0&-X^{2N}+(-1)^NY^{2N}\end{pmatrix},\\
&\zeta_3=\begin{pmatrix}0&X^{2j}\\Y^{2j}&0\end{pmatrix},\\
&\zeta_4=\begin{pmatrix}0&X^{2j+1}Y\\-XY^{2j+1}&\end{pmatrix},\\
&\zeta_5=\begin{pmatrix}0&(-1)^NY^{2N-2j}\\X^{2N-2j}&0\end{pmatrix},\\
&\zeta_6=\begin{pmatrix}0&-(-1)^NXY^{2N-2j+1}\\X^{2N-2j+1}Y&0\end{pmatrix}.
\end{align*}
In terms of ground forms (which are defined by $\DD_N$) the primary invariants are $\theta_1=\fal^2$ and $\theta_2=(\fbe+\fga)^2+(-1)^N(\fbe-\fga)^2$. These generators can be obtained in a similar manner as is shown for the dihedral group. A detailed walk through does not add value to this thesis and is left out. 

The Lie structure can be computed directly.
\[
\begin{array}{lll}
[\zeta_1,\zeta_2]=0&[\zeta_1,\zeta_3]=2\zeta_4&[\zeta_1,\zeta_4]=2\theta_1\zeta_3\\

[\zeta_2,\zeta_3]=2\theta_2\zeta_3-4\theta_1^j\zeta_5&[\zeta_2,\zeta_4]=2\theta_2\zeta_4+4\theta^j\zeta_6&[\zeta_2,\zeta_5]=4(-1)^N\theta_1^{N-j}\zeta_3-2\theta_2\zeta_5\\

[\zeta_3,\zeta_4]=-2\theta_1^j\zeta_1&[\zeta_3,\zeta_5]=\zeta_2&[\zeta_3,\zeta_6]=\theta_2\zeta_1\\

[\zeta_4,\zeta_5]=\theta_2\zeta_1&[\zeta_4,\zeta_6]=\theta_1\zeta_2&\\

[\zeta_5,\zeta_6]=2\theta_1^{N-j}\zeta_1&&[\zeta_1,\zeta_5]=-2\zeta_6\\

&[\zeta_1,\zeta_6]=-2\theta_1\zeta_5&[\zeta_2,\zeta_6]=-4(-1)^N\theta_1^{N-j}\zeta_4-2\theta_2\zeta_6
\end{array}\]
Despite looking completely different, this Polynomial ALiA gives exactly the same ALiA as ${\palia[\DD_N]{sl}{V_{\psi_j}}}$ from Example \ref{ex:dihedral palia} does, when the prehomogenisation and homogenisation operators from Section \ref{sec:homogenisation} are applied. This illustrates some significant differences between ALiAs and Polynomial ALiAs.
\end{Example}
No attempt has been made to find a set of generators for this module of invariants that makes the Lie algebra structure more transparent, and it is obvious from the example that without such effort it is difficult to make sense of this structure, even when there are merely $6$ generators. To cure this problem we introduce a \emph{normal form} in Section \ref{sec:normal form}.
But first we will find some actual ALiAs as Lie subalgebras of Polynomial ALiAs. 

\section{Automorphic Lie Algebras}
\label{sec:alias}

Automorphic Lie Algebras are Lie subalgebras of current algebras of Krichever-Novikov type i.e.~the tensor product $\mf{g}\otimes\mero_\Gamma$, also known as \emph{loop algebras} (usually $\Gamma=\{0,\infty\}$ so that $\mero_\Gamma=\CC[\lambda,\lambda^{-1}]$). Current algebras are commonly denoted by $\overline{\mf{g}}$ (cf.~\cite{babelon2003introduction,schlichenmaier2003higher,vasil2014harmonic}) and we introduce a notation for Automorphic Lie Algebras in line with this convention:
\[\salia{g}{V}{\Gamma}=\alia{g}{V}{\Gamma}.\]
Polynomial ALiAs serve as a stepping stone to reach ALiAs in this exposition. The step can be taken by the prehomogenisation and homogenisation operators from Section \ref{sec:homogenisation}, thanks to Lemma \ref{lem:only quotients of invariants} which shows that the full ALiA is obtained in this way. Lemma \ref{lem:prehomogenisation mod f is homogenisation} explains the relation between the prehomogenised Polynomial ALiAs and ALiAs. 
In fact, since the Lie algebra structure of the ALiA depends only on ring structure, that is, addition and multiplication, the Lemma gives the following isomorphism of Lie algebras. We adopt the notation from Chapter \ref{ch:A}.
\begin{Proposition} 
\label{prop:prehomogenised palia is alia} Let $G<\Aut(\overline{\CC})$ be a finite polyhedral group and  $\Gamma\in\bigslant{\overline{\CC}}{G}$. Moreover, let $G^\flat$ be a sufficient extension of $G$. If $|G|$ divides $d\in\NN$ then there is a Lie algebra isomorphism
\[\p{d}\palia[G^\flat]{g}{V}\textrm{mod }F_\Gamma\cong \salia{g}{V}{\Gamma}.\]
In particular, any Automorphic Lie Algebra is a quotient of $\p{|G|}\palia[G^\flat]{g}{V}$.
\end{Proposition}
The significance of this result lies in the fact that $\p{|G|}\palia[G^\flat]{g}{V}$ is independent of the orbit $\Gamma$. Therefore, one can study all ALiAs $\left\{\salia{g}{V}{\Gamma}\;|\;\Gamma\in\bigslant{\overline\CC}{G}\right\}$ by studying  the one Lie algebra $\p{|G|}\palia[G^\flat]{g}{V}$. Note that this does not imply that ALiAs with different pole orbits are isomorphic. 

\begin{Example}[${\salia[\DD_N]{sl}{V_{\psi_j}}{\Gamma}}$]
\label{ex:dihedral alia}
Consider the Polynomial ALiA $\palia[\DD_N]{sl}{V_{\psi_j}}$ from Example \ref{ex:dihedral palia} as starting point. To determine the image of this space under the prehomogenisation projection $\p{|\DD_N|}=\p{2N}$ we need to know whether $N$ is odd or even. First suppose $N$ is odd. The Poincar\'e series, with assisting dummy symbols $\al$ and $\be$,
\begin{multline*}
P\left(\palia[\DD_N]{sl}{V_{\psi_j}},t\right)=\frac{t^{2j}+t^{N-2j}+t^{N}}{(1-\al t^2)(1-\be t^N)}\\
=\frac{(t^{2j}+t^{N-2j}+t^{N})(1+\al t^2+\ldots+\al^{N-1}t^{2N-2})(1+\be t^N)}{(1-\al^N t^{2N})(1-\be^2 t^{2N})}
\end{multline*}
is mapped to
\[\p{2N}P\left(\palia[\DD_N]{sl}{V_{\psi_j}},t\right)=\frac{\al^{N-j}t^{2N}+\al^j\be t^{2N}+\be t^{2N}}{(1-\al^N t^{2N})(1-\be^2 t^{2N})}.\]
This is in agreement with (\ref{eq:prehomogenised generating function}). We have found the prehomogenised Lie algebra
\begin{equation}
\label{eq:prehomogenised dihedral palia}
\p{2N}\palia[\DD_N]{sl}{V_{\psi_j}}=\CC[\fal^N,\fbe^2]\left(\fal^{N-j}\eta_1\oplus \fal^j\fbe\eta_2\oplus\fbe \eta_3\right).
\end{equation}

In order to find ALiAs with $\DD_N$ symmetry via Polynomial ALiAs we need a sufficient extension of the dihedral group. If $N$ is odd then $\DD_N$ itself will suffice, but if $N$ is even we need something more, for instance $\DD_{2N}$, cf.~Secion \ref{sec:binary polyhedral groups}.

Clearly $\palia[\DD_{2N}]{sl}{V_{\psi_j}}$ is obtained by the substitution $N\mapsto 2N$ in the expression for $\palia[\DD_{N}]{sl}{V_{\psi_j}}$ of Example \ref{ex:dihedral palia}. We would, however, like to express the primary invariants in terms of the ground forms as defined by the action of the original group on the Riemann sphere, rather than its extension. We take the primary invariants $\fal$ and $\fbe$ of $\DD_N$ to construct primary invariants $\fal$ and $\fbe^2$ for $\DD_{2N}$. The Poincar\'e series reads
\begin{align*}
P\left(\palia[\DD_{2N}]{sl}{V_{\psi_j}},t\right)&=\frac{t^{2j}+t^{2N-2j}+t^{2N}}{(1-\al t^2)(1-\be^2 t^{2N})}\\
&=\frac{(t^{2j}+t^{N-2j}+t^{N})(1+\al t^2+\ldots+\al^{N-1}t^{2N-2})}{(1-\al^N t^{2N})(1-\be^2 t^{2N})},
\end{align*}
and is mapped to
\[\p{2N}P\left(\palia[\DD_N]{sl}{V_{\psi_j}},t\right)=\frac{\al^{N-j}t^{2N}+\al^{j} t^{2N}+t^{2N}}{(1-\al^N t^{2N})(1-\be^2 t^{2N})}.\]
We find the prehomogenised Lie algebra
\[\p{2N}\palia[\DD_{2N}]{sl}{V_{\psi_j}}=\CC[\fal^N,\fbe^2]\left(\fal^{N-j}\eta_1\oplus \fal^{j}(\eta_2|_{N\mapsto 2N})\oplus (\eta_3|_{N\mapsto 2N})\right).\]
Let us compare this module with the previous, (\ref{eq:prehomogenised dihedral palia}). Their first generators are identical and the last only differs by a factor $\nicefrac{1}{2}$. 
Finally we notice that
\begin{align*}
2\fal^j\fbe\eta_2&=(XY)^j(X^N+Y^N)\begin{pmatrix}0&Y^{N-2j}\\X^{N_2j}&0\end{pmatrix}\\
&=\begin{pmatrix}0&(XY)^{N-j}X^{2j}+(XY)^jY^{2N-2j}\\(XY)^{N-j}Y^{2j}+(XY)^jX^{2N-2j}&0\end{pmatrix}\\
&=\fal^{N-j}\eta_1+\fal^j\eta_2|_{N\mapsto 2N}.
\end{align*}
In other words, we have found the same module as before.
Regardless of whether $N$ is odd or even, if a sufficient extension of the dihedral group $\DD_N$ is used, say $\DD_N^\flat$, then the image of the prehomogenisation operator equals (\ref{eq:prehomogenised dihedral palia}). One could even try the $\bd{N}$-symmetric Polynomial ALiA from Example \ref{ex:dicyclic palia}.
This reflects the fact that $\p{|G|}\palia[G^\flat]{g}{V}\textrm{mod }F_\Gamma \cong \salia{g}{V}{\Gamma}$ is a space of $G$-invariants rather than $G^\flat$-invariants (also observed in Example \ref{ex:automorphic forms}). 

We can conclude with
\[\p{2N}\palia[\DD_{N}^\flat]{sl}{V_{\psi_j}}=\CC[\fal^N,\fbe^2]\left(\tilde{\eta}_1\oplus \tilde{\eta}_2\oplus \tilde{\eta}_3\right),\]
where
\begin{equation}
\label{eq:prehomogenised dihedral invariant matrices}
\begin{array}{lcccl}
\tilde{\eta}_1&=&\fal^{N-j}\eta_1&=&\fal^N\begin{pmatrix}0&\lambda^j\\\lambda^{-j}&0\end{pmatrix},\\
\tilde{\eta}_2&=& \fal^j\fbe\eta_2&=&\fal^N\begin{pmatrix}0&\nicefrac{1}{2}(\lambda^j+\lambda^{j-N})\\\nicefrac{1}{2}(\lambda^{-j}+\lambda^{N-j})&0\end{pmatrix},\\
\tilde{\eta}_3&=&\fbe \eta_3&=&\fbe\fga \begin{pmatrix}1&0\\0&-1\end{pmatrix}.
\end{array}
\end{equation}
One readily computes the Lie brackets
\begin{align*}
&[\tilde{\eta}_2,\tilde{\eta}_3]=2(-{F_\be }^{2}\tilde{\eta}_1+{F_\be }^2 \tilde{\eta}_2),\\
&[\tilde{\eta}_3,\tilde{\eta}_1]=2({F_\be }^2 \tilde{\eta}_1-{F_\al}^{N} \tilde{\eta}_2),\\
&[\tilde{\eta}_1,\tilde{\eta}_2]=2F_\al^N\tilde{\eta}_3,
\end{align*}
and notices the structure constants are indeed elements of $\p{|G|}R^{G^\flat}=\CC[\fal^\nal,\fbe^\nbe]$.
\end{Example}
One can apply the homogenisation operator $\h{\Gamma}$ to obtain the ALiA. 
The particular choice of generators from the previous example is very similar to the generators found in other works.
\begin{Example}[Comparison of dihedral results]
It is useful to compare ours with previous results. In particular, let us consider \cite{LM05comm} and \cite{bury2010automorphic}, which contain explicit descriptions of ALiAs with dihedral symmetry with poles restricted to the orbit of two points, which we call $\Gal=\{0,\infty\}$.

If we take $\eta_1$, $\eta_2$ and $\eta_3$ to be the generators of the Polynomial ALiA of Example \ref{ex:dihedral palia}, then we can recover the generators (27) in \cite{LM05comm}, page 190, where $N=2$ and $j=1$.
These are nothing but
$\h{\Gal}(\eta_1)$, $\h{\Gal}(\eta_2)$ and $2\h{\Gal}(\eta_3)$, respectively.

Similarly,  generators (4.10) in \cite{bury2010automorphic}, page 81, for general $N$ and $j$, are $\h{\Gal}(\eta_1)$, $\h{\Gal}(2\eta_2-\eta_1)$ and $4\h{\Gal}(\eta_3)$ respectively, if $N$ is odd, and $\h{\Gal}(\eta_1)$, $\h{\Gal}(\eta_2)$ and $2\h{\Gal}(\eta_3)$ respectively, if $N$ is even.
%
%
\end{Example}

\subsection{A Cartan-Weyl Normal Form}
\label{sec:normal form}
The normal form introduced in this section improves our understanding of the structure of ALiAs and it is an aid in establishing isomorphisms between these Lie algebras. 

Throughout this section let $\mf{g}$ be a semisimple complex Lie algebra of dimension $k$ and rank $\ell$,  
\[k=\dim\mf{g},\qquad \ell=\dim\mf{h},\]
where $\mf{h}$ is a Cartan subalgebra (CSA) of $\mf{g}$. The root system \cite{fulton1991representation, humphreys1972introduction, knapp2002lie} of $\mf{g}$ is denoted by $\roots$ and we fix a subset of positive roots $\roots^+$ and simple roots $\sroots$. Moreover, we introduce a shorthand notation for free $R$-modules,
\begin{equation}
\label{eq:short hand notation free module}
R\left(a_i\;|\;i\in I\right)=\bigoplus_{i\in I} Ra_i,
\end{equation}
where $R$ is a ring.

We define a Cartan-Weyl normal form for ALiAs similar to the Cartan-Weyl basis for semisimple Lie algebras, that is, a basis that identifies a Cartan subalgebra and diagonalises its adjoint action. 
\begin{Definition}[Cartan-Weyl normal form]
\label{def:cartan-weyl normal form}
The $k$ generators of an ALiA $\salia{g}{V}{\Gamma}$ as a $\CC[\II]$-module are in Cartan-Weyl normal form 
if they diagonalise the adjoint action of $\ell$ of these generators, with the same roots as $\mf{h}$ in the Cartan-Weyl basis for $\mf{g}$. In other words, the set $\{h_\alpha, e_\beta\;|\;\alpha\in\sroots,\;\beta\in\roots\}$ defines a Cartan-Weyl normal form if 
\[\salia{g}{V}{\Gamma}=\CC[\II]\left(h_\alpha, e_\beta\;|\;\alpha\in\sroots,\;\beta\in\roots\right)\] and 
\begin{align*}
&[h,h']=0,\\
&[h,e_{\beta}]=\beta (h)e_{\beta},
\end{align*}
for all $h, h'\in \CC[\II]\left(h_\alpha\;|\;\alpha\in\sroots\right)$ and $\beta\in\roots$. Here a root $\beta$ is considered as a $\CC[\II]$-linear form on $\CC[\II]\left(h_\alpha\;|\;\alpha\in\sroots\right)$ defined by $\beta (h_\alpha)=\frac{2(\beta,\alpha)}{(\alpha,\alpha)}$, where $(\cdot,\cdot)$ is the inner product on $\CC\sroots$ induced by the Killing form of $\mf{g}$.
\end{Definition}
Notice that $\CC[\II]\left(h_\alpha\;|\;\alpha\in\sroots\right)$ is an abelian-, in particular nilpotent subalgebra which is selfnormalising, i.e.~a CSA of $\salia{g}{V}{\Gamma}$.
It follows from the eigenvalues of this CSA and Jacobi's identity that $[e_\beta,e_\gamma]=0$ if $\beta+\gamma\notin\roots$ and  $[e_\beta,e_\gamma]\in\CC[\II]e_{\beta+\gamma}$ if $\beta+\gamma\in\roots$. The restriction of the Killing form of $\salia{g}{V}{\Gamma}$ to a CSA depends only on the Lie brackets appearing in the definition of the Cartan-Weyl normal form, and is therefore identical to the Killing form on the CSA of the base Lie algebra. 
The Cartan element $\frac{(\alpha,\alpha)}{2}h_\alpha$ is dual to $\alpha$ with respect to this inner product, and invariance of the Killing form $K(\cdot,\cdot)$ implies $[e_\alpha,e_{-\alpha}]=K(e_\alpha,e_{-\alpha})\frac{(\alpha,\alpha)}{2}h_\alpha$, cf.~\cite{humphreys1972introduction}. It is in the last two types of brackets where ALiAs differ from semisimple complex Lie algebras. Further specification of these Lie brackets leads to the so called \emph{Chevalley basis} in the complex case and we aim to define an analogues normal form in Definition \ref{def:lie algebra associated to a cocycle} for Lie algebras over a graded ring.

The existence of a Cartan-Weyl normal form relies on the fact that the ALiA $\overline{\mf{g}(V)}^G_\Gamma$ can be generated by $k$ elements over $\CC[\II]$, as was shown in Theorem \ref{thm:dimV generators}. This is necessary but not sufficient to prove the existence of a Cartan-Weyl normal form for a general ALiA, and the problem is still open.
\begin{Conjecture}
\label{conj:existence serre normal form}
For all Automorphic Lie Algebras, a Cartan-Weyl normal form exists.
\end{Conjecture}
There are several reasons why we deem this plausible. First of all, many particular cases are computed. For instance, in Section \ref{sec:dihedral alias} we will find the normal form for all ALiAs with dihedral symmetry, for any orbit of poles, and in \cite{knibbeler2014higher} the $\mf{sl}_n(\CC)$-based ALiAs with exceptional pole orbits are classified and a normal form for each case is computed.
To discuss general reasons for the existence conjecture we break it down into easier problems.
First we need a Cartan subalgebra. This is perhaps the hardest part of Conjecture \ref{conj:existence serre normal form} due to its nonlinear nature.
\begin{Conjecture}
\label{conj:existence csa}
For any orbit $\Gamma$, one can find $\ell$ diagonalisable and commuting elements in $\alia{g}{V}{\Gamma}$ with eigenvalues in $\CC$. 
\end{Conjecture}

Notice that by Theorem \ref{thm:evaluated alias} there is a CSA for the base Lie algebra $\mf{g}(V)$ available in the evaluation of the natural representation of the ALiA at any point of the Riemann sphere. 
This not being the case would be an obstruction to the existence of a CSA for the ALiA.

The linear part of the construction of Cartan-Weyl normal form generators is the diagonalisation of the adjoint action of the CSA. 
\begin{Conjecture}
\label{conj:diagonalisation}
Suppose $\CC[\II](h_1,\ldots,h_\ell)$ is a CSA for $\salia{g}{V}{\Gamma}$.
That is, $h_1,\ldots,h_\ell$ are diagonalisable, pairwise commuting, and have $\lambda$-independent eigenvalues.
Then there exists a nonzero solution $e_\alpha\in\salia{g}{V}{\Gamma}$ to the equations
\[[h_r, e_\alpha]=\alpha(h_r) e_\alpha,\qquad 1\le r\le \ell,\]
for all roots $\alpha\in\roots$ of the base Lie algebra.
\end{Conjecture}
It can be helpful to approach this problem in terms of homogeneous matrices over $\CC[U]$ rather than the matrices over $\mero$, using the identification $\salia{g}{V}{\Gamma}\cong \p{|G|}\palia{g}{V}\mod F_\Gamma$ from Proposition \ref{prop:prehomogenised palia is alia}. One can then restrict the diagonalisation to a fixed degree (the degree of $e_\alpha$ is unknown, but it can be taken as a variable), by which the $\CC$-linear problem becomes finite dimensional. In the case $\mf{g}=\mf{sl_2}(\CC)$ one can then simply count the number of equations and the number of variables and conclude that the solutions $e_\alpha$ exists at some degree. For general Lie algebras $\mf{g}$ it is more difficult.

\begin{Conjecture}
Let  $\CC[\II](h_1,\ldots,h_\ell)$ be a CSA of $\salia{g}{V}{\Gamma}$ and 
$e_\alpha$, $\alpha\in\roots$,  a set of solutions of the equations $[h_r, e_\alpha]=\alpha(h_r) e_\alpha$, $1\le r\le \ell$ as in the above conjecture. Suppose $e_\alpha$ does not have a factor of $\CC[\II]$ vanishing on a generic orbit. Then the transformation $(\eta_1,\ldots,\eta_k)\mapsto(h_r, e_\alpha\;|\;1\le r \le \ell, \alpha\in\roots)$ is invertible over $\CC[\II]$, i.e.~the determinant of the corresponding matrix is in $\CC^\ast$, i.e.~\[\salia{g}{V}{\Gamma}=\CC[\II]\left(\eta_1,\ldots,\eta_{k}\right)=\CC[\II]\left(h_r,\;e_\alpha\;|\;1\le r\le \ell,\;\alpha\in\roots\right)\]
and $\left\{h_r,\;e_\alpha\;|\;1\le r\le \ell,\;\alpha\in\roots\right\}$ defines a Cartan-Weyl normal form for $\salia{g}{V}{\Gamma}$.
\end{Conjecture}

This statement becomes plausible when considering evaluations of the invariant matrices.
If $\mu\in\overline{\CC}\setminus\Gamma$ then the evaluated generators $h_r(\mu)$ form a basis for a CSA $\mf{h}$ of $\mf{g}$ by assumption. In particular, one can find an element $h$ in the $\CC$-span of $\{h_r(\mu)\;|\;1\le r\le \ell\}$ such that $\alpha(h)\ne\beta(h)$ if $\alpha\ne\beta$ are two roots.
Nonzero elements of $\{e_\alpha(\mu)\;|\;\alpha\in\roots\}$ are eigenvectors of $h$ with distinct eigenvalues and as such linearly independent. 

By Proposition \ref{prop:evaluating invariant vectors}, the original generators $\eta_1(\mu),\ldots, \eta_{k}(\mu)$ evaluated at $\mu$ are independent if $\mu$ is a generic point outside $\Gamma$. Thus if $\mu$ is such a point and $e_\alpha(\mu)=f^\alpha_1(\mu)\eta_1(\mu)+\ldots+f^\alpha_{k}(\mu)\eta_{k}(\mu)=0$, then all the functions $f^\alpha_1,\ldots,f^\alpha_{k}\in\CC[\II]$ vanish at $\mu$. This implies they have a common factor, contradicting the assumption that factors of $e_\alpha$ are removed.
We have proved that the transformation $(f^i_j(\mu))$ between the original generators $\eta_j(\mu)$ and the realisation of a Cartan-Weyl normal form $\{h_r,\;e_\alpha\;|\;1\le r\le \ell,\;\alpha\in\roots\}$ is invertible if $\mu$ is generic and not in $\Gamma$. 

What if $\mu_i\in\Gamma_i\ne\Gamma$ is an exceptional value? The matrices $\{h_\ell(\mu_i)\}$ are independent so in that case the matrices $\{e_\alpha(\mu_i)\;|\;\alpha\in\roots\}$ are dependent by Proposition \ref{prop:evaluating invariant vectors}. However, all the \emph{nonzero} matrices $e_\alpha(\mu_i)$ have different eigenvalues relative to the CSA, so \emph{these} are independent. Also, the space spanned by $\{e_\alpha(\mu_i)\;|\;\alpha\in\roots\}$ is a subspace of the space spanned by $\{\eta_j(\mu_i)\;|\;1\le j\le k\}$.  Therefore, the minimal number of vanishing eigenvectors $e_\alpha(\mu_i)$ is $k-\dim \CC\langle \eta_1(\mu_i),\ldots, \eta_{k}(\mu_i)\rangle=\codim\,\mf{g}(V)^{g_i}$. In the next section we will see all these things happen when we construct a Cartan-Weyl normal form for the dihedral ALiAs.

\subsection{Dihedral Automorphic Lie Algebras}
\label{sec:dihedral alias} 

This section is the culmination of the example sequence related to dihedral invariants and describes the main result of \cite{ knibbeler2014automorphic}.
We find the Lie algebra structure of ALiAs with dihedral symmetry and an arbitrary pole orbit, by a construction of a Cartan-Weyl normal form.



\begin{Theorem}
\label{thm:main2}

Let $\DD_N$ act faithfully on the Riemann sphere and let $\Gamma$ be a single orbit therein.
Let $(c_{\al},c_\be)\in\CC P^1$ be such that $\cN(\Gamma)$ is generated by \[F=F_\Gamma^{\nu_\Gamma}=c_{\al}F_\al^N+c_\be F_\be ^2\] where $F_i$ is given by (\ref{eq:Fs}). 
Then, adopting the previous notation,
\[\salia[\DD_N]{sl}{V_{\psi_j}}{\Gamma}=\CC[\II]h\oplus \CC[\II]e_+\oplus \CC[\II]e_-\]
and 
\begin{align*}
&[h,e_\pm]=\pm 2e_\pm, \\
&[e_+,e_-]=\ial\ibe\iga h,
\end{align*}
where $\ii$, $i\in\Omega$, is given by (\ref{eq:automorphic function}).
A set of generators for this normal form is given by
\begin{equation}
\label{eq:dihedral invariant matrices in normal form}
\begin{array}{lcl}
h&=&\begin{pmatrix}
\frac{c_{\be}\fbe \fga}{F}   & (c_{\al}\lambda^j+\nicefrac{c_{\be}}{2}(\lambda^j+\lambda^{j-N}))\ial\\
(c_{\al}\lambda^{-j}+\nicefrac{c_{\be}}{2}(\lambda^{-j}+\lambda^{N-j}))\ial & \frac{-c_{\be}F_\be F_\ga}{F}  
\end{pmatrix},\\
e_+&=&\frac{F_\al^N F_\be F_\ga}{2F^2}\begin{pmatrix}
1&-\lambda^j\\
\lambda^{-j}&-1
\end{pmatrix},\\
e_-
&=&\frac{F_\al^NF_\be F_\ga  }{2F^2}
\begin{pmatrix}
c_{\al}+c_{\be}+c_{\al}c_{\be}\iga&\left(c_{\al}(c_{\al}+c_{\be})\ial+c_{\be}\lambda^{-N}\right)\lambda^{j}\\
-\left(c_{\al}(c_{\al}+c_{\be})\ial+c_{\be}\lambda^N\right)\lambda^{-j}& -c_{\al}-c_{\be}-c_{\al}c_{\be}\iga
\end{pmatrix}\,,
\end{array}
\end{equation}
in the basis corresponding to (\ref{eq:standard matrices}).
\end{Theorem}
\begin{proof}
We prove the theorem by an explicit construction.
Consider the prehomogenised Lie algebra $\p{|2N|}\palia[\DD_N^\flat]{sl}{V_{\psi_j}}=\CC[F_\al^N,F_\be ^2]\left(\tilde{\eta}_1\oplus\tilde{\eta}_2\oplus\tilde{\eta}_3\right)$ from Example \ref{ex:dihedral alia}. Define $(\tilde{h},\tilde{e}_+,\tilde{e}_-)$ by \[(\tilde{h},\tilde{e}_+,\tilde{e}_-)=(\tilde{\eta}_1,\tilde{\eta}_2,\tilde{\eta}_3)T\] where 
\begin{equation}
\label{eq:transformation to normal form}
T=\frac{1}{2}\begin{pmatrix}
2c_{\al}&-F_\be ^2&c_{\al}^2F_\al^N{F_\ga}^2+F^2\\
2c_{\be}&F_\al^N&-c_{\be}^2F_\be ^2{F_\ga}^2-F^2\\
2c_{\be}&F_\al^N&F_\al^N(c_{\al}c_{\be}{F_\ga}^2+(c_{\al}+c_{\be})F)
\end{pmatrix}\,.
\end{equation}
It is a straightforward though very tedious exercise to check that this transformation has determinant $\nicefrac{1}{2}F^3$ and is therefore invertible over the ring $\CC[\fal^\nal, \fbe^\nbe]\mod F\cong\CC[\II] $. Moreover, one computes the commutators
\begin{align*}
&[\tilde{h},\tilde{e}_\pm]=\pm 2F\tilde{e}_\pm ,
\\&[\tilde{e}_+,\tilde{e}_-]=FF_\al^NF_\be ^2F_\ga^2 \tilde{h} \,,
\end{align*}
for instance using the commutation relations of $\tilde{\eta}_i$ in Example \ref{ex:dihedral alia}.
Hence the homogenised matrices $h=\h{\Gamma}(\tilde{h})$ and $e_{\pm}=\h{\Gamma}(\tilde{e}_{\pm})$ define a Cartan-Weyl normal form as described in the theorem.
\end{proof}
We observe that different choices of transformation groups have been made in previous works. Here we follow \cite{LS10} and allow invertible transformations $T$ on the generators $\tilde{\eta}_i$ over the ring of invariant forms, contrary to \cite{bury2010automorphic}, where only $\CC$-linear transformations on the generators are considered. 
The last approach preserves the quasigrading of the Lie algebra (cf.~\cite{LM05comm}).

To find the desired isomorphism $T$ one first looks for a matrix $\tilde{h}\in \CC[F_\al^N,F_\be ^2]\left(\tilde{\eta}_1\oplus\tilde{\eta}_2\oplus\tilde{\eta}_3\right)$ which has eigenvalues $\pm\mu$ that are  units in the localised ring, that is, powers of $F$. This yields the first column of $T$. The other two columns are found by diagonalising $\ad(\tilde{h})$, i.e.~solving $[\tilde{h},\tilde{e}_\pm]=\pm 2 \mu \tilde{e}_\pm$ for $\tilde{e}_\pm\in\CC[F_\al^N,F_\be ^2]\left(\tilde{\eta}_1\oplus\tilde{\eta}_2\oplus\tilde{\eta}_3\right)$. Then one has to check that the transformation is invertible.

In Table \ref{tab:invariant matrices with exceptional poles} we present the invariant matrices when $\Gamma$ is one of the exceptional orbits $\Gamma_i$. In other words, these are the matrices (\ref{eq:dihedral invariant matrices in normal form}) with $(c_{\al},c_{\be})=(1,0)$ for poles at $\Gal$, $(c_{\al},c_{\be})=(0,1)$ for poles at $\Gbe$ and $(c_{\al},c_{\be})=(-1,1)$ for poles at $\Gga$.

\begin{table}[h!]
\caption{Generators in normal form of Automorphic Lie Algebras with dihedral symmetry and exceptional pole orbit.}
\label{tab:invariant matrices with exceptional poles}
\begin{center}
\scalebox{0.914}{
\begin{tabular}{lccc}
 $ $ & $h$ & $e_+$ & $e_-$ \\
\hline
$\Gal$ &
$\begin{pmatrix}0&\lambda^j\\\lambda^{-j}&0\end{pmatrix}$ & 
$\frac{\lambda^{2N}-1}{8\lambda^N}\begin{pmatrix}1&-\lambda^j\\\lambda^{-j}&-1\end{pmatrix}$ & 
$\frac{\lambda^{2N}-1}{8\lambda^N}\begin{pmatrix}1&\lambda^j\\-\lambda^{-j}&-1\end{pmatrix}$ \\
$\Gbe$ &
$\frac{1}{\lambda^N+1}\begin{pmatrix}\lambda^N-1&-2\lambda^j\\2\lambda^{N-j}&-\lambda^N+1\end{pmatrix}$ & 
$\frac{2(\lambda^N-1)\lambda^N}{(\lambda^N+1)^3}\begin{pmatrix}1&-\lambda^j\\\lambda^{-j}&-1\end{pmatrix}$ & 
$\frac{2(\lambda^N-1)\lambda^N}{(\lambda^N+1)^3}\begin{pmatrix}1&\lambda^{j-N}\\-\lambda^{N-j}&-1\end{pmatrix}$ \\
$\Gga$ &
$\frac{1}{\lambda^N-1}\begin{pmatrix}\lambda^N+1&-2\lambda^j\\2\lambda^{N-j}&-\lambda^N-1\end{pmatrix}$ & 
$\frac{2(\lambda^N+1)\lambda^N}{(\lambda^N-1)^3}\begin{pmatrix}1&-\lambda^j\\\lambda^{-j}&-1\end{pmatrix}$ & 
$\frac{2(\lambda^N+1)\lambda^N}{(\lambda^N-1)^3}\begin{pmatrix}-1&\lambda^{j-N}\\-\lambda^{N-j}&1\end{pmatrix}$ \\
\end{tabular}
}
\end{center}
\end{table}


Theorem \ref{thm:main2} describes ALiAs with dihedral symmetry, with poles at any orbit. Its proof exhibits the consequence of Proposition \ref{prop:prehomogenised palia is alia} that one can compute all these Lie algebras in one go. In particular, one does not need to distinguish between exceptional orbits and generic orbits. 

The resulting Lie algebras differ only in the bracket $[e_+,e_-]=\ial\ibe\iga h$. Here one can discriminate between the case of generic and exceptional orbits since precisely one factor $\ii$ equals $1$ if and only if $\Gamma$ is an exceptional orbit. In other words, $\ial\ibe\iga$ is a polynomial in $\II$ (for any $1\ne\II=\ii$, $i\in\Omega$) of degree $2$ or $3$ if $\Gamma$ is exceptional or generic respectively. Notice that this degree equals the complex dimension of the quotient 
\[\bigslant{\salia[\DD_N]{sl}{V_{\psi_j}}{\Gamma}}{\left[\salia[\DD_N]{sl}{V_{\psi_j}}{\Gamma},\salia[\DD_N]{sl}{V_{\psi_j}}{\Gamma}\right]}\cong \bigslant{\CC[\II]}{\CC[\II]\ial\ibe\iga}\]
in agreement with results in \cite{bury2010automorphic}.

Another way to analyse the ALiAs
is by considering their values at a particular point $\lambda$. This can be done in the Lie algebra, i.e.~the structure constants, or in the natural representation of the Lie algebra, i.e.~the invariant matrices (cf.~Theorem \ref{thm:evaluated alias}).
In the first case we have a three-dimensional Lie algebra, equivalent to $\mf{sl}_2(\CC)$ if $\lambda\in \overline{\CC}\setminus (\Gamma\cup\Gga\cup\Gbe\cup\Gal)$. When on the other hand $\lambda\in \Gamma_i\ne\Gamma$ we obtain the Lie algebra 
\[
[h,e_\pm]=\pm 2 e_\pm\,,\quad [e_+,e_-]=0 \,.
\]

Evaluating the invariant matrices in $\lambda\in \overline{\CC}\setminus (\Gamma\cup\Gga\cup\Gbe\cup\Gal)$ also results in $\mf{sl}_2(\CC)$. If, on the other hand, $\lambda\in\Gamma_i\ne\Gamma$ then two generators $e_\pm$ vanish and one is left with a one-dimensional, and in particular commutative, Lie algebra, in agreement with Theorem \ref{thm:evaluated alias}.

\subsection{Matrices of Invariants}
\label{sec:moi}

Invariant matrices act on invariant vectors by multiplication. The description of the invariant matrices in terms of this action yields greatly simplified matrices, which we call \emph{matrices of invariants}, while preserving the structure constants of the Lie algebra. This provides a convenient representation for Automorphic Lie Algebras. We follow \cite{knibbeler2014higher}.

The claimed action relies on the fact that a product $\eta \upsilon$  of an invariant matrix $\eta$ and an invariant vector $\upsilon$ is again an invariant vector. We check 
\[g\cdot (\eta(\lambda)\upsilon(\lambda))=\tau_g\eta(g^{-1}\lambda) \upsilon(g^{-1}\lambda) =(\tau_g\eta(g^{-1}\lambda)\tau_g^{-1})(\tau_g \upsilon(g^{-1}\lambda) )=\eta(\lambda) \upsilon(\lambda)\]
where $\tau:G^\flat\rightarrow \GL(V)$ defines the action on the underlying vector space $V$.

Let the generators $R^{\chi_V}_{|G|}$ of the module of invariant vectors be $\upsilon_1,\ldots,\upsilon_n$, $n=\dim V$,  i.e.~\[\p{|G|}\left(V\otimes R\right)^{G^\flat}=\bigoplus_{j=1}^{n}\CC[\fal^\nal,\fbe^\nbe]\upsilon_j\] and define the \emph{matrix of invariants} $A=(Q_{r,s}(\fal^\nal,\fbe^\nbe))$ related to an invariant matrix $\eta$  by \[\eta\upsilon_s=\sum_{r=1}^{n} Q_{r,s}\upsilon_{r}.\]
In terms of the square matrix $P=(\upsilon_1,\ldots,\upsilon_n)$ this reads
\[\eta P=PA.\]
The entries $Q_{r,s}$ of $A$ are forms of degree $\frac{\deg\eta}{|G|}$ in $\fal^\nal$ and $\fbe^\nbe$.
Important to note is that $A$ is not an invariant matrix in the usual sense. This is the reason for using a Roman instead of a Greek letter.

We know that $\det P=\det R^\chi_{|G|}=\prod_{i\in\Omega}F_i^{\kappa(\chi)_i\nu_i}$ where $\kappa(\chi)_i=\nicefrac{1}{2}\;\codim\,V_\chi^{\langle g_i \rangle}$ (cf.~Definition \ref{def:determinant of invariant vectors} and Theorem \ref{thm:determinant of invariant vectors}). In particular, $P$ is not invertible on the whole Riemann sphere. 
The singularity of $P$ means that $A$ cannot simply be seen as $\eta$ in a different basis, but we want to make sure that the Lie algebra structure is preserved nonetheless. Suppose \[[\eta_i,\eta_j]=\sum_{k=1}^{n}c^k_{i,j}\eta_k,\] where $c^k_{i,j}\in\CC[\fal^\nal,\fbe^\nbe]$. 
By definition
\[0=\left([\eta_i,\eta_j]-\sum_{k=1}^{n}c^k_{i,j}\eta_k\right)P=P\left([A_i,A_j]-\sum_{k=1}^{n}c^k_{i,j}A_k\right)\]
where $A_i$ is the matrix of invariants related to the invariant matrix $\eta_i$. At the points on the Riemann sphere where $P$ is invertible we thus find that $[A_i,A_j]-\sum_{k}c^k_{i,j}A_k=0$. Since these regular points include all generic orbits they constitute a non-discrete set and we can conclude that the polynomial entries of $[A_i,A_j]-\sum_{k}c^k_{i,j}A_k$ are in fact identically zero,
\[[A_i,A_j]=\sum_{k=1}^n c^k_{i,j}A_k,\]
and the transformation $\eta_i\mapsto A_i$ defines a Lie algebra morphism (see \cite{knibbeler2014higher} for an algebraic argument to this end).

To illustrate we compute the matrices of invariants for the dihedral group. We consider the original generators (\ref{eq:prehomogenised dihedral invariant matrices}) of prehomogenised $\DD_N$-invariant matrices first. We could use the generators of the normal form (\ref{eq:dihedral invariant matrices in normal form}), but these are more involved (only the structure constants are simpler) and we will consider them later. The transformation on the generators of the Lie algebra, such as the one achieving a Cartan-Weyl normal form, and the transformation on the matrix representatives of these generators, such as the transition between invariant matrices and matrices of invariants, commute. Hence we can apply transformation (\ref{eq:transformation to normal form}) of Theorem \ref{thm:main2} afterwards if we wish to.
\begin{Example}[Matrices of invariants for $\DD_N$]
\label{ex:moi}
As starting point we consider the original generators (\ref{eq:prehomogenised dihedral invariant matrices}) of the prehomogenised module of $\DD_N$-invariant matrices,
\begin{align*}
&\tilde{\eta}_1=\fal^N\begin{pmatrix}0&\lambda^j\\\lambda^{-j}&0\end{pmatrix},\\
&\tilde{\eta}_2=\fal^N\begin{pmatrix}0&\nicefrac{1}{2}(\lambda^j+\lambda^{j-N})\\\nicefrac{1}{2}(\lambda^{-j}+\lambda^{N-j})&0\end{pmatrix},\\
&\tilde{\eta}_3=\fbe\fga \begin{pmatrix}1&0\\0&-1\end{pmatrix}.
\end{align*}
Suppose $j$ is even and consider the invariant vectors of degree $|G|$ 
\[\upsilon_1=\fal^{N-\nicefrac{j}{2}}\begin{pmatrix}X^j\\Y^j\end{pmatrix},\qquad \upsilon_2=\fal^{\nicefrac{j}{2}}\fbe\begin{pmatrix}Y^{N-j}\\X^{N-j}\end{pmatrix}.\]
If $P=(\upsilon_1, \upsilon_2)$ and $\tilde{\eta}_i P=P A_i$ then the matrices $A_i$ are given by
\begin{equation}
\label{eq:dihedral moi}
A_1=\begin{pmatrix}\fal^N&2\fbe^2\\0&-\fal^N\end{pmatrix},\qquad
A_2=\begin{pmatrix}0&\fbe^2\\\fal^N&0\end{pmatrix},\qquad
A_3=\begin{pmatrix}\fbe^2&\fbe^2\\-\fal^N&-\fbe^2\end{pmatrix}
\end{equation}
and we see that all entries are in $R^{G^\flat}_{|G|}=\CC\fal^\nal\oplus\CC \fbe^\nbe$ as expected.

\end{Example}

We defined the determinant of invariant vectors in Definition \ref{def:determinant of invariant vectors}. This can be computed just as well when the base vector space is a space of matrices. But since then one can also consider the determinant of one vector, we call the former determinant the \emph{total determinant of invariant matrices}. That is, relative to a basis $\{e_1,\ldots,e_k\}$ for $\mf{g}(V)$, if the invariant matrices are given by $\eta_i=f_{i,j}e_j$, then the total determinant of invariant matrices is $\det f_{i,j}$. A change of basis $e_i=c_{i,j}e'_j$ will only change the total determinant by a factor $\det c_{i,j}\in\CC^\ast$, therefore the total determinant is well defined.
\begin{Lemma}
\label{lem:determinant of moi}
The total determinants of invariant matrices and their related matrices of invariants are related by a factor $\pm 1$; 
\[\det(\eta_1,\ldots,\eta_k)=\pm 1 \det(A_1,\ldots,A_k).\]
In particular, they are identical up to scalar factor.
\end{Lemma}
\begin{proof}
Let $P$ be the matrix of invariant vectors such that $\eta_j P=P A_j$ and suppose $D\subset\CC^2$ is the set of points $(X,Y)$ such that $\frac{X}{Y}\in\overline{\CC}\setminus(\Gal\cup\Gbe\cup\Gga)$, i.e.~$D$ is the preimage under the canonical quotient map $\CC^2\rightarrow\CC P^1\cong\overline{\CC}$ of all points in the Riemann sphere with trivial stabiliser subgroup. The topological closure of $D$ is $\CC^2$ and $D$ is connected, as can be seen from its continuous image in the Riemann sphere.

The total determinants of the invariant matrices and the matrices of invariants are polynomial functions on $\CC^2$ and in particular continuous. 
The former is nonzero on $D$, by Proposition \ref{prop:evaluating invariant vectors}, as is the determinant of $P$.
Therefore, in the domain $D$, we have the following properties. Conjugation with $P$ induces a basis transformation of $\mf{g}(V)$ with nonzero determinant $\det \Ad P$ and 
\[\det(\eta_1,\ldots,\eta_k)(X,Y)=\det \Ad P(X,Y) \det(A_1,\ldots,A_k)(X,Y),\qquad (X,Y)\in D.\]  
The basis transformation given by $\Ad P(X,Y)$ is an automorphism of a semisimple Lie algebra, for which Lemma \ref{lem:determinant of automorphism} shows $\det \Ad P(X,Y) =\pm 1$. Moreover, since both total determinants are continuous, so is $\det \Ad P:D\rightarrow \{1,-1\}$. Finally, because $D$ is connected, $\det \Ad P$ is constant. By continuity of the total determinants we can fill in the gaps $\CC^2\setminus D$ since $D$ is dense in $\CC^2$.
\end{proof}

\begin{Example}[The total determinant]
\label{ex:total determinant}
Consider the $\mf{sl}_2(\CC)$-basis \[\left\{ \begin{pmatrix}1&0\\0&-1\end{pmatrix},\;\begin{pmatrix}0&1\\0&0\end{pmatrix},\;\begin{pmatrix}0&0\\1&0\end{pmatrix}\right\}.\] The $\DD_N$-invariant matrices (\ref{eq:prehomogenised dihedral invariant matrices}) have coefficients $\tilde{\eta}_1\cong (0,\fal^{N-j}X^{2j},\fal^{N-j}Y^{2j})$ etc. We check that $R^{\chi_2+\psi_{2j}}_{|G|}=\CC\fal^N\fbe^2\fga^2$ as predicted by Theorem \ref{thm:determinant of invariant vectors} and Table \ref{tab:1/2codimV^g},
\[\det(\tilde{\eta}_1,\tilde{\eta}_2,\tilde{\eta}_3)=
\det\begin{pmatrix}0&\fal^{N-j}X^{2j}&\fal^{N-j}Y^{2j}\\
0&\fal^{j}\fbe Y^{N-2j}&\fal^{j}\fbe X^{N-2j}\\
\fbe\fga&0&0\end{pmatrix}=2\fal^N\fbe^2\fga^2.\]
To illustrate Lemma \ref{lem:determinant of moi}, we compute the total determinant of the matrices of invariants (\ref{eq:dihedral moi}) as well. Notice first that conjugation gives an inner automorphism of the special linear algebra, so we expect the same sign. 
\begin{align*}
&\det(A_1,A_2,A_3)=
\det\begin{pmatrix}\fal^N&2\fbe^2&0\\
0&\fbe^2&\fal^N\\
\fbe^2&\fbe^2&-\fal^N\end{pmatrix}=\\
&\fal^N(-2\fal^N\fbe^2)+2\fbe^2(\fal^N\fbe^2)=2\fal^N\fbe^2(\fbe^2-\fal^N)=2\fal^N\fbe^2\fga^2.
\end{align*}
\end{Example}

\begin{Example}[Matrices of invariants in normal form]
\label{ex:moi in normal form}
Let us apply transformation (\ref{eq:transformation to normal form}), achieving the Cartan-Weyl normal form, to the matrices of invariants (\ref{eq:dihedral moi}). The homogenised version \footnote{It is arguably easier to postpone the homogenisation procedure, because in order to obtain the given simplified matrices one has to replace $1$ by appropriate powers of $c_\al\ial+c_\be\ibe$ in various places.}\[(H, E_+, E_-)=\h{\Gamma}\left((A_1,A_2,A_3)T\right)=(\h{\Gamma}A_1,\h{\Gamma}A_2,\h{\Gamma}A_3)\h{\Gamma}T\]
takes the form
\[\h{\Gamma}A_1=\begin{pmatrix}\ial&2\ibe\\0&-\ial\end{pmatrix},\qquad
\h{\Gamma}A_2=\begin{pmatrix}0&\ibe\\\ial&0\end{pmatrix},\qquad
\h{\Gamma}A_3=\begin{pmatrix}\ibe&\ibe\\-\ial&-\ibe\end{pmatrix},\]
\begin{equation}
\h{\Gamma}T=\frac{1}{2}\begin{pmatrix}
2c_{\al}&-\ibe&c_{\al}^2\ial\iga+1\\
2c_{\be}&\ial&-c_{\be}^2\ibe\iga-1\\
2c_{\be}&\ial&\ial(c_{\al}c_{\be}\iga+c_{\al}+c_{\be})
\end{pmatrix}\,.
\end{equation}
After simplification, using the two relations
\begin{align*}
\ial-\ibe+\iga&=0,\\
c_\al\ial+c_\be\ibe&=1,
\end{align*}
one finds
\begin{align*}
H&=\begin{pmatrix}1&2(c_\al+c_\be)\ibe\\0&-1\end{pmatrix},\\
E_+&=\begin{pmatrix}0&-\ibe\iga\\0&0\end{pmatrix},\\
E_-&=\begin{pmatrix}
(c_\al+c_\be)\ial\ibe
&(c_\al+c_\be)^2\ial\ibe^2\\
-\ial&
-(c_\al+c_\be)\ial\ibe
\end{pmatrix}.
\end{align*}
The matrices have been simplified significantly compared to (\ref{eq:dihedral invariant matrices in normal form}) but there is one more thing to do. By construction, $H$ is diagonalisable, so let us get this done as well \footnote{The diagonalisation was not carried out for the matrices (\ref{eq:dihedral invariant matrices in normal form}) because it interferes with the group action (as does the basis of invariant vectors). The diagonalising basis for $V$ depends on $\lambda$ and therefore the matrices representing $G$ by $\tau:G\rightarrow\GL(V)$ become dependent on $\lambda$ as well (cf.~twisted reduction group, \cite{Lombardo}). Theorem \ref{thm:main2} is concerned with invariant matrices, contrary to this example, where we are only concerned with the Lie structure.}, with basis
\[P=\begin{pmatrix}1&(c_\al+c_\be)\ibe\\0&-1\end{pmatrix}.\]
This matrix is invertible over $\CC[\II]$. In fact, $P^2=\Id$, and 
\[P^{-1}\begin{pmatrix}u&v\\w&-u\end{pmatrix}P=
\begin{pmatrix}
u+(c_\al+c_\be)\ibe w&
2(c_\al+c_\be)\ibe u-v+(c_\al+c_\be)^2\ibe^2 w\\
-w&
-u-(c_\al+c_\be)\ibe w\end{pmatrix}\] 
and
\begin{equation}
\label{eq:dihedral alias serre model}
H\sim\begin{pmatrix}1&0\\0&-1\end{pmatrix},\qquad
E_+\sim\begin{pmatrix}0&\ibe\iga\\0&0\end{pmatrix},\qquad
E_-\sim\begin{pmatrix}0&0\\\ial&0\end{pmatrix}.
\end{equation}
This is a set of generators that reflect the simplicity of the structure constants clearly.
\end{Example}

\section{Applications of Invariants}
\label{sec:structure theory for alias}

If we assume the existence of a Cartan-Weyl normal form for any ALiA, i.e.~Conjecture \ref{conj:existence serre normal form}, then we can say quite a few things about this Lie algebra. 
The investigation of the predictive power of the invariants of ALiAs naturally leads to a cohomology theory for root systems. 


\subsection{Root Cohomology}

This section presents a tentative setup for a cohomology theory of root systems and explains its relation to Automorphic Lie Algebras and their representations.

Let \(q\in \mb{N}\) and let $\roots$ be a root system. Put $\roots_0=\roots\cup\{0\}$.
We define $1$-chains $C_1(\roots)$ by formal $\ZZ$-linear combinations of roots.
Inductively we define $m$-chains $C_m(\roots)$ by formal $\ZZ$-linear combinations of $m$-tuples of roots $(\alpha_1 ,\ldots, \alpha_m)\in\roots_0^m$ with the property that $ (\alpha_1 ,\ldots,\alpha_j+\alpha_{j+1},\ldots, \alpha_m)\in C_{m-1}(\roots)$ for all $1\le j< m$.

Dually, let the \emph{$m$-cochains} be the $\ZZ$-linear maps \[C^m(\roots,\mb{Z}^q)=\Hom(C_m(\roots),\mb{Z}^q).\] 
The maps $\mathsf{d}^m:C^m(\roots,\mb{Z}^q)\rightarrow C^{m+1}(\roots,\mb{Z}^q)$ can then be defined in the standard way \cite{harrison1962commutative} (with trivial \(\roots\)-action on \(\mb{Z}^q\)).
\begin{align*}
\mathsf{d}^1 \omega^1(\alpha_0,\alpha_1)&= \omega^1(\alpha_1)-\omega^1(\alpha_0+\alpha_1)+\omega^1(\alpha_0)\\
\mathsf{d}^2 \omega^2(\alpha_0,\alpha_1,\alpha_2)&= \omega^2(\alpha_1,\alpha_2)-\omega^2(\alpha_0+\alpha_1,\alpha_2)+\omega^2(\alpha_0,\alpha_1+\alpha_2)-\omega^2(\alpha_0,\alpha_1)\\
\mathsf{d}^m \omega^m(\alpha_0,\ldots,\alpha_m)&=\omega^m(\alpha_1,\ldots,\alpha_m)\\
&+\sum_{j=1}^m (-1)^{j} \omega^m(\alpha_0,\ldots,\alpha_{j-1}+\alpha_j,\cdots,\alpha_m)\\
&-(-1)^m\omega^m(\alpha_0,\ldots,\alpha_{m-1}).
\end{align*}
We have $\mathsf{d}^{m+1}\mathsf{d}^m=0$ and define the group of \emph{$m$-cocycles}, \emph{$m$-coboundaries} and the \emph{$m^{\text{th}}$-cohomology group with coefficients in $\ZZ^q$} respectively by
\begin{align*}
Z^m(\roots,\mb{Z}^q)&=\ker \mathsf{d}^m,\\
B^m(\roots,\mb{Z}^q)&=\mathsf{d}^{m-1} C^{m-1}(\roots,\mb{Z}^q),\\
H^m(\roots,\mb{Z}^q)&=\bigslant{Z^m(\roots,\mb{Z}^q)}{B^m(\roots,\mb{Z}^q)}.
\end{align*}

If \(\omega^m\in C^m(\roots,\mb{Z}^q)\) \emph{and} \(\mathsf{d}^m \omega^m\in C^{m+1}(\roots,\mb{Z}^q)\) take their values in \(\NN_0^q\), where $\NN_0=\NN\cup\{0\}$,
we say that \(\omega^m\in C^m(\roots,\NN_0^q)\). Observe that if $\omega^m\in  C^{m}(\roots,\NN_0^q)$ then $\mathsf{d}^m\omega^m\in  C^{m+1}(\roots,\NN_0^q)$, since \(\mathsf{d}^{m+1}\mathsf{d}^m=0\).
Thus \(\mathsf{d}^m\) is well defined on  \(C^m(\roots,\NN_0^q)\).

In what follows we will only use $1$-(co)chains and $2$-(co)chains (with coefficients in $\NN_0$). 
The basis of the $2$-chains $C_2(\roots)$ are the pairs of roots whose sum is a root as well. The $1$-cochains are maps $\omega^1:\roots_0\rightarrow\NN_0$ such that $\mathsf{d}^1\omega^1(\alpha,\beta)=\omega^1(\beta)-\omega^1(\alpha+\beta)+\omega^1(\alpha)\in\NN_ 0$ for all $(\alpha, \beta)\in C_2(\roots)$. 
Symmetric $2$-cocycles define a class of Lie algebras of our interest and $1$-cochains provide representations for such Lie algebras, as is shown below.

\begin{Definition}[Associated Lie algebras $\AL{\roots}$ in Chevalley normal form] 
\label{def:lie algebra associated to a cocycle}
Let  \(\omega^2\in Z^2(\roots,\NN_0^q)\) be symmetric. Define $\AL{\roots}$ as the Lie algebra associated to the $2$-cocycle \(\omega^2\), in the following sense.
The Lie algebra $\AL{\roots}$ is the free $\CC[\II_1,\cdots,\II_q]$-module with generators  $\{h_\alpha, e_\beta\;|\;\alpha\in\sroots,\;\beta\in\roots\}$ and $\CC[\II_1,\cdots,\II_q]$-linear Lie bracket
\[\begin{array}{ll}
[h,h']=0&\text{if }h,h'\in \CC[\II_1,\cdots,\II_q]\left(h_\alpha\;|\;\alpha\in\sroots\right),\\
{[}h, e_{\beta}]=\beta(h)e_{\beta}&\text{if }h\in \CC[\II_1,\cdots,\II_q]\left(h_\alpha\;|\;\alpha\in\sroots\right)\text{ and }\beta\in\roots,\\
{[}e_\beta,e_{-\beta}]=\II^{\omega^2(\beta,-\beta)} \sum_{i=1}^d h_{\alpha_i} & \text{if }\beta=\alpha_1+\ldots+\alpha_d,\;\alpha_i\in\sroots,\\
{[} e_\beta, e_\gamma]=\pm(r+1) \II^{\omega^2(\beta,\gamma)} e_{\beta+\gamma}&\text{if } \beta,\,\gamma,\,\beta+\gamma\in \roots,\\
{[} e_\beta, e_\gamma]=0& \text{if }\beta,\,\gamma\in \roots\text{ and }\beta+\gamma\notin\roots_0.
\end{array}\]
Here we use a multi-index notation  $\II^{\omega^2(\alpha,\beta)}=\prod_{i=1}^q\ii^{\omega^2(\alpha,\beta)_i}$. Moreover, a root $\beta$ is considered as a $\CC[\II_1,\cdots,\II_q]$-linear form on $\CC[\II_1,\cdots,\II_q]\left(h_\alpha\;|\;\alpha\in\sroots\right)$ 
defined by $\beta (h_\alpha)=\frac{2(\beta,\alpha)}{(\alpha,\alpha)}$, where $(\cdot,\cdot)$ is the inner product on $\CC\sroots$ induced by the Killing form of $\mf{g}$.
The integer $r\in\ZZ$ is the largest satisfying $\gamma-r\beta\in\roots$. A consistent choice of signs is originally given by Tits \cite{tits1966constantes} and also described by Kac \cite{kac1994infinite} in terms of a $2$-cocycle on the root lattice with values in $\{\pm1\}$. If the generators of a Lie algebra over a graded ring satisfy the above Lie brackets it is said to be in Chevalley normal form.
\end{Definition}

To prove that $\AL{\roots}$ is a Lie algebra, one has to confirm that the Jacobi identity is satisfied. Given the classical existence proof for a Lie algebra with root system $\roots$ and its Chevalley basis \cite{humphreys1972introduction}, 
the Jacobi identity for $\AL{\roots}$ is equivalent to the assumption that $\mathsf{d}^2\omega^2=0$ and $\omega^2(\alpha,\beta)=\omega^2(\beta,\alpha)$ for all $(\beta,\alpha)\in C_2(\roots)$. Indeed, for the expression $[e_\alpha,[e_\beta,e_\gamma]]+[e_\beta,[e_\gamma,e_\alpha]]+[e_\gamma,[e_\alpha,e_\beta]]$ to vanish, it is necessary and sufficient that each term has the same $\II$-multidegree. That is \[\omega^2(\alpha,\beta+\gamma)+\omega^2(\beta,\gamma)=\omega^2(\beta,\gamma+\alpha)+\omega^2(\gamma,\alpha)=\omega^2(\gamma,\alpha+\beta)+\omega^2(\alpha,\beta).\]
Under the symmetry assumption this is equivalent to \[\mathsf{d}^2\omega^2(\alpha,\beta,\gamma)= \omega^2(\beta,\gamma)-\omega^2(\alpha+\beta,\gamma)+\omega^2(\alpha,\beta+\gamma)-\omega^2(\alpha,\beta)=0.\]
We remark that the Killing form $K$ of $\AL{\roots}$ on its generators is identical to the complex Lie algebra related to $\roots$ apart from the values \[K(e_\alpha,e_{-\alpha})=\II^{\omega^2(\alpha,-\alpha)},\qquad \alpha\in\roots.\]

If $\omega^2$ is a coboundary then $\AL{\roots}$ can be realised in the following way.
\begin{Example}
\label{ex:representation by 1-cochain}
Suppose $\left\{H_\alpha, E_\beta\;|\;\alpha\in\sroots,\; \beta\in\roots\right\}$ is a Chevalley basis for a simple complex Lie algebra and let $\omega^1\in C^1(\roots,\NN_0^q)$ be a $1$-cochain. Then the Lie algebra associated to its coboundary can be faithfully represented by 
\[\AL[\mathsf{d}^1\omega^1]{\roots}\cong\CC[\II_1,\ldots,\II_q]\left(H_\alpha, \II^{\omega^1(\beta)}E_\beta\;|\;\alpha\in\sroots,\; \beta\in\roots\right).\]
Indeed, any $2$-coboundary $\mathsf{d}^1\omega^1$ is a symmetric $2$-cocycle and the commutation relations are easy to check. For instance, if $\beta$, $\gamma$, $\beta+\gamma\in\roots$ then
\begin{align*}
[\II^{\omega^1(\beta)}E_\beta,\II^{\omega^1(\gamma)}E_\gamma]&=\II^{\omega^1(\beta)}\II^{\omega^1(\gamma)}[E_\beta,E_\gamma]\\
&=\pm(r+1)\II^{\omega^1(\beta)+\omega^1(\gamma)}E_{\beta+\gamma}\\
&=
\pm(r+1)\II^{\mathsf{d}^1\omega^1(\beta,\gamma)}(\II^{\omega^1(\beta+\gamma)}E_{\beta+\gamma}).
\end{align*}
We call this concretisation the \emph{canonical representation associated to the 1-cochain $\omega^1$}.
Recall the matrices of invariants (\ref{eq:dihedral alias serre model}) for example.
\end{Example}

The following result is based on the invariants of ALiAs gathered throughout this thesis. For convenience we first define the \emph{norm of a $1$-cochain}, \[\|\omega^1\|=\sum_{\alpha\in\roots}\omega^1(\alpha)\in\NN_0^q,\qquad \omega^1\in C^1(\roots, \NN_0^q),\]
and more generally one can define the norm of a $m$-cochain as the sum of all values on the basis elements of the $m$-chains.
\begin{Theorem}
\label{thm:alias satisfy game of roots}
Let $V$ be an irreducible representation of a binary polyhedral group $G^\flat$ and $\mf{g}(V)$ a  $G^\flat$-submodule of $\mf{gl}(V)$ and a simple Lie subalgebra of rank $\ell$ and with root system $\roots$. Suppose the Automorphic Lie Algebra $\salia{g}{V}{\Gamma}$ allows a Cartan-Weyl normal form 
\[\salia{g}{V}{\Gamma}\cong\CC[\II]\left( h_\alpha, e_\beta\;|\;\alpha\in\sroots,\, \beta\in\roots\right).\]
Then there exists a $1$-cochain $\omega^1\in C^1(\roots, \NN_0^{|\Omega|})$ with $\omega^1(0)=0$ such that \[\salia{g}{V}{\Gamma}\cong \AL[\mathsf{d}^1\omega^1]{\roots}.\]
Moreover, for $i\in\Omega$ and $(\alpha,\beta)\in C_2(\roots)$,
\begin{align}
\label{eq:noI}
\|\omega^1\|_i&=\delta_{\Gamma\Gamma_i} \kappa(\roots)_i\\
\label{eq:omega1le1}
\omega^1(\alpha)_i&\in \{0,1\}\\
\label{eq:d1omega1le1}
 \mathsf{d}^1\omega^1(\alpha,\beta)_i&\in \{0,1\}
\end{align}
where $\delta$ is the Kronecker delta and $\kappa(\roots)$ of Definition \ref{def:kappa(roots)} is given in Table \ref{tab:noI}.
\end{Theorem}
\begin{center}
\begin{table}[h!] 
\caption{
The integers $\kappa(\roots)_i=\nicefrac{1}{2}\,\codim\mf{g}(V)^{\langle g_i\rangle}$, $i\in\Omega$.}
\label{tab:noI}
\begin{center}
\begin{tabular}{cccccccccccc} \hline
$\roots$&$A_1$&$A_2$&$B_2$&$A_3$&
$C_3$&$A_4$&$A_5$
\\
\hline
$\al$&$1$&$3$&${4}$&$6$&
${8}$&$10$&$14$\\
$\be$&$1$&$3$&${3}$&$5$&
${7}$&$8$&$12$\\
$\ga$&$1$&$2$&${3}$&$4$&
${6}$&$6$&$9$\\
\hline
$\Sigma$&$3$&$8$&$10$&$15$&$21$&$24$&$35$\\
\hline 
\end{tabular}
\end{center}
\end{table}
\end{center}
\begin{proof}
Notice that \(\mathsf{d}^1 \omega^1\) is a symmetric $2$-cocycle so that $\AL[\mathsf{d}^1\omega^1]{\roots}$ is a well defined Lie algebra.
Suppose the generators of the Automorphic Lie Algebra are represented by matrices of invariants. Then the generators $\{h_\alpha\,|\,\alpha\in\sroots\}$ of the Cartan subalgebra (CSA), if diagonalised, are constant matrices, which span a CSA for $\mf{g}(V)$. 

First we consider the case $\mf{g}=\mf{sl}$.
If $\mf{sl}(V)_\beta=\CC E_\beta$ is the (one-dimensional) weight space in $\mf{sl}(V)$ corresponding to the root $\beta$, relative to the CSA, then the generator $e_\beta$ is an element of $\CC[\II]E_\beta$. Say $e_\beta=f_\beta E_\beta$ with $f_\beta\in\CC[\II]$. The total determinant of all the generators is the product of the polynomials $f_\beta$ and this determinant is known to be \[\prod_{\beta\in\roots}f_\beta=\prod_{i\in\Omega}\ii^{\delta_{\Gamma\Gamma_i}\kappa(\roots)_i}\] where $\kappa(\roots)_i$ is given in Table \ref{tab:noI} (as this follows from Theorem \ref{thm:determinant of invariant vectors} and Lemma \ref{lem:determinant of moi}). In particular, all polynomials $f_\beta$ are monomials in $\ial$, $\ibe$ and $\iga$.
By defining $\omega^1(\beta)_i$ to be the power of $\ii$ in the monomial $f_\beta$, \[\II^{\omega^1(\beta)}=f_\beta,\] we see that the ALiA is faithfully represented by the canonical representation associated to  a $1$-cocycle $\omega^1$ satisfying (\ref{eq:noI}). In particular \[\salia{sl}{V}{\Gamma}\cong \AL[\mathsf{d}^1\omega^1]{\roots}\]
as Lie algebras.

If $\mf{g}=\mf{so}$ or $\mf{g}=\mf{sp}$ there is an additional complication. Let $B$ be the nondegenerate bilinear form that defines $\mf{g}(V)=\mf{g}_B(V)$; it can be represented by a constant matrix. The transformation $P(\lambda)$ of invariant vectors taking invariant matrices to matrices of invariants (cf.~Section \ref{sec:moi}) transforms the constant bilinear form to a form $P^T(\lambda)BP(\lambda)$ which is $\lambda$-dependent. But because the representation $\tau:G^\flat\rightarrow\GL(V)$ preserves the original bilinear form (that is, $\tau_g^{-T}B \tau_g^{-1}=B$) the new bilinear form is in fact only $\II$-dependent. Indeed, it is invariant under the action on $\lambda$:
\[P^T(g^{-1}\lambda)BP(g^{-1}\lambda)=(P^T(g^{-1}\lambda)\tau_g^T)(\tau_g^{-T}B \tau_g^{-1})(\tau_g P(g^{-1}\lambda))=P^T(\lambda)BP(\lambda)\] 
thus $P^T(\lambda)BP(\lambda)=B'(\II)$.
This implies that the one-dimensional weight spaces $\CC E_\beta(\II)$ in $\mf{g}_{B'(\II)}(V)$ relative to the CSA spanned by $\{h_\alpha\,|\,\alpha\in\sroots\}$ are $\II$-dependent as well, and the generators of a Cartan-Weyl normal form are of the form $e_\beta=f_\beta E_\beta(\II)$, with $f_\beta\in\CC[\II]$. 

Nonetheless the complex Lie algebra $\mf{g}_{B'(\II)}(V)=\CC\left(h_\alpha, E_\beta(\II)\;|\;\alpha\in\sroots,\, \beta\in\roots\right)$ remains constant up to isomorphism, because the transition to matrices of invariants preserves the Lie bracket. In particular, the total determinant of this basis is a nonvanishing meromorphic function and is therefore constant. Hence the total determinant of the natural representation of the ALiA is once more given by $\prod_{\beta\in\roots}f_\beta=\prod_{i\in\Omega}\ii^{\delta_{\Gamma\Gamma_i}\kappa(\roots)_i}$ and we can define a $1$-cochain $\omega^1$ by $\II^{\omega^1(\beta)}=f_\beta$ as before.
Moreover, the ALiA can be represented by $\CC[\II]\left(h_\alpha,\, f_\beta E_\beta\;|\;\alpha\in\sroots, \,\beta\in\roots\right)$ where $\left\{h_\alpha,\,E_\beta\;|\;\alpha\in\sroots, \,\beta\in\roots\right\}$ defines a Chevalley basis for $\mf{g}_B(V)$, that is, the canonical representation associated to $\omega^1$ (even though such a representation might not be obtainable through the sequence of transformations that was applied on the ALiAs with dihedral symmetry throughout this thesis).

To justify condition (\ref{eq:omega1le1}) and (\ref{eq:d1omega1le1}) we will show that
\begin{equation}
\label{eq:killing form minimal rank}
\mathsf{d}^1\omega^1(\alpha,-\alpha)_i\in\{0,1\},\qquad i\in\Omega,\; \alpha\in\roots,
\end{equation}
using evaluations. But before doing so we make sure that (\ref{eq:killing form minimal rank}) implies  (\ref{eq:omega1le1}) and (\ref{eq:d1omega1le1}). Since $\mathsf{d}^1\omega^1(\alpha,-\alpha)_i=\omega^1(\alpha)_i+\omega^1(-\alpha)_i$ and $\omega^1$ takes values in $\NN_0^{|\Omega|}$, condition (\ref{eq:omega1le1}) follows immediately.
Notice also that (\ref{eq:killing form minimal rank}) implies 
\begin{equation}
\label{eq:mle1/2delta}
\|\omega^1\|_i\le\nicefrac{1}{2}|\roots|,\qquad i\in\Omega
\end{equation}
since $\|\omega^1\|_i=\sum_{\alpha\in\roots^+}(\omega^1(\alpha)_i+\omega^1(-\alpha)_i)=\sum_{\alpha\in\roots^+}\mathsf{d}^1\omega^1(\alpha,-\alpha)_i\le |\roots^+|$.

The implication (\ref{eq:d1omega1le1}) is a bit more hidden. Granted (\ref{eq:omega1le1}) we only need to exclude the possibility of roots $\alpha$ and $\beta$ such that $\omega^1(\alpha)-\omega^1(\alpha+\beta)+\omega^1(\beta)=2$ i.e.~$\omega^1(\alpha)=\omega^1(\beta)=1$ and $\omega^1(\alpha+\beta)=0$. If such roots exists then $\omega^1(-\alpha)=0$ by (\ref{eq:killing form minimal rank}) and $\mathsf{d}^1\omega^1(\alpha+\beta,-\alpha)=\omega^1(-\alpha)-\omega^1(\beta)+\omega^1(\alpha+\beta)=0-1+0$ contradicting the fact that $\mathsf{d}^1\omega^1$ takes values in $\NN_0^{|\Omega|}$. Hence (\ref{eq:killing form minimal rank}) implies (\ref{eq:d1omega1le1}).

We turn to the proof of (\ref{eq:killing form minimal rank}). Let $\salia{g}{V}{\Gamma}(\mu)$ be the space spanned by the invariant matrices evaluated at $\mu\in\overline{\CC}$. This is the Lie algebra $\mf{g}(V)^{G_\mu}$ and we studied it extensively in Section \ref{sec:evaluating alias}. Now we define another Lie algebra, $\salia[]{g}{V}{\Gamma}(\mu)$, as the complex vector space with abstract basis $\{h_\alpha,\, e_\beta\;|\;\alpha\in\sroots,\,\beta\in \roots\}$ and with structure constants given by the structure constants of the Automorphic Lie Algebra evaluated in $\mu$. In general the first is a Lie subalgebra of the latter \[\salia{g}{V}{\Gamma}(\mu)<\salia[]{g}{V}{\Gamma}(\mu),\] and these two Lie algebras coincide and equal $\mf{g}(V)$ if $\mu$ belongs to a generic orbit, that is, if $G_\mu=1$.

For $i\in\Omega$ we can split the set of roots \[\roots=\roots_i\sqcup\roots_i^c\] such that $\salia{g}{V}{\Gamma}(\mu_i)$ has basis $\{h_\alpha(\mu_i),\, e_\beta(\mu_i)\;|\;\alpha\in\sroots,\,\beta\in \roots_i\}$ if $\mu_i\in\Gamma_i$ and $h_\alpha$ and $e_\beta$ are invariant matrices. In other words, $e_\beta(\mu_i)=0$ if and only if $\beta\in \roots_i^c$. This does \emph{not} work with the matrices of invariants; the singularity of the transformation $P$ of invariant vectors at $\mu_i$ allows evaluated invariant matrices to be zero while the corresponding matrix of invariants is nonzero.

The evaluated invariant matrices $\salia{g}{V}{\Gamma}(\mu)$ constitute a reductive Lie algebra given by Table \ref{tab:evaluated alias}, that is, a direct sum of a semisimple and an abelian Lie algebra. The elements $e_\beta(\mu_i)$ are in the semisimple summand. In particular, if $\beta\in\roots_i$ then $-\beta\in\roots_i$ (a standard fact, see e.g.~\cite{fulton1991representation, humphreys1972introduction, knapp2002lie}) hence if $\beta\in\roots_i^c$ then $-\beta\in\roots_i^c$.

Take $\beta=\alpha_1+\ldots+\alpha_d\in \roots_i^c$, $\alpha_j\in\sroots$. In $\salia[]{g}{V}{\Gamma}(\mu_i)$ we can consider the bracket $[e_\beta,e_{-\beta}]=f_\beta(\mu_i)f_{-\beta}(\mu_i)\sum_{j=1}^{d}h_{\alpha_j}\in\CC\mf{h}$ where $f_\beta$ and $f_{-\beta}$ are the monomials occurring in the matrices of invariants. Here we use that $h_\alpha$, as a diagonal matrix of invariants, is constant. Since $e_\beta(\mu_i)=0$ as an invariant matrix, this bracket is zero, and because the sum of evaluated invariant matrices $\sum_{j=1}^{d}h_{\alpha_j}(\mu_i)\ne 0$ we can conclude $f_\beta(\mu_i)f_{-\beta}(\mu_i)=0$, i.e.~$\ii$ occurs at least once in at least one of the monomials $f_\beta$ and $f_{-\beta}$. 

Condition (\ref{eq:killing form minimal rank}) now follows by counting. Indeed, $|\roots_i^c|=\dim \mf{g}(V)-\dim \mf{g}(V)^{\langle g_i\rangle}=\codim \mf{g}(V)^{\langle g_i\rangle}=2\kappa(\roots)_i$. Under condition (\ref{eq:noI}) we only have half this many $\ii$ at our disposal to make a representation for the ALiA out of matrices of invariants. Therefore, in order to satisfy $f_\beta(\mu_i)f_{-\beta}(\mu_i)=0$ if $\beta\in\roots_i^c$ we have to be economical and only use one $\ii$ for each pair $\{\beta,-\beta\}\subset\roots_i^c$. Thus condition (\ref{eq:killing form minimal rank}).
\end{proof}

\begin{Corollary} 
\label{cor:necessary conditions for isomorphism} 
If $\salia{g}{V}{\Gamma}\cong \salia[G']{g'}{V'}{\Gamma'}$ then the base Lie algebras are isomorphic,
$\roots\cong\roots'$, and $\{\delta_{\Gamma\Gamma_i}\kappa(\roots)_i\;|\;i\in\Omega\}=\{\delta_{\Gamma'\Gamma'_i}\kappa(\roots')_i\;|\;i\in\Omega\}$.
\end{Corollary}

\begin{proof}
The first necessary condition to have an isomorphism $\salia{g}{V}{\Gamma}\cong\salia[G']{g'}{V'}{\Gamma'}$ was already evident from the evaluations considered in Proposition \ref{prop:evaluating invariant vectors}, where we found that $\salia{g}{V}{\Gamma}(\mu)=\mf{g}(V)$ for generic $\mu\in\overline{\CC}$, thus any isomorphism of ALiAs evaluates on the Riemann sphere to an isomorphism of the base Lie algebras. 
To obtain the second necessary condition we consider the Killing form of $\salia{g}{V}{\Gamma}$ and $\salia[G']{g'}{V'}{\Gamma'}$. We observed before that this bilinear form on $\salia{g}{V}{\Gamma}$ can be represented by a $\dim \mf{g} \times \dim \mf{g}$ matrix over $\CC[\II]$ with determinant in $\CC^\ast\prod_{i\in\Omega}\ii^{\delta_{\Gamma\Gamma_i}2\kappa(\roots)_i}$. An isomorphism $\salia{g}{V}{\Gamma}\cong \salia[G']{g'}{V'}{\Gamma'}$ is realised by a $\dim \mf{g} \times \dim \mf{g}$ matrix over $\CC[\II]$ with determinant in $\CC^*$ and acts on a bilinear form as $K\mapsto T^t K T$, hence the determinant of the Killing form is preserved up to scalar under Lie algebra isomorphisms.
\end{proof}

Are the necessary conditions for an isomorphism given in Corollary \ref{cor:necessary conditions for isomorphism} also sufficient? This is the isomorphism question, \ref{q:isomorphism question}, or rather a slightly stronger version of it.
In the following subsection we investigate the collections of $\AL[\mathsf{d}^1\omega^1]{\roots}$ where $\omega^1$ satisfies the conditions (\ref{eq:noI}, \ref{eq:omega1le1}, \ref{eq:d1omega1le1}) and will find that they often contain only one Lie algebra in which case we find an ALiA, granted it allows a Cartan-Weyl normal form.  However, we will also find cases where there are still various Lie algebras in this class and the isomorphism question remains open.

\subsection{The Isomorphism Question}

This section investigates the consequences of Theorem \ref{thm:alias satisfy game of roots} for the isomorphism question; whether Automorphic Lie Algebras themselves are invariants of Automorphic Lie Algebras. In Table \ref{tab:group decompositions of simple lie algebras} we see that there are only four root systems involved in the isomorphism question: $A_1$, $A_2$, $B_2$ and $A_3$ \cite{fulton1991representation, humphreys1972introduction, knapp2002lie}. For each we discuss and illustrate their $2$-coboundaries satisfying the conditions of Theorem \ref{thm:alias satisfy game of roots}.

Because we are ultimately interested in Lie algebras up to isomorphism we define a notion of isomorphism for $2$-cochains.
\begin{Definition}[Isomorphism of $m$-cochains] 
\label{def:isomorphism of cochains}
Two cochains $\omega^m, \bar{\omega}^m\in C^m(\roots,\NN_0^q)$ are \emph{isomorphic} if there exists an isomorphism $\phi:\roots\rightarrow\roots$ of root systems such that \[\omega^m(\phi(\alpha_1),\ldots,\phi(\alpha_m))= \bar{\omega}^m(\alpha_1,\ldots,\alpha_m)\]
for all $m$-chains $(\alpha_1,\ldots,\alpha_m)\in C_m(\roots)$.
\end{Definition}
\newcommand{\liehom}{\phi^{\ast\ast}}If two symmetric $2$-cocycles $\omega^2, \bar{\omega}^2\in Z^2(\roots,\NN_0^q)$ are isomorphic then their associated Lie algebras 
$\AL{\roots}$ and $\AL[\bar{\omega}^2]{\roots}$ with respective generators $h_\alpha,\,e_\beta$ and $\bar{h}_\alpha,\,\bar{e}_\beta$, $\alpha\in\sroots,\,\beta\in\roots$,
are isomorphic as Lie algebras. Indeed, the isomorphism $\phi$ of the root system can be used to construct an isomorphism $\liehom: \AL[\bar{\omega}^2]{\roots}\rightarrow \AL{\roots}$ by $\CC[\II_1,\ldots,\II_q]$-linear extension of $\liehom(\bar{h}_\alpha)=h_{\phi(\alpha)}$ and $\liehom(\bar{e}_\beta)=e_{\phi(\beta)}$, $\alpha\in\sroots,\;\beta\in\roots$. This map has inverse $(\liehom)^{-1}(h_\alpha)=\bar{h}_{\phi^{-1}(\alpha)}$, $(\liehom)^{-1}(e_\beta)=\bar{e}_{\phi^{-1}(\beta)}$ and preserves the Lie bracket. We check the case $(\beta,\gamma)\in C_2(\roots)$,
\begin{align*}
\liehom[\bar{e}_\beta,\bar{e}_\gamma]&=\pm\liehom (r+1)\II^{\bar{\omega}^2(\beta,\gamma)}\bar{e}_{\beta+\gamma}\\
&= \pm(r+1)\II^{\bar{\omega}^2(\beta,\gamma)}\liehom\bar{e}_{\beta+\gamma}\\
&= \pm(r+1)\II^{\omega^2(\phi(\beta),\phi(\gamma))}e_{\phi(\beta)+\phi(\gamma)}\\
&= [e_{\phi(\beta)}, e_{\phi(\gamma)}]\\
&= [\liehom \bar{e}_\beta,\liehom\bar{e}_\gamma].
\end{align*}
The remaining brackets are left for the reader to check, using $\phi(\beta)(h_{\phi(\alpha)})=\frac{2(\phi(\beta),\phi(\alpha))}{(\phi(\alpha),\phi(\alpha))}=\frac{2(\beta,\alpha)}{(\alpha,\alpha)}=\beta(h_\alpha)$.

In order to study the Lie algebra $\AL{\roots}$ using cocycles, it is clearly desirable to have a notion of isomorphism on symmetric $2$-cocycles that coincides with the Lie algebra isomorphism of the associated Lie algebras. The current attempt possibly defines too few isomorphisms. That is, if two symmetric $2$-cochains are not isomorphic by Definition \ref{def:isomorphism of cochains}, there is no guarantee that the associated Lie algebras are different. This is an open problem.

\subsubsection{Root System $A_1$}
\label{sec:A1}

The root systems $A_1$ and $A_1\times A_1$ for $\mf{sl}_2(\CC)$ and $\mf{so}_4(\CC)\cong\mf{sl}_2(\CC)\oplus \mf{sl}_2(\CC)$ are shown in Figure \ref{fig:A1} and Figure \ref{fig:A1xA1} respectively.
\begin{figure}[h!]
\begin{minipage}{0.5\linewidth}
\caption{The root system $A_1$.}
\label{fig:A1}
\begin{center}
\begin{tikzpicture}[scale=0.6]
  \path node at ( 0,0) [scale=\nodescale,shape=circle,draw,label=270: $ $] (zero) {$ $}
  	node at ( 2,0) [scale=\nodescale,shape=circle,draw,label=0: $\alpha$] (one) {$ $}
  	node at ( -2,0) [scale=\nodescale,shape=circle,draw,label=180: $-\alpha$] (two) {$ $};
	\draw (zero) to node [sloped,below]{} (one);
	\draw (zero) to node [sloped,above]{} (two);
\end{tikzpicture}
\end{center}
\end{minipage}%
\begin{minipage}{0.5\linewidth}
\caption{The root system $A_1\times A_1$.}
\label{fig:A1xA1}
\begin{center}
\begin{tikzpicture}[scale=0.6]
  \path node at ( 0,0) [scale=\nodescale,shape=circle,draw,label=270: $ $] (zero) {$ $}
  	node at ( 2,0) [scale=\nodescale,shape=circle,draw,label=0: $\alpha$] (one) {$ $}
  	node at ( -2,0) [scale=\nodescale,shape=circle,draw,label=180: $-\alpha$] (two) {$ $}
  	node at ( 0,2) [scale=\nodescale,shape=circle,draw,label=90: $\beta$] (three) {$ $}
  	node at ( 0,-2) [scale=\nodescale,shape=circle,draw,label=270: $-\beta$] (four) {$ $};
	\draw (zero) to node [sloped,below]{} (one);
	\draw (zero) to node [sloped,above]{} (two);
	\draw (zero) to node [sloped,below]{} (three);
	\draw (zero) to node [sloped,above]{} (four);
\end{tikzpicture}
\end{center}
\end{minipage}
\end{figure}

There are only two maps $\omega^1:C_1(A_1)\rightarrow \NN _0$ satisfying conditions $\omega^1(0)=0$ and $\omega^1(-\alpha)+\omega^1(0)+\omega^1(\alpha)=1$ from Theorem \ref{thm:alias satisfy game of roots}. Either $\omega^1(-\alpha)=1$ or $\omega^1(\alpha)=1$ and the other values are zero. Both of them map to the same $2$-coboundary; $\mathsf{d}^1\omega^1(-\alpha,\alpha)=1$ and $\mathsf{d}^1\omega^1(0,\alpha)=\mathsf{d}^1\omega^1(0,-\alpha)=0$. In terms of the representation (\ref{eq:dihedral alias serre model}) this statement means that it does not matter whether $\ii$ is in the top right entry or in the bottom left. Hence, if $\dim V=2$, i.e.~if
\[V\in\{\psi_j, \btiiii, \btiiiii, \btiiiiii, \boiii, \boiiii, \boiiiii, \byii, \byiii\}\] 
then
\begin{align*}
\salia{sl}{V}{\Gamma}&\cong\CC[\II]\left(h,  e_+, e_-\right)\\
{[}h,e_{\pm}]&=\pm 2 e_\pm\\
{[}e_+,e_-]&=\ial\ibe\iga h
\end{align*}
for all orbits $\Gamma$, if we take into account that $\ii=1\Leftrightarrow \Gamma=\Gamma_i$.
This is conform the dihedral examples. 
Moreover, considering a switch $1\ne\II_i\leftrightarrow\II_j\ne 1$ to be an isomorphism we obtain 
\[\salia{sl}{V}{\Gal}\cong \salia{sl}{V}{\Gbe}\cong \salia{sl}{V}{\Gga}\ncong\salia{sl}{V}{\Gamma\ne\Gamma_i}.\]
In particular, we obtain the main result of \cite{LS10} for arbitrary polyhedral groups by taking $\Gamma=\Gal$, and the $\mf{sl}_2(\CC)$-based isomorphism results of \cite{bury2010automorphic}.

For root system $A_1\times A_1$, another constraint plays a part as (\ref{eq:d1omega1le1}) forces one $\ii$ at each simple component of the root system. Thus, for all four-dimensional irreducible representations $V$ of real type, that is, only for $V=\byiiiiii$, we find $\salia{so}{\byiiiiii}{\Gamma}$ to be isomorphic to $\salia{sl}{V}{\Gamma}\oplus\salia{sl}{V}{\Gamma}$ if $\dim V=2$.
This can be generalised to $n$ (orthogonal) copies of $A_1$ and $\kappa((A_1)^n)=(n,n,n)$.

\subsubsection{Root System $A_2$}
\label{sec:A2}

The planar root system $A_2$ for $\mf{sl}_3(\CC)$ consists of $6$ roots arranged in a regular hexagon. In the diagram there is a node for each root. The addition of roots is inherited from the ambient vector space with origin in the centre of the hexagon. 

For this root system the defining properties of chains and cochains play an important role, contrary to $A_1$. Recall that $2$-chains are formal $\ZZ$-linear combinations of pairs of roots $(\alpha,\beta)$ such that $\alpha+\beta$ is again a root. We depict $(\alpha,\beta)\in C_2(\roots)$ in the diagram by an edge between $\alpha$ and $\beta$, cf.~Figure \ref{fig:C_2(A_2)}, using the fact that we are only interested in symmetric $2$-cocycles for now.
\begin{figure}[h!]
\caption{The basis for $C_2(A_2)$.}
\label{fig:C_2(A_2)}
\begin{center}
\begin{tikzpicture}[scale=0.6]
  \path node at ( 0,0) [scale=\nodescale,shape=circle,draw,label=270: $ $] (zero) {$ $}	
	node at ( 4,0) [scale=\nodescale,shape=circle,draw,label=0: $\alpha_1$] (one) {$ $}
  	node at ( 2,3.464) [scale=\nodescale,shape=circle,draw,label=60: $\alpha_2$] (two) {$ $}
  	node at ( -2,3.464) [scale=\nodescale,shape=circle,draw,label=120: $\alpha_3$] (three) {$ $}
	node at ( -4,0) [scale=\nodescale,shape=circle,draw,label=180: $\alpha_4$] (four) {$ $}
	node at ( -2,-3.464) [scale=\nodescale,shape=circle,draw,label=240: $\alpha_5$] (five) {$ $}
	node at ( 2,-3.464) [scale=\nodescale,shape=circle,draw,label=300: $\alpha_6$] (six) {$ $};
	\draw (zero) to node [sloped,below]{} (one);
	\draw (zero) to node [sloped,above]{} (two);
	\draw (zero) to node [sloped,above]{} (three);
	\draw (zero) to node [sloped,above]{} (four);
	\draw (zero) to node [sloped,above]{} (five);
	\draw (zero) to node [sloped,above]{} (six);
	\draw (one) to node [sloped,above]{$ $} (three);
	\draw (two) to node [sloped,above]{$ $} (four);
	\draw (three) to node [sloped,above]{$ $} (five);
	\draw (four) to node [sloped,above]{$ $} (six);
	\draw (five) to node [sloped,above]{$ $} (one);
	\draw (six) to node [sloped,above]{$ $} (two);
\end{tikzpicture}
\end{center}
\end{figure}
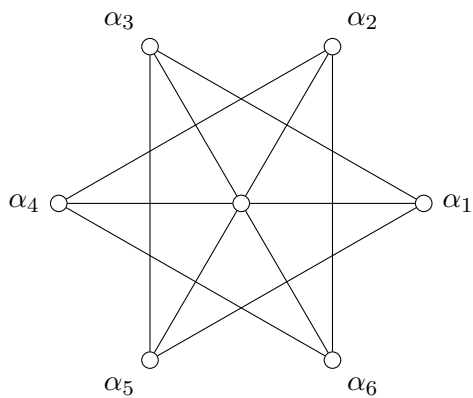

We are interested in $1$-cochains $\omega^1$ that satisfy (\ref{eq:noI}, \ref{eq:omega1le1}, \ref{eq:d1omega1le1}). Table \ref{tab:noI} gives \[\|\omega^1\|=(3,3,2).\] Components of $\omega^1$ can be studied separately and we start with the smallest sum $\|\omega^1\|_\ga=2$.

The condition in the cochain definition that $\omega^1$ be such that $\mathsf{d}^1\omega^1$ takes nonnegative values is very restrictive. For instance, if $\omega^1(\alpha_1)>0$ then either $\omega^1(\alpha_2)>0$ or $\omega^1(\alpha_6)>0$, since otherwise $\mathsf{d}^1\omega^1(\alpha_2,\alpha_6)=\omega^1(\alpha_2)-\omega^1(\alpha_1)+\omega^1(\alpha_6)<0$. In particular there is no $1$-cochain with norm $1$. And, more relevant, there is a unique $1$-cochain with norm $2$, up to isomorphism. Indeed, it was just argued that two neighbours in the hexagon have to take positive value, and if we define $\omega^1(\alpha_1)=\omega^1(\alpha_2)=1$ and  $\omega^1(\alpha_j)=0$ for $j=3,\ldots,6$ then $\mathsf{d}^1\omega^1\ge 0$ so $\omega^1$ is a cochain.

In the root system diagram we depict a $1$-cochain $\omega^1$ by filling the node $\alpha$ if and only if $\omega^1(\alpha)=1$ and we depict the corresponding $2$-coboundary $\mathsf{d}^1\omega^1$ by drawing and edge $(\alpha,\beta)\in C_2(\roots)$ if and only if $\mathsf{d}^1\omega^1(\alpha,\beta)=1$, cf.~Figure \ref{fig:B^2(A_2), m=2,3}. This provides all information under the conditions (\ref{eq:omega1le1}, \ref{eq:d1omega1le1}).

The equation $\|\omega^1\|_\be=3$ has a unique solution in $C^1(A_2,\NN_0)$ as well, up to isomorphism. This follows from the same argument. Each root $\alpha$ for which $\omega^1(\alpha)=1$ requires at least one neighbour in the hexagon to take value $1$ as well. Therefore, three subsequent roots take value $1$ and there is one way to do this, up to isomorphism. The solution together with its $2$-coboundary is shown in the right root system of Figure \ref{fig:B^2(A_2), m=2,3}. Notice that the associated Lie algebra of the latter is not generated as a Lie algebra by $6$ elements, such as the base Lie algebra $\mf{sl}_3(\CC)$, whereas the Lie algebra associated to the former coboundary is.

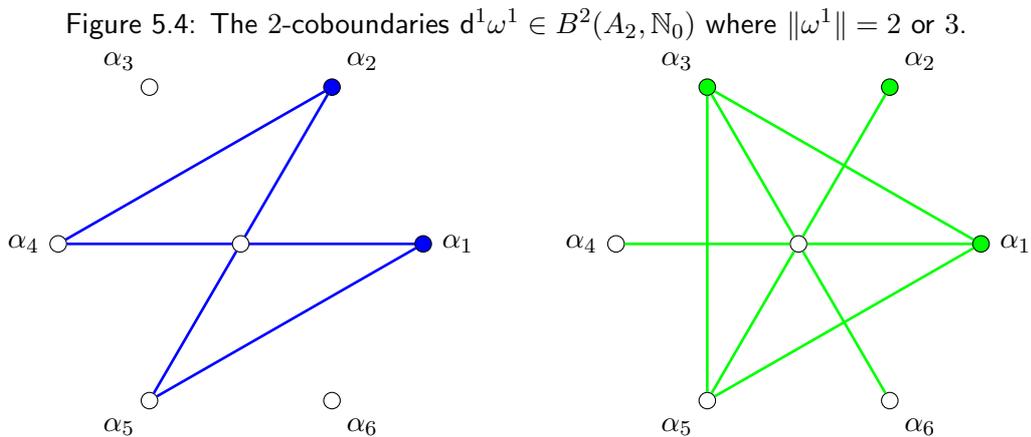
\begin{figure}[h!]
\caption{The $2$-coboundaries $\mathsf{d}^1\omega^1\in B^2(A_2,\NN_0)$ where $\|\omega^1\|=2$ or $3$.}
\label{fig:B^2(A_2), m=2,3}
\begin{minipage}{0.5\linewidth}
\begin{center}
\begin{tikzpicture}[scale=0.6]
  \path node at ( 0,0) [scale=\nodescale,shape=circle,draw,label=270: $ $] (zero) {$ $}	
	node at ( 4,0) [scale=\nodescale,shape=circle,fill=\colorga,draw,label=0: $\alpha_1$] (one) {$ $}
  	node at ( 2,3.464) [scale=\nodescale,shape=circle,fill=\colorga,draw,label=60: $\alpha_2$] (two) {$ $}
  	node at ( -2,3.464) [scale=\nodescale,shape=circle,draw,label=120: $\alpha_3$] (three) {$ $}
	node at ( -4,0) [scale=\nodescale,shape=circle,draw,label=180: $\alpha_4$] (four) {$ $}
	node at ( -2,-3.464) [scale=\nodescale,shape=circle,draw,label=240: $\alpha_5$] (five) {$ $}
	node at ( 2,-3.464) [scale=\nodescale,shape=circle,draw,label=300: $\alpha_6$] (six) {$ $};
	\draw [\colorga, line width=1pt](one) to node [sloped,above,near end]{} (zero);
	\draw [\colorga, line width=1pt](two) to node [sloped,below,near end]{} (zero);
	\draw [\colorga, line width=1pt](four) to node [sloped,below,near end]{} (zero);
	\draw [\colorga, line width=1pt](five) to node [sloped,below,near end]{} (zero);
	\draw [\colorga, line width=1pt](two) to node [sloped,above]{} (four);
	\draw [\colorga, line width=1pt](five) to node [sloped,below]{} (one);
\end{tikzpicture}
\end{center}
\end{minipage}
\begin{minipage}{0.5\linewidth}
\begin{center}
\begin{tikzpicture}[scale=0.6].  
  \path node at ( 0,0) [scale=\nodescale,shape=circle,draw,label=270: $ $] (zero) {$ $}	
	node at ( 4,0) [scale=\nodescale,shape=circle,fill=\colorbe,draw,label=0: $\alpha_1$] (one) {$ $}
  	node at ( 2,3.464) [scale=\nodescale,shape=circle,fill=\colorbe,draw,label=60: $\alpha_2$] (two) {$ $}
  	node at ( -2,3.464) [scale=\nodescale,shape=circle,fill=\colorbe,draw,label=120: $\alpha_3$] (three) {$ $}
	node at ( -4,0) [scale=\nodescale,shape=circle,draw,label=180: $\alpha_4$] (four) {$ $}
	node at ( -2,-3.464) [scale=\nodescale,shape=circle,draw,label=240: $\alpha_5$] (five) {$ $}
	node at ( 2,-3.464) [scale=\nodescale,shape=circle,draw,label=300: $\alpha_6$] (six) {$ $};
	\draw [\colorbe, line width=1pt](one) to node [sloped,above,near end]{} (zero);
	\draw [\colorbe, line width=1pt] (two) to node [sloped,below,near end]{} (zero);
	\draw [\colorbe, line width=1pt] (three) to node [sloped,below,near end]{} (zero);
	\draw [\colorbe, line width=1pt] (four) to node [sloped,below,near end]{} (zero);
	\draw [\colorbe, line width=1pt] (five) to node [sloped,below,near end]{} (zero);
	\draw [\colorbe, line width=1pt] (six) to node [sloped,below,near end]{} (zero);
	\draw [\colorbe, line width=1pt] (one) to node [sloped,above]{} (three);
	\draw [\colorbe, line width=1pt] (three) to node [sloped,below]{} (five);
	\draw [\colorbe, line width=1pt] (five) to node [sloped,below]{} (one);
\end{tikzpicture}
\end{center}
\end{minipage}
\end{figure}

Any solutions to $\|\omega^1\|=\kappa$ with $\kappa\ge4$ violates (\ref{eq:d1omega1le1}) as noticed before: (\ref{eq:mle1/2delta}).


To make a statement on ALiAs the above results have to be combined. For instance, if $\dim V=3$, the ALiAs $\salia{sl}{V}{\Gal}$ and $\salia{sl}{V}{\Gbe}$ are associated to $\mathsf{d}^1\omega^1$ for some $\omega^1\in C^1(A_2,\NN_0^2)$ with \[\|\omega^1\|=(3,2)\] by Theorem \ref{thm:alias satisfy game of roots}. Since there is only one $2$-coboundary with norm $2$ and with norm $3$, the number of coboundaries that could define $\salia{sl}{V}{\Gal}$ and $\salia{sl}{V}{\Gbe}$ is the number of ways the previous two coboundaries can be combined. In this case we are in luck because there is again just one possibility up to isomorphism, shown in Figure \ref{fig:B^2(A_2), m=(3,2)}. 
In order to incorporate multidimensional cochains in the root system diagrams we define a colouring as follows.
\begin{equation}
\label{eq:colouring}
\text{\Large{\color{\coloral} \bf $\al$}, {\color{\colorbe} \bf $\be$}, {\color{\colorga} \bf $\ga$}.}
\end{equation}

\begin{figure}[h!]
\caption{The $2$-coboundary $\mathsf{d}^1\omega^1\in B^2(A_2,\NN_0^2)$ where $\|\omega^1\|=(3,2)$.}
\label{fig:B^2(A_2), m=(3,2)}
\begin{center}
\begin{tikzpicture}[scale=0.6]
  \tikzset{edgeal/.style ={green, line width=1pt}}
  \tikzset{edgebe/.style ={red, line width=1pt}}
  \tikzset{edgega/.style ={blue, line width=1pt}}
  \path node at ( 0,0) [scale=\nodescale,shape=circle,draw,label=270: $ $] (zero) {$ $}	
	node at ( 4,0) [scale=\nodescale,shape=circle,draw,label=0: $\alpha_1$] (one) {$ $}
  	node at ( 2,3.464) [scale=\nodescale,shape=circle,draw,label=60: $\alpha_2$] (two) {$ $}
  	node at ( -2,3.464) [scale=\nodescale,shape=circle,draw,label=120: $\alpha_3$] (three) {$ $}
	node at ( -4,0) [scale=\nodescale,shape=circle,draw,label=180: $\alpha_4$] (four) {$ $}
	node at ( -2,-3.464) [scale=\nodescale,shape=circle,draw,label=240: $\alpha_5$] (five) {$ $}
	node at ( 2,-3.464) [scale=\nodescale,shape=circle,draw,label=300: $\alpha_6$] (six) {$ $};
	\draw (one)[\colorbe, line width=1pt] to node [sloped,above,near end]{} (zero);
	\draw (two)[\colorbe, line width=1pt] to node [sloped,below,near end]{} (zero);
	\draw (three)[\colorbe, line width=1pt] to node [sloped,below,near end]{} (zero);
	\draw (four)[\colorbe, line width=1pt] to node [sloped,below,near end]{} (zero);
	\draw (five)[\colorbe, line width=1pt] to node [sloped,below,near end]{} (zero);
	\draw (six)[\colorbe, line width=1pt] to node [sloped,below,near end]{} (zero);
	\draw (one)[\colorbe, line width=1pt] to node [sloped,above]{} (three);
	\draw (three)[\colorbe, line width=1pt] to node [sloped,below]{} (five);
	\draw (five)[\colorbe, line width=1pt] to node [sloped,below]{} (one);


	\draw (one.210)[\colorga, line width=1pt] to node [sloped,above,near end]{} (zero.330);
	\draw (two.270)[\colorga, line width=1pt] to node [sloped,below,near end]{} (zero.30);
	\draw (four.330)[\colorga, line width=1pt] to node [sloped,below,near end]{} (zero.210);
	\draw (five.30)[\colorga, line width=1pt] to node [sloped,below,near end]{} (zero.270);
	\draw (two)[\colorga, line width=1pt] to node [sloped,above]{$ $} (four);
	\draw (five.0)[\colorga, line width=1pt] to node [sloped,below]{} (one.240);
\end{tikzpicture}
\end{center}
\end{figure}

One can conclude that the existence of the Cartan-Weyl normal form (Conjecture \ref{conj:existence serre normal form}) implies through Theorem \ref{thm:alias satisfy game of roots} that $\mf{sl}_3(\CC)$-based ALiAs with poles restricted to one of the smallest two orbits are isomorphic. This includes all combinations of \[V\in\{\btiiiiiii,\boiiiiii, \boiiiiiii, \byiiii, \byiiiii\}, \qquad \Gamma\in\{\Gal, \Gbe\}.\] 
Better yet, it describes the Lie algebra structure explicitly. 
With a choice of simple roots  $\alpha=\alpha_1$ and $\beta=\alpha_3$  this Lie algebra
\[
\salia[]{sl}{V}{\Gamma_j}=\CC[\II]\left( h_\alpha, h_\beta, e_{\alpha}, e_{\beta}, e_{\alpha+\beta}, e_{-\alpha}, e_{-\beta}, e_{-\alpha-\beta}\right)\]
has commutation relations
\[
\begin{array}{llllll}
{[}e_\alpha,e_{-\alpha}]&=&\ii\iga h_\alpha&\qquad{[}e_\alpha,e_{\beta}]&=&\ii e_{\alpha+\beta}\\
{[}e_\beta,e_{-\beta}]&=&\ii h_\beta&\qquad{[}e_{\alpha+\beta},e_{-\alpha}]&=&-\iga e_{\beta}\\
{[}e_{\alpha+\beta},e_{-\alpha-\beta}]&=&\ii\iga (h_\alpha+h_\beta)&\qquad{[}e_\beta,e_{-\alpha-\beta}]&=&\ii e_{-\alpha}\\
&&&\qquad{[}e_{-\alpha},e_{-\beta}]&=&- e_{-\alpha-\beta}\\
&&&\qquad{[}e_{-\alpha-\beta},e_{\alpha}]&=&\ii\iga e_{-\beta}\\
&&&\qquad{[}e_{-\beta},e_{\alpha+\beta}]&=&- e_{\alpha}
\end{array}
\]
where $i=\be$ if $j=\al$ and vice versa.
The other brackets are given by the roots (and identical to the base Lie algebra). Arguably the full Lie algebra structure is displayed in a more transparent manner by the coboundary in Figure \ref{fig:B^2(A_2), m=(3,2)} than through this list of Lie brackets. 

The remaining $\mf{sl}_3(\CC)$-based ALiAs with exceptional pole orbits are related to cochains $\omega^1$ with norm $\|\omega^1\|=(3,3)$. There are two distinct\footnote{Distinct in the sense of Definition \ref{def:isomorphism of cochains}} $2$-coboundaries resulting from such cochains, depicted in Figure \ref{fig:B^2(A_2), m=(3,3)}. It is not known at the time of writing whether the Lie algebras associated to these coboundaries are isomorphic.
\begin{figure}[h!]
\caption{The two $2$-coboundaries $\mathsf{d}^1\omega^1\in B^2(A_2,\NN_0^2)$ where $\|\omega^1\|=(3,3)$.}
\label{fig:B^2(A_2), m=(3,3)}
\begin{minipage}{0.5\linewidth}
\begin{center}
\begin{tikzpicture}[scale=0.6]
  \path node at ( 0,0) [scale=\nodescale,shape=circle,draw,label=270: $ $] (zero) {$ $}	
	node at ( 4,0) [scale=\nodescale,shape=circle,draw,label=0: $ $] (one) {$ $}
  	node at ( 2,3.464) [scale=\nodescale,shape=circle,draw,label=60: $ $] (two) {$ $}
  	node at ( -2,3.464) [scale=\nodescale,shape=circle,draw,label=120: $ $] (three) {$ $}
	node at ( -4,0) [scale=\nodescale,shape=circle,draw,label=180: $ $] (four) {$ $}
	node at ( -2,-3.464) [scale=\nodescale,shape=circle,draw,label=240: $ $] (five) {$ $}
	node at ( 2,-3.464) [scale=\nodescale,shape=circle,draw,label=300: $ $] (six) {$ $};

	\draw[\coloral, line width=1pt] (one.150) to node [sloped,above,near end]{$ $} (zero.30);
	\draw[\coloral, line width=1pt] (two.210) to node [sloped,below,near end]{} (zero.90);
	\draw[\coloral, line width=1pt] (three.330) to node [sloped,below,near end]{$ $} (zero.90);
	\draw[\coloral, line width=1pt] (four.30) to node [sloped,below,near end]{} (zero.150);
	\draw[\coloral, line width=1pt] (five.90) to node [sloped,below,near end]{$ $} (zero.210);
	\draw[\coloral, line width=1pt] (six.90) to node [sloped,below,near end]{} (zero.330);
	\draw[\coloral, line width=1pt] (two) to node [sloped,above]{$ $} (four);
	\draw[\coloral, line width=1pt] (four) to node [sloped,below]{$ $} (six);
	\draw[\coloral, line width=1pt] (six) to node [sloped,below]{$ $} (two);

	\draw[\colorbe, line width=1pt] (one) to node [sloped,above,near end]{$ $} (zero);
	\draw[\colorbe, line width=1pt] (two) to node [sloped,below,near end]{} (zero);
	\draw[\colorbe, line width=1pt] (three) to node [sloped,below,near end]{$ $} (zero);
	\draw[\colorbe, line width=1pt] (four) to node [sloped,below,near end]{} (zero);
	\draw[\colorbe, line width=1pt] (five) to node [sloped,below,near end]{$ $} (zero);
	\draw[\colorbe, line width=1pt] (six) to node [sloped,below,near end]{} (zero);
	\draw[\colorbe, line width=1pt] (one) to node [sloped,above]{$ $} (three);
	\draw[\colorbe, line width=1pt] (three) to node [sloped,below]{$ $} (five);
	\draw[\colorbe, line width=1pt] (five) to node [sloped,below]{$ $} (one);
\end{tikzpicture}
\end{center}
\end{minipage}%
\begin{minipage}{0.5\linewidth}
\begin{center}
\begin{tikzpicture}[scale=0.6]
  \path  node at ( 0,0) [scale=\nodescale,shape=circle,draw,label=270: $ $] (zero) {$ $}	
	node at ( 4,0) [scale=\nodescale,shape=circle,draw,label=0: $ $] (one) {$ $}
  	node at ( 2,3.464) [scale=\nodescale,shape=circle,draw,label=60: $ $] (two) {$ $}
  	node at ( -2,3.464) [scale=\nodescale,shape=circle,draw,label=120: $ $] (three) {$ $}
	node at ( -4,0) [scale=\nodescale,shape=circle,draw,label=180: $ $] (four) {$ $}
	node at ( -2,-3.464) [scale=\nodescale,shape=circle,draw,label=240: $ $] (five) {$ $}
	node at ( 2,-3.464) [scale=\nodescale,shape=circle,draw,label=300: $ $] (six) {$ $};

	\draw[\coloral, line width=1pt] (one.150) to node [sloped,above,near end]{$ $} (zero.30);
	\draw[\coloral, line width=1pt] (two.210) to node [sloped,below,near end]{} (zero.90);
	\draw[\coloral, line width=1pt] (three.330) to node [sloped,below,near end]{$ $} (zero.90);
	\draw[\coloral, line width=1pt] (four.30) to node [sloped,below,near end]{} (zero.150);
	\draw[\coloral, line width=1pt] (five.90) to node [sloped,below,near end]{$ $} (zero.210);
	\draw[\coloral, line width=1pt] (six.90) to node [sloped,below,near end]{} (zero.330);
	\draw[\coloral, line width=1pt] (one.120) to node [sloped,above]{$ $} (three.0);
	\draw[\coloral, line width=1pt] (three.240) to node [sloped,below]{$ $} (five.120);
	\draw[\coloral, line width=1pt] (five.60) to node [sloped,below]{$ $} (one.180);

	\draw[\colorbe, line width=1pt] (one) to node [sloped,above,near end]{$ $} (zero);
	\draw[\colorbe, line width=1pt] (two) to node [sloped,below,near end]{} (zero);
	\draw[\colorbe, line width=1pt] (three) to node [sloped,below,near end]{$ $} (zero);
	\draw[\colorbe, line width=1pt] (four) to node [sloped,below,near end]{} (zero);
	\draw[\colorbe, line width=1pt] (five) to node [sloped,below,near end]{$ $} (zero);
	\draw[\colorbe, line width=1pt] (six) to node [sloped,below,near end]{} (zero);
	\draw[\colorbe, line width=1pt] (one) to node [sloped,above]{$ $} (three);
	\draw[\colorbe, line width=1pt] (three) to node [sloped,below]{$ $} (five);
	\draw[\colorbe, line width=1pt] (five) to node [sloped,below]{$ $} (one);
\end{tikzpicture}
\end{center}
\end{minipage}
\end{figure}
An $\mf{sl}_3(\CC)$-based ALiA $\salia{sl}{V}{\Gga}$ allowing a Cartan-Weyl normal form is isomorphic to (at least) one of these.
A solution to this problem also solves the case with poles at generic orbits, because the coboundary of the $1$-cochain on $A_2$ with norm $2$ (cf.~Figure \ref{fig:B^2(A_2), m=2,3} on the left) is included in the coboundaries of Figure \ref{fig:B^2(A_2), m=(3,3)} in a unique way, up to isomorphism.

\subsubsection{Root System $B_2$}
\label{sec:B2}
The root system $B_2$ related to the isomorphic Lie algebras $\mf{so}_5(\CC)$ and $\mf{sp}_4(\CC)$ has four short roots and four long roots. Its $2$-chains are shown in Figure \ref{fig:B2}.
\begin{figure}[h!]
\caption{The basis for $C_2(B_2)$.}
\label{fig:B2}
\begin{center}
\begin{tikzpicture}[scale=0.7]
  \path node at ( 0,0) [scale=\nodescale,shape=circle,draw,label=270: $ $] (zero) {$ $}	
	node at ( 4,0) [scale=\nodescale,shape=circle,,draw,label=0: $\alpha_1$] (one) {$ $}
  	node at ( 4,4) [scale=\nodescale,shape=circle,,draw,label=45: $\alpha_2$] (two) {$ $}
  	node at ( 0,4) [scale=\nodescale,shape=circle,,draw,label=90: $\alpha_3$] (three) {$ $}
	node at ( -4,4) [scale=\nodescale,shape=circle,draw,label=135: $\alpha_4$] (four) {$ $}
	node at (-4,0) [scale=\nodescale,shape=circle,draw,label=180: $\alpha_5$] (five) {$ $}
	node at ( -4,-4) [scale=\nodescale,shape=circle,draw,label=225: $\alpha_6$] (six) {$ $}
	node at (0,-4) [scale=\nodescale,shape=circle,draw,label=270: $\alpha_7$] (seven) {$ $}
	node at (4,-4) [scale=\nodescale,shape=circle,draw,label=315: $\alpha_8$] (eight) {$ $};
	\draw (one) to node [sloped,above,near end]{} (zero);
	\draw (two) to node [sloped,above,near end]{} (zero);
	\draw (three) to node [sloped,below,near end]{} (zero);
	\draw (four) to node [sloped,below,near end]{} (zero);
	\draw (five) to node [sloped,below,near end]{} (zero);
	\draw (six) to node [sloped,below,near end]{} (zero);
	\draw (seven) to node [sloped,below,near end]{} (zero);
	\draw (eight) to node [sloped,below,near end]{} (zero);
	\draw (one) to node [sloped,above]{$ $} (three);
	\draw (three) to node [sloped,above]{$ $} (five);
	\draw (five) to node [sloped,above]{$ $} (seven);
	\draw (seven) to node [sloped,above]{$ $} (one);
	\draw (two) to node [sloped,above]{$ $} (five);
	\draw (two) to node [sloped,above]{$ $} (seven);
	\draw (four) to node [sloped,above]{$ $} (seven);
	\draw (four) to node [sloped,above]{$ $} (one);
	\draw (six) to node [sloped,above]{$ $} (one);
	\draw (six) to node [sloped,above]{$ $} (three);
	\draw (eight) to node [sloped,above]{$ $} (three);
	\draw (eight) to node [sloped,above]{$ $} (five);
\end{tikzpicture}
\end{center}
\end{figure}
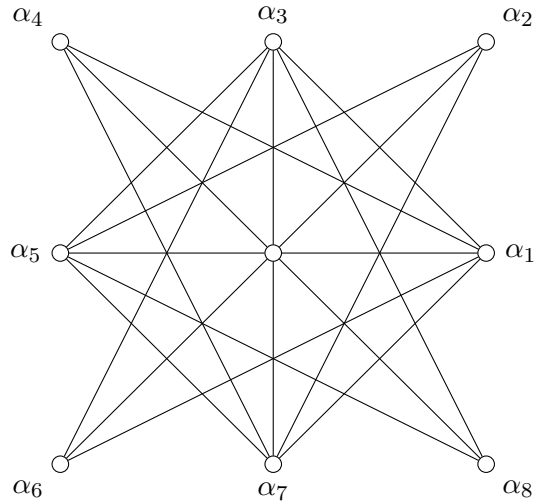

There are no $1$-cochains $\omega^1$ with norm $1$ or $2$.
For norm $1$ it is noticed that any root $\alpha$ is a sum of two roots. Therefore if $\omega^1(\alpha)=1$ and all other values are zero, then $\mathsf{d}^1\omega^1(\beta,\gamma)=-1$ if $\beta+\gamma=\alpha$, which is not allowed.

If $\omega^1$ has norm $2$ and is positive on a short root, say $\omega^1(\alpha_3)=1$, then, in order for $\mathsf{d}^1\omega^1$ to be nonnegative, either $\omega^1(\alpha_1)=1$ or  $\omega^1(\alpha_4)=1$, \emph{and} either $\omega^1(\alpha_2)=1$ or  $\omega^1(\alpha_5)=1$, which is impossible if $\sum_{B_2}\omega^1(\alpha)=2$. If a positive value is assigned to a long root, say $\omega^1(\alpha_2)=1$, then at least one of the values $\omega^1(\alpha_1)$ and $\omega^1(\alpha_3)$ is positive, and it was argued that this is impossible.



There are two $1$-cochains of norm $3$, as can be deduced by an ad hoc analysis similar to the above. Their coboundaries, depicted in the left two root systems of Figure \ref{fig:B^2(B_2), m=3,4}, are nonisomorphic. Norm $4$ provides another occasion where additional constraints from Theorem \ref{thm:alias satisfy game of roots} play a part, since there are multiple $1$-cochains of norm $4$ on $B_2$ but only one of these satisfies (\ref{eq:d1omega1le1}) (or equivalently (\ref{eq:killing form minimal rank})) and is depicted in the right root system of Figure \ref{fig:B^2(B_2), m=3,4}.

\begin{figure}[h!]
\caption{The $2$-coboundaries $\mathsf{d}^1\omega^1\in B^2(B_2,\NN_0)$ where $\|\omega^1\|=3$ or $4$, satisfying (\ref{eq:d1omega1le1}).}
\label{fig:B^2(B_2), m=3,4}
\begin{minipage}{0.32\linewidth}
\begin{center}
\begin{tikzpicture}[scale=0.45]
  \path node at ( 0,0) [scale=\nodescale,shape=circle,draw,label=270: $ $] (zero) {$ $}	
	node at ( 4,0) [scale=\nodescale,shape=circle,draw,label=0: $ $] (one) {$ $}
  	node at ( 4,4) [scale=\nodescale,shape=circle,fill=\colorga,draw,label=45: $ $] (two) {$ $}
  	node at ( 0,4) [scale=\nodescale,shape=circle,fill=\colorga,draw,label=90: $ $] (three) {$ $}
	node at ( -4,4) [scale=\nodescale,shape=circle,fill=\colorga,draw,label=135: $ $] (four) {$ $}
	node at (-4,0) [scale=\nodescale,shape=circle,draw,label=180: $ $] (five) {$ $}
	node at ( -4,-4) [scale=\nodescale,shape=circle,draw,label=225: $ $] (six) {$ $}
	node at (0,-4) [scale=\nodescale,shape=circle,draw,label=270: $ $] (seven) {$ $}
	node at (4,-4) [scale=\nodescale,shape=circle,draw,label=315: $ $] (eight) {$ $};
	\draw [\colorga, line width=1pt](two) to node [sloped,above,near end]{$ $} (zero);
	\draw [\colorga, line width=1pt](three) to node [sloped,below,near end]{$ $} (zero);
	\draw [\colorga, line width=1pt](four) to node [sloped,below,near end]{$ $} (zero);
	\draw [\colorga, line width=1pt](six) to node [sloped,below,near end]{} (zero);
	\draw [\colorga, line width=1pt](seven) to node [sloped,below,near end]{} (zero);
	\draw [\colorga, line width=1pt](eight) to node [sloped,below,near end]{} (zero);
	\draw [\colorga, line width=1pt](two) to node [sloped,above]{$ $} (seven);
	\draw [\colorga, line width=1pt](four) to node [sloped,above]{$ $} (seven);
	\draw [\colorga, line width=1pt](six) to node [sloped,above]{$ $} (three);
	\draw [\colorga, line width=1pt](eight) to node [sloped,above]{$ $} (three);
\end{tikzpicture}
\end{center}
\end{minipage}
\begin{minipage}{0.32\linewidth}
\begin{center}
\begin{tikzpicture}[scale=0.45]
  \path node at ( 0,0) [scale=\nodescale,shape=circle,draw,label=270: $ $] (zero) {$ $}	
	node at ( 4,0) [scale=\nodescale,shape=circle,fill=\colorbe,draw,label=0: $ $] (one) {$ $}
  	node at ( 4,4) [scale=\nodescale,shape=circle,fill=\colorbe,draw,label=45: $ $] (two) {$ $}
  	node at ( 0,4) [scale=\nodescale,shape=circle,fill=\colorbe,draw,label=90: $ $] (three) {$ $}
	node at ( -4,4) [scale=\nodescale,shape=circle,draw,label=135: $ $] (four) {$ $}
	node at (-4,0) [scale=\nodescale,shape=circle,draw,label=180: $ $] (five) {$ $}
	node at ( -4,-4) [scale=\nodescale,shape=circle,draw,label=225: $ $] (six) {$ $}
	node at (0,-4) [scale=\nodescale,shape=circle,draw,label=270: $ $] (seven) {$ $}
	node at (4,-4) [scale=\nodescale,shape=circle,draw,label=315: $ $] (eight) {$ $};
	\draw [\colorbe, line width=1pt](one) to node [sloped,above,near end]{$ $} (zero);
	\draw [\colorbe, line width=1pt](two) to node [sloped,above,near end]{$ $} (zero);
	\draw [\colorbe, line width=1pt](three) to node [sloped,below,near end]{$ $} (zero);
	\draw [\colorbe, line width=1pt](five) to node [sloped,below,near end]{} (zero);
	\draw [\colorbe, line width=1pt](six) to node [sloped,below,near end]{} (zero);
	\draw [\colorbe, line width=1pt](seven) to node [sloped,below,near end]{} (zero);
	\draw [\colorbe, line width=1pt](one) to node [sloped,above]{$ $} (three);
	\draw [\colorbe, line width=1pt](three) to node [sloped,above]{$ $} (five);
	\draw [\colorbe, line width=1pt](seven) to node [sloped,above]{$ $} (one);
	\draw [\colorbe, line width=1pt](six) to node [sloped,above]{$ $} (one);
	\draw [\colorbe, line width=1pt](six) to node [sloped,above]{$ $} (three);
\end{tikzpicture}
\end{center}
\end{minipage}
\begin{minipage}{0.32\linewidth}
\begin{center}
\begin{tikzpicture}[scale=0.45]
 \path node at ( 0,0) [scale=\nodescale,shape=circle,draw,label=270: $ $] (zero) {$ $}	
	node at ( 4,0) [scale=\nodescale,shape=circle,fill=\coloral,draw,label=0: $ $] (one) {$ $}
  	node at ( 4,4) [scale=\nodescale,shape=circle,fill=\coloral,draw,label=45: $ $] (two) {$ $}
  	node at ( 0,4) [scale=\nodescale,shape=circle,fill=\coloral,draw,label=90: $ $] (three) {$ $}
	node at ( -4,4) [scale=\nodescale,shape=circle,fill=\coloral,draw,label=135: $ $] (four) {$ $}
	node at (-4,0) [scale=\nodescale,shape=circle,draw,label=180: $ $] (five) {$ $}
	node at ( -4,-4) [scale=\nodescale,shape=circle,draw,label=225: $ $] (six) {$ $}
	node at (0,-4) [scale=\nodescale,shape=circle,draw,label=270: $ $] (seven) {$ $}
	node at (4,-4) [scale=\nodescale,shape=circle,draw,label=315: $ $] (eight) {$ $};
	\draw [\coloral, line width=1pt](one) to node [sloped,above,near end]{$ $} (zero);
	\draw [\coloral, line width=1pt](two) to node [sloped,above,near end]{$ $} (zero);
	\draw [\coloral, line width=1pt] (three) to node [sloped,below,near end]{$ $} (zero);
	\draw [\coloral, line width=1pt] (four) to node [sloped,below,near end]{$ $} (zero);
	\draw [\coloral, line width=1pt] (five) to node [sloped,below,near end]{} (zero);
	\draw [\coloral, line width=1pt] (six) to node [sloped,below,near end]{} (zero);
	\draw [\coloral, line width=1pt] (seven) to node [sloped,below,near end]{} (zero);
	\draw [\coloral, line width=1pt] (eight) to node [sloped,below,near end]{} (zero);
	\draw [\coloral, line width=1pt] (one) to node [sloped,above]{$ $} (three);
	\draw [\coloral, line width=1pt] (seven) to node [sloped,above]{$ $} (one);
	\draw [\coloral, line width=1pt] (four) to node [sloped,above]{$ $} (seven);
	\draw [\coloral, line width=1pt] (four) to node [sloped,above]{$ $} (one);
	\draw [\coloral, line width=1pt] (six) to node [sloped,above]{$ $} (one);
	\draw [\coloral, line width=1pt] (six) to node [sloped,above]{$ $} (three);
\end{tikzpicture}
\end{center}
\end{minipage}
\end{figure}



This analysis of $B_2$ provides limited information about the ALiAs $\salia{sp}{\boiiiiiiii}{\Gamma}$, $\salia{sp}{\byiiiiiii}{\Gamma}$ and $\salia{so}{\byiiiiiiii}{\Gamma}$. Not only because there are two $1$-cochains of norm $3$ with nonisomorphic coboundaries, but also because the various ways to combine the cochains of this section into an element of $B^1(B_2,\NN_0^3)$ of norm $(4,3,3)$ result in various coboundaries.

\subsubsection{Root System $A_3$}
\label{sec:A3}
There is only one three-dimensional root system that is involved in the isomorphism question for ALiAs, namely $A_3$.
The roots are the midpoints of the edges of a cube (or octahedron) and they all have the same length. 
In Figure \ref{fig:A3 and edges not incident to 0} all basis elements of the $2$-chains are shown except the ones incident to zero.
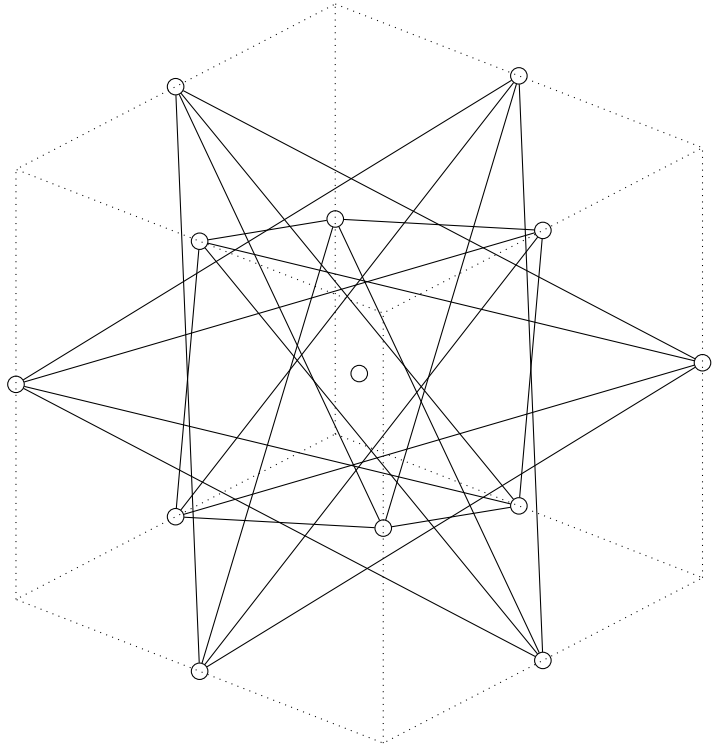
\begin{figure}
\caption{Basis elements of $C_2(A_3)$ not incident to $0$.}
\label{fig:A3 and edges not incident to 0}
\begin{center}
\tdplotsetmaincoords{63}{131}
\begin{tikzpicture}
		[tdplot_main_coords,
			cube/.style={dotted},
			grid/.style={very thin,gray},
			axis/.style={->,blue,thick},scale=3.2]



	\draw[cube] (-1,-1,-1) -- (-1,1,-1) -- (1,1,-1) -- (1,-1,-1) -- cycle;
	\draw[cube] (-1,-1,1) -- (-1,1,1) -- (1,1,1) -- (1,-1,1) -- cycle;
	
	\draw[cube] (-1,-1,-1) -- (-1,-1,1);
	\draw[cube] (-1,1,-1) -- (-1,1,1);
	\draw[cube] (1,-1,-1) -- (1,-1,1);
	\draw[cube] (1,1,-1) -- (1,1,1);

 \path node at (0, 0,0) [scale=\nodescale,shape=circle,draw,label=270: $ $] (zero) {$ $}
	node at (1,-1,0) [scale=\nodescale,shape=circle,draw,label=0: $ $] (one) {$ $}
  	node at (0,1,-1) [scale=\nodescale,shape=circle,draw,label=60: $ $] (two) {$ $}
  	node at (1,0,1) [scale=\nodescale,shape=circle,draw,label=120: $ $] (three) {$ $}
	node at (-1,1,0) [scale=\nodescale,shape=circle,draw,label=180: $ $] (four) {$ $}
	node at ( 0,-1,1) [scale=\nodescale,shape=circle,draw,label=240: $ $] (five) {$ $}
	node at (-1,0,-1) [scale=\nodescale,shape=circle,draw,label=300: $ $] (six) {$ $}

	node at (1,0,-1) [scale=\nodescale,shape=circle,draw,label=0: $ $] (seven) {$ $}
  	node at ( 1,1,0) [scale=\nodescale,shape=circle,draw,label=60: $ $] (eight) {$ $}
  	node at (0,1,1) [scale=\nodescale,shape=circle,draw,label=120: $ $] (nine) {$ $}
	node at  ( -1,0,1) [scale=\nodescale,shape=circle,draw,label=180: $ $] (ten) {$ $}
	node at(-1,-1,0) [scale=\nodescale,shape=circle,draw,label=240: $ $] (eleven) {$ $}
	node at  ( 0,-1,-1) [scale=\nodescale,shape=circle,draw,label=300: $ $] (twelve) {$ $}
	;

	\draw (seven) to node [sloped,above]{$ $} (nine);
	\draw (eight) to node [sloped,above]{$ $} (ten);
	\draw(nine) to node [sloped,above]{$ $} (eleven);
	\draw (ten) to node [sloped,above]{$ $} (twelve);
	\draw(eleven) to node [sloped,above]{$ $} (seven);
	\draw(twelve) to node [sloped,above]{$ $} (eight);

	\draw (seven) to node [sloped,above]{$ $} (five);
	\draw (seven) to node [sloped,above]{$ $} (four);
	\draw (eight) to node [sloped,above]{$ $} (five);
	\draw (eight) to node [sloped,above]{$ $} (six);
	\draw (nine) to node [sloped,above]{$ $} (one);
	\draw (nine) to node [sloped,above]{$ $} (six);
	\draw (ten) to node [sloped,above]{$ $} (one);
	\draw (ten) to node [sloped,above]{$ $} (two);
	\draw (eleven) to node [sloped,above]{$ $} (three);
	\draw (eleven) to node [sloped,above]{$ $} (two);
	\draw (twelve) to node [sloped,above]{$ $} (three);
	\draw (twelve) to node [sloped,above]{$ $} (four);

	\draw (one) to node [sloped,above]{$ $} (two);
	\draw (two) to node [sloped,above]{$ $} (three);
	\draw (three) to node [sloped,above]{$ $} (four);
	\draw (four) to node [sloped,above]{$ $} (five);
	\draw (five) to node [sloped,above]{$ $} (six);
	\draw (six) to node [sloped,above]{$ $} (one);
\end{tikzpicture}
\end{center}
\end{figure}

The computer package Sage \cite{sage} is used to compute all solutions $\omega^1\in C^1(A_3,\NN_0)$ to the equation $\|\omega^1\|=\kappa$ satisfying (\ref{eq:omega1le1}) and (\ref{eq:d1omega1le1}).
Multiple solutions are found with norm $\kappa=3,\ldots,6$, and no solutions for other values of $\kappa$. 
Interestingly, for all the relevant norms, $\kappa=4, 5$ or $6$ (cf.~Table \ref{tab:noI}), these solutions provide just one coboundary $\mathsf{d}^1\omega^1$ up to isomorphism. They are shown in Figure \ref{fig:B^2(A_3), m=4,5,6}, together with a choice of $1$-cochain integrating them. 

\begin{figure}[h!]
\caption{The $2$-coboundaries $\mathsf{d}^1\omega^1\in B^2(A_3,\NN_0)$ where $\|\omega^1\|=4,5$ or $6$, satisfying (\ref{eq:omega1le1}, \ref{eq:d1omega1le1}).}
\label{fig:B^2(A_3), m=4,5,6}
\begin{minipage}{0.32\linewidth}
\begin{center}
\tdplotsetmaincoords{63}{131}
\begin{tikzpicture}
		[tdplot_main_coords,
			cube/.style={dotted},
			grid/.style={very thin,gray},
			axis/.style={->,blue,thick},scale=1.5]


	\draw[cube] (-1,-1,-1) -- (-1,1,-1) -- (1,1,-1) -- (1,-1,-1) -- cycle;
	\draw[cube] (-1,-1,1) -- (-1,1,1) -- (1,1,1) -- (1,-1,1) -- cycle;
	
	\draw[cube] (-1,-1,-1) -- (-1,-1,1);
	\draw[cube] (-1,1,-1) -- (-1,1,1);
	\draw[cube] (1,-1,-1) -- (1,-1,1);
	\draw[cube] (1,1,-1) -- (1,1,1);

 \path node at (0, 0,0) [scale=\nodescale,shape=circle,draw,label=270: $ $] (zero) {$ $}
	node at (1,-1,0) [scale=\nodescale,shape=circle,draw,fill=\colorga,label=0: $ $] (one) {$ $}
  	node at (0,1,-1) [scale=\nodescale,shape=circle,draw,label=60: $ $] (two) {$ $}
  	node at (1,0,1) [scale=\nodescale,shape=circle,draw,fill=\colorga,label=120: $ $] (three) {$ $}
	node at (-1,1,0) [scale=\nodescale,shape=circle,draw,label=180: $ $] (four) {$ $}
	node at ( 0,-1,1) [scale=\nodescale,shape=circle,draw,label=240: $ $] (five) {$ $}
	node at (-1,0,-1) [scale=\nodescale,shape=circle,draw,label=300: $ $] (six) {$ $}

	node at (1,0,-1) [scale=\nodescale,shape=circle,draw,fill=\colorga,label=0: $ $] (seven) {$ $}
  	node at ( 1,1,0) [scale=\nodescale,shape=circle,draw,fill=\colorga,label=60: $ $] (eight) {$ $}
  	node at (0,1,1) [scale=\nodescale,shape=circle,draw,label=120: $ $] (nine) {$ $}
	node at  ( -1,0,1) [scale=\nodescale,shape=circle,draw,label=180: $ $] (ten) {$ $}
	node at(-1,-1,0) [scale=\nodescale,shape=circle,draw,label=240: $ $] (eleven) {$ $}
	node at  ( 0,-1,-1) [scale=\nodescale,shape=circle,draw,label=300: $ $] (twelve) {$ $}
	;
	\draw [\colorga,line width=1pt](one) to node [sloped,above,near end]{$ $} (zero);
	\draw [\colorga,line width=1pt] (three) to node [sloped,below,near end]{$ $} (zero);
	\draw [\colorga,line width=1pt] (four) to node [sloped,below,near end]{$ $} (zero);
	\draw [\colorga,line width=1pt] (six) to node [sloped,below,near end]{} (zero);
	\draw [\colorga,line width=1pt] (seven) to node [sloped,below,near end]{} (zero);
	\draw [\colorga,line width=1pt] (eight) to node [sloped,below,near end]{} (zero);
	\draw [\colorga,line width=1pt] (ten) to node [sloped,below,near end]{} (zero);
	\draw [\colorga,line width=1pt] (eleven) to node [sloped,below,near end]{} (zero);
	\draw [\colorga,line width=1pt] (eight) to node [sloped,above]{$ $} (ten);
	\draw [\colorga,line width=1pt] (eleven) to node [sloped,above]{$ $} (seven);
	\draw [\colorga,line width=1pt] (seven) to node [sloped,above]{$ $} (four);
	\draw [\colorga,line width=1pt] (eight) to node [sloped,above]{$ $} (six);
	\draw [\colorga,line width=1pt] (ten) to node [sloped,above]{$ $} (one);
	\draw [\colorga,line width=1pt] (eleven) to node [sloped,above]{$ $} (three);
	\draw [\colorga,line width=1pt] (three) to node [sloped,above]{$ $} (four);
	\draw [\colorga,line width=1pt] (six) to node [sloped,above]{$ $} (one);
\end{tikzpicture}
\end{center}
\end{minipage}
\begin{minipage}{0.32\linewidth}
\begin{center}
\tdplotsetmaincoords{63}{131}
\begin{tikzpicture}
		[tdplot_main_coords,
			cube/.style={dotted},
			grid/.style={very thin,gray},
			axis/.style={->,blue,thick},scale=1.5]


	\draw[cube] (-1,-1,-1) -- (-1,1,-1) -- (1,1,-1) -- (1,-1,-1) -- cycle;
	\draw[cube] (-1,-1,1) -- (-1,1,1) -- (1,1,1) -- (1,-1,1) -- cycle;
	
	\draw[cube] (-1,-1,-1) -- (-1,-1,1);
	\draw[cube] (-1,1,-1) -- (-1,1,1);
	\draw[cube] (1,-1,-1) -- (1,-1,1);
	\draw[cube] (1,1,-1) -- (1,1,1);

 \path node at (0, 0,0) [scale=\nodescale,shape=circle,draw,label=270: $ $] (zero) {$ $}
	node at (1,-1,0) [scale=\nodescale,shape=circle,draw,fill=\colorbe,label=0: $ $] (one) {$ $}
  	node at (0,1,-1) [scale=\nodescale,shape=circle,draw,label=60: $ $] (two) {$ $}
  	node at (1,0,1) [scale=\nodescale,shape=circle,draw,fill=\colorbe,label=120: $ $] (three) {$ $}
	node at (-1,1,0) [scale=\nodescale,shape=circle,draw,label=180: $ $] (four) {$ $}
	node at ( 0,-1,1) [scale=\nodescale,shape=circle,draw,fill=\colorbe,label=240: $ $] (five) {$ $}
	node at (-1,0,-1) [scale=\nodescale,shape=circle,draw,label=300: $ $] (six) {$ $}

	node at (1,0,-1) [scale=\nodescale,shape=circle,draw,label=0: $ $] (seven) {$ $}
  	node at ( 1,1,0) [scale=\nodescale,shape=circle,draw,fill=\colorbe,label=60: $ $] (eight) {$ $}
  	node at (0,1,1) [scale=\nodescale,shape=circle,draw,fill=\colorbe,label=120: $ $] (nine) {$ $}
	node at  ( -1,0,1) [scale=\nodescale,shape=circle,draw,label=180: $ $] (ten) {$ $}
	node at(-1,-1,0) [scale=\nodescale,shape=circle,draw,label=240: $ $] (eleven) {$ $}
	node at  ( 0,-1,-1) [scale=\nodescale,shape=circle,draw,label=300: $ $] (twelve) {$ $}
	;
	\draw [\colorbe,line width=1pt] (one) to node [sloped,above,near end]{$ $} (zero);
	\draw [\colorbe,line width=1pt] (two) to node [sloped,above,near end]{$ $} (zero);
	\draw [\colorbe,line width=1pt] (three) to node [sloped,below,near end]{$ $} (zero);
	\draw [\colorbe,line width=1pt] (four) to node [sloped,below,near end]{$ $} (zero);
	\draw [\colorbe,line width=1pt] (five) to node [sloped,below,near end]{} (zero);
	\draw [\colorbe,line width=1pt] (six) to node [sloped,below,near end]{} (zero);
	\draw [\colorbe,line width=1pt] (eight) to node [sloped,below,near end]{} (zero);
	\draw [\colorbe,line width=1pt] (nine) to node [sloped,below,near end]{} (zero);
	\draw [\colorbe,line width=1pt] (eleven) to node [sloped,below,near end]{} (zero);
	\draw [\colorbe,line width=1pt] (twelve) to node [sloped,below,near end]{} (zero);
	\draw [\colorbe,line width=1pt] (nine) to node [sloped,above]{$ $} (eleven);
	\draw [\colorbe,line width=1pt] (twelve) to node [sloped,above]{$ $} (eight);
	\draw [\colorbe,line width=1pt] (eight) to node [sloped,above]{$ $} (five);
	\draw [\colorbe,line width=1pt] (eight) to node [sloped,above]{$ $} (six);
	\draw [\colorbe,line width=1pt] (nine) to node [sloped,above]{$ $} (one);
	\draw [\colorbe,line width=1pt] (nine) to node [sloped,above]{$ $} (six);
	\draw [\colorbe,line width=1pt] (one) to node [sloped,above]{$ $} (two);
	\draw [\colorbe,line width=1pt] (four) to node [sloped,above]{$ $} (five);
	\draw [\colorbe,line width=1pt] (five) to node [sloped,above]{$ $} (six);
	\draw [\colorbe,line width=1pt] (six) to node [sloped,above]{$ $} (one);
\end{tikzpicture}
\end{center}	
\end{minipage}
\begin{minipage}{0.32\linewidth}
\begin{center}
\tdplotsetmaincoords{63}{131}
\begin{tikzpicture}
		[tdplot_main_coords,
			cube/.style={dotted},
			grid/.style={very thin,gray},
			axis/.style={->,blue,thick},scale=1.5]


	\draw[cube] (-1,-1,-1) -- (-1,1,-1) -- (1,1,-1) -- (1,-1,-1) -- cycle;
	\draw[cube] (-1,-1,1) -- (-1,1,1) -- (1,1,1) -- (1,-1,1) -- cycle;
	
	\draw[cube] (-1,-1,-1) -- (-1,-1,1);
	\draw[cube] (-1,1,-1) -- (-1,1,1);
	\draw[cube] (1,-1,-1) -- (1,-1,1);
	\draw[cube] (1,1,-1) -- (1,1,1);

 \path node at (0, 0,0) [scale=\nodescale,shape=circle,draw,label=270: $ $] (zero) {$ $}
	node at (1,-1,0) [scale=\nodescale,shape=circle,draw,fill=\coloral, label=0: $ $] (one) {$ $}
  	node at (0,1,-1) [scale=\nodescale,shape=circle,draw, fill=\coloral, label=60: $ $] (two) {$ $}
  	node at (1,0,1) [scale=\nodescale,shape=circle,draw,fill=\coloral,label=120: $ $] (three) {$ $}
	node at (-1,1,0) [scale=\nodescale,shape=circle,draw,label=180: $ $] (four) {$ $}
	node at ( 0,-1,1) [scale=\nodescale,shape=circle,draw,label=240: $ $] (five) {$ $}
	node at (-1,0,-1) [scale=\nodescale,shape=circle,draw,label=300: $ $] (six) {$ $}

	node at (1,0,-1) [scale=\nodescale,shape=circle,draw,fill=\coloral, label=0: $ $] (seven) {$ $}
  	node at ( 1,1,0) [scale=\nodescale,shape=circle,draw,fill=\coloral,label=60: $ $] (eight) {$ $}
  	node at (0,1,1) [scale=\nodescale,shape=circle,draw,fill=\coloral,label=120: $ $] (nine) {$ $}
	node at  ( -1,0,1) [scale=\nodescale,shape=circle,draw,label=180: $ $] (ten) {$ $}
	node at(-1,-1,0) [scale=\nodescale,shape=circle,draw,label=240: $ $] (eleven) {$ $}
	node at  ( 0,-1,-1) [scale=\nodescale,shape=circle,draw,label=300: $ $] (twelve) {$ $}
	;
	\draw [\coloral,line width=1pt](one) to node [sloped,above,near end]{$ $} (zero);
	\draw [\coloral,line width=1pt](two) to node [sloped,above,near end]{$ $} (zero);
	\draw [\coloral,line width=1pt](three) to node [sloped,below,near end]{$ $} (zero);
	\draw [\coloral,line width=1pt](four) to node [sloped,below,near end]{$ $} (zero);
	\draw [\coloral,line width=1pt](five) to node [sloped,below,near end]{} (zero);
	\draw [\coloral,line width=1pt](six) to node [sloped,below,near end]{} (zero);
	\draw [\coloral,line width=1pt](seven) to node [sloped,below,near end]{} (zero);
	\draw [\coloral,line width=1pt](eight) to node [sloped,below,near end]{} (zero);
	\draw [\coloral,line width=1pt](nine) to node [sloped,below,near end]{} (zero);
	\draw [\coloral,line width=1pt](ten) to node [sloped,below,near end]{} (zero);
	\draw [\coloral,line width=1pt](eleven) to node [sloped,below,near end]{} (zero);
	\draw [\coloral,line width=1pt](twelve) to node [sloped,below,near end]{} (zero);
	\draw [\coloral,line width=1pt](seven) to node [sloped,above]{$ $} (nine);
	\draw [\coloral,line width=1pt](nine) to node [sloped,above]{$ $} (eleven);
	\draw [\coloral,line width=1pt](eleven) to node [sloped,above]{$ $} (seven);
	\draw [\coloral,line width=1pt](nine) to node [sloped,above]{$ $} (one);
	\draw [\coloral,line width=1pt](nine) to node [sloped,above]{$ $} (six);
	\draw [\coloral,line width=1pt](ten) to node [sloped,above]{$ $} (one);
	\draw [\coloral,line width=1pt](ten) to node [sloped,above]{$ $} (two);
	\draw [\coloral,line width=1pt](eleven) to node [sloped,above]{$ $} (three);
	\draw [\coloral,line width=1pt](eleven) to node [sloped,above]{$ $} (two);
	\draw [\coloral,line width=1pt](one) to node [sloped,above]{$ $} (two);
	\draw [\coloral,line width=1pt](two) to node [sloped,above]{$ $} (three);
	\draw [\coloral,line width=1pt](six) to node [sloped,above]{$ $} (one);
\end{tikzpicture}
\end{center}	
\end{minipage}

\end{figure}




\cleardoublepage\phantomsection
\addcontentsline{toc}{chapter}{Conclusions}
\chapter*{Conclusions\markboth{Conclusions}{}}
\label{ch:conclusions}

Alongside the computational classification project of Automorphic Lie Algebras a theoretical treatment of the subject takes shape. This theory aims to explain the observations of the computational results. Most notably the uniformity over different reduction groups and the simplicity of the Lie algebra structure that becomes visible by the construction of a Cartan-Weyl normal form. Ultimately, the theory aims to answer the isomorphism question, assist in the classification project and provide information about Automorphic Lie Algebras which are computationally inaccessible.

The invariants of Automorphic Lie Algebras obtained in this thesis provide various predictions. They determine the structure of Automorphic Lie Algebras as a $\CC[\II]$-module and provide important information about the faithful Lie algebraic representation given by matrices of invariants. Moreover, the invariants put severe constraints on the Lie algebra structures of Automorphic Lie Algebras. 

\section*{Summary of the Results\markboth{Conclusions}{Summary of the Results}}

Chapter \ref{ch:A} discusses how the workplace of Automorphic Lie Algebras can be moved from the Riemann sphere and its automorphism group $\Aut(\overline{\CC})\cong \PSL_2(\CC)$ to the linear space $\CC^2$ and the group $SL_2(\CC)$. A method is devised to study Automorphic Lie Algebras for all pole orbits through a single Lie algebra. Using well known theory for the representations of $SL_2(\CC)$ it is then shown that Automorphic Lie Algebras are free modules over the polynomial ring in one variable. Moreover, the number of generators equals the dimension of the base Lie algebra. Since this is true for any Automorphic Lie Algebra it is an invariant. It also allows the definition of the determinant of invariant vectors, for which a simple formula is found.

Chapter \ref{ch:B} discusses the behavior of complex Lie algebras when they are acted upon by a finite group. Formulas to decompose classical Lie algebras into irreducible group representations are established and used to decompose all the base Lie algebras occurring in the set up of Automorphic Lie Algebras in this thesis. Furthermore, it is explained why only inner automorphisms occur in the representations. 
Finally, the natural representation of Automorphic Lie Algebras is defined and a full classification of their evaluations is given. In particular, these complex Lie algebras turn out to depend solely on the base Lie algebra and the type of orbit containing the point of evaluation, hence this is another invariant of Automorphic Lie Algebras. Combining this result with the formula for the determinant of invariant vectors, obtained in Chapter \ref{ch:A}, yields a third invariant. 

In Chapter \ref{ch:C} the Cartan-Weyl normal form for Automorphic Lie Algebras is introduced. 
It is illustrated how this normal form presents Automorphic Lie Algebras in a clear and familiar manner. A representation by matrices of invariants is then used to concretise the Lie algebra structure.
This motivates the introduction of root system cohomology. The constraints on the Lie algebra structure imposed by the invariants of Automorphic Lie Algebras are formulated in terms of cochains and their boundaries. This immediately describes $\mf{sl}_2(\CC)$-based Automorphic Lie Algebras with all pole orbits and $\mf{sl}_3(\CC)$-based Automorphic Lie Algebras whose poles are in one of the two smallest exceptional orbits, $\Gal$ or $\Gbe$. The implications for the isomorphism question are discussed in general. Moreover, necessary conditions for isomorphisms between Automorphic Lie Algebras are established using the invariants and the cohomological set up.

Throughout the main body of this thesis the results are illustrated by explicit calculations for the dihedral symmetric case. The dihedral group is suitable for this purpose because it contains most of the group theoretical difficulties one encounters with e.g.~Schur covers and exponents, but at the same time its polynomial invariants are simple enough to tackle by hand. As a byproduct the complete classification of dihedral Automorphic Lie Algebras \cite{knibbeler2014automorphic} is established. 

\section*{Open Questions and Research Directions\markboth{Conclusions}{Open Questions and Research Directions}}

\subsubsection*{Existence of a Cartan-Weyl Normal Form for Automorphic Lie Algebras}
The Cartan-Weyl normal form for Automorphic Lie Algebras, introduced in Section \ref{sec:normal form}, is a set of generators of the Lie algebra as a $\mero_\Gamma^G$-module that identifies a Cartan subalgebra and diagonalises its adjoint action, analogous to the Cartan-Weyl basis for semisimple complex Lie algebras. The first step in a general existence proof is given by Theorem \ref{thm:dimV generators} stating that any Automorphic Lie Algebra $\left(\mf{g}\otimes \mero_\Gamma\right)^G$ is freely generated by $\dim \mf{g}$ elements as a $\mero_\Gamma^G$-module.
Identifying a Cartan subalgebra boils down to the search for $\ell=\rank \mf{g}$ invariant matrices which are simultaneously diagonalisable and have eigenvalues in $\CC$ (that is, independent of $\lambda$). The second part of this problem is the identification of the root spaces of this Cartan subalgebra.
Proof of existence of such generators, conjectured in Section \ref{sec:normal form}, would be a major contribution to the subject and provide the final piece of the proof that justifies the use of root cohomology in the study of Automorphic Lie Algebras.

\subsubsection*{An Equivalence Relation on $2$-Cocycles}
If an Automorphic Lie Algebra allows a Cartan-Weyl normal form, as we know many do, cf.~\cite{knibbeler2014higher, knibbeler2014automorphic, LS10}, then the Lie algebra structure can be described by a function sending two roots of the base Lie algebra to a triple of natural numbers. A cohomology theory on root systems can be defined such that the $2$-cocycles are exactly those functions that describe a Lie algebra in Cartan-Weyl normal form over a graded ring. Thus to study such Lie algebras one can investigate the more tractable $2$-cocycles. However, for this to be effective it is crucial to know what $2$-cocycles produce isomorphic Lie algebras. This is not a straightforward problem as the cocycles only describe the Lie algebra in normal form, hence an isomorphism of Lie algebras needs to be composed with an isomorphism that takes it to a normal form before the effect on the cocycle can be studied.

The invariants of Automorphic Lie Algebras obtained in this thesis put severe constraints on the $2$-cocycles that describe Automorphic Lie Algebras, leaving only a handful of options. It is at the time of writing not known weather these cocycles generate distinct or isomorphic Lie algebras. If the desired equivalence relation on $2$-cocycles is obtained there is a chance that the answer of the isomorphism question soon follows, and the classification project of Automorphic Lie Algebras leaps ahead.

\subsubsection{Second Cohomology Groups}
The tentative set up for a cohomology theory of root systems could prove an interesting direction of further study by itself. The cohomology groups for instance are not yet studied in this thesis. The second cohomology group in particular has an obvious interpretation in terms of Lie algebras over graded rings and their representations, as it measures the amount of such Lie algebras that do not allow a representation given by a $1$-cochain in the canonical way described in Example \ref{ex:representation by 1-cochain}. Preliminary unpublished research in this direction looks promising and the subject has raised much interest from the integrable system community.

\subsubsection{Generalisations of Automorphic Lie Algebras}
The ingredients that go into an Automorphic Lie Algebra: the base Lie algebra, its field (Automorphic Lie Algebras over finite fields have been proposed), the automorphisms of the Lie algebra in the reduction group, the Riemann surface; they all can be generalised. To complete the classification of Automorphic Lie Algebras based on simple complex Lie algebras one has to study the exceptional Lie algebras. The icosahedral group at least can be represented in the exceptional Lie groups \cite{lusztig2003homomorphisms}.

The study of Automorphic Lie Algebras on a compact Riemann surface of positive genus was started in \cite{chopp2011lie}. Little of the current thesis can be used directly to continue this research, as the results rely from the very beginning on the Riemann sphere, through the properties of the polyhedral groups. However, the classification project of finite groups of automorphisms of compact Riemann surfaces is well ahead. 

\subsubsection{Kac-Moody Affiliations}
Kac describes in his monograph \cite{kac1994infinite} how central extensions of current algebras, possibly reduced by a cyclic reduction group, and adjoint by a derivation, give rise to all Kac-Moody algebras of affine type. One could mimic this construction starting from the Automorphic Lie Algebras that we have come to know, or even the larger class associated to the $2$-cocycles, and aim to describe this Lie algebra by a generalised Cartan matrix. Research of this nature can illuminate the relation between Automorphic Lie Algebras and Kac-Moody algebras. Moreover, it can enable the application of Automorphic Lie Algebras in conformal field theory and related areas of theoretical physics, where Kac-Moody algebras play an important role.

\subsubsection{Integrable Systems}
One could also stay true to the original motivation of the subject, integrable partial differential equations, and continue research in line with \cite{bury2010automorphic}. It would be interesting to investigate the consequences of the isomorphism theorem of \cite{knibbeler2014higher} or to investigate the consequences of the isomorphism conjecture. Integrable partial differential equations can be constructed using Automorphic Lie Algebras, as demonstrated in \cite{bury2010automorphic,Lombardo}. We emphasise that these equations depend on the Lie algebra structure only. Therefore they are independent of the choice of reduction group, by the isomorphism theorem of \cite{knibbeler2014higher}. Yet, for the study of this equation one can use a Lax pair from the natural representation of the Automorphic Lie Algebra, whose analytic structure \emph{does} depend on the reduction group. Hence, for each polyhedral group appearing as a reduction group for this Automorphic Lie Algebra, one could expect a class of solutions to the equation with symmetry properties of this particular group.

Naturally it would also be interesting to start an investigation on the relation between
Automorphic Lie Algebras and discrete integrable systems
as the latter display rich algebraic structures. This is an application for pioneers, yet it would be timely because of the rising awareness of the importance of discrete integrable systems in fundamental physics.

\cleardoublepage\phantomsection
\addcontentsline{toc}{chapter}{Index of Notation}
\chapter*{Index of Notation\markboth{Index of Notation}{}} 
\label{ch: Index of symbols}

\begin{center}
\begin{longtable}{ll}
\hline 
\multicolumn{1}{l}{Notation} & \multicolumn{1}{l}{Description} \\ \hline 
\endfirsthead

\multicolumn{2}{c}
{{\bfseries  continued from previous page}} \\
\hline 
\multicolumn{1}{l}{Notation} & \multicolumn{1}{l}{Description} \\ \hline 
\endhead

\hline \multicolumn{2}{r}{{Continued on next page}} \\ 
\endfoot

\hline 
\endlastfoot
ALiAs & Automorphic Lie Algebras\\
form & A polynomial of homogeneous degree\\
gcd & Greatest common divisor\\
lcm & Least common multiple\\
$\Hom(U,V)$ & Homomorphisms $U\rightarrow V$\\
$\End(V)$ & Endomorphisms, $\Hom(V,V)$\\
$\Aut(V)$ & Automorphisms, invertible endomorphisms\\
$\GL(V)$ & The general linear group, $\Aut(V)$ where $V$ is a vector space\\
$\NN_0$ & The nonnegative integers $\NN\cup\{0\}$\\
$\CC^\ast$ & The multiplicative group of nonzero complex numbers $\CC\setminus\{0\}$\\
$\overline{\CC}$ & The one-point compactification of the complex plane, also known as\\
&the Riemann sphere and the complex projective line $\CC P^1$\\
$G$ & A group, often a polyhedral group, i.e.~a finite subgroup of $\Aut(\overline{\CC})$\\
$\triv$ & The trivial character of a group $G$, i.e.~the map $\triv:G\rightarrow \{1\}\subset\CC^\ast$\\
$|G|$ & The order of a group $G$\\
$\|G\|$ & The exponent of a group $G$, Definition \ref{def:exponent of a group}\\
$G^\flat$ & The binary polyhedral group related to the polyhedral group $G$\\
$G_\lambda$ & Stabiliser subgroup $\{g\in G\;|\;g\lambda=\lambda\}$\\
$\Gamma$ & $G$-orbit in $\overline{\CC}$\\
$\Omega$ & Index set for all exceptional orbits of $G.\;\Omega=\{\al, \be\}\text{ if }G\text{ is cyclic and }$\\
&$\Omega=\{\al, \be, \ga\}\text{ if }G\text{ is a non-cyclic polyhedral group}$\\
$\Gamma_i, \;i\in\Omega$ & $G$-orbit in $\overline{\CC}$ of size $< |G|\text{, called \emph{exceptional orbit}}$\\
$d_i , \;i\in\Omega$& {The size $|\Gamma_i|$ of the exceptional orbit $\Gamma_i$}\\
$\nu_i , \;i\in\Omega$& {The size $\frac{|G|}{d_i}$ of the nontrivial stabiliser subgroups at }$\Gamma_i$\\
$\omega_N$& {A primitive $N$-th root of unity, e.g. }$e^{\frac{2\pi i}{N}}$\\
$C_A(B)$ & {The centraliser of $B$, }$\{a\in A\;|\;ab=ba,\,\forall b\in B\}$\\
$Z(G)$ & {The centre of $G$, }$\{g\in G\;|\;gh=hg,\,\forall h\in G\}=C_G(G)$\\
$V_\chi$ & The vector space of the representation affording the character $\chi$\\
$\chi_V$ & The character of the representation with vector space $V_\chi$\\
$\kappa(\chi)_i,\;i\in\Omega$ & the half integer $\nicefrac{1}{2}\,\codim V_\chi^{G_\mu}$, where $\mu\in\Gamma_i$, Definition \ref{def:kappa(chi)}\\
$V^\ast$ & {The dual of a vector space $V$, i.e.~}$\Hom(V,k)${ where $k$ is the field of $V$}\\
$\CC[U]$ & {The ring of polynomials whose variables are basis elements of $U^\ast$}\\
$R$&$\CC[U]${ where $U$ is the natural representation of the involved group}\\
$\cV(F)$ & {If $F\in\CC[U]$ then }$\cV(F)=\{u\in U\;|\;F(u)=0\}${, Definition \ref{def:algebraic set}}\\
$\cN(\Gamma)$ & {If $\Gamma\subset U$ then }$\cN(\Gamma)=\{F\in \CC[U]\;|\;F|_{\Gamma}=0\}${, Definition \ref{def:algebraic set}}\\
$F_\Gamma$ &  {The generator of }$\cN(\Gamma)${, Definition \ref{def:form of orbit}}\\
$F_i$ & Short hand notation for $F_{\Gamma_i}$\\
$\p{d}$ & {Prehomogenisation, Definition \ref{def:prehomogenisation}}\\
$\h{\Gamma}$ & {Homogenisation, Definition \ref{def:homogenisation}}\\
$\ii$& The quotient $\h{\Gamma}F_i^{\nu_i}=\frac{F_i^{\nu_i}}{F_{\Gamma}^{\nu_{\Gamma}}}${ where }$\nu_{\Gamma}=\frac{|G|}{|\Gamma|}$, i.e. the meromorphic\\
& function on $\overline{\CC}$ with divisor $\nu_i\Gamma_i-\nu_\Gamma\Gamma$\\
$\II$& A nonconstant function of the set $\{\ii\,|\,i\in\Omega\}$.
{ This notation is only}\\ & {used in a context where the particular choice is irrelevant, e.g.~}$\CC[\II]$\\
$\mf{g}$ & {A Lie algebra over the complex numbers, unless otherwise stated}\\
$\mf{h}${, CSA} & {A Cartan subalgebra, a nilpotent selfnormalising subalgebra of }$\mf{g}$\\
$\ell$& {The rank $\rank \mf{g}=\dim\mf{h}$ of the Lie algebra}\\
$K(\cdot,\cdot)$& {The killing form }$\mf{g}\times\mf{g}\ni(a,b)\mapsto K(a,b)=\tr(\ad(a)\ad(b))\in\CC$\\
$(\cdot,\cdot)$& {The killing form restricted to }$\mf{h}\times\mf{h}$\\
$\roots$ & {The roots of a semisimple Lie algebra}\\
$\roots_0$ & $\roots\cup\{0\}$\\
$\sroots$ & {A choice of simple roots in }$\roots$\\
$\kappa(\roots)_i,\;i\in\Omega$ & The integer $\nicefrac{1}{2}\,\codim\mf{g}(V)^{G_\mu}${ where $\mf{g}(V)$ is a Lie algebra with root }\\& {system $\roots$ and a $G$-module as induced by a $G$-action on $V$, and $\mu\in\Gamma_i$}, \\ &Definition \ref{def:kappa(roots)}\\
$\mero$ & {The field of rational functions on the Riemann sphere}\\
$\mero_\Gamma$ & {The ring of functions in $\mero$ with poles restricted to }$\Gamma\subset\overline{\CC}$\\
$\overline{V}$ & The tensor product $V\otimes\mero$, i.e. $V$-valued meromorphic functions \\
$\overline{V}_\Gamma$ & The tensor product $V\otimes\mero_\Gamma$ \\
$C_1(\roots)$ & $1$-chains, formal $\ZZ$-span of the eigenvalues of a CSA, $\ZZ\langle\roots_0\rangle$\\
$C_m(\roots)$ & $m$-chains, formal $\ZZ$-span of $m$-tuples $(\alpha_1 ,\ldots, \alpha_m)\in\roots_0^m$ such that \\
&$(\alpha_1 ,\ldots,\alpha_j+\alpha_{j+1},\ldots, \alpha_m)\in C_{m-1}(\roots)\text{ for all } 1\le j< m$\\
$C^m(\roots,X)$ & $m$-ochains, $\Hom(C_m(\roots),X)$ where $X$ is an abelian group\\
$\mathsf{d}^m$ & The map $C^m(\roots,X)\ni\omega^m\mapsto \mathsf{d}^m\omega^m\in C^{m+1}(\roots,X)$ defined by\\&
$\mathsf{d}^m \omega^m(\alpha_0,\ldots,\alpha_m)=\omega^m(\alpha_1,\ldots,\alpha_m)$\\&$
+\sum_{j=1}^m (-1)^{j} \omega^m(\alpha_0,\ldots,\alpha_{j-1}+\alpha_j,\cdots,\alpha_m)$\\&$
-(-1)^m\omega^m(\alpha_0,\ldots,\alpha_{m-1})$\\
$C^m(\roots,\NN_0^q)$ & The set of all $\omega^m\in C^m(\roots,\ZZ^q)$ such that $\omega^m$ and $\mathsf{d}^m\omega^m$ take \\&values in $\NN_0^q$\\ 
$\|\omega^m\|$& {The sum $\sum\omega^m(\alpha)$ ranging over the basis elements $\alpha$ of }$C_m(\roots)$\\
$Z^m(\roots,\NN_0^q)$ & The kernel of $\mathsf{d}^m:C^m(\roots,\NN_0^q)\rightarrow C^{m+1}(\roots,\NN_0^q)$, the set of\\& $m$-cocycles or closed cochains\\
$B^m(\roots,\NN_0^q)$ & The image of $\mathsf{d}^{m-1}:C^{m-1}(\roots,\NN_0^q)\rightarrow C^{m}(\roots,\NN_0^q)$, the set of\\& $m$-coboundaries or exact cochains\\
$\AL{\roots}$& The Lie algebra over a polynomial ring associated to a $2$-cocycle $\omega^2$\\& on the root system $\Phi$, Definition \ref{def:lie algebra associated to a cocycle}
\end{longtable}
\end{center}

\cleardoublepage\phantomsection
\addcontentsline{toc}{chapter}{Bibliography}
\bibliography{BackMatter/bibliography}

\def\cprime{$'$}
\begin{thebibliography}{10}

\bibitem{babelon2003introduction}
Olivier Babelon, Denis Bernard, and Michel Talon.
\newblock {\em Introduction to {C}lassical {I}ntegrable {S}ystems}.
\newblock Cambridge University Press, 2003.

\bibitem{beschesmallgroups}
H.~U. Besche, B.~Eick, and E.~O'Brien.
\newblock Small{G}roups-library of all ``small'' groups, {GAP} package,
  {V}ersion included in {GAP} 4.4. 12, {T}he {GAP} {G}roup, 2002.

\bibitem{bourbaki1998lie}
Nicolas Bourbaki.
\newblock {\em Lie {G}roups and {L}ie {A}lgebras. {C}hapters 1--3}.
\newblock Springer-Verlag, Berlin, 1998.
\newblock Translated from the French, Reprint of the 1989 English translation.

\bibitem{bury2010automorphic}
Rhys Bury.
\newblock {\em Automorphic {L}ie {A}lgebras, {C}orresponding {I}ntegrable
  {S}ystems and their {S}oliton {S}olutions}.
\newblock {PhD} in {A}pplied {M}athematics, The University of Leeds, School of
  Mathematics, Department of Applied Mathematics, 2010.

\bibitem{chopp2011lie}
Mika\"el Chopp.
\newblock {\em Lie-admissible structures on {W}itt type algebras and
  automorphic algebras}.
\newblock {PhD} in {M}athematics, Universit\'e du Luxembourg and Universit\'e
  Paul Verlaine-Metz, 2011.

\bibitem{curtis1999pioneers}
Charles~W. Curtis.
\newblock {\em Pioneers of {R}epresentation {T}heory: {F}robenius, {B}urnside,
  {S}chur, and {B}rauer}, volume~15.
\newblock American Mathematical Soc., 1999.

\bibitem{dolgachev2009mckay}
I.~Dolgachev.
\newblock Mc{K}ay correspondence. {W}inter 2006/07.
\newblock {\em Lecture notes}, 2009.

\bibitem{eisenbud1995commutative}
David Eisenbud et~al.
\newblock {\em Commutative {A}lgebra with a {V}iew {T}oward {A}lgebraic
  {G}eometry}, volume~27.
\newblock Springer New York, 1995.

\bibitem{etinghof}
P.~Etinghof, O.~Golberg, S.~Hensel, T.~Liu, A.~Schwendner, D.~Vaintrob, and
  E.~Yudovina.
\newblock {\em Introduction to {R}epresentation {T}heory}, volume~59 of {\em
  Student Mathematical Library}.
\newblock AMS, 2011.

\bibitem{ford1951automorphic}
Lester~R. Ford.
\newblock {\em Automorphic {F}unctions}, volume~85.
\newblock American Mathematical Soc., 1951.

\bibitem{fossum}
R.~M. Fossum.
\newblock Invariant theory, representation theory, commutative algebra -
  m\'enage \`a trois.
\newblock {\em Lecture Notes in Mathematics 867, Springer-Verlag}, pages 1--37,
  1981.

\bibitem{fulton1991representation}
William Fulton and Joe Harris.
\newblock {\em Representation {T}heory: {A} {F}irst {C}ourse}, volume 129.
\newblock Springer, 1991.

\bibitem{GAP}
The GAP~Group.
\newblock {\em {{G}{A}{P} -- {G}roups, {A}lgorithms, and {P}rogramming, Version
  4.4.12}}, 2008.

\bibitem{MR2363237}
Gert-Martin Greuel and Gerhard Pfister.
\newblock {\em A {\bf {S}ingular} {I}ntroduction to {C}ommutative {A}lgebra}.
\newblock Springer, Berlin, extended edition, 2008.
\newblock With contributions by Olaf Bachmann, Christoph Lossen and Hans
  Sch{\"o}nemann, With 1 CD-ROM (Windows, Macintosh and UNIX).

\bibitem{harrison1962commutative}
D.~K. Harrison.
\newblock Commutative algebras and cohomology.
\newblock {\em Transactions of the American Mathematical Society}, pages
  191--204, 1962.

\bibitem{hochster1971cohen}
Melvin Hochster and John~A. Eagon.
\newblock Cohen-{M}acaulay rings, invariant theory, and the generic perfection
  of determinantal loci.
\newblock {\em American Journal of Mathematics}, pages 1020--1058, 1971.

\bibitem{humphreys1972introduction}
James~E. Humphreys.
\newblock {\em Introduction to {L}ie {A}lgebras and {R}epresentation {T}heory},
  volume~9.
\newblock Springer, 1972.

\bibitem{humphreys1996course}
John~F. Humphreys.
\newblock {\em A {C}ourse in {G}roup {T}heory}, volume~6.
\newblock Oxford University Press, 1996.

\bibitem{jacobson1979lie}
Nathan Jacobson.
\newblock {\em Lie {A}lgebras}.
\newblock Number~10. Courier Dover Publications, 1979.

\bibitem{kac1969automorphisms}
Victor~G Kac.
\newblock Automorphisms of finite order of semisimple {L}ie algebras.
\newblock {\em Funktsional'nyi Analiz i ego prilozheniya}, 3(3):94--96, 1969.

\bibitem{kac1994infinite}
Victor~G Kac.
\newblock {\em Infinite-dimensional {L}ie algebras}, volume~44.
\newblock Cambridge university press, 1994.

\bibitem{Klein56}
Felix Klein.
\newblock {\em Lectures on the {I}cosahedron and the {S}olution of {E}quations
  of the {F}ifth {D}egree}.
\newblock Dover Publications Inc., New York, revised edition, 1956.
\newblock Translated into English by George Gavin Morrice.

\bibitem{Klein93}
Felix Klein.
\newblock {\em Vorlesungen \"uber das {I}kosaeder und die {A}ufl\"osung der
  {G}leichungen vom {F}\"unften {G}rade}.
\newblock Birkh\"auser Verlag, Basel, 1993.
\newblock Reprint of the 1884 original, Edited, with an introduction and
  commentary by Peter Slodowy.

\bibitem{knapp2002lie}
Anthony~W. Knapp.
\newblock {\em Lie {G}roups: {B}eyond an {I}ntroduction}, volume 140.
\newblock Springer, 2002.

\bibitem{knibbeler2014higher}
V.~Knibbeler, S.~Lombardo, and J.~A. Sanders.
\newblock Higher dimensional {A}utomorphic {L}ie {A}lgebras.
\newblock {\em In preparation}.

\bibitem{knibbeler2014automorphic}
V.~Knibbeler, S.~Lombardo, and J.~A. Sanders.
\newblock Automorphic {L}ie algebras with dihedral symmetry.
\newblock {\em J. Phys. A: Math. Theor.}, 47(36):365201, 2014.

\bibitem{Form00}
Jan Kuipers, Takahiro Ueda, Jos A.~M. Vermaseren, and Jens Vollinga.
\newblock {FORM} version 4.0.
\newblock {\em Computer Physics Communications}, 184(5):1453--1467, 2013.

\bibitem{LM04iop}
S.~Lombardo and A.~V. Mikhailov.
\newblock Reductions of integrable equations: dihedral group.
\newblock {\em J. Phys. A}, 37(31):7727--7742, 2004.

\bibitem{LM05comm}
S.~Lombardo and A.~V. Mikhailov.
\newblock Reduction {G}roups and {A}utomorphic {L}ie {A}lgebras.
\newblock {\em Comm. Math. Phys.}, 258(1):179--202, 2005.

\bibitem{LS10}
S.~Lombardo and J.~A. Sanders.
\newblock {O}n the {C}lassification of {A}utomorphic {L}ie {A}lgebras.
\newblock {\em Comm. Math. Phys.}, 29(3):793, 2010.

\bibitem{Lombardo}
Sara Lombardo.
\newblock {\em Reductions of {I}ntegrable {E}quations and {A}utomorphic {L}ie
  {A}lgebras}.
\newblock {PhD} in {A}pplied {M}athematics, The University of Leeds, School of
  Mathematics, Department of Applied Mathematics, 2004.

\bibitem{lusztig2003homomorphisms}
G.~Lusztig.
\newblock Homomorphisms of the alternating group $\mathcal{A}_5$ into reductive
  groups.
\newblock {\em Journal of Algebra}, 260(1):298--322, 2003.

\bibitem{mikhailov1979integrability}
A.~V. Mikhailov.
\newblock Integrability of a two-dimensional generalization of the {T}oda
  chain.
\newblock {\em JETP Lett}, 30(7):414--418, 1979.

\bibitem{mikhailov1980reduction}
A.~V. Mikhailov.
\newblock Reduction in integrable systems. {T}he reduction group.
\newblock {\em JETP Lett}, 32(2):187--192, 1980.

\bibitem{Mikhailov81}
A.~V. Mikhailov.
\newblock The reduction problem and the inverse scattering method.
\newblock {\em Physica D}, 3(1\&2):73--117, 1981.

\bibitem{MPW2014}
A.~V. Mikhailov, G.~Papamikos, and J.~P. Wang.
\newblock Darboux transformation with {D}ihedral reduction group.
\newblock {\em arXiv preprint}, arXiv:1402.5660, 2014.

\bibitem{MR2500567}
Iku Nakamura.
\newblock Mc{K}ay correspondence.
\newblock In {\em Groups and symmetries}, volume~47 of {\em CRM Proc. Lecture
  Notes}, pages 267--298. Amer. Math. Soc., Providence, RI, 2009.

\bibitem{neusel2007invariant}
Mara~D. Neusel.
\newblock {\em Invariant {T}heory}, volume~36 of {\em Student Mathematical
  Library}.
\newblock AMS, 2007.

\bibitem{olver2006multivariate}
Peter~J. Olver.
\newblock On multivariate interpolation.
\newblock {\em Studies in Applied Mathematics}, 116(2):201--240, 2006.

\bibitem{paule2008algorithms}
Peter Paule and Bernd Sturmfels.
\newblock {\em Algorithms in {I}nvariant {T}heory}.
\newblock Springer, 2008.

\bibitem{read1976schur}
E.~W. Read.
\newblock On the {S}chur multipliers of the finite imprimitive unitary
  reflection groups ${G} (m, p, n)$.
\newblock {\em Journal of the London Mathematical Society}, 2(1):150--154,
  1976.

\bibitem{schlichenmaier2003higher}
Martin Schlichenmaier.
\newblock Higher genus affine algebras of {K}richever-{N}ovikov type.
\newblock {\em Moscow Math. J}, 3(4):1395--1427, 2003.

\bibitem{schur1904darstellung}
J.~Schur.
\newblock {\"U}ber die {D}arstellung der endlichen {G}ruppen durch gebrochen
  lineare {S}ubstitutionen.
\newblock {\em J. Reine Angew. Math.}, 127:20--50, 1904.

\bibitem{schur1911darstellung}
J.~Schur.
\newblock {\"U}ber die {D}arstellung der symmetrischen und der alternierenden
  {G}ruppe durch gebrochene lineare {S}ubstitutionen.
\newblock {\em J. Reine Angew. Math.}, 139:155--250, 1911.

\bibitem{serre1977linear}
Jean~Pierre Serre and Leonard~L. Scott.
\newblock {\em Linear {R}epresentations of {F}inite {G}roups}, volume~42.
\newblock Springer, 1977.

\bibitem{MR1328644}
Larry Smith.
\newblock {\em Polynomial {I}nvariants of {F}inite {G}roups}, volume~6 of {\em
  Research Notes in Mathematics}.
\newblock A K Peters Ltd., Wellesley, MA, 1995.

\bibitem{springer1987poincare}
T.~A. Springer.
\newblock Poincar{\'e} {S}eries of {B}inary {P}olyhedral {G}roups and
  {M}c{K}ay's {C}orrespondence.
\newblock {\em Mathematische Annalen}, 278(1):99--116, 1987.

\bibitem{stanley1979invariants}
Richard~P. Stanley.
\newblock Invariants of finite groups and their applications to combinatorics.
\newblock {\em Bulletin of the American Mathematical Society}, 1(3):475--511,
  1979.

\bibitem{sage}
W.\thinspace{}A. Stein et~al.
\newblock {\em {S}age {M}athematics {S}oftware ({V}ersion 6.3)}.
\newblock The Sage Development Team, 2014.
\newblock {\tt http://www.sagemath.org}.

\bibitem{suter2007quantum}
Ruedi Suter.
\newblock Quantum affine {C}artan matrices, {P}oincar{\'e} series of binary
  polyhedral groups, and reflection representations.
\newblock {\em Manuscripta Mathematica}, 122(1):1--21, 2007.

\bibitem{tits1966constantes}
Jacques Tits.
\newblock Sur les constantes de structure et le th{\'e}oreme d'existence des
  algebres de {L}ie semi-simples.
\newblock {\em Publications Math{\'e}matiques de l'IH{\'E}S}, 31(1):21--58,
  1966.

\bibitem{toth2002finite}
Gabor Toth.
\newblock {\em Finite {M}{\"o}bius {G}roups, {M}inimal {I}mmersions of
  {S}pheres, and {M}oduli}.
\newblock Springer, 2002.

\bibitem{toth2002glimpses}
Gabor Toth.
\newblock {\em Glimpses of {A}lgebra and {G}eometry}.
\newblock Springer, 2002.

\bibitem{vasil2014harmonic}
Alexander Vasil'ev.
\newblock {\em Harmonic and {C}omplex {A}nalysis and {I}ts {A}pplications}.
\newblock Springer, 2014.

\bibitem{zakharov1974scheme}
V.~E. Zakharov and A.~B. Shabat.
\newblock {A} scheme for integrating the nonlinear equations of mathematical
  physics by the method of the inverse scattering problem. {I}.
\newblock {\em Functional analysis and its applications}, 8(3):226--235, 1974.

\end{thebibliography}

\end{document}